%% file: main-arxiv.tex
\documentclass[acmsmall,nonacm]{acmart}
\AtBeginDocument{%
  }

\setcopyright{acmcopyright}
\copyrightyear{2023}
\acmYear{2023}
\acmDOI{XXXXXXX.XXXXXXX}

\acmConference[PODC 2023]{ACM Symposium on Principles of Distributed Computing}{June 19--23,
  2023}{Orlando, FL}
\acmPrice{15.00}
\acmISBN{978-1-4503-XXXX-X/18/06}

\usepackage{graphicx}

\usepackage{subfigure}
\usepackage{epsfig}
\usepackage{wrapfig}
\usepackage{hyperref}
\usepackage{amsmath}

\usepackage[ruled,vlined,linesnumbered,titlenotnumbered,noend]{algorithm2e}
\usepackage{tikz}
\usepackage{tikz-network}
\usepackage{multirow}
\input{mathsprograms.sty}
\usepackage[ruled,vlined,linesnumbered,titlenotnumbered,noend]{algorithm2e}
\input{macros}

\input{comm}

\begin{document}


\title{Optimal Stateless Model Checking of  Transactional Programs under Causal Consistency}         

\author{Parosh Aziz Abdulla}
\affiliation{
  \institution{Uppsala University}    
  \country{Sweden}        
}
\author{Mohamed Faouzi Atig}
\affiliation{
  \institution{Uppsala University}            
\country{Sweden}
}
\author{Ashutosh Gupta}
\affiliation{
  \institution{IIT Bombay}            
\country{India}
}
\author{Shankaranarayanan Krishna}
\affiliation{
  \institution{IIT Bombay}            
\country{India}
}
\author{Omkar Tuppe}
\affiliation{
  \institution{IIT Bombay}            
\country{India}
}

\begin{abstract}
 We present a framework for 
efficient stateless model checking (SMC) of concurrent programs   under five prominent models of causal consistency, $\ccvt, \cm, \cc, \readat, \readc$.  Our approach is based on exploring traces under the program order $\tpo$ and the reads from $\trf$ relations. 
Our SMC algorithm is provably optimal in the sense that it explores each $\tpo$ and $\trf$ relation exactly once. We have implemented our framework in a tool called \ourtool{}. Experiments show that \ourtool{} performs 
well in detecting anamolies in  classical distributed databases benchmarks. 
\end{abstract}

\maketitle

\input{intro-trans.tex}

\input{prelims.tex}

\input{saturated.tex}

\input{dpor.tex}
\input{experiments.tex}

\input{conclusion}

 \bibliographystyle{ACM-Reference-Format}
\bibliography{references}

\newpage
\appendix
\input{ccv-proofs}

\input{cc-proofs}

\input{cm-proofs.tex}

\input{ra-proofs}
\input{rc-proofs}

\input{appendix-expts}

\end{document}

%% file: macros.tex
\def\plog#1{$\mathtt{#1}$}

\newcommand{\ourtool}{\textsc{Tranchecker}}
\newcommand{\nidhugg}{\textsc{Nidhugg}}
\newcommand{\wt}{\mathsf{wt}}
\newcommand{\rd}{\mathsf{rd}}

\newcommand{\vars}{\mathcal{V}}

\definecolor{mygreen}{rgb}{0.05, 0.5, 0.06}

\newcommand{\tpo}{\textcolor{red}{po}}
\newcommand{\trf}{\textcolor{blue}{\it{rf}}}
\newcommand{\ow}{\textcolor{mygreen}{\mathsf{OW}}}
\newcommand{\tco}{\textcolor{orange}{co}}
\newcommand{\tto}{\textcolor{mygreen}{to}}

\newcommand{\co}{\mathsf{CO}}
\newcommand{\init}{\mathsf{init}}
\newcommand{\cfr}{\textcolor{violet}{\mathsf{CF}}}

\newcommand{\cp}{\mathsf{Causal Past}}

\newcommand{\cccyc}{\mathsf{CCcycle}}
\newcommand{\ccvcyc}{\mathsf{CCvcycle}}
\newcommand{\readccyc}{\mathsf{RComcycle}}
\newcommand{\readatcyc}{\mathsf{RAtcycle}}
\newcommand{\cmcyc}{\mathsf{CMcycle}}
\newcommand{\xcyc}{\mathsf{Xcycle}}
\newcommand{\extend}{\mathsf{extend}}

%% file: comm.tex
\definecolor{myorange}{rgb}{0.93, 0.49, 0.1}
\definecolor{myred}{rgb}{0.82, 0.1, 0.26}
\definecolor{myblue}{rgb}{0.01, 0.28, 1.0}
\definecolor{myviolet}{rgb}{0.6, 0.4, 0.8}
\definecolor{mygray}{rgb}{0.9, 0.89, 0.89}
\definecolor{mypurple}{rgb}{0.41, 0.16, 0.38}

\newcommand\bjcomcom[1]{}

\newcommand\nat{{\mathbb N}}
\newcommand\alphabet\Sigma
\newcommand{\tuple}[1]{\left\langle#1\right\rangle}
\newcommand{\setcomp}[2]{\left\{{#1}\mid {#2}\right\}}

\newcommand\app\bullet

\newcommand\emptyword\epsilon

\newcommand\ii{i}
\newcommand\jj{j}

\newcommand\nn{n}

\newcommand\xvar{x}
\newcommand\yvar{y}

\newcommand\initof[1]{\init_{#1}}
\newcommand\initx{\initof\xvar}

\newcommand\nmodels\nvDash
\newcommand\restrict[2]{#1\raise-.5ex\hbox{\ensuremath|}_{#2}}

\newcommand\rf{{\color{mygreen}{\tt rf}}}

\newcommand\prog{{\mathcal P}}

\newcommand\conf\gamma

\newcommand\runs{{\it Runs}}

\newcommand\threadset{{\mathsf{Procs}}}

\newcommand\thread{{\it proc}}

\newcommand\areg{{\tt a}}
\newcommand\breg{{\tt b}}
\newcommand\creg{{\tt c}}

\newcommand\run\rho
\newcommand\thrun\pi
\newcommand\pth\pi
\newcommand\varset{{\mathbb X}}

\newcommand\rtype{{\mathsf{rd}}}
\newcommand\trace\tau
\newcommand\otrace\sigma
\newcommand\emptytrace{\trace_\emptyset}

\newcommand\tlub\sqcup
\newcommand\tglb\sqcap
\newcommand\tequiv\sim
\newcommand\ttequiv\equiv
\newcommand\ctordering\sqSubset
\newcommand\tordering\sqsubseteq
\newcommand\stordering\sqsubset
\newcommand\eventset{{\tt E}}

\newcommand\inittranset{\tran_{\it init}}

\newcommand\swap{\texttt{Swappable}}
\newcommand\swapof[1]{\swap\!\left(#1\right)}

\newcommand\typeof[1]{{#1}.{\it type}}
\newcommand\tranof[1]{{#1}.{\it tran}}

\newcommand\varof[1]{{#1}.{\it var}}

\newcommand\id{{\it id}}

\newcommand\tracetuple{\tuple{\tran,\tpo,\trf}}
\newcommand\tracetuplep{\tuple{\tran',\tpo',\trf'}}

\newcommand\zeroval{0}
\newcommand\valset{{\mathbb V}}

\newcommand\add\odot

\newcommand\lbl\ell
\newcommand\action\lbl
\newcommand\silentlbl\varepsilon
\newcommand\expr{e}

\newcommand\movesto[1]{\xrightarrow{#1}{}}

\newcommand\succof[1]{{\tt succ}\!\left(#1\right)}

\newcommand\lmovesto[1]{\stackrel{#1}\leadsto}
\newcommand\assigned\leftarrow

\newcommand\rschedule{\textsc{RunSchedule}}
\newcommand\explore{\textsc{ExploreTraces}}
\newcommand\declarepostponed{\textsc{CreateSchedule}}

\newcommand\scheduledof[1]{\texttt{Schedules}\!\left(#1\right)}

\newcommand\schedule\beta

\newcommand\setname[1]{{\mathcal A}}

\newcommand\true{{\tt true}}
\newcommand\false{{\tt false}}

\newcommand\obs\alpha
\newcommand\obsseq\pi
\newcommand\obsseqsub\preceq
\newcommand\obsseqminus\ominus

\crefformat{section}{\S#2#1#3} 
\crefformat{subsection}{\S#2#1#3}
\crefformat{subsubsection}{\S#2#1#3}
\crefformat{subsubsection}{\S#2#1#3}
\crefformat{appendix}{Appendix~#2#1#3}
\crefformat{subappendix}{Appendix~#2#1#3}
\crefformat{subsubappendix}{Appendix~#2#1#3}

\definecolor{mGreen}{rgb}{0,0.6,0}
\definecolor{mGray}{rgb}{0.5,0.5,0.5}
\definecolor{mPurple}{rgb}{0.58,0,0.82}
\definecolor{backgroundColour}{rgb}{0.95,0.95,0.92}

\lstdefinestyle{CStyle}{
    backgroundcolor=\color{white},   
    commentstyle=\color{mGreen},
    keywordstyle=\color{magenta},
    numberstyle=\tiny\color{mGray},
    stringstyle=\color{mPurple},
    basicstyle=\linespread{0.6}\footnotesize,
    breakatwhitespace=false,         
    breaklines=true,                 
    captionpos=b,                    
    keepspaces=true,                 
    numbers=left,                    
    numbersep=5pt,                  
    showspaces=false,                
    showstringspaces=false,
    showtabs=false,                  
    tabsize=2,
    language=C
}

\SetAlgoNoLine
\SetInd{0.8em}{0.8em}
\SetNoFillComment

\SetCommentSty{mycommfont}
\DontPrintSemicolon

%% file: intro-trans.tex
\section{Introduction}
Transactions are a proven abstraction mechanism in database systems 
for constructing reusable parallel computations. 
 A transactional program is a concurrent program, which has  \emph{transactions} in the code of its processes. Transactional programming  \cite{DBLP:series/synthesis/2010Harris} attempts to adapt the powerful transaction processing technology originally developed for database concurrency, to 
the setting of general purpose concurrent programs. Using transactions as the main enabling construct 
for shared memory concurrency obliviates the programmer 
from  using low level mechanisms like locks, mutexes and semaphores to prevent two concurrent threads from interfering.

 A \emph{transaction} is a 
code segment consisting of several instructions which must be executed \emph{atomically}. Each process in a transactional program is a sequence of transactions, 
 written as $\mathtt{t1}[\dots]; \mathtt{t2}[\dots]; \dots ;\mathtt{tn}[\dots]$, where $\mathtt{t}[\dots]$ represents a 
transaction, with the brackets [ and ] demarcating the beginning and end.  The general idea is that each process $p$ executes its transactions following the program order, and, when a process $p$ finishes executing a transaction $\mathtt{t}[\dots]$, it ``delivers'' the writes/data updates made in 
$\mathtt{t}[\dots]$ to all processes. The manner in which the  receiving processes make use of this delivery is dependent  on the underlying consistency model of transactional programs.

The strongest consistency model is the one ensuring serializability \cite{DBLP:journals/jacm/Papadimitriou79b}, where every execution of a program is equivalent  
 to another one where the transactions are executed serially one after the other, with no interference. In the non-transactional setting, 
 this model corresponds to sequential consistency (SC) \cite{DBLP:journals/tc/Lamport79}. While serializability and 
 SC are easy to understand for programmers,  they are too strong, since they 
  need global synchronization between all processes, which makes  it difficult to achieve good performance guarantees 
  \cite{DBLP:journals/jacm/FischerLP85,DBLP:journals/sigact/GilbertL02}.  Thus, modern databases 
 ensure weaker consistency guarantees. 
 Causal consistency \cite{DBLP:journals/cacm/Lamport78}, is one of the  fundamental models implemented in many leading production databases such as 
 AntidoteDB, CockroachDB and MongoDB. Unlike serializability, causal consistency is a weaker concurrency notion 
 and  allows conflicting transactions having reads/writes to a common variable  to be executed in different orders by different processes as long as they are not causally related. The sets of updates visible to different processes may differ and read actions may return values that cannot be obtained in SC executions.    There are many variations of causal consistency introduced in the literature, some of the prominent ones being weak causal consistency (\wcct{})~\cite{10.1145/3016078.2851170,10.1145/3093333.3009888}, causal memory (\scct{})~\cite{10.1007/BF01784241,10.1145/3016078.2851170}, causal convergence (\ccvt{})~\cite{DBLP:journals/ftpl/Burckhardt14}, read atomicity ($\readat$)\cite{DBLP:conf/concur/Cerone0G15}, \cite{10.1145/3360591} as well as read committed ($\readc$) \cite{10.1145/3360591}, \cite{10.1145/223784.223785}.

 In a recent survey 
 \cite{DBLP:conf/sigmod/Pavlo17} of database administrators, more than 85\%  of the participants responded that almost all transactions in their database execute at these weak levels of concurrency.   A weaker concurrency notion allows more behaviours than stronger notions. The onus of ensuring that a database application 
  can tolerate this larger set of behaviours is on the developers. These weaker models, while being deployed in most distributed databases, are hard to reason about \cite{DBLP:conf/popl/BrutschyD0V17}. The  resulting application bugs can cause huge losses \cite{DBLP:conf/sigmod/WarszawskiB17}. Thus, there is a tradeoff between good performance as in the weak models versus correctness as in serializability. One can easily reason about database applications under serializability; however doing this for the weak models is challenging since they have exponentially more behaviours. Given the 
  widespread use of weak models in modern distributed databases, we need    verification techniques which ensure  that applications are not only 
  efficient but also correct.

\noindent{\bf{Verification of Database Applications}}. 	This paper focuses on verifying the correctness of concurrent programs under causal consistency against assertion violations. Model checking  \cite{DBLP:conf/popl/ClarkeES83} is the most prominent technique for algorithmic verification of programs which explores systematically, all executions of a program. The limiting factor in the applicability of model checking is state space explosion; the number of executions grows exponentially in the number of processes. 

Stateless model checking (SMC), one of the most  
successful techniques for finding concurrency bugs \cite{10.1145/263699.263717}  is useful to counter this state space explosion. The name ``stateless" model checking comes from the fact that only a small fraction of the visited states are stored when  exploring executions. While this  avoids excessive memory consumption to some extent, we still need to cope up with the large amount of non-determinism which gives rise to exponentially many interleavings.  
 To address this, SMC is often combined with Partial Order reduction (POR).  POR \cite{DBLP:journals/sttt/ClarkeGMP99,DBLP:books/sp/Godefroid96,DBLP:conf/cav/Peled93}
is a technique that limits the number of executions explored without compromising on the coverage of program behaviour.  We achieve this 
by avoiding the analysis of equivalent executions. 
In POR, two executions are considered equivalent if one can be obtained from the other by swapping consecutive independent execution steps. In POR, all the executions are partitioned among equivalence classes, and at least one execution from each representative equivalence class is explored.  Dynamic partial order reduction (DPOR), which combines 
SMC and POR, 
 explores 
the executions by computing the equivalence between executions on-the-fly.

DPOR was first developed for concurrent programs under SC \cite{10.1145/2578855.2535845,10.5555/1763218.1763234}. 
Recent years have seen  DPOR adapted to language induced weak memory models  \cite{DBLP:journals/pacmpl/Kokologiannakis18,10.1145/2806886,10.1145/3276505}, as well as hardware-induced relaxed memory models \cite{10.1007/978-3-662-46681-0_28,10.1145/2737924.2737956}. 
 Under sequential consistency,  the equivalence classes are called \emph{Mazurkiewicz} traces  \cite{10.5555/25542.25553}, while for relaxed memory models, the generalization of these are called \emph{Shasha-Snir} traces \cite{DBLP:journals/toplas/ShashaS88}.   A Shasha-Snir trace characterizes an execution of a concurrent program by the relations (1) $\tpo$ program order, which totally orders events of each process, (2),  $\trf$ reads from, which connects each read with the write it reads from, (3) $\tco$ coherence order, which totally orders writes to the same shared variable. DPOR can be optimized further by observing that the assertions to be verified at the end of an execution does not depend on the coherence order of shared variables, and hence it suffices to consider traces over $\tpo$-$\trf$. Based on this observation,  the DPOR algorithms for programs under the release-acquire semantics (RA) and SC  
   \cite{10.1145/3276505}, \cite{10.1145/3360576} explores traces with $\tpo, \trf$ and $\tco$ where the $\tco$ edges are added on the fly. 
  The equivalence classes are considered  wrt $\tpo$-$\trf$, reducing 
  the number of distinct traces to be analyzed.

  Since a program may behave entirely differently when run on different consistency models, verification frameworks such as DPOR need to be resigned from one model to the next.
In particular, it is not possible to migrate existing DPOR algorithms for models such as SC and RA 
to causal consistency.
The crucial challenge in the design of such an algorithm is (i) to compute the set the write events from which 
read events can fetch their values, and (ii) calculate the causalities that such a read operation implies.
These operations are intricately dependent on the consistency model.
For instance, they have an exponential cost in SC, but only a polynomial cost in RA \cite{10.1145/3360576,10.1145/3276505}.
This paper aims to develop DPOR algorithms for transactional programs under causal consistency. 
To that end, we face two obstacles, namely that (i) we are dealing with  new consistency models (compared to existing DPOR algorithms), 
and (ii)  we consider transactions (rather than simple transitions).
Concretely, we 
determine efficiently, on-the-fly, which transaction to read from whenever we encounter a read
instruction. Whether a transaction is ``readable'' depends on the individual model under consideration, which is different for all of them. Second, each time we read from one of the readable transactions, specific causal
dependencies are created amongst the others. These must be resolved so that the resultant execution remains
consistent concerning the model. 
We instantiate our framework for the five variants 
$\cc, \ccvt, \cm, \readc, \readat$. Despite this heterogeneity, we propose a uniform DPOR algorithm across all five models, which computes the readable set of transactions and resolves causal dependencies in polynomial time. 

\smallskip 

 \noindent{\bf{Contributions}}. We propose a DPOR based SMC algorithm for the  models $X {\in} \{\cc, \ccvt, \cm, \readc, \readat\}$ which explores systematically, all the distinct $\tpo$-$\trf$ traces covering all possible executions of the program. 
    We develop a uniform  algorithm for all models which is sound and complete : that is, all traces explored are consistent wrt the model $X$ under consideration, and all such consistent  traces are explored. 
        Moreover, our algorithm is \emph{optimal} in the sense that, each consistent $\tpo$-$\trf$ trace is explored  
   exactly once. One of the key challenges during the trace exploration is to maintain the consistency of the traces wrt the model under consideration. We tackle this by defining a \emph{trace semantics}
   which ensures that the traces generated in each step only contain edges which will be present in any consistent trace.  We implement our algorithms    in a tool \ourtool{} which is, to the best of our knowledge, the first of its kind to perform SMC on transactional programs wrt these five prominent causal consistency models. \ourtool{} checks for assertion violation of programs under $\cc, \ccvt, \cm, \readc, \readat$.      

  We generated 7507 litmus tests using  Herd \cite{10.1145/2627752}
  and evaluate \ourtool{} on them. \ourtool{} takes around 
  570s to finish running all 7507 programs. 
 Then we proceed with experimental  evaluation on a wide range of benchmarks  from distributed databases.  
We showed that  (i) $\ourtool$ correctly detects 
known consistency bugs  \cite{bernardi2016robustness}
, \cite{10.1145/3360591},  \cite{lmcs:7149} and 
\cite{beillahi2021checking} 
under $\ccvt, \cm, \cc, \readat, \readc$, 
(ii) $\ourtool$ correctly detects known assertion 
violations in applications  \cite{10.14778/2732240.2732246}, \cite{10.1145/2987550.2987559}, \cite{lmcs:7149}, \cite{10.1007/978-3-030-44914-8_20}. 
 We also did a stress test of \ourtool{} on some  
 parameterized benchmarks which   
 resulted in a large number (7.7 million) of traces. 

\smallskip

\noindent{\bf{Related Work}}.
We compare our contribution with two lines of work. On one hand, 
we review work  on the  verification of causal 
consistency models, and on the other hand, on 
SMC algorithms.

There has been a spate of recent work on causal consistency models 
\cite{10.1007/BF01784241}. The question of verifying automatically, whether a given transactional program is robust against causal consistency has been explored in many papers, 
\cite{bernardi2016robustness}, \cite{DBLP:conf/popl/BrutschyD0V17}, \cite{DBLP:conf/pldi/BrutschyD0V18}, 
 \cite{DBLP:journals/jacm/CeroneG18}, 
\cite{DBLP:conf/concur/NagarJ18} via under and over approximation 
techniques. Checking robustness amounts to asking whether the behaviours  obtained while executing over a causally consistent database are serializable. The latest in this 
direction providing a precise robustness algorithm is  \cite{lmcs:7149}, which showed that 
  verifying robustness against causal consistency reduces to the reachability problem under SC. As a means to analyse the robustness problem for $\cc$,  \cite{cerone_et_al:LIPIcs:2017:7794}  provide an algebraic connection between the operational and axiomatic semantics of $\cc$. In the non-transactional setting, \cite{DBLP:journals/siglog/Lahav19} 
  surveys WRA (analogue of $\cc$), SRA (analogue of $\ccvt$) and discusses the operational and axiomatic semantics of both models. \cite{DBLP:conf/vmcai/RaadLV19} proposes a lock-based 
  reference implementation for \emph{snapshot isolation} (SI), a weak consistency model used in databases such as Oracle and MS SQL server, while prior work \cite{DBLP:conf/esop/RaadLV18} did the same for \emph{parallel snapshot isolation} (PSI), a model closely related to   snapshot isolation. \cite{DBLP:conf/ecoop/XiongCRG19} identify a   consistency model \emph{weak snapshot isolation} (WSI) that sits in between SI and PSI and propose an operational semantics for client observable behaviours of atomic transactions. 
        \cite{DBLP:conf/pldi/ChongSW18} considers the 
    interplay between weak memory and transactional memory \cite{DBLP:conf/isca/HerlihyM93} by extending existing axiomatic weak memory models (x86, Power, ARMv8, and C++) with new rules for transactional memory. They also synthesize tests for validating these models on existing implementations. 

      In contrast to our paper, 
   none of the earlier works provide algorithms that can 
   systematically explore all executions of a concurrent program 
   under causal consistency and certify them safe against concurrency bugs. 
 
 \begin{figure}[t]
 	\includegraphics[width=15cm]{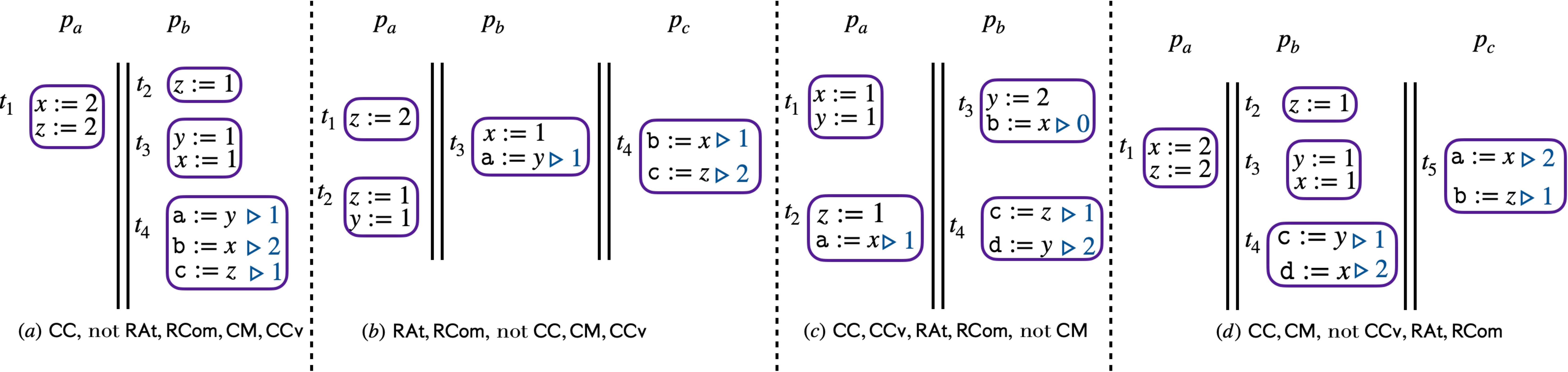}
 \caption{Differences between consistency models. 
 The \textcolor{blue!80!green}{$\triangleright v$} denotes the expected 
 return value of the read event.  }
 \label{fig:bad}
 \end{figure}
 
 \section{Overview}
In this section, we give an overview of some 
 well-known variants of causal consistency  
  as well as the working of our DPOR algorithm. Later, we give the denotational semantics  needed for our DPOR algorithm (section \ref{sec:tr}),
  two efficient operations to explore traces (section \ref{sec:saturated}), 
    the DPOR algorithm (section \ref{sec:dpor}) and our experimental results 
    (section \ref{sec:experiments}).

  Two key notions needed across the  models are 
  that of \emph{causal delivery} and 
  \emph{transaction isolation}. 
  We discuss these before going into 
  specific models. 
   A process is a sequence of transactions, and a transaction is a sequence 
  of read/write instructions. A transaction is executed atomically,
   and no interleaving happens within a transaction.  
    A \emph{transaction log} is a sequence $(x_1,v_1) \dots, (x_n, v_n)$ 
    obtained from the writes $x_i:=v_i$  made in order in the transaction. 
    Here, $x_i$ is a shared variable and $v_i$ is a value. 
     The log is  reset (to $\epsilon$) before a transaction begins, and every write $x:=v$ in the transaction appends $(x,v)$ to the log.
Thus, the log of transaction $t_1$ in the program Figure \ref{fig:bad}(a) is $(x,2)(z,2)$.  When a transaction finishes, its log is available to be delivered to all processes. This is done respecting causal delivery :  if the log of a transaction $t$ is received by a process before it starts executing a transaction $t'$, then the log of $t$ is received by all processes before receiving the log of $t'$. Also, if $t, t'$ are two transactions in program order in a process $p$, with $t'$ coming later than  $t$, then the log of $t$  is always (self)delivered to $p$ before the execution of $t'$. Hence, the log of $t$ will be delivered to any process before delivering the log of $t'$. Transaction isolation says that while a process is executing a transaction, it will not receive any delivery. Thus, deliveries are received only when a process has finished executing a transaction and not started executing the next transaction. An execution under any of the models respect causal delivery and transaction isolation. Each process has a \emph{local store}, keeping track of the latest values delivered to it, per variable. In a program having shared variables $x_1, \dots, x_n$, the local store of all processes is $\{(x_1,0), \dots, (x_n,0)\}$ initially. 0 is the initial value of all variables. 

When a process $p$ executes a read $r=x$ in its transaction $t$, it first checks  $t$ made a write to $x$. If yes, it reads the latest value written to $x$ in $t$. Otherwise, it takes the value of $x$ from its local store.
\smallskip 

\noindent{\bf Overview of $\cm$}.   Under $\cm$, whenever a process receives a delivery, it updates its local store using the latest writes per variable from the log received. For instance, if the 
log received is $(x,1)(z,2)(x,2)$, and the local store is $\{(x,3),(y,0),(z,1)\}$, then the store is updated to $\{(x,2),(y,0),(z,2)\}$.

 The program in Figure \ref{fig:bad}(d) has a $\cm$ consistent execution : 
 First $p_b$ executes $t_2, t_3$ delivering $t_2, t_3$ to itself in order. Then $p_a$ executes $t_1$, delivers the log to $p_b$. The local store of $p_b$ then is $\{(x,2),(y,1),(z,2)\}$. Then $t_4$ is executed, and $p_b$ reads off ${\tt{c}}$ as 1 and ${\tt{d}}$ as 2 from its local store.  This is followed by delivering to $p_c$, $t_1, t_2$ in order before executing $t_5$ so that the local store of $p_c$ is $\{(x,2),(y,0),(z,1)\}$. Then $p_c$ can read 2 into ${\tt{a}}$ and 1 into ${\tt{b}}$. $t_4$ can be delivered to $p_c$ at the end.

\smallskip

\noindent{\bf Overview of $\ccvt$}.  In $\ccvt$, transactions are totally ordered by unique ids ($\in \mathbb{N}$)  associated to 
transactions. Transactions in program order of a process 
have monotonically increasing ids.  Each process $p$  maintains per variable $x$,   the timestamp of $x$, $\timestamp(x)$. 
 $\timestamp(x)$ is the largest transaction id among all  transaction logs  delivered to $p$, which has a write to $x$. On receiving the delivery of a log from a transaction $t$,  the value of $x$ in the local store of process $p$ is updated 
 with the latest value of $x$ from the log 
 only if $\timestamp(x)$ in $p$ is smaller than $t.id$. In this case, $\timestamp(x)$ is also updated to $t.id$.

 Consider the program in Figure \ref{fig:bad}(a). This is not consistent wrt $\ccvt$. We know that $t_2.id < t_3.id$. To read 2 into ${\tt{b}}$, we need $t_1.id > t_3.id > t_2.id$. However, to read 1 into ${\tt{c}}$, we need $t_2.id > t_1.id$.  
\smallskip 

\noindent{\bf Overview of $\cc$}.  In $\cc$, transaction ids are vector clocks whose dimension equals the number of processes. Transactions in program order of a process have increasing ids. 
 The local store 
of each process is a set of triples of the form $(x, v, t.id)$, where 
$t.id$ is the transaction id of the transaction which wrote $v$ to $x$.  
For each process $p$, on receiving delivery of the transaction log of transaction $t$,  (1) $(x,v,t.id)$ is added to the local store of $p$ if $t.id$ is incomparable with
 all the transaction ids of existing triples 
 in the store of $p$, (2) $(x,v,t.id)$ replaces 
 all $(x, v', t'.id)$ triples in the store of $p$ if $t.id > t'.id$, 
  and (3) $(x,v)$ is the latest write to $x$ in the log of $t$.  
   Each time a new transaction begins execution in a process $p$, it maintains a \emph{snapshot} of its local store per variable.  It assigns to each variable $x$,  the latest value $snapshot(x)$ obtained from a linearization of the store wrt $x$. 
   A read instruction $r=x$ in a transaction $t$ first checks 
   if $t$ has written to $x$, and reads the latest such write if any.
   If $t$ has not written to $x$,  it reads $snapshot(x)$.  

The program in Figure \ref{fig:bad}(a) is $\cc$ consistent. 
Choose $t_1.id$ incomparable with $t_2.id, t_3.id$. Let $t_1.id=(0,5)$, 
$t_2.id=(1,2), t_3.id=(1,3)$. $p_b$'s local store after $t_2, t_3$  is $\{(z,1,t_2.id),(y,1,t_3,id), 
(x,1,t_3.id)\}$. On delivering the log of $t_1$ to $p_b$,  $(x,2,t_1.id), (z,2,t_1.id)$ are added to the local store of $p_b$. 
 When $p_b$ starts $t_4$,  it maintains 
 $snapshot(x)=2$ (linearizing $\{(x,2,t_1.id), (x,1,t_3.id)\}$ 
 as $t_3t_1$),    $snapshot(z)=1$ (linearizing $\{(z,2,t_1.id), (z,1,t_2.id)\}$  as $t_1t_2$),   and   $snapshot(y)=1$. 
 This enables the reads in $t_4$.  
\smallskip

\noindent{\bf Overview of $\readc, \readat$}.  Like $\ccvt$, in $\readc, \readat$ also, transactions are totally ordered by unique ids ($\in \mathbb{N}$).  
In $\readc$, for transactions $t_1, t_2$ writing to some variable $x$, if the latest writes on $x$ by $t_1, t_2$  
are read by events $r_1, r_2$ appearing in a transaction $t$ such that $r_2$ appears later than $r_1$ in $t$, 
then $t_1.id < t_2.id$.  Any $\readc$ consistent execution orders transactions monotonically wrt respective 
monotonicity of reads.  
 The program Fig.\ref{fig:bad}(a) is not $\readc$ consistent.   Consider the read events $e_1,e_2$  as ${\tt{a}}:=y$ and  ${\tt{b}}:=x$ in $t_4$ reading respectively from $t_3, t_1$.  Then we have $t_3.id < t_1.id$. Likewise,  
  $t_4$ has the read event $e_3$ given by ${\tt{c}}:=z$ reading from $t_2$. 
   Then $t_1.id < t_2.id$ which is not possible. $\readat$ is a strengthening of $\readc$ wrt the ordering of transactions. 
     For all variables $x$ and all $t_1 \neq t_2 \in \tran^{\wt,x}$ 
   if a transaction $t_3$ in some process $p$ reads from the latest write  of $t_1$ on $x$, then any transaction  $t_2$ which is $\tpo$-before $t_3$ in $p$, or which reads from $t_3$    
   must be such that 
  $t_2.id < t_1.id$.  For the program Fig.\ref{fig:bad}(a), we have
     $t_1, t_3 \in \tran^{\wt,x}$ and $t_4$ reads from $t_1$. Since $t_3$ is $\tpo$-before $t_4$, we have $t_3.id < t_1.id$. $t_4$ reads from $t_2$ with $t_1, t_2 \in \tran^{\wt,z}$ and since 
     $t_4$ reads from $t_1$ (for reading into ${\tt{b}}$), we have  $t_1.id < t_2.id$. This gives $t_3.id < < t_1.id < t_2.id$ which is not possible.

\begin{figure}[t]
 	\includegraphics[width=10cm]{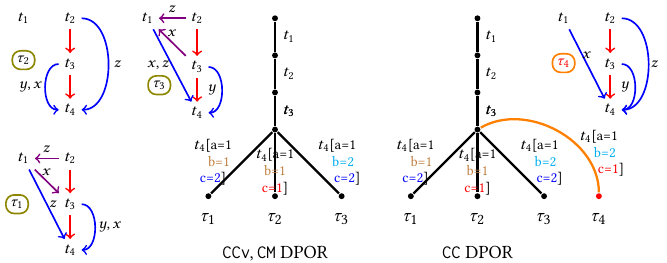}
 	\caption{The DPOR run tree on program Fig \ref{fig:bad}(a) wrt $\cc, \cm, \ccvt$}
 	 	\label{dpor-over}
 \end{figure}
 \smallskip 
 
\noindent{\bf{Overview of DPOR}}. 
We give a quick overview of the key elements in 
our DPOR algorithm.  
Consider the program in Figure \ref{fig:bad}(a) 
and the assertion $\varphi : {\tt{a}}=1 \wedge {\tt{b}}=2 \wedge {\tt{c}}=1$. 
The DPOR algorithm non-deterministically chooses to execute $t_1$ or $t_2$ since 
the program can start with $t_1$ or $t_2$.  Assume that the DPOR algorithm non-deterministically
executes $t_1, t_2, t_3$ 
in order, and finally comes to $t_4$. This choice does not affect the final outcome.  The algorithm will not have a redundant  exploration where it begins with $t_2$. 
      The reason is that once we fix a $\trf$ for the reads in $t_4$, all the three executions $t_1 \rightarrow t_2 \rightarrow  t_3 \rightarrow  t_4$, 
   $t_2 \rightarrow t_1 \rightarrow  t_3 \rightarrow  t_4$, 
$t_2 \rightarrow t_3 \rightarrow  t_1 \rightarrow  t_4$ 
result in the same $\tpo$-$\trf$ trace and  
are  equivalent. The optimality of our algorithm 
 ensures that we do not explore any trace more than once. Let us first consider the $\ccvt$ DPOR.

\noindent $[\bf{\textcolor{red}{(1)}}]$ The first read event of $t_4$ 
has only one source
 to read from, namely the write $y:=1$ in $t_3$. 
  
 \noindent $[\bf{\textcolor{red}{(2)}}]$ The next read event of $t_4$ is on ${\tt{b}}$; 
 it has two possible sources to read from : $t_1$ and 
 $t_3$. The DPOR detects the 
  \emph{readable set} of transactions at this point as  
 $\{t_1, t_3\}$, and proceeds with 
two branches corresponding to these choices. 

\noindent $[\bf{\textcolor{red}{(3)}}]$ Consider 
the branch   where ${\tt{b}}$  reads from the 
write $x:=1$ of $t_3$. 
 It then proceeds with the next read in $t_4$ on ${\tt{c}}$. It creates two branches corresponding to the two sources to read from, namely, $t_1, t_2$. 
  If it reads from $t_2$, we obtain trace $\tau_2$.

 \noindent $[\bf{\textcolor{red}{(4)}}]$
 Since it did not witness $\varphi$, the algorithm backtracks to the first point where it had the alternate 
source to read for ${\tt{c}}$, namely, $t_1$. Here, it detects  
two causal dependencies created by reading ${\tt{c}}$ from $t_1$. 
First it detects a causality between $t_2, t_1$ since $t_2$ also writes on $z$.  To resolve this, it orders $t_2$ before $t_1$. Second, 
it also detects a causality between $t_1$ and $t_3$ 
since both write to $x$, and both can reach $t_4$.  
 Since it reads $x$ from $t_3$, it orders $t_1$ 
before $t_3$. 
The obtained trace is $\tau_1$. The purple edges (see Figure \ref{dpor-over}) between $t_3, t_2$ and $t_1$ represent the orderings 
  made by the DPOR to   resolve the causal dependencies 
 after the read of    ${\tt{c}}$. The blue edges labeled with variables $x,y,z$ represents  reading from the writes of $x,y,z$.  The red edges represent the program order.

\noindent$[\bf{\textcolor{red}{(5)}}]$
Once again, since it did not witness $\varphi$, it backtracks to the first point where it had the alternate 
source to read for ${\tt{b}}$, namely from the write 
$x:=2$ in $t_1$.  Here it detects a causality between $t_3$ and $t_1$  
since both write on $x$, and both can reach $t_4$. To resolve this, it orders $t_3$ before $t_1$.  It then proceeds with the next read in $t_4$ on ${\tt{c}}$. Here, it detects that ${\tt{c}}$ cannot read from $t_2$ since  $t_3$ is ordered before $t_1$.  
 It therefore reads from the write $z:=2$ of $t_1$. After the read, it orders $t_2$ before $t_1$.   
  The obtained trace  
  is $\tau_3$.  
 At the end of this branch, it returns having finished all possible exlporations  and concludes that 
$\varphi$ cannot be witnessed 
in any $\ccvt$ consistent execution. 

However, if we run the algorithm under $\cc$, 
then we get an extra trace $\tau_4$. This is obtained since the algorithm 
allows in step $(\bf{5})$ above, to read ${\tt{c}}$ 
from $t_2$. 

Under $\cc$, the algorithm detects that even though 
it ordered $t_3$ before $t_1$ wrt $x$ (as in step $(\bf{5})$ above), it can now 
order $t_1$ before $t_2$ since it is dealing with a read on a variable $z \neq x$.  This allows it to witness the assertion 
$\varphi$. Finally, if we consider $\cm$, once again, we do not witness 
$\varphi$. The algorithm in this case produces three traces 
which are similar to those in $\ccvt$.

%% file: prelims.tex
\section{Preliminaries}
\noindent\textcolor{red}{\emph{Programs}}.  
We consider transactional programs $\prog$ consisting  of a finite set of \emph{threads} or  \emph{processes}  $p \in \threadset$ that share a finite set $\varset$  of shared variables 
ranging over 
a domain $\valset$ of {\em values} that includes a special
value $\zeroval$.
  A process  has a finite set
of local registers that store values from $\valset$.
Each process runs a deterministic code, built in a standard way
from expressions and atomic commands, using standard control
flow constructs (sequential composition, selection, and bounded loop constructs). %
Throughout the paper, we use $\xvar, \yvar$ for shared variables,
$\areg, \breg, \creg$ for registers, and $\expr$ for expressions.
Expressions do not contain any shared variables. 

Each process is an ordered sequence of \emph{transactions}.  
 A transaction is an ordered  sequence of \emph{labeled instructions}. 
 Let $\mathcal{T}$ denote the set of all transactions in $\prog$. Transactions are denoted by $t_1, t_2, t_3, \dots$. 
A transaction $t \in \mathcal{T}$ in process $p$ starts with a \plog{begin}$(p,t)$ instruction and ends with an \plog{end}$(p,t)$ instruction (in the figures, we simply enclose the code of a  transaction in a box without explicit \plog{begin} and \plog{end} statements). 
 The instructions in a transaction are write/read/assign instructions apart from standard conditional flow constructs. 
  Write instructions $x:=e$ write to the shared variable $x$ while read instructions $r:=x$ read from the shared variables $x$. Assignments have the form $r:=e$ involving local registers and expressions. 
 We assume that the control cannot pass from one transaction to another in a process without going through the \plog{begin} and \plog{end} instructions.  

 Each transaction has a transaction identifier $t.id$ (which is 
  in $\Nats$ or $\Nats^{|\threadset|}$ depending on the model $X \in \{\cc, \ccvt, \cm, \readc, \readat\}$ we consider).    
 If process $p$ is the ordered sequence $t_1, t_2, \dots, t_n$ of transactions, then, $t_i.id < t_j.id$ for $i < j$. 
 
 The local state of a process $\thread\in\threadset$ 
is defined  by its program counter (which also makes it clear which transaction $t$ the control resides in each process), the contents
of its registers and shared variables. 
A {\it configuration} of $\prog$
is made up of the local states of all the processes.
A program execution is a sequence of transitions between configurations, starting with the initial configuration $\conf^{\init}$.  
Each transition corresponds to one process performing an instruction in one of its transactions. 
A transition between two configurations $\conf$ and $\conf'$ is of form
$\conf\lmovesto{\lbl}\conf'$,
where  $\lbl \in \{\beginact(p,t), \commitact(p,t), x:=v, {\tt{a}}:=x\}$.  
Since assignments involve only registers, they are not visible to other
processes, we will not represent them explicitly 
in the transition relation. Instead,
we let each transition represent the combined effect of some 
finite sequence of assignments  in a transaction followed
by a read/write/\plog{end} statement in  the same transaction.

The labels represent starting and ending of transaction $t$ in process $p$,  executing a read/write instruction in a transaction. Between labels  
$\beginact(p,t), \commitact(p,t)$, we only execute instructions of transaction $t$. We follow a more succinct representation 
for a run by replacing the contiguous sequence of transitions $\sigma_i \xrightarrow[]{{\tt{begin}}(p,t)} \sigma_{i+1} \dots \xrightarrow[]{{\tt{end}}(p,t)} \sigma_j$ corresponding to the execution of a transaction with $\sigma_i \xrightarrow[]{{\tt{issue}}(p,t)} \sigma_j$. 
 \plog{issue}$(p,t)$ represents  \emph{issuing} the transaction $t$ in  process $p$, and signifies executing all instructions of $t$ in order, 
starting with \plog{begin}$(p,t)$ and ending with \plog{end}$(p,t)$.

The values which can be read on a read instruction appearing in a transaction depend on the consistency model under consideration. This is taken care of by associating runs with so-called {\em traces} which we define in section \ref{sec:tr}. Traces
 tell us the values  reads can obtain  from available  
writes. A  causal consistency model $X \in \{\readat, \readc,
\cc, \ccvt, \cm\}$ 
is formulated by imposing 
restrictions on traces, thereby also restricting the possible
runs that are associated with them.

 A configuration $\conf$ is said to be {\it terminal} if 
$\succof\conf=\emptyset$, i.e., no process can issue any transaction from
$\conf$.
A {\it run} $\run$ from $\conf$ is a sequence
$\conf_0\movesto{\lbl_1}\conf_1
\movesto{\lbl_2}\cdots\movesto{\lbl_\nn}
\conf_\nn$ 
such that $\conf_0=\conf$.
We say that $\run$ is {\it terminated} if $\conf_\nn$ is terminal.
We let $\runs(\conf)$ denote the set of runs
from $\conf$.

\smallskip 

\noindent{\emph{\textcolor{red}{Events}}}. 
An event corresponds to a particular execution
of a statement in a run of $\prog$.
A {\it write event} $\event$ in a transaction $t$ is given by $(id,\thread,t,\wt(x,\val))$ where 
$id \in \nat$ is the identifier of the event, $\thread$ is the process containing the event in transaction $t$,
$\xvar\in\varset$ is a variable, and
$\val\in\valset$ is a value.
This 
corresponds to a write event 
happening in  transaction $t$ of process $\thread$ 
writing the value $\val$ to variable $\xvar$.
Likewise, a {\it read event} $\event$ is given by 
$(id, \thread, t, \rd(x))$ where  $\xvar \in \varset$. 
The read event $\event$ does not specify the particular 
value it reads; 
this value will be defined in a trace by specifying a write event 
from which $\event$ fetches its value.

For each variable $\xvar\in\varset$, we assume
a special transaction containing only (apart from \plog{begin,end}) the write event $\initx=\wt(x,0)$
called the {\it initializer} transaction $t_{\initx}$ for $\xvar$.
This transaction is not part of any of the processes in $\threadset$, and
writes the value $\zeroval$ to $\xvar$.
We define $\inittranset:=\setcomp{t_{\initx}}{\xvar\in\varset}$ as
the set of initializer transactions.   
Let $\tran^{\wt,x}, \tran^{\rd,x}$ respectively denote the set of transactions having a write and read instruction on $x$.  
 If $\eventset$ is a set of events, we define
subsets of $\eventset$ characterized
by particular attributes of its events.
For instance, for a  variable $\xvar$, we let $\eventset^{\rtype,\xvar,t}$
denote $\setcomp{\event \in \eventset}
{ \typeof\event=\rtype \land \varof{\event}=\xvar \land \tranof{\event}=t}$, for all read events in transaction $t$ on variable $x$. 

\smallskip 

\noindent{\emph{\textcolor{red}{Traces}}}.
A {\it trace}     
$\trace$ is a tuple $\tracetuple$,
where $\tran$ is the set of all {\it transactions}
including the set $\inittranset$ of initializer transactions,  $\tpo$ (program order),
$\trf$ (read-from) are binary relations on $\tran$ that satisfy:

$\bullet$ $t ~\tpo~ t'$ if $\threadof{t} = \threadof{t'}$ and   $t.id < t'.id$.  The relation $\tpo$
  totally orders the transactions of each
  individual process. We assume that 
$t_{\initx}~\tpo~t$ for all transactions $t \in \tran$ for all variables $x$.

$\bullet$ $ t~\trf~t'$  if  $t \in \tran^{\wt,x}$, $t' \in \tran^{\rd,x}$, and the value read is the value written by the last write $w$ in $t$.

We can view $\trace = \tracetuple$ as a graph whose nodes
are the transactions $\tran$ and whose edges are defined by the 
relations $\tpo$, $\trf$. $\tpo$ is depicted by red solid edges and captures the order of transactions in each process while $\trf$ edges are depicted as solid blue edges.

We define the {\it empty trace}
$\emptytrace:=\tuple{\inittranset,\emptyset,\emptyset}$,
containing only the initializer transactions,
and all relations empty.

We define when a trace can be associated with a run.
Consider a run $\run$ 
$\conf_0\movesto{\lbl_1}\cdots\movesto{\lbl_\nn}\conf_\nn$,

and let $\trace=\tracetuple$ be a trace.
We write $\run\models\trace$ to denote that
the following conditions are satisfied:

(i) $\tran=\set{t_1,\ldots,t_\nn}$,
i.e.,
each transaction $t_i \in \tran$ corresponds exactly to one label $\issueact(p,t_i)$ in $\run$.

(ii) For $\ell_i=\issueact(p, t_i), \ell_j=\issueact(p,t_j)$, 
$i < j$ iff $t_i.id < t_j.id$.  

(iii)
   $t_\ii\;\trf\;t_\jj$ iff $t_i \in \tran^{\wt,x}, t_j \in \tran^{\rd,x}$ 
   for some variable $x$, and $\ell_k=\issueact(p_i, t_i)$,  
   $\ell_h=\issueact(p_j, t_j)$ for $k < h$, and the values written and read in $t_i, t_j$ are the same.  

(iv)
  if $t_{\initx}~\trf~t_i$, then $t_i \in \tran^{\rd,x}$
  and 
  the value read is 0.

\section{Causally Consistent Models}
\label{sec:tr}
Our DPOR algorithm relies on the \emph{declarative} definition  \cite{10.1145/3009837.3009888}, \cite{10.1145/3360591} of causal consistency models. Under the declarative semantics, the program is associated with a set of \emph{traces}. 
Each trace summarizes a particular program execution and describes all accesses to the shared variables and the relations between them in that run.  To define the five models \cite{10.1145/3360591}, \cite{10.1145/3009837.3009888} formally, we introduce a function that, for each model, extends a given trace uniquely by a set of new edges. Then we define the model by requiring that the extended trace does not contain any cycles. A run of the program satisfies a consistency model 
$X \in \{\readat, \readc, \cc, \ccvt, \cm\}$ 
 if its associated extended trace has no cycles.

 Let $\co$, called \emph{causality order} represent $(\tpo \cup \trf)^+$.  Two transactions $t_1, t_2$ are \emph{causally related} if either $t_1~\co~t_2$  or $t_2~\co~t_1$.

\begin{figure}[t]
 	\includegraphics[width=9cm]{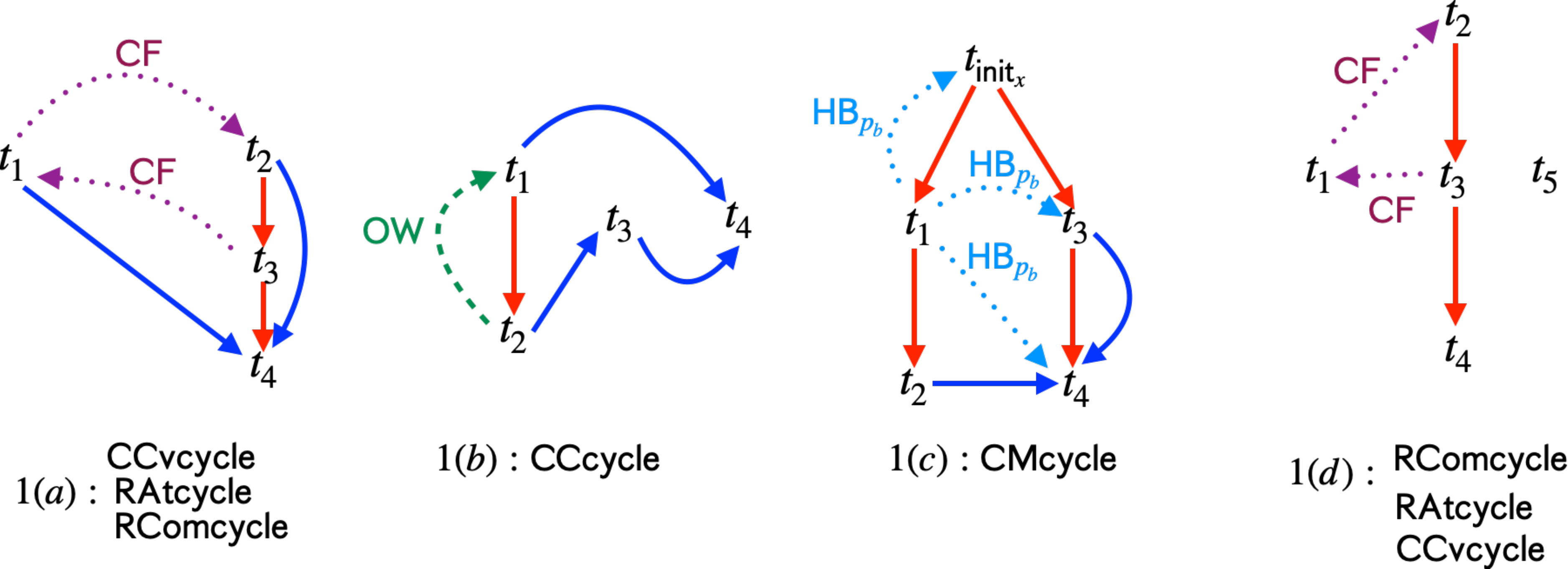}
 	\caption{solid red, blue edges are $\tpo$,$\trf$.}
 	\label{traces-badd}
 \end{figure}
 
 \smallskip 
\noindent{\bf{Weak Causal Consistency $\cc$}}.
We start presenting one of the weakest notions of causal consistency,  $\cc$ \cite{jad-thesis}, \cite{10.1007/BF01784241}. 
In $\cc$, transactions writing to a variable which are not causally related can be ordered differently in different transactions in a process. If $t_1, t_2 \in \tran^{\wt,x}$ are causally unrelated, for $t, t' \in \tran^{\rd,x}$, $t$ can read from the latest write of $t_1$ or $t_2$ (choose one); likewise $t'$ can independently  read from one of $t_1, t_2$.  

To illustrate, consider the program Fig.\ref{fig:bad}(a). 
 The transaction $t_1$ is not causally related to transactions $t_2, t_3$. 
 Hence, the read from $x$ in $t_4$ can read from either of $t_1, t_3 \in \tran^{\wt,x}$; likewise, the read from $z$ in $t_4$ can also choose to read from either of $t_1, t_2 \in \tran^{\wt,z}$.

      Finally, if in $p_b$ has transaction $t_5$ (after $t_4$) having  read ${\tt{d}}:=x$,  ${\tt{d}}$ can read 1 by choosing to read  from $t_3$ (and not $t_1$).

A trace $\tau$ does not violate $\cc$ as long as there is a causality order which explains the return value of each read event. 
To capture traces violating $\cc$, we define a relation $\ow$ between transactions ($\ow$ represents  overwrite) which write to the same variable.   For  transactions $t_1, t_2 \in \tran^{\wt,x}$
and $t_3 \in \tran^{\rd,x}$,  if 	$t_1~\co~t_2~\co~t_3$, and $t_1~\trf~t_3$, then   
$t_2~\ow~t_1$. This says that $t_3$ reads an earlier  write by $t_1$, resulting in a $\co~\cup~\ow$ cycle. We refer to $\co\cup\ow$ cycles as $\cccyc$. 

We define a function
$\extend_{\cc}(\tau)$ which extends a trace 
 $\tau=\tracetuple$ by adding all possible $\ow$ edges   
 between transactions which write on the same variable.  

For a trace $\tau=\tracetuple$,
 $\tau\models \cc$ iff $\extend_{\cc}(\tau)$ 
does not have  
 $\cccyc$.

\begin{figure}[t]
\includegraphics[width=13cm]{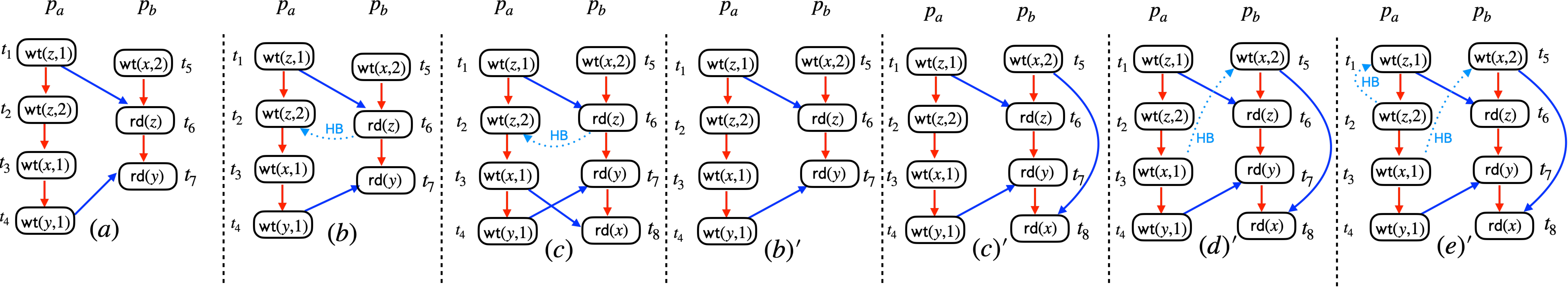}	
\caption{Start with (a). In (b) we add the $\hb{}$ edge from $t_6$ to $t_2$ following condition (ii). 
 (c) is obtained on adding  $t_8$ and  $t_3~\trf~t_8$.  In contrast, (b)' does not follow condition (ii). Hence, when $t_8$ is added in (c)', 
 $t_5$ is available to be read (since we do not have the path from $t_5$ to $t_3$, making $t_5$ an overwritten write on $x$).  
   Choosing $t_5~\trf~t_8$ necessitates adding $t_3~\hb{}~t_5$ in (d)' by condition (i).
  This necessitates adding 
 $t_2~\hb{}~t_1$ in (e)' creating $\cmcyc$. 
 }
\label{fig:hb}
\end{figure}

\smallskip 

 \noindent{\emph{\textcolor{red}{Examples}}}. 
 Program Fig. \ref{fig:bad}(b) is not $\cc$ since there is no causality 
 order which explains the return values of the read events. If we consider any trace (Fig. \ref{traces-badd}) 
 of the program Fig.\ref{fig:bad}(b), we find that $t_1~\trf~t_4$ (for reading 2 into ${\tt{c}}$), $t_2~\trf~t_3$ (for reading 1 into 
 ${\tt{a}}$), and $t_3~\trf~t_4$(for reading 1 into ${\tt{b}}$). 
 This induces $t_1~\co~t_2~\co~t_4$ ($t_1~\tpo~t_2~\trf~t_3~\trf~t_4$). 
 Since $t_1, t_2 \in \tran^{\wt,z}$, $t_1~\co~t_2~\co~t_4$  with 
   $t_1~\trf~t_4$ creates $t_2~\ow~t_1$ witnessing $\cccyc$. 
  
\smallskip 
 
\noindent{\bf{Causal Convergence $\ccvt$}}.
Under $\ccvt$, we need a total order on all transactions which write to the same variable.  This order, called \emph{arbitration order}, is an abstraction of how conflicts are resolved by all processes to agree upon one  ordering among transactions which are not causally related. 
Thus, unlike $\cc$,
different transactions in a process cannot independently choose to read from different causally unrelated transactions.  
 To enforce a total order between transactions writing to the same variable, we use a new relation $\cfr$ called conflict relation on transactions writing to the same variable.  
  For all variables $x \in \vars$,  transactions $t_1,t_2 \in \tran^{\wt,x}$ and $t_3 \in \tran^{\rd,x}$, 
   if  $t_1 ~\co~t_3$,  and $t_2~\trf~t_3$   then $t_1 ~\cfr~t_2$. 
  
  We define a function
$\extend_{\ccvt}(\tau)$ which extends a trace 
 $\tau=\tracetuple$ by adding all possible  $\ow, \cfr$ edges   
 between transactions writing on the same variable. 
     Traces violating $\ccvt$ exhibit a 
 $\co~\cup~\cfr~\cup~\ow$ cycle in $\extend_{\ccvt}(\tau)$, which we refer to as $\ccvcyc$.
   We say that $\tau \models \ccvt$ iff $
 \extend_{\ccvt}(\tau)$ does not contain a $\ccvcyc$. 

\smallskip

  \noindent{\emph{\textcolor{red}{Examples}}}.  For the program Fig.\ref{fig:bad}(a) and any trace $\tau$,   
  $\extend_{\ccvt}(\tau)$  
  has a $\ccvcyc$ (see Fig.\ref{traces-badd})  
    since in any trace, we have $t_3~\co~t_4$, $t_1~\trf~t_4$ (this $\trf$ is unavoidable to read 2 into ${\tt{b}}$) with 
    $t_1, t_3 \in \tran^{\wt,x}$ giving $t_3~\cfr~t_1$. 
    Now we have $t_1~\co~t_4$ with $t_1, t_2 \in \tran^{\wt,z}$ and $t_2~\trf~t_4$ (this $\trf$ is needed to read 1 into ${\tt{c}}$). 
    This gives  $t_1~\cfr~t_2$  resulting in $t_1~\cfr~t_2~\tpo~t_3~\cfr~t_1$. 
    Intuitively,   we cannot find a total order amongst the transactions 
  $t_1, t_3 \in \tran^{\wt,x}$. Likewise,  Fig.\ref{fig:bad}(d) exhibits $\ccvcyc$.

  \smallskip 
  
 \noindent{\bf{Causal Memory $\cm$}}. 
 The $\cm$ model is stronger than $\cc$ and incomparable to $\ccvt$. 
 In $\cm$, a process can diverge from another one in its ordering of transactions which are not causally related. However, once a process chooses an ordering of such transactions,  all reading transactions in it adhere to it; this makes it stronger than $\cc$ and incomparable to $\ccvt$. 
 
 A \emph{happened before} relation per process fixes the per process ordering of transactions.    
  For a transaction $t$ in a trace, the \emph{Causal Past} of $t$, $\cp(t)=\{t' \in \tran \mid t' ~\co~t\}$ is the set of transactions  
which are in the causal past of $t$.

 For a transaction $t$, the happened before relation $\hb{t}$  
 is the smallest relation on transactions which is transitive, and is such that for all transactions $t_1, t_2 \in$  $\cp(t)$,  $t_1~\co~t_2 \Rightarrow t_1~\hb{t}~ t_2$. In other words, $\co_{\mid\cp(t)} \subseteq \hb{t}$ : $\hb{t}$ contains all pairs of transactions obtained by restricting $\co$ to the transactions in the causal past of $t$.   For any variable $x$,  if we have transactions 
 $t_3, t_4 \in \tran^{\wt,x}$, $t_5 \in \tran^{\rd,x}$ 
 such that 
   (i) $t_5{=}t$  or $t_5~\tpo~t$,  $t_4~\trf~t_5$, and   $t_3~\hb{t}~t_5$, then $t_3~\hb{t}~ t_4$, and 
     (ii) if $t_3~ \hb{t}~ t_4	$, $t_3 ~\trf~t_5$, then 
     $t_5~\hb{t}~t_4$. 
  
 Let $t_p$ be the $\tpo$-last transaction of process $p$: that is,  for all transactions $t$ in process $p$, $t=t_p$ or $t~\tpo~t_p$. Since  $\hb{t} \subseteq \hb{t_p}$ for all transactions $t$ in process $p$, $\hb{t_p}$  fixes the ordering among all causally unrelated transactions for process $p$. We write $\hb{p}$ instead of $\hb{t_p}$.

 We define a function $\extend_{\cm}$ which extends a trace $\tau=\tracetuple$ by adding all possible $\ow, \hb{p}$ edges for all processes $p$. Traces violating $\cm$ exhibit 
 a $\ow \cup \hb{p}$ cycle, called a $\cmcyc$
   in $\extend_{\cm}(\tau)$ for some process $p$. 
 We say that $\tau \models \cm$ iff $\extend_{\cm}(\tau)$ does not contain a $\cmcyc$. Figure \ref{fig:hb}  motivates conditions (i), (ii) to add $\hb{}$ edges so that 
  $\extend_{\cm}(\tau)$ does not contain $\cmcyc$.  
 
 \smallskip 
 
 \noindent{\emph{\textcolor{red}{Examples}}}. 
 For the program Fig.\ref{fig:bad}(a) and any trace $\tau$, $\extend_{\cm}(\tau)$ contains $\cmcyc$. Consider 
 transaction $t_4$ in $p_b$. We have $t_1~\trf~t_4$, and 
 $t_3~\hb{p_b}~t_4$ with $t_1, t_3 \in \tran^{\wt,x}$, obtaining 
 $t_3~\hb{p_b}~t_1$.  Now, we have $t_1~\co~t_4$, that is, 
 $t_1~\hb{p_b}~t_4$ and $t_2~\trf~t_4$, with $t_1, t_2 \in \tran^{\wt,z}$
 obtaining $t_1~\hb{p_b}~t_2$. This gives $t_3~\hb{p_b}~t_1	~\hb{p_b}~t_2~\hb{p_b}~t_3$ resulting in $\cmcyc$. 	Likewise, in the program Fig.\ref{fig:bad}(b), we have $t_2~\hb{p_c}~t_4$ 
 and  $t_1~\trf~t_4$ with $t_1, t_2 \in \tran^{\wt,z}$ resulting in 
 $t_2~\hb{p_c}~t_1$ and therefore the $\cmcyc$ $t_1~\tpo~t_2~\hb{p_c}~t_1$. 
  For the program in Fig.\ref{fig:bad}(c), we have  
 $t_2~\trf~t_4$ inducing $t_1~\hb{p_b}~t_4$. Now, we have $t_3~\trf~t_4$ 
 and $t_1, t_3 \in \tran^{\wt,y}$, obtaining $t_1~\hb{p_b}~t_3$. 
 We also have $t_{\initx}~\trf~t_3$. Along with  
 $t_1~\hb{p_b}~t_3$ and $t_1, t_{\initx} \in \tran^{\wt,x}$ gives 
 $t_1~\hb{p_b}~t_{\initx}$. Then we have the $t_{\initx}~\tpo~t_1~\hb{p_b}~t_{\initx}$ obtaining $\cmcyc$. 	
 
 \smallskip

  \noindent{\bf{Read Committed $\readc$}}. 
  $\readc$ is one of the weakest models for causal consistency \cite{DBLP:conf/sigmod/BerensonBGMOO95}.   Like $\ccvt$, the consistency condition is defined by extending traces with the $\cfr$ relation between transactions writing to the same variable. 
 For transactions $t_1, t_2 \in \tran^{\wt,x}$ for some variable $x$, 
 $t_1~\cfr~t_2$  is defined by ensuring that the  
 corresponding read events  appearing in a transaction $t\in \tran^{\rd,x}$ are \emph{monotonic}. To capture  
   monotonicity between events in a transaction, we use the \emph{transaction order} $\tto$ which orders events within a transaction. If $e_1, e_2, \dots, e_k$ is the ordered sequence  of events in a transaction $t$, define $e_i~\tto~e_j$ for $i < j$. 
   The intuition is that for $e_1, e_2 \in \eventset^{\rtype,t}$ with  $e_1~\tto~e_2$, if $e_1 \in \eventset^{\rtype,\xvar,t}$ reads from a transaction $t_1$, and $e_2 \in \eventset^{\rtype,\yvar, t}$ reads from a transaction $t_2$, and  $t_1, t_2 \in \tran^{\wt,x}$, then 
   $t_2$ must be ''$\cfr$-later'' than $t_1$, that is, $t_1~\cfr~t_2$. 
 That is, for all variables $x$, and all transactions $t_1 \neq t_2 \in \tran^{\wt,x}$, if there exist read events $e_1~\tto~e_2$ in $t \in \tran$ such that  $t_1~\trf~e_1, t_2~\trf~e_2$, then $t_1~\cfr~t_2$.

Define a function $\extend_{\readc}(\tau)$ which extends a trace $\tau=(\tran, \tpo, \trf)$ by 
  adding all possible $\cfr$ edges between transactions  which write on the same variable. 
  In $\extend_{\readc}(\tau)$, we do not draw the $\trf$ edges 
  from transactions $t \in \tran^{\wt,x}$ to read events 
 $e \in \eventset^{\rtype,\xvar, t'}$; these are used to draw the $\cfr$ edges between transactions.     
 Traces violating $\readc$ exhibit a $\tpo \cup \cfr$ cycle in $\extend_{\readc}(\tau)$, which we refer to as a $\readccyc$. We say that $\tau \models \readc$ iff $\extend_{\readc}(\tau)$ does not contain a $\readccyc$. 
\smallskip 

  \noindent{\emph{Examples}}. For the program Fig.\ref{fig:bad}(a), and any trace $\tau$, $\extend_{\readc}(\tau)$ contains $\readccyc$. 
  Consider the read events $e_1={\tt{a}}:=y$, $e_2={\tt{b}}:=x$ in $t_4$, 
   such that $e_1~\tto~e_2$ and $t_3~\trf~e_1$, $t_1~\trf~e_2$, with $t_1, t_3 \in \tran^{\wt,x}$. Then we have $t_3~\cfr~t_1$. Likewise, 
   $t_4$ has the read event $e_3={\tt{c}}:=z$ with $e_2~\tto~e_3$. 
   Since $t_2~\trf~e_3$, we obtain $t_1~\cfr~t_2$ resulting in a 
   cycle. Likewise, for the  program Fig.\ref{fig:bad}(d), we obtain $t_1~\cfr~t_2$ (reads in $t_5$) as well as $t_3~\cfr~t_1$ (reads in $t_4$).  

 \smallskip 

  \noindent{\bf{Read Atomic $\readat$}}. This model \cite{DBLP:conf/concur/Cerone0G15} is a strengthening 
  of $\readc$ in how we add $\cfr$ edges between transactions writing  
  to the same variable.  While $\readc$ enforced a $\cfr$-ordering 
  between transactions $t_1, t_2$ writing to the same variable 
  based on the $\tto$ ordering of the corresponding read events, this 
  is strengthened in $\readat$ as follows.   
   For a transaction $t_1 \in \tran^{\wt,x}$ which is read by a read event in a transaction $t_2 \in \tran^{\rd,x}$, any transaction $t' \in \tran^{\wt,x}$  for which 
      $t' ~\trf~t_2$ or $t'~\tpo~t_2$ must be such that  $t'~\cfr~t_1$. 
   
   For all variables $x$ and all $t_1 \neq t_2 \in \tran^{\wt,x}$ such that $t_1~\trf~t_3$, $t_3 \in \tran^{\rd,x}$, if $t_2~\trf~t_3$     or $t_2~\tpo~t_3$, then we have $t_2~\cfr~t_1$. 
      Define a function $\extend_{\readat}(\tau)$ which extends a trace $\tau=(\tran, \tpo, \trf)$ by 
  adding all possible $\cfr$ edges between transactions  which write on the same variable. 
  Traces violating $\readat$ exhibit a $\co \cup \cfr$ cycle in 
  $\extend_{\readat}(\tau)$, which we refer to as a $\readatcyc$. We say that $\tau \models \readat$ iff $\extend_{\readat}(\tau)$ does not contain a $\readatcyc$.  
\smallskip 
   
    \noindent{\emph{Examples}}. For the program Fig.\ref{fig:bad}(a), we have $t_1~\trf~t_4$ with $t_1, t_3 \in \tran^{\wt,x}$ and $t_3~\tpo~t_4$ giving $t_3~\cfr~t_1$. Likewise, we also have $t_2~\trf~t_4$ with $t_1, t_2 \in \tran^{\wt,z}$ and 
    $t_1~\rf~t_4$ (for reading into ${\tt{b}}$) giving $t_1~\cfr~t_2$ 
    resulting in $t_1~\cfr~t_2~\tpo~t_3~\cfr~t_1$ exhibiting a $\readatcyc$. Likewise, for the program in Fig.\ref{fig:bad}(d), 
    we have $t_3~\cfr~t_1$ (since $t_1~\trf~t_4$, $t_3~\tpo~t_4$ and $t_1, t_3 \in \tran^{\wt,x}$)
    and also $t_1~\cfr~t_2$ (since $t_2~\trf~t_5$, $t_1~\trf~t_5$ and  
        $t_1, t_2 \in \tran^{\wt,z}$) giving a $\readatcyc$.

  A run $\rho$ satisfies a model $X\in \{\cc, \ccvt, \cm, \readc, \readat\}$ if there exists a trace $\tau$ such that $\rho \models \tau$ and $\tau \models X$. 
       Define 
$
\conf_{X}:=
\{\trace_{X}\mid \exists\run\in\runs(\conf).
  \run\models\trace_{X}\land\trace_{X}\models X\}$, 
 the set of traces generated
under  $X$ from a given configuration $\conf$.

%% file: saturated.tex
\section{Trace Semantics}
\label{sec:saturated}
To analyse a program $\prog$ under a model $X \in \{\cc, \ccvt, \cm, \readat,\readc\}$, 
all runs of $\prog$ must be explored. We do this by exploring the associated traces. In fact, two runs having the same associated traces are equivalent 
 since the assertions to be checked at the end of a run depend only on $\tpo, \trf$. We begin with the empty trace, and continue exploration 
 by adding enabled read/write events to the traces generated so far. While doing this, we must ensure that the generated traces $\tau$ are s.t. $\tau \models X$. 
  We present two efficient operations to add a new transaction to 
 a trace $\tau$ obtaining a trace $\tau'$ so that $\extend_X(\tau')$  does not contain a $\xcyc$. 

\smallskip 

\noindent{\emph{Readability and Visibility}}. For all 5 models, readability identifies the transactions $t'$
from which read events of a newly added transaction $t$ can read/fetch their value. Visibility is used to add, in the case of $\ccvt$, $\readat$ and  $\readc$  
new $\cfr$ edges (and in the case of $\cm$, new $\hb{}$ edges) that are implied by the fact that the read event in the new transaction $t$ reads from $t'$. Let $\tau=\tracetuple$ be a trace, and $\tau_X=\extend_X(\tau)$. 
	Let $\tau^{t}_X$ denote adding $t$ to $\tau_X$.

We define the readable set $\rbl(\tau^t_X, t, x)$ for a transaction $t$ containing a read event on variable $x$ as the set of transactions 
$t' \in \tran^{\wt,x}$ from which $t$ can read from. 	Intuitively, $\rbl(\tau^t_X,t,x)$ contains all $t' \in \tran^{\wt,x}$ which are not hidden in $\tau^t_X$ by other $t'' \in \tran^{\wt,x}$.  The newly added transaction $t \in \tran^{\rd,x}$ can fetch its value from a transaction in $\rbl(\tau^t_X,t,x)$. 

\begin{enumerate}
\item For $X=\readc$, we define $\rbl(\tau^t_X, t, x)$ 
as the set of all transactions  $t' \in \tran^{\wt,x}$ provided we do not have a transaction $t'' \in \tran^{\wt,x}$ such that the following is true. 
Assume  $\beta~\tto~\alpha$ are two read events in $t$ such that $\beta$ reads from $t''$, $\alpha$ (current read event) reads from $t'$ and 
$t' ~(\co~\cup \cfr)^+~t''$. Note that having such a $t''$ induces  
$t''~\cfr~t'$ and a $\readccyc$.

	\item 
For $X=\cc$, 	$\rbl(\tau^t_X, t, x)$ is defined as the set of all transactions  $t' \in \tran^{\wt,x}$ s.t. 
\begin{itemize}

	\item  there is no transaction $t'' \in \tran^{\wt,x}$  s.t. $t'~\co~t''~\co~t$ in $\tau^t_X$. 
Allowing $t'~\trf~t$ (wrt $x$) in the presence of such a $t''$ gives $t''~\ow~t'$ and $\cccyc$.
\item 
	 In case $t' \in \tran^{\wt,y}$ ($t'$ also writes on some $y \neq x$) and $t \in \tran^{\rd,y}$,  
there is no transaction $t'' \in \tran^{\wt,y}$ such that $t'' ~\trf~t$ (wrt $y$)  and $t''~ \co ~t'$.  Note that having such a $t''$, and 
allowing $t'~\trf~t$ (wrt $x$) 
results in $t'~\ow~t''$ (wrt $y$)  and $\cccyc$.
\item 

if $t \in \tran^{\rd,y}$ (for some $y \neq x$) there are no transactions $t_1, t_2 \in \tran^{\wt,y}$ such that  
$t_1 ~\trf~ t$ (wrt $y$),  $t_1~\co~t_2~\co~t'$. Assuming we have  
this, then allowing $t'~\trf~t$ (wrt $x$) gives $t_1~\co~t_2~\co~t$ 
and hence $t_2~\ow~t_1$ (wrt $y$) and $\cccyc$. 
\end{itemize}	

\item For $X=\ccvt$, $\rbl(\tau^t_X, r, x)$ is defined as the set of all transactions $t' \in 
	\tran^{\wt,x}$ s.t. 
\begin{itemize}
	\item there is no transaction $t'' \in \tran^{\wt,x}$ such that 
	$t' ~(\co\cup \cfr)^+~t''~\co~t$. If we have this $t''$, allowing $t'~\trf~t$ gives $t''~\cfr~t'$ and $\ccvcyc$. 
	\item If $t' \in \tran^{\wt,y}$ ($t'$ also writes on some $y \neq x$), there is no transaction $t'' \in \tran^{\wt,x} \cap \tran^{\wt,y}$ such that $t''~\trf~t$ (wrt $y$).  
		Having such a $t''$ with $t'~\trf~t$ (wrt $x$) results in $t'~\cfr~t''$ (wrt $x$) and $t''~\cfr~t'$ (wrt $y$), and $\ccvcyc$. 
	\item If $t' \in \tran^{\wt,y}$ ($t'$ also writes on some $y \neq x$), there is no transaction 
	$t'' \in \tran^{\wt,y}$ 	such that $t''~\trf~t$ (wrt $y$) and 
	$t''~(\co~\cup~\cfr)^+t'$. Having such a $t''$ along with 
	$t'~\trf~t$ (wrt $x$) results in $t'~\cfr~t''$ (wrt $y$) and $\ccvcyc$. 
	\item There are no transactions $t_1, t_2 \in \tran^{\wt,y}$ (for $y \neq x$) such that  
$t_1~\trf~t$ (wrt $y$) and  $t_1~(\co \cup \cfr)^+~t_2~\co~t$. Having  such $t_1, t_2$, along with $t'~\trf~t$ gives $t_1~(\cfr \cup \co)^+~t_2~\co~t$. With  $t_1~\trf~t$, we get  $t_2~\cfr~t_1$ and $\ccvcyc$.
\item If $t \in \tran^{\rd,y}, t' \in \tran^{\wt,y}$ (for some $y \neq x$),  there are no transactions $t_1 \in \tran^{\wt,y}, t_2 \in \tran^{\wt,x}$ such that  
$t_1~(\co \cup \cfr)^+~t_2~\co~t'$. If we have this, allowing 
$t'~\trf~t$ gives $t_2~\cfr~t'$ and  $t'~\cfr~t_1$, resulting 
in $\ccvcyc$. 
\end{itemize}	
\item Let $t$ be a transaction in process $p$. For $X=\cm$, $\rbl(\tau^t_X, t, x)$ is defined as the set of all transactions  $t' \in \tran^{\wt,x}$ s.t. 
\begin{itemize}
\item 	there is no transaction $t'' \in \tran^{\wt,x}$ such that 
	$t' ~\hb{p}~t''~\hb{p}~t$. Having such a $t''$ and allowing $t'~\trf~t$ gives $t''~\hb{p}~t'$ and $\cmcyc$.
	\item  If $t' \in \tran^{\wt,y}$ (for $y \neq x$), there is no transaction $t'' \in \tran^{\wt,x} \cap \tran^{\wt,y}$ such that $t''~\trf~t$ (wrt $y$).  
		Having such a $t''$ with $t'~\trf~t$ (wrt $x$) results in $t''~\hb{p}~t'$, $t~\hb{p}~t''$  and $\cmcyc$.
		\item If $t' \in \tran^{\wt,y}$ (for $y \neq x$), there is no  transaction $t'' \in \tran^{\wt,y}$ such that $t''~\hb{p}~t'$ and $t''~\trf~t$ (wrt $y$).  Having such a $t''$ with  $t'~\trf~t$ gives $t'~\hb{p}~t''$ giving $\cmcyc$.
		\item  There are no transactions $t_1, t_2 \in \tran^{\wt,y}$ such that  
$t_1~\trf~t$ (wrt $y$), and $t_1~\hb{p}~t_2~\hb{p}~t'$.
 Having such $t_1, t_2$ with $t'~\trf~t$ gives  $t_2~\hb{p}~t$. With 
 $t_1~\trf~t$ we get $t_2~\hb{p}~t_1$ and $\cmcyc$.	
\item If $t' \in \tran^{\wt,y}$ (for $y \neq x$), there are no transactions $t_1 \in \tran^{\wt,y}, t_2 \in \tran^{\wt,x}$ such that $t_1~\trf~t$ (wrt $y$), $t_1~\hb{p}~t_2~\hb{p}~t$. If so, allowing $t'~\trf~t$ gives 
$t_2~\hb{p}~t'$. Also, $t_1 ~\trf~ t, t_1~\hb{p}~t_2$ gives $t~\hb{p}~t_2$. 
Now, we have $t_2~\hb{p}~t'~\trf~t~\hb{p}~t_2$ obtaining $\cmcyc$.
\end{itemize}
 \item  For $X=\readat$, $\rbl(\tau^t_X, t, x)$ is defined as the set of all transactions  $t' \in \tran^{\wt,x}$ s.t.
\begin{itemize}
	\item there is no transaction $t'' \in \tran^{\wt,x}$ such that 
	$t' ~(\co~\cup \cfr)^+~t''~(\tpo \cup \trf)~t$. Having such a $t''$ with $t'~\trf~t$ gives $t''~\cfr~t'$ and $\readatcyc$. 
		\item  If $t' \in \tran^{\wt,y}$ (for $y \neq x$) there is no transaction $t'' \in \tran^{\wt,x} \cap \tran^{\wt,y}$ such that $t''~\trf~t$ (wrt $y$). Allowing $t'~\trf~t$ in this case creates $t'~\cfr~t''$ and 
	$t''~\cfr~t'$ and $\readatcyc$.
	\item If $t' \in \tran^{\wt,y}$ (for $y \neq x$), there is no transaction $t'' \in \tran^{\wt,y}$ such that $t''~\trf~t$ (wrt $y$) and $t''~(\co \cup \cfr)^+t'$. Having such a $t''$ with $t'~\trf~t$ gives 
	$t'~\cfr~t''$ and $\readatcyc$.
	\item If $t' \in \tran^{\wt,y}$ (for $y \neq x$), there are no transactions $t_1 \in \tran^{\wt,y}, t_2 \in \tran^{\wt,x}$ such that $t_1~\trf~t$ (wrt $y$), $t_1~(\co \cup \cfr)^+~t_2~(\tpo\cup \trf)~t$. 
	 If so, allowing $t'~\trf~t$ gives $t_2~\cfr~t'$ ($t_2, t' \in \tran^{\wt,x}$) as well as 
	 $t'~\cfr~t_1$  ($t_1, t' \in \tran^{\wt,y}$). This gives 
	 $t_1~(\co \cup \cfr)^+~t_2~\cfr~t'\cfr~t_1$ and $\readatcyc$. 
	 
\end{itemize}

	\end{enumerate}
After adding $t'~\trf~t$, we must check that 
there are no consistency violations.  
The check set $\vbl(\tau^t_X, t,x)$ is defined as the set of transactions which turn 	``sensitive''  
on adding the new  edge  $t'~\trf~t$. 
Unless  appropriate edges are added involving 
these sensitive transactions, we may get consistency violating cycles
in the resultant trace.
Let $\tau^{tt'}$ denote the trace obtained by adding the new transaction $t$ and the edge $t'~\trf~t$ to trace $\tau$. Now, we 
identify the \emph{visible set}, that is, 
the ``sensitive transaction set''  $\vbl(\tau^t_X, t,x)$ and the new edges which must be added to $\tau^{tt'}$ to obtain a consistent extended trace.

 \begin{enumerate}
\item For $X=\readc$, $\vbl(\tau^t_X,t,x)=\{t''  \mid$ there is a read event $\beta$ in $t$ reading from $t''$, 
$\beta~\tto~\alpha$, and $\alpha$ reads from $t'\}$. The newly added $t'~\rf~t$ is due to this $\alpha$. 
 This necessitates 
adding $t''~\cfr~t'$ to preserve consistency. 	Then $\extend_X(\tau^{tt'})$ contains $\{(t'',t') \mid t'' \in \vbl(\tau^t_X, t,x)\}$.

	\item For $X=\readat$, $\vbl(\tau^t_X,t,x)$ is classified into two categories. 
	\begin{itemize}
	\item The first kind of transactions in $\vbl(\tau^t_X,t,x)$  are  
$\{t''\in \rbl(\tau^t_X, t,x) \mid t''~\tpo \cup \trf~t\}$.  
Then we add  from each $t'' \in \vbl(\tau^t_X,t,x)$ a $\cfr$ edge to $t'$ in $\extend_X(\tau^{tt'})$. 
	\item The second kind of transactions in $\vbl(\tau^t_X,t,x)$  are
	transactions $t'' \in \tran^{\wt,y}$ such that $t''~\trf~t$ when 
	$t' \in  \tran^{\wt,y}$. Then we add $t'~\cfr~t''$ to $\extend_X(\tau^{tt'})$.
		\end{itemize}

\item For $X=\ccvt, \vbl(\tau^t_X,t,x)$ is classified into three categories. 
		\begin{itemize}
		\item[(a)] The first kind of transactions in $\vbl(\tau^t_X,t,x)$  are  
	$\{t''\in \rbl(\tau^t_X, t,x) \mid t''~\co~t\}$.  
		Then we add  from each $t'' \in \vbl(\tau^t_X,t,x)$ a $\cfr$ to $t'$ in $\extend_X(\tau^{tt'})$. 
		\item[(b)] The second kind of  transactions in $\vbl(\tau^t_X,t,x)$  are   $t_1 \in \tran^{\wt,y}$ such that $t_1~\trf~t$ (wrt $y$) when $t' \in \tran^{\wt,y}$. Then we  add  $t'~\cfr~t_1$ to $\extend_X(\tau^{tt'})$. 
		\item[(c)] The third kind of  transactions in $\vbl(\tau^t_X,t,x)$  are 
		 $t_2 \in \tran^{\wt,y}$ such that $t_2~\co~t'$ 
		and we have $t_1~\trf~t$ (wrt $y$) for some $t_1 \in \tran^{\wt,y}$. In this case, we add $t_2~\cfr~t_1$  to $\extend_X(\tau^{tt'})$. 
			\end{itemize}

\item For $X=\cm$, $\vbl(\tau^t_X,t,x)$ is classified into three categories. Let transaction $t$ be in process $p$.
\begin{itemize}
\item The first kind of transactions in $\vbl(\tau^t_X,t,x)$  are  
$\{t''\in \rbl(\tau^t_X, t,x) \mid t''~\hb{p}~t\}$.  
Then we add  from each $t'' \in \vbl(\tau^t_X,t,x)$ a $\hb{p}$ edge to $t'$ in $\extend_X(\tau^{tt'})$.
\item The second kind of transactions in $\vbl(\tau^t_X,t,x)$  are 
$t_1 \in \tran^{\wt,y}$ such that $t_1~\trf~t$ when $t' \in \tran^{\wt,y}$. Then we add $t'~\hb{p}~t_1$ to $\extend_X(\tau^{tt'})$.
\item The third kind of transactions in $\vbl(\tau^t_X,t,x)$  are 
 $t_2 \in \tran^{\wt,y}$ such that $t_2~\hb{p}~t'$ 
when we have $t_1~\trf~t$  for some $t_1 \in \tran^{\wt,y}$. 
Then we add $t_2~\hb{p}~t_1$ to $\extend_X(\tau^{tt'})$.
Adding these $\hb{p}$ edges, can result in $t_3~\hb{p}~t_4$ for some $t_3,t_4\in \tran^{\wt,y}$.
If $t_3 ~\trf~ t_1$ and $t_1~\tpo~t$ then we add $t_1~\hb{p}~t_4$
to $\extend_X(\tau^{tt'})$. 
 
\end{itemize}
\end{enumerate}

 The \emph{trace semantics} for a model $X \in \{\readc, \readat, \cc, \ccvt, \cm\}$ is given as the transition relation $\xrightarrow[]{}_{X{-}\sat}$, defined as  
$\tau_X \xrightarrow[]{\alpha}_{X{-}\sat}\tau'_X$ where $\extend_X(\tau)=\tau_X, \extend_X(\tau')=\tau'_X$. 
 The label $\alpha$ is one of $\readact(t,t'), \beginact(t), \commitact(t), \writeact(t)$
 representing  respectively, a transaction $t$ reading from a transaction $t'$, the beginning, end  
 of a transaction $t$, as well as a write event in $t$.    
An important property of $\tau_X \xrightarrow[]{\alpha}_{X{-}\sat} \tau'_X$ is that if $\tau_X$ does not have $\xcyc$, then $\tau'_X$ also does 
not have $\xcyc$; in other words, if $\tau \models X$, then $\tau' \models X$. 

   We now describe the transitions  $\tau_X\xrightarrow[]{\alpha}_{X{-}\sat} \tau'_X$ where $\extend_X(\tau) =\tau_X, \extend_X(\tau')=\tau'_X$, $\tau=\tracetuple, \tau'=\tracetuplep$. We start from the empty trace $\tau_0$, $\extend_X(\tau_0)=\tau_0$. For $\alpha \in \{\commitact(t), \writeact(t)\}$, $\tau'_X=\tau_X$.

\begin{itemize}

\item  Assume we observe the start of a transaction $t$ in $\tau_X$.  Then the label $\alpha$ is $\beginact(t)$, and
$\tran'= \tran \cup \{t\}$, $\trf'=\trf$ and $\textcolor{red}{po}' = \textcolor{red}{po} \cup t' [\textcolor{red}{po}] t$ where $t',t$ are transactions in the same process $p$, $t'$ is the last transaction executed in $p$, and 
$t$ is the next transaction in $\tpo$ order of $p$. This gives $\tau'_X$.

 \item   Assume that $t$ reads from $t'$ in $\tau_X$. Then  the  label $\alpha$ is $\readact(t,t')$. 
           Add the transaction $t$, a $\tpo$ edge from the $\tpo$-latest transaction in process $p$ containing $t$ in $\tau_X$ to $t$.       Add $\trf$  from a $t' \in \rbl(\tau^t_X,t,x)$  to $t$ obtaining   $\tau^{tt'}$. 
            Then obtain 
      $\tau'_X=\extend_X(\tau^{tt'})$ as discussed above.  
\end{itemize}

\begin{lemma}
    If $\tau_X=\extend_X(\tau)$  with $\tau \models X$, and $\tau_X \xrightarrow[]{\alpha}_{\sat} \tau'_X=\extend_X(\tau')$,  then $\tau' \models X$ for $X \in \{\cc, \ccvt, \cm, \readat, \readc\}$. 
\label{lem:satcons}
\end{lemma}

\noindent{\emph{Efficiency and  Correctness}}. Each step of  $\xrightarrow[]{\alpha}_{\sat}$ is computable in polynomial time(see the supplementary material). This is based on the fact that  readable and visible sets are computable in polynomial time. 

The  correctness of the trace semantics for a model $X$ 
stems from the fact that it generates only those $X$-extensions which do not have cycles (Lemma \ref{lem:satcons}). The design 
of the transitions ensures that the resultant extended trace 
does not have a cycle.

%% file: dpor.tex
\section{DPOR Algorithm for $\cc, \ccvt, \cm , \readat , \readc$}
\label{sec:dpor}

We present our DPOR algorithm, which systematically explores,	 for any terminating program under the consistency models $X \in \{\cc, \ccvt, \cm, \readat, \readc\}$, all traces $\tau_X$ wrt $X$ which can be generated by the trace semantics.  Enabled transactions from any of the processes are added to the trace generated so far, and we proceed with the next transaction. For a transaction $t$ with a read event,  we explore in separate branches, all possible transactions $t'$ with write events from which $t$ can read. Each such branch is  a sequence of transactions also called a \emph{schedule}. 
There may be transactions $t''$ which will be added to the trace later in the exploration, from which $t$ can also read.   Such transactions $t''$ are called \emph{postponed} wrt $t$; when $t''$ is added to the trace later, the algorithm will have a branch where $t$ can read from $t''$. In that branch,
the algorithm reorders transactions in the sequence s.t. $t''$ and $t$ exchange places, and all transactions which are needed for $t''$ to occur are also
placed before $t''$ ($\declarepostponed$). All generated schedules will be executed by $\rschedule$. 
The algorithm is uniform across the models, 
with the main technical differences being taken care of
by the respective trace semantics which guides the exploration 
of traces.
\SetKwInput{KwInput}{Input}  
\SetAlCapFnt{}
\begin{algorithm}[ht]
  \footnotesize{
  \DontPrintSemicolon
  {\KwInput{{$X \in \{\readc , \readat, \cc, \ccvt , \cm \}$ , transactional program $\mathcal{P}$}}}
  
  $\textsc{inpProgram} = \mathcal{P}$ \tcc*[f]{global variable storing the program}
  
  $\explore(X , \tau_\emptyset , \epsilon)$
  }
  \caption{ $\textsc{DPOR}$($X$ , $\mathcal{P}$) \label{alg:algm}}
\end{algorithm}

\smallskip 

\noindent The main procedure is $\textsc{DPOR}(X,\mathcal{P})$ which invokes the procedure 
$\explore(X,\tau_\emptyset,\epsilon)$, to explore all traces of the input program $\mathcal{P}$ under  model $X$.

 \begin{figure}[t]
       \includegraphics[width=6in]{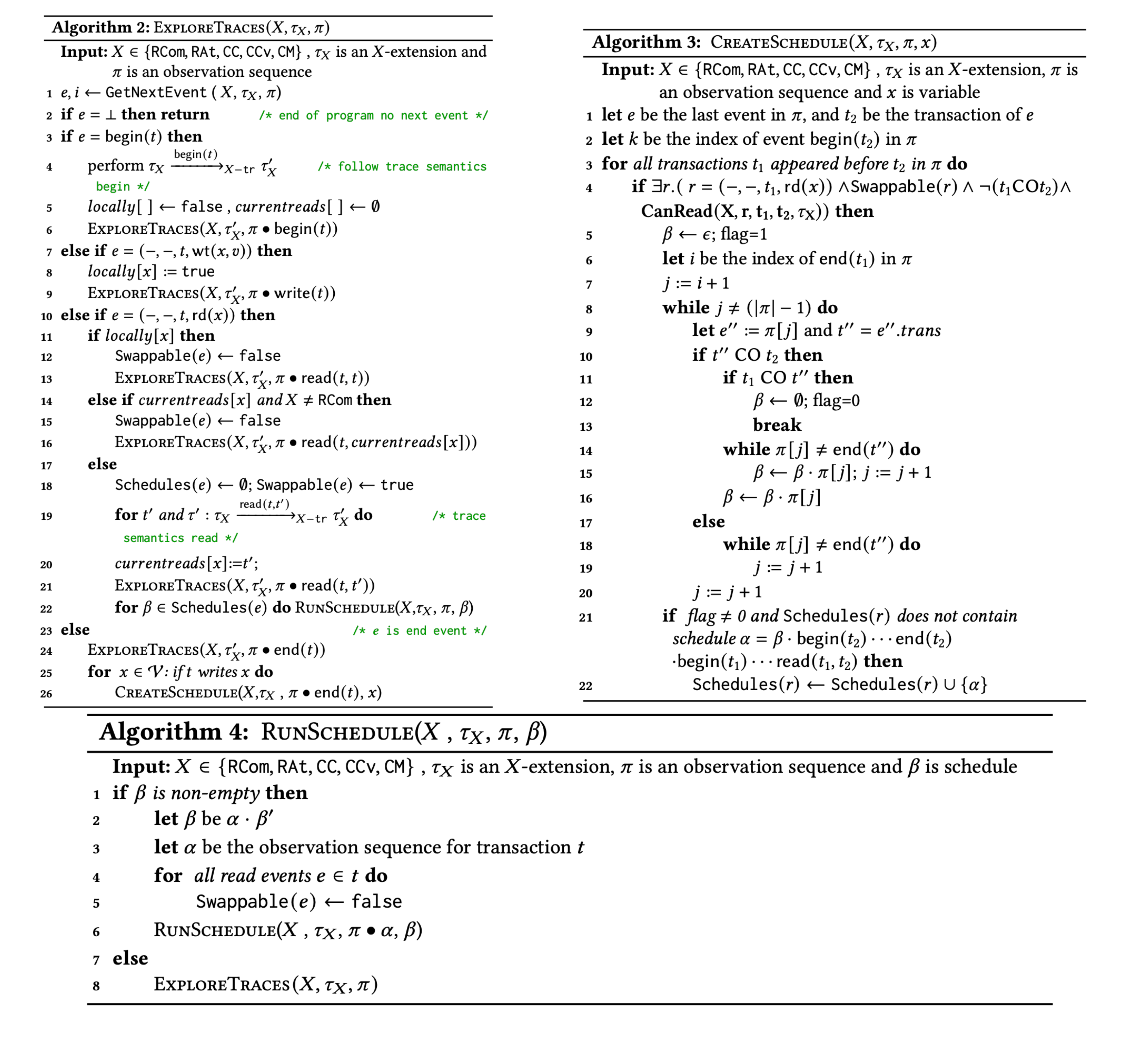}
      \label{alg:alg1}
            \end{figure}

\smallskip 
\noindent{\bf{$\explore$}}.  This algorithm takes as input,  a consistency model $X$, an  
$X$-extension $\tau_X$  and an \emph{observation sequence} $\pi$. 
 $\pi$ is a sequence of events 
 $ev \in \{\beginact(t)$,$\commitact(t)$,$\writeact(t),$ 
 $\readact(t,t')\}$. 
 Events from the same transaction appear contiguously in $\pi$. 
  The initial invocation is with the empty trace $\tau_0$ 
   and observation sequence $\pi=\epsilon$. 
 The observation sequence is 
used to swap transactions $t$ having read events with transactions $t''$  that are \emph{postponed}  
wrt them. 
 A Boolean array $locally[~]$ of size $|\vars|$ is reset to 
   $\false$ for all variables at the start of any transaction; 
 when a write event on $x$ takes place in the current transaction, 
 then  $locally[x]$ is set to $\true$. \textit{This signifies the local availability 
 of an $(x,v)$ in the log}.   We also maintain a map $currentreads$ from variables to transactions : $currentreads[x]=t'$ 
 means that $x$ read from transaction $t'$; hence, 
 all later read events of $x$ in the current transaction $t$ must read from $t'$, 
 unless $t$ writes to $x$ between an earlier read (from $t'$) and 
 the next read\footnote{This case does apply for $\readc$, later reads of x can read 
 from transaction other than $t'$.}.  
\texttt{GetNextEvent}() returns the next event $e$ and the last index $i$ of $\pi$. 

\noindent $\bullet$
If we are at  the  last event of $\pi$ and 
   it is not $\commitact(t)$ 
     then \texttt{GetNextEvent}() returns the next event $e$ (read, write or $\commitact$) and its index from the current transaction $t$.

            \noindent $\bullet$ If $e$  writes  on $x$,  set the flag $locally[x]$ TRUE to indicate that a local write on $x$ is available.   Update $\pi {\leftarrow} \pi.e$ and recursively call on $\explore$. Subsequent reads on $x$ from  $t$ will read from $t$.  
        
                 \noindent $\bullet$ Let $e$ be a read event on $x$ in transaction $t$. $e$ can read from the latest write in $t$ if $locally[x]$ is true. Otherwise, it reads from the transaction  stored in $currentreads[x]$ (this is a map from variables to transactions which maintains for each variable, the transaction it can be read  from).  In this case, 
          $\swapof{e}$ is $\false$, indicating there is no need 
          to look for a source transaction to read from. 
           If $currentreads[x]$ is empty, then  the set $\tt{Sche}$$\tt{dules(e)}$ is initialized and the Boolean flag $\swapof{e}$ is set to $\true$. 
          $\swapof{e}$ indicates that we must find, according to the trace semantics,  readable transactions $t' \in \tran^{w,x}$ where $e$ can read from, leading to $\tau'_X$. 
           For each such $t'$, let $\pi=\zeta \cdot \beginact(t) \dots \readact(t,t')$ be the observation sequence obtained 
            on choosing to read from $t'$. $currentreads[x]$ is set to $t'$ and   $\explore$() is called recursively.   
            If transaction $t'$ has not yet appeared in $\pi$, the write in $t'$  is called a \emph{postponed write}.

 \noindent $\bullet$ If $e$ is $\commitact(t)$, then we recursively call on $\explore$ after appending $\commitact(t)$ to $\pi$, exploring 
 other transactions.  When the recursive calls return to this point when the last  event of $\pi$ is $\commitact(t)$,  we do the following. 
 Corresponding to the last write for each variable in transaction $t$, we call $\declarepostponed$.  
 This call  enables us to pair these writes as potential sources to be read from, for the reads in transactions which are dependent on postponed writes. 
 This call also creates schedules for such reads which are identified in $\declarepostponed$. 
 These schedules are then explored when the recursive call to $\explore$ returns. 
 For a read event $e$ in a transaction $t$, each schedule $\beta \in \scheduledof{e}$ explores a new source transaction 
 $t'$ to read from, and extends the trace.

\smallskip       
 
\noindent{\bf{$\declarepostponed$}}.  
 The inputs to this algorithm  are a consistency model $X$, trace $\tau_X$ wrt $X$, a variable $x$, and  
 an observation sequence $\pi$ whose last element is a $e=\commitact(t_2)$ event of a transaction 
 $t_2 \in \tran^{\wt,x}$. 
   $\declarepostponed$ looks for read events $r$ occurring in a transaction $t_1$ in $\pi$ 
   for which $t_2$ is postponed. Indeed, 
   (i) $r$ must be swappable, 
   (ii)  $t_1$ must not be  $\co$ before $t_2$, and (iii) 
$\tt{\bf{CanRead}(X,r,t_1,t_2,\tau_X)}$ must be  true (that is, $t_2$ is readable for $r$  according to the trace semantics).  We begin with the closest (from  $e=\commitact(t_2)$)  transaction $t_1$ containing such a read $r$. A schedule $\beta$ is created for $r$. The schedule contains 
all elements following $\commitact(t_1)$ in $\pi$ from transactions 
following $t_1$ in $\pi$ and preceding $t_2$ wrt $\co$. The schedule ends with $\beginact(t_2)\dots,\commitact(t_2)\beginact(t_1)\dots \readact(t_1,t_2)$ enabling the read $r$ in $t_1$ to read from $t_2$. This schedule 
is added to ${\tt Schedules}(r)$ if it does not contain a  $\beta'$ 
which has the same set of observations.

 \smallskip 
 
 \noindent{\bf{$\rschedule$}}.  The inputs are a model $X$,  trace $\tau_X$, an observation sequence $\pi$ and a schedule $\beta$. 
  The schedule of observations in $\beta$  is explored one by one, 
by recursively calling itself, and updating the trace. The reads 
in the schedule are not swappable, preventing a redundant exploration 
for them (schedules where these are swapped with respective writes will 
be created by $\declarepostponed$).

\input{app-dpor.tex}

\subsection{Complexity, Soundness, and Completeness of the DPOR algorithm}

\begin{theorem}
The  DPOR algorithm for $\cc, \ccvt, \cm, \readc, \readat$ is sound, complete and optimal.
\label{thm:dporgood}
\end{theorem}

The proof of Theorem \ref{thm:dporgood} is in Appendices \ref{app:satcons},
\ref{app:complete-ccv}, \ref{app:ccful}, \ref{app:complete-cc}, \ref{cm:prop}, \ref{app:complete-cm}, \ref{rato:prop}, \ref{rato:comp}, \ref{rcom:prop}, \ref{rcom:comp}. Our DPOR algorithm spends polynomial time per explored trace : (1) each schedule has a length  bounded by the program length, (2) each schedule we generate  corresponds to a part of a trace that has not been generated before and is added to the set of schedules only if it has not been added.  Checking if a new schedule is equal to an existing one is done in poly time:  arrange all created schedules as a tree and  compare from the root,  wrt new schedule, going over the $\co$ relations. 

\smallskip

\noindent{\bf{Soundness, Optimality and Completeness}}. The algorithm is sound : if we initiate Algorithm 2 from $(X, \tau_{\emptyset},\epsilon)$, 
  all explored traces
$\tau_X$ are s.t. $\tau_X \models X$. This follows from the fact that the exploration uses the 
$\xrightarrow[]{}_{X{-}\sat}$ relation.
The algorithm is optimal in the sense that, for any two different recursive calls to Algorithm 2 with arguments 
$(X, \tau^1_X, \pi_1)$ and $(X, \tau^2_X, \pi_2)$, if $\tau^1_X, \tau^2_X$ are  etendible, then $\tau^1_X \neq \tau^2_X$. 
This follows from  (i) each time we run the for loop in line 19 of Algorithm 2, the  read event  will read from a different source, (ii) in each schedule $\beta \in \scheduledof{e}$ in line 22  of Algorithm 2, the event $e$ reads from a transaction $t$ which is different from all transactions it reads from in line 19 of Algorithm 2, (iii) Any two schedules added to $\scheduledof{e}$ at line 22 of Algorithm 3 will have a read event reading from different transactions. 
The algorithm is complete in the sense that, it produces all  traces corresponding to terminating runs. 

%% file: app-dpor.tex
 \section{An Illustrating Example for $\ccvt$ DPOR}
\label{app:dpor}

 We illustrate the DPOR algorithm on Figure \ref{eg:run}.
We have presented all the traces as well as the recursion tree in Figure \ref{eg:dpor-app}. 
The nodes $n_0 \dots n_{35}$ represent recursive calls to $\explore$ and $\rschedule$. 
$n_0$ is the first call to $\explore$ with an empty trace and  observation sequence. The detailed description can be seen in the supplementary material.

From $n_0$, consider the case when $\explore$ selects $t1$, making recursive calls to itself after each instruction in $t1$. 
The two reads of $t1$  each have two sources to read from, $init_x$ or $t4$, and $init_y$ or $t2$. 
The reads from $t4,t2$ will be scheduled later, by a call to $\declarepostponed$, as indicated by the orange arrows. 
The last call to $\explore$  in $t1$ is made at $\commitact(t1)$
reaching node $n_1$. 
Next, $\explore$ selects $t2$, and the last recursive call 
in $t2$ leads to node $n_5$, where $t3$ is selected. 
The read in $t3$ has two choices to read from : $t1$ and $t2$, 
 and makes two recursive calls after each choice.
  The respective last recursive calls to $\explore$ leading to $n_9,n_{10}$. 
$t5$ is selected next, leading to nodes 
$n_{16}$ and $n_{17}$. Finally, $t4$ is chosen.  $t4$ has two  
sources to read from, namely, $t5$ and $t3$, leading to nodes $n_{23}$ to $n_{26}$. 
This gives traces $\tau_1, \dots, \tau_4$, and respective observation sequences. 
For eg., the observation sequence for $\tau_3$ is 
 $\beginact(t1),\readact(t1,init_x),\readact(t1,init_y)$,
 $\writeact(t1),\commitact(,t1),\beginact(t2),\dots,\commitact(t2)$,
 $\beginact(t3),\readact(t3,t1),\writeact(t3)$,
  $\commitact(t3)$,
 $\beginact(t5)$,\\  $\dots,\commitact(t5),\beginact(t4),\readact(t4,t5)$,  $\writeact(t4),\commitact(t4)$. 
The recursive call to $\explore$ from  $n_{23}, n_{24}, n_{25}, n_{26}$  returns as there are no more transactions, and $\declarepostponed$ is called from these since $t4$  
has write instructions.  

\begin{figure}[H]
  \includegraphics[width=1.8in]{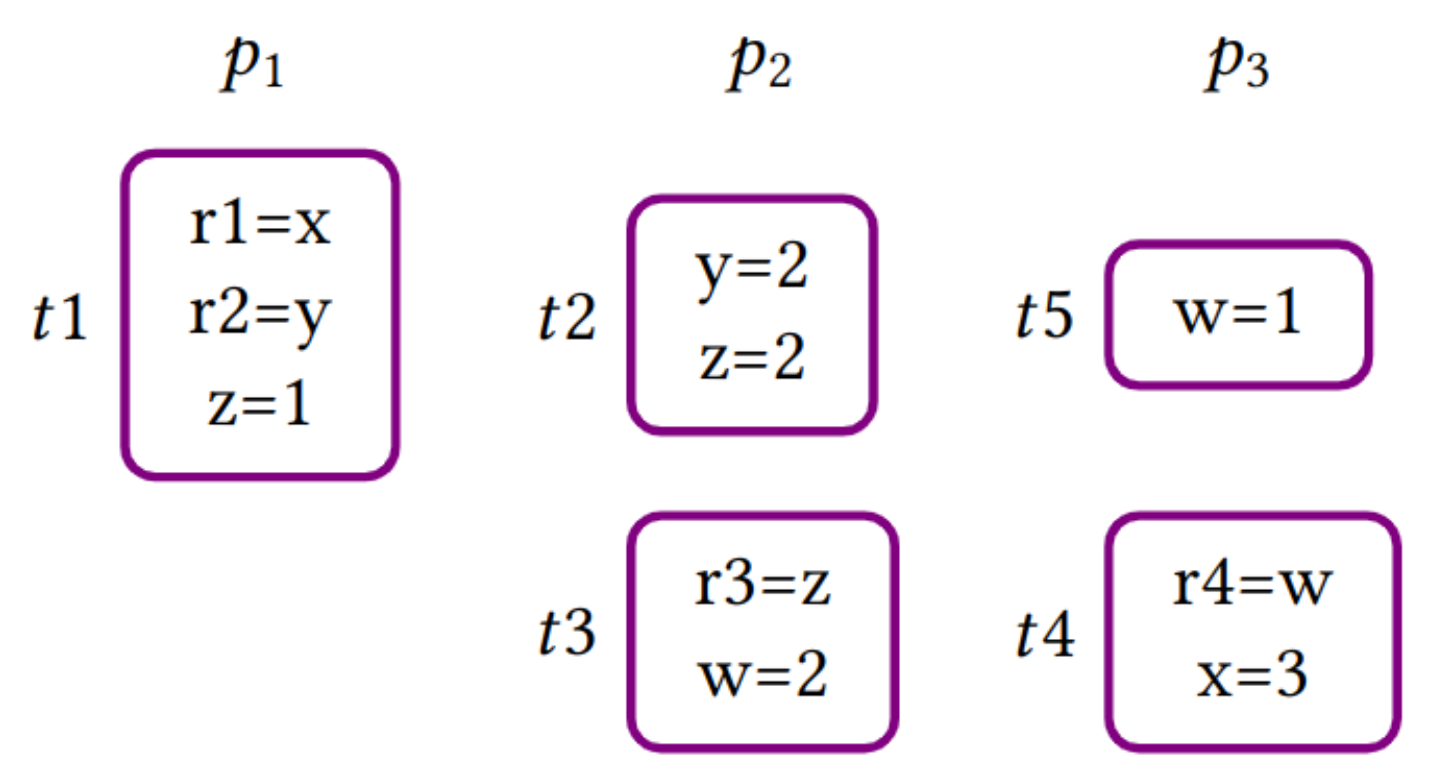}
  \caption{Program to illustrate DPOR}
\label{eg:run}
\end{figure}

\begin{figure}[H]  
 \begin{center}
 \includegraphics[width=5.2in]{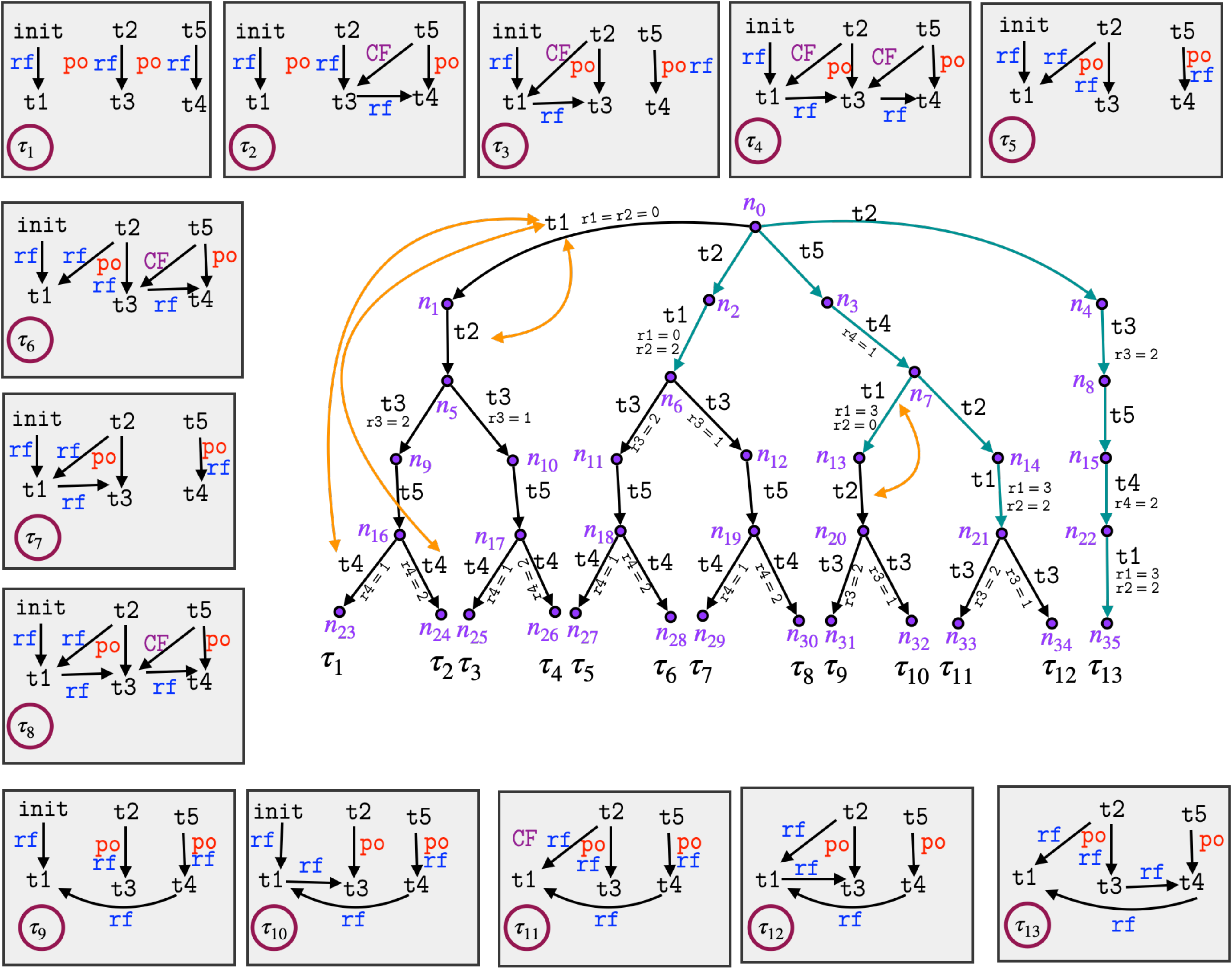}
 \end{center}
 \caption{Traces and recursion tree for the program in Figure \ref{eg:run} under $\ccvt$. 
  The nodes $n_0$ to $n_{35}$ represent recursive calls to Algorithm 2. The black edges come from Algorithm 2, green edges represent 
   Algorithm 4 running a schedule created by a call to Algorithm 3.
   The orange lines denote the pairs of  writes and reads, resulting in schedules.           The edges are decorated with the transactions executed, and register values read.
  }
 \label{eg:dpor-app}
 \end{figure}

 From $n_{23}$, $\declarepostponed$ is invoked when the recursive call of $\explore$ returns
  to the point after $\commitact(t4)$.
We traverse back in the observation sequence till  
  $\beginact(t4)$ (line 2 of $\declarepostponed$).  
We keep moving back till obtaining the read event $e'$ of $r1$ in $t1$, for which  
  $\swapof{e'}$ set to $\true$ in $\explore$. 
We check if $t1$ happens before $t4$, which it does not,  and also check that $t4$ is readable for $t1$($\tt{\bf{CanRead()}}$). 
Thus, $t4$ and $t1$ can be swapped.
The schedule $\beta_1$ is initialized. 
We then check whether the transactions $t2,t3,t5$ 
which happen between $\commitact(t1)$ and $\beginact(t4)$  
should be retained before $t4$ in the schedule for $e'$. 
    Since $\neg t2~\co~t4, \neg t3~\co~t4$, 
   the schedule need not contain $t2,t3$, and we skip over them. 
However,  $t5 ~\co~t4$  must be retained in the schedule.  
So we obtain $\beta_1=\beginact(t5)\dots\commitact(t5)\cdot$ $\beginact(t4)$ $\dots$ $\commitact(t4)\beginact(t1)$ $\dots$ $\readact(t1,t4)$ in $\scheduledof{e'}$. 
After this, the call from $\declarepostponed$ returns, and  we trace back to  $n_{16}$ and change the source (from the readable set) for $r4$, so $r4 = 2$ and we get the new trace $\tau_2$. 
The call  to $\explore$ at node $n_{24}$ returns to the point after $\commitact(t4)$, after which procedure $\declarepostponed$ is invoked. The corresponding read event is again $e'$ of $r1$ in $t1$.   
 As seen above, we 
check whether the transactions $t2,t3,t5$ 
which happen between $\commitact(t1)$ and $\beginact(t4)$  
should be retained before $t4$ in the schedule for $e'$. 
Indeed, $t3 ~\co~t4, t2 ~\co~t3, 
t5 [\tpo]  t4$. 
Hence we 
obtain another schedule  $\beta_2$=$\beginact(t2)$ $\dots\commitact(t2)$ $\beginact(t3)$ $\dots$ $\commitact(t3)\beginact(t5)$ $\dots$ $\commitact(t5)\beginact(t4)$ 
    $\dots \commitact(t4)\beginact(t1) \dots \readact(t1,t4)$ in $\scheduledof{e'}$. 
Then we backtrack to $n_{16}$ and call \\ $\declarepostponed$ at $\commitact(t5)$ which returns an empty schedule since there are no corresponding read events.  
Then we backtrack to node $n_9$ and then to $n_5$ till we can change the source for $r3$, which reads from $t1$. We continue exploring traces with $\explore$ and obtain  $\tau_3, \tau_4$. 
We invoke $\declarepostponed$ from nodes $n_{25}$ and $n_{26}$, however these do not add any more schedules to $e'$, since they are equivalent to what has  been added already in $\scheduledof{e'}$.  
From $n_{25}$, we backtrack and reach node $n_5$ and invoke 
$\declarepostponed$ at $\commitact(t2)$ which creates a schedule
$\gamma_1$ corresponding to the read event $e''$ of $r2$ in $t1$.  
 $\gamma_1$=$\beginact(t2)\dots \commitact(t2)\beginact(t1)
  \readact(t1,init_x)$ $\readact(t1,t2)$ is added to $\scheduledof{e''}$. 
We then backtrack till the point $r2=0$  
and find that $\scheduledof{e''}$  is non-empty. 
$\rschedule$  runs the schedule $\gamma_1$  (shown in green edges $n_0-n_2-n_6$), followed by recursive calls to $\explore$, resulting in traces $\tau_5, \tau_6, \tau_7, \tau_8$. 
 At node $n_{30}$ on return from $\explore$, we backtrack till the point $r1=0$ and find two schedules $\beta_1, \beta_2$ in $\scheduledof{e'}$. 
  We invoke $\rschedule$ first for $\beta_1$, 
      obtaining traces $\tau_9, \tau_{10}$ ending at nodes $n_{31}, n_{32}$. At nodes $n_{31}, n_{32}$, we invoke $\declarepostponed$ at $\commitact(t3)$ but these will not generate any schedules for the read $r4$ in $t4$  since $\tt{Swappable}$ is $\false$ for this read event. Note that 
$t3 [\trf] t4$ (along with $t4 [\trf]t1, t2 [\trf] t1$) is covered by trace $\tau_{13}$; thus,  we optimize on redundant traces by keeping  $\tt{Swappable}$  $\false$. 
We backtrack till $\commitact(t2)$ at $n_{20}$ and 
      invoke $\declarepostponed$ corresponding to the write in $t2$, which can be a source for the read in $t1$. 
      This adds the schedule $\gamma_2$=  $\beginact(t2)\dots \commitact(t2)$ $\beginact(t1)  \readact(t1,t4)\readact(t1,t2)$ to $\scheduledof{e''}$. 
On returning from this call, we backtrack till the read $e''$ of $r2=0$ in $t1$  and invoke $\rschedule$ at node $n_7$ (the green edges $n_7-n_{14}-n_{21}$) for $\gamma_2$.  
After $\rschedule$ finishes $n_{21}$, we invoke $\explore$ and obtain traces $\tau_{11}, \tau_{12}$.  
  
At node $n_{34}$ on return from $\explore$, we backtrack till the read event $e'$ of $r1$ in $t1$  and invoke $\rschedule$ on the remaining schedule $\beta_2$, obtaining trace $\tau_{13}$.

%% file: experiments.tex
\section{Experiments}
\label{sec:experiments}
We describe the implementation of our optimal DPOR algorithm for 
the causal consistency models $\cc, \ccvt, \cm, \readc, \readat$ as a tool $\ourtool$.  
To the best of our knowledge, \ourtool{} is the first stateless model checking tool for the causal consistency models $\cc, \ccvt, \cm, \readc, \readat$. 

\smallskip 

\noindent{\bf{\ourtool{}}}. 
\ourtool{}  extends  \nidhugg~\cite{10.1007/978-3-662-46681-0_28} and   
works at LLVM IR level accepting a C language program as input. 
At runtime, \ourtool{} controls the exploration of the input program 
until it has explored all the traces using the DPOR algorithm.   
It can detect user-provided assertion violations by analyzing the generated traces. We conduct all experiments on a Ubuntu 22.04.1 LTS with Intel Core i7-1165G7 and 16 GB RAM. 
We evaluate \ourtool~on the following categories of benchmarks, as seen below.

\begin{table}[t]
\footnotesize
\centering
\begin{minipage}{0.4\textwidth}
\caption{\ourtool~on  litmus tests.}
\label{tab:litmus}
\centering
\begin{tabular}{@{} l l l l@{}}
\hline\hline
{   \bfseries    Model}   &  {  \bfseries    Allow    }  & {	\bfseries   Forbid   } & {	\bfseries  Time} \\[0.5ex]
\hline
 \quad$\readc$  &   \quad 5867   &   \quad 1640  &  \quad 570s\\
\quad $\readat$ &  \quad  5867 &  \quad 1640 & \quad 550s \\
\quad $\cc$ &  \quad 5093   &  \quad 2414 & \quad 555s  \\
\quad $\ccvt$ &  \quad 5093   &  \quad 2414 & \quad 550s  \\
\quad $\cm$ &  \quad 5093   &  \quad 2414 & \quad 560s  \\
\hline
\end{tabular}
\end{minipage}
\hfill
\begin{minipage}{0.5\textwidth}
\caption{Evaluating \ourtool~on Database Applications. Tr is number of traces.}
\centering
\begin{tabular}{@{} l c c c c c c@{}}
\hline\hline
{   \bfseries   Applications} &   { $\bf{Tr}$}  &   { $\cc$} & {$\ccvt$  } & { $\cm$ } & { $\readat$  } & { $\readc$} \\[0.5ex]
\hline
\quad Vote \cite{10.14778/2732240.2732246}  & 10 & 0.1s & 0.06s & 0.06s & 0.13s & 0.13s\\  
\quad Twitter \cite{10.1145/2987550.2987559}& 4 & 0.1s & 0.07s & 0.08s & 0.12s & 0.13s  \\
\quad FusionTicket \cite{10.1145/2987550.2987559}&  8 & 0.51s & 0.80s & 1s & 1.13s & 13s\\
\quad Auction \cite{10.1007/978-3-030-44914-8_20} & 8 & 0.1s & 0.1s & 0.11s & 0.13s & 0.15s \\
\quad Auction-2 \cite{10.1007/978-3-030-44914-8_20} & 8 & 0.27s & 0.36s & 0.70s & 0.55s & 29s\\
\quad Group & 6 & 0.08s & 0.1s & 0.12s& 0.13s & 0.12s\\
\hline
\end{tabular}
\end{minipage}
\hspace*{\fill}

\vspace{7pt}
\caption{ Evaluating behaviour of \ourtool~on classical benchmarks}
\centering
\begin{tabular}{@{}l c c c c c@{}}
\hline\hline
{\bfseries $\;\;$ Program  } &  { $\cc$} & {$\ccvt$  } & { $\cm$ } & { $\readat$  } & { $\readc$}\\
\hline
 Causality Violation \cite{bernardi2016robustness} & \textcolor{cyan}{SAFE}& \textcolor{cyan}{SAFE}& \textcolor{cyan}{SAFE}& \textcolor{cyan}{SAFE}& \textcolor{brown}{UNSAFE}\\
Causal Violation \cite{10.1145/3360591} & \textcolor{cyan}{SAFE}& \textcolor{cyan}{SAFE}& \textcolor{cyan}{SAFE}& \textcolor{brown}{UNSAFE} & \textcolor{brown}{UNSAFE}  \\
 Delivery Order \cite{lmcs:7149} & \textcolor{brown}{UNSAFE} & \textcolor{brown}{UNSAFE} & \textcolor{cyan}{SAFE}& \textcolor{brown}{UNSAFE} & \textcolor{brown}{UNSAFE}\\
 Long Fork \cite{bernardi2016robustness} & \textcolor{brown}{UNSAFE}  & \textcolor{brown}{UNSAFE} & \textcolor{brown}{UNSAFE}  & \textcolor{brown}{UNSAFE} & \textcolor{brown}{UNSAFE} \\
 Lost Update \cite{bernardi2016robustness} & \textcolor{brown}{UNSAFE} & \textcolor{brown}{UNSAFE} & \textcolor{brown}{UNSAFE} & \textcolor{brown}{UNSAFE} & \textcolor{brown}{UNSAFE}\\
 Message Passing \cite{beillahi2021checking}& \textcolor{cyan}{SAFE}& \textcolor{cyan}{SAFE}& \textcolor{cyan}{SAFE}& \textcolor{cyan}{SAFE} & \textcolor{brown}{UNSAFE}\\
Modification Order & \textcolor{brown}{UNSAFE}& \textcolor{cyan}{SAFE}& \textcolor{cyan}{SAFE}& \textcolor{cyan}{SAFE} & \textcolor{brown}{UNSAFE}\\
 Conflict violation \cite{10.1145/3360591} & \textcolor{brown}{UNSAFE} & \textcolor{brown}{UNSAFE} & \textcolor{brown}{UNSAFE} & \textcolor{brown}{UNSAFE} & \textcolor{brown}{UNSAFE}\\
 Read Atomicity \cite{10.1145/3360591} & \textcolor{cyan}{SAFE}& \textcolor{cyan}{SAFE}& \textcolor{cyan}{SAFE}& \textcolor{cyan}{SAFE}& \textcolor{cyan}{SAFE}\\
 Read Committed \cite{10.1145/3360591} & \textcolor{brown}{UNSAFE} & \textcolor{brown}{UNSAFE} & \textcolor{brown}{UNSAFE} & \textcolor{brown}{UNSAFE} & \textcolor{brown}{UNSAFE}\\
 Repeated Read \cite{10.1145/3360591} & \textcolor{cyan}{SAFE}& \textcolor{cyan}{SAFE}& \textcolor{cyan}{SAFE}& \textcolor{cyan}{SAFE}& \textcolor{brown}{UNSAFE} \\
Load Buffer & \textcolor{cyan}{SAFE}& \textcolor{cyan}{SAFE}& \textcolor{cyan}{SAFE}& \textcolor{cyan}{SAFE}& \textcolor{cyan}{SAFE}\\
 Store Buffer & \textcolor{brown}{UNSAFE} & \textcolor{brown}{UNSAFE} & \textcolor{brown}{UNSAFE} & \textcolor{brown}{UNSAFE} & \textcolor{brown}{UNSAFE}\\
 Write Skew \cite{bernardi2016robustness} & \textcolor{brown}{UNSAFE} & \textcolor{brown}{UNSAFE} & \textcolor{brown}{UNSAFE} & \textcolor{brown}{UNSAFE} & \textcolor{brown}{UNSAFE}\\
\hline
\end{tabular}
\label{tab:database}
\end{table}

\smallskip 

\noindent{\bf{Evaluation on Litmus Tests}}. 
\label{sec:lit}
The first set of experiments (Table 1) includes litmus tests. 
These litmus tests are non-transactional programs with assertions at end of each program generated using the Herd tool \cite{10.1145/2627752}. 
 as a sanity check in tool development.  
We generated 7507 litmus tests   
and  instrumented   them into  transactional programs.  
Given the small number of instructions
 per process in the litmus tests 
of \cite{10.1145/2627752}, we group all instructions in a process as a single transaction. 
In these litmus tests, the processes execute concurrently, and we check if an assertion is violated or not. 
Table \ref{tab:litmus} 
describes the outcomes observed on the litmus tests under all five consistency models.  
An outcome $\bf{Allow}$ in Table \ref{tab:litmus} means that the program has an execution that violates the assertion under the respective consistency model; otherwise, the outcome is $\bf{Forbid}$. The time taken per model 
to evaluate all  7507 tests is atmost 570 seconds. 
 
 \smallskip 
 
\noindent{\bf{Evaluation on Classical Transactional Benchmarks}}. 
\label{sec:ana}
The second set (Table \ref{tab:database}) consists of classical benchmark examples from databases on well-known transaction related anomalies 
 \cite{bernardi2016robustness}, \cite{10.1145/3360591},  \cite{lmcs:7149} and \cite{beillahi2021checking}. 
 We consider 14 such examples. 
 While most of them are taken from  cited sources, the ones 
without a citation in Table \ref{tab:database} 
are created by us  (see Appendix \ref{app:tab3-p}).

For each of these examples, we have made three versions (see Appendix 
\ref{app:ext-tables})  
by parameterizing on the number of processes and transactions. 
In version 1,  we have two to four processes  per program and two transactions per process. Version 2  is obtained 
 allowing each process to have three to four transactions. Version 3
expands version 2 by allowing each program to have four to five processes.  
Versions 2,3 serve as a stress test for \ourtool{} as increasing the number of transactions and 
processes  increases the number of executions to be explored and the number 
of consistent traces.
 \ourtool~took less than four second to finish running all version 1 programs,  
 about 40\footnote{under $\readc$ consistency it took 160s, where the program \emph{Conflict Violation} took about 85s to generate 551781 traces.} seconds to finish running all 
 version 2 programs, and about 450 to 1200\footnote{under $\readc$ consistency it took 3900s, where the program \emph{Writeskew} took about 2360s to generate 7726230 traces.} seconds to finish running all version 3 programs 
 across the models. 
 All programs in Table \ref{tab:database} are classified as \textcolor{cyan}{SAFE} or \textcolor{brown}{UNSAFE}, 
 depending on whether or not they have executions violating the assertion.  
Since programs under $\cc,\readc$ semantics may exhibit more behaviours than $\ccvt, \cm, \readat$ 
we observe more \textcolor{brown}{UNSAFE} under $\cc$ and $\readc$ semantics 
while they are \textcolor{cyan}{SAFE} under the $\ccvt,\cm$ and $\readat$ semantics.
\smallskip

\noindent{\bf{Evaluation on some Database Applications}}. 
\label{sec:appl}
The third set (Table 2) of benchmarks consists of transactional programs inspired 
from six distributed systems and database applications \cite{10.14778/2732240.2732246}, \cite{10.1145/2987550.2987559}, 
\cite{10.1007/978-3-030-44914-8_20}. For each application,  $\ourtool{}$ 
verifies the assertions given in these papers on all possible executions.  For instance, Auction \cite{10.1007/978-3-030-44914-8_20} is a buggy application where 
the highest bid may not win, while  Auction-2 \cite{10.1007/978-3-030-44914-8_20} is the non buggy version where only the highest bid will win. $\ourtool{}$ detected the execution leading to the bug in Auction \cite{10.1007/978-3-030-44914-8_20}, and also verified that in all executions of Auction-2, the highest bid wins.  
 Group is a synthetic application created by us, inspired from whatsapp groups where 
 we verify that a new user can be added to a whatsapp group only once.

%% file: conclusion.tex
\section{Conclusion}
\vspace{-.1cm}
In this paper, we have provided a DPOR algorithm using the $\tpo$-$\trf$ equivalence for transactional programs under five prominent causal consistency models, and also implemented the same in a tool \ourtool{}. This is the first tool for stateless model checking of transactional causal consistency models. As future work, we plan to adapt 
our algorithm to  two of the strongest consistency notions known,  namely snapshot isolation(SI) and serializability (SER). The main challenge in  handling these two models is that the complexity of checking if a given $\tpo$-$\trf$ trace is SI/SER consistent is an NP-complete problem \cite{10.1145/3009837.3009888}. Hence, we need to come up with some 
heuristics to have an algorithm which will scale well in practice.

%% file: ccv-proofs.tex
\newpage
\centerline{\Large{Appendix}}

The appendix contains the following sections.
\begin{enumerate}
\item Programs under $\ccvt$ 
\begin{enumerate}
\item Readability and Visibility conditions for $\ccvt$ (Appendix \ref{ccv:rv})
		\item Readablity and Visibility checks in polynomial time  (Appendix \ref{ccv:rv})

	\item Soundness- Properties of Trace semantics 
 (Appendix \ref{app:satcons})
	
	\item Proof of the completeness of the $\ccvt$ DPOR algorithm (Appendix \ref{app:complete-ccv}).
\end{enumerate}

\item Programs under $\cc$ 
\begin{enumerate}
\item Readability for $\cc$ (Appendix \ref{cc:r})
 \item Readablity checks in polynomial time  (Appendix \ref{lemma:cc-poly})

	\item Soundness-Properties of Trace semantics 
 (Appendix \ref{app:ccful})
	\item Proof of the completeness of the $\cc$ DPOR algorithm (Appendix \ref{app:complete-cc}).
\end{enumerate}

\item Programs under $\cm$ 
\begin{enumerate}
\item Readability and Visibility conditions for $\cm$ (Appendix \ref{cm:rv})
	\item Readablity and Visibility checks in polynomial time  (Appendix \ref{lemma:cm-poly})

	\item Soundness-Properties of Trace semantics (Appendix \ref{cm:prop})
 	\item Proof of the completeness of the $\cm$ DPOR algorithm (Appendix \ref{app:complete-cm}).
\end{enumerate}

\item Programs under $\readat$ 
\begin{enumerate}
\item Readability and Visibility conditions for $\readat$ (Appendix \ref{rato:rv})
	\item Readablity and Visibility checks in polynomial time  (Appendix \ref{rato:rv})
	\item Soundness-Properties of Trace semantics 
  (Appendix \ref{rato:prop})
	\item Proof of the completeness of the $\readat$ DPOR algorithm (Appendix \ref{rato:comp}).
\end{enumerate}

\item Programs under $\readc$ 
\begin{enumerate}
\item Readability and Visibility for $\readc$ (Appendix \ref{rcom:rv})
	\item Readablity and Visibility checks in polynomial time  (Appendix \ref{rcom:rv})
	\item Soundness-Properties of Trace semantics 
 (Appendix \ref{rcom:prop})
	\item Proof of the completeness of the $\readc$ DPOR algorithm (Appendix \ref{rcom:comp})
\end{enumerate}

\item Experimental Evaluation
\begin{itemize}
	\item[(a)] Time taken by $\ccvt, \cc , \cm , \readat , \readc$ for  versions-2 and version-3 of the classical benchmarks (Appendix \ref{app:ext-tables})
        \item [(b)] Uncited program used in classical benchmarks (Appendix \ref{app:tab3-p})

  \end{itemize}
\end{enumerate}

\newpage

\begin{center}
\Large{\bf{Causal Convergence $\ccvt$ }}
	
\end{center}
\section{$\ccvt$ Trace Semantics}
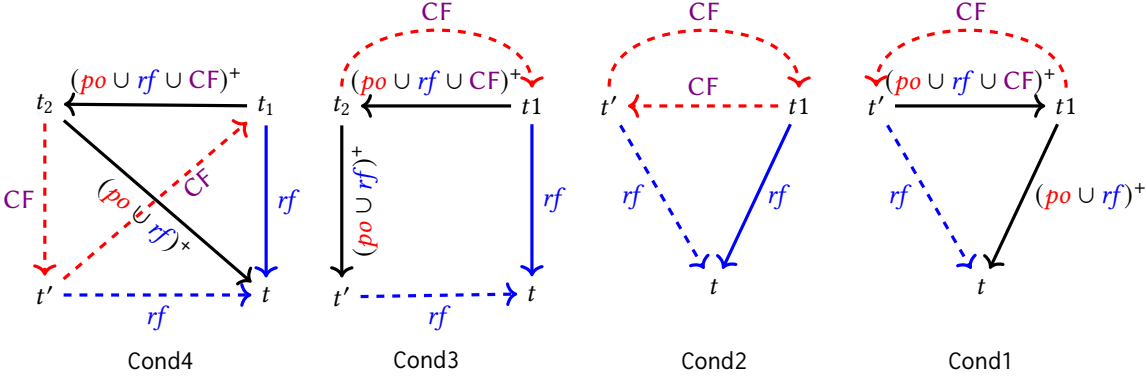
\begin{figure}[H]
        \begin{tikzpicture}[node distance=10mm , thick, main/.style= {draw,circle},];
        \node[] (t) {$t$};
        \node[] (t1) [above = 20mm of t] {$t_1$};
        \node[] (t') [left = 25mm of t] {$t'$};
        \node[] (t2) [above = 20mm of t']  {$t_2$};
        \node[] (1) [below=5mm of t] {};
        \node[] (cap) [left=7mm of 1] {$\tt{Cond4}$};
        \draw[->,blue,line width=1.2pt] (t1) -- node[right] {$\trf$} (t);
        \draw[->,black,line width=1.2pt] (t1) -- node[above] {$(\tpo \cup \trf \cup \cfr)^+$} (t2);
        \draw[->,black,line width=1.2pt] (t2) --node[rotate=-40, below] {$(\tpo \cup \trf)^+$} (t);
        \draw[dashed , ->,blue,line width=1.2pt] (t') -- node[below] {$\trf$} (t);
        \draw[dashed , ->,red,line width=1.2pt] (t2) -- node[left] {$\cfr$} (t');
        \draw[dashed , ->,red,line width=1.2pt] (t') -- node[rotate=40,right=3mm] {$\cfr$} (t1);

        \node[] (3at2) [right= 5mm of t1] {$t_2$};
        \node[] (3at1) [right= 20mm of 3at2] {$t1$};
        \node[] (3at) [below=20mm of 3at1] {$t$};
        \node[] (3at') [below=20mm of 3at2] {$t'$};
        \node[] (31) [below=5mm of 3at] {};
        \node[] (3acap) [left=7mm of 31] {$\tt{Cond3}$};

        \draw[->,black,line width=1.2pt] (3at1) -- node[above] {$(\tpo \cup \trf \cup \cfr)^+$} (3at2);
         \draw[->,blue,line width=1.2pt] (3at1) -- node[right] {$\trf$} (3at);
        \draw[dashed , ->,blue,line width=1.2pt] (3at') -- node[below] {$\trf$} (3at);
        \draw[dashed , ->,red,line width=1.2pt] (3at2) to [out=90, in=90] node[above] {$\cfr$} (3at1);
        \draw[->,black,line width=1.2pt] (3at2) -- node[rotate=90,below] {$(\tpo \cup \trf)^+$} (3at');
        
        \node[] (2at') [right= 5mm of 3at1] {$t'$};
        \node[] (2at1) [right= 20mm of 2at'] {$t1$};
        \node[] (21) [below=20mm of 2at1] {};
        \node[] (2at) [left=8mm of 21] {$t$};
        \node[] (2acap) [below=5mm of 2at] {$\tt{Cond2}$};
        
        \draw[dashed , ->,red,line width=1.2pt] (2at1) -- node[above] {$\cfr$} (2at');
         \draw[->,blue,line width=1.2pt] (2at1) -- node[right] {$\trf$} (2at);
        \draw[dashed , ->,blue,line width=1.2pt] (2at') -- node[left=1mm] {$\trf$} (2at);
        \draw[dashed , ->,red,line width=1.2pt] (2at') to [out=90, in=90] node[above] {$\cfr$} (2at1);

        \node[] (1t') [right= 5mm of 2at1] {$t'$};
        \node[] (1t1) [right= 20mm of 1t'] {$t1$};
         \node[] (111) [below=20mm of 1t1] {};
        \node[] (1t) [left=8mm of 111] {$t$};
        \node[] (1cap) [below=5mm of 1t] {$\tt{Cond1}$};

        \draw[->,black,line width=1.2pt] (1t') -- node[above] {$(\tpo \cup \trf \cup \cfr)^+$} (1t1);
        \draw[->,black,line width=1.2pt] (1t1) -- node[right] {$(\tpo \cup \trf)^+$} (1t);
        \draw[dashed , ->,blue,line width=1.2pt] (1t') -- node[left=1mm] {$\trf$} (1t);
        \draw[dashed , ->,red,line width=1.2pt] (1t1) to [out=90, in=90] node[above] {$\cfr$} (1t');
        \end{tikzpicture}
        \caption{Readability for $\ccvt$}
        \label{read-ccv}
    \end{figure}

\subsection{Readability and Visibility  for $\ccvt$}
\label{ccv:rv}
For $X=\ccvt$, $\rbl(\tau^t_X, t, x)$ is defined as the set of all transactions $t' \in 
	\tran^{\wt,x}$ s.t. 
\begin{itemize}
	\item[\texttt{cond1}] there is no transaction $t'' \in \tran^{\wt,x}$ such that 
	$t' ~(\co\cup \cfr)^+~t''~\co~t$. If we have this $t''$, allowing $t'~\trf~t$ gives $t''~\cfr~t'$ and $\ccvcyc$. 
	\item[\texttt{cond2}] If $t' \in \tran^{\wt,y}$ ($t'$ also writes on some $y \neq x$), there is no transaction $t'' \in \tran^{\wt,x} \cap \tran^{\wt,y}$ such that $t''~\trf~t$ (wrt $y$).  
		Having such a $t''$ with $t'~\trf~t$ (wrt $x$) results in $t'~\cfr~t''$ (wrt $x$) and $t''~\cfr~t'$ (wrt $y$), and $\ccvcyc$. 
	\item [\texttt{cond3a}] If $t' \in \tran^{\wt,y}$ ($t'$ also writes on some $y \neq x$), there is no transaction 
	$t'' \in \tran^{\wt,y}$ 	such that $t''~\trf~t$ (wrt $y$) and 
	$t''~(\co~\cup~\cfr)^+t'$. Having such a $t''$ along with 
	$t'~\trf~t$ (wrt $x$) results in $t'~\cfr~t''$ (wrt $y$) and $\ccvcyc$. 
	\item[\texttt{cond3b}] There are no transactions $t_1, t_2 \in \tran^{\wt,y}$ (for $y \neq x$) such that  
$t_1~\trf~t$ (wrt $y$) and  $t_1~(\co \cup \cfr)^+~t_2~\co~t'$. Having  such $t_1, t_2$, along with $t'~\trf~t$ gives $t_1~(\cfr \cup \co)^+~t_2~\co~t$. With  $t_1~\trf~t$, we get  $t_2~\cfr~t_1$ and $\ccvcyc$.
\item[\texttt{cond4}] If $t \in \tran^{\rd,y}, t' \in \tran^{\wt,y}$ (for some $y \neq x$),  there are no transactions $t_1 \in \tran^{\wt,y}, t_2 \in \tran^{\wt,x}$ such that  
$t_1~(\co \cup \cfr)^+~t_2~\co~t$ and $t_1~\trf~t$ (wrt $y$).  If we have this, allowing 
$t'~\trf~t$ gives $t_2~\cfr~t'$ and  $t'~\cfr~t_1$, resulting 
in $\ccvcyc$. 
\end{itemize}	

Figure \ref{read-ccv} explains cycles created after violating conditions $\tt{cond1}-\tt{cond4}$.
$\tt{cond3} $ covers $\tt{cond3a} $  and $\tt{cond3b} $.
For $\tt{cond3a}$,  $t_2$ coincides with $t'$.

After adding $t'~\trf~t$, we must check that 
there are no consistency violations.  
The check set $\vbl(\tau^t_X, t,x)$ is defined as the set of transactions which turn 	``sensitive''  
on adding the new  edge  $t'~\trf~t$. 
Unless  appropriate edges are added involving 
these sensitive transactions, we may get consistency violating cycles
in the resultant trace.
Let $\tau^{tt'}$ denote the trace obtained by adding the new transaction $t$ and the edge $t'~\trf~t$ to trace $\tau$. Now, we 
identify $\vbl(\tau^t_X, t,x)$ and the edges which must be added to $\tau^{tt'}$ to obtain a consistent extended trace.

For $X=\ccvt, \vbl(\tau^t_X,t,x)$ is classified into three categories. 
	\begin{itemize}
	\item[(a)] The first kind of transactions in $\vbl(\tau^t_X,t,x)$  are  
$\{t''\in \rbl(\tau^t_X, t,x) \mid t''~\co~t\}$.  
	Then we add  from each $t'' \in \vbl(\tau^t_X,t,x)$ a $\cfr$ to $t'$ in $\extend_X(\tau^{tt'})$. 
	\item[(b)] The second kind of  transactions in $\vbl(\tau^t_X,t,x)$  are   $t_1 \in \tran^{\wt,y}$ such that $t_1~\trf~t$ (wrt $y$) when $t' \in \tran^{\wt,y}$. Then we  add  $t'~\cfr~t_1$ to $\extend_X(\tau^{tt'})$. 
	\item[(c)] The third kind of  transactions in $\vbl(\tau^t_X,t,x)$  are 
	 $t_2 \in \tran^{\wt,y}$ such that $t_2~\co~t'$ 
	and we have $t_1~\trf~t$ (wrt $y$) for some $t_1 \in \tran^{\wt,y}$. In this case, we add $t_2~\cfr~t_1$  to $\extend_X(\tau^{tt'})$. 
		\end{itemize}

\begin{lemma}
    Given trace $\tau$, a transaction $t$ and a variable $x$, we can construct the sets  $\rbl(\tau, t, x)$, $\vbl(\tau, t, x)$ in polynomial time.
    \label{lemma:ccv-poly}
\end{lemma}
\begin{proof}
    We prove that set $\rbl(\tau, t, x)$ can be computed in polynomial time ($O(|\tran|^3)$ time), by designing an algorithm to generate set $\rbl(\tau, t, x)$. The algorithm consists of the following  steps:
    \begin{itemize}
        \item[(i)] First we compute the  transitive closure of the relations
        $[\tpo \cup \trf \cup \cfr]$ and $[\tpo \cup \trf]$, i.e, we compute $[\tpo \cup \trf \cup \cfr ]^+$ and $[\tpo \cup \trf]^+$.  We can use the Floyd-Warshall algorithm \cite{10.5555/1614191} to compute the transitive closure. This will take $O(|\tran|^3)$ time.
        \item[(ii)] We compute the set $\tran'=\{t'\mid t' [\tpo \cup \trf]^+ t\}$. This takes $O(|\tran|)$ time.
 
        \item[(iii)] For each transaction $t' \in \tran^{wt,x}$, we perform following checks:
        \begin{itemize}
            \item Check $\mathtt{cond 1 :} $ Check whether there exists a transaction $t'' \in \tran' \cap \tran^{wt,x}$            
            with $t'$ $[\tpo  \cup \trf \cup \cfr]^+$ $t''$. 
            If there is no such $t''$ then proceed to $\mathtt{cond 2}$. 
            If we find such $t''$, then $t' \notin \rbl(\tau , t, x)$. 
            This step will take $O(|\tran|^2)$ time.
            \item Check $\mathtt{cond 2 :}$ Check whether there exists a transaction $t''$ such that $t'' [\trf^y] t$ and $t',t'' \in \tran^{wt,x} \cap \tran^{wt,y}$. 
            If there is no such $t''$ then proceed to $\mathtt{cond 3}$. 
            If we find such $t''$, then $t' \notin \rbl(\tau , t, x)$. 
            This step will take $O(|\tran|)$ time.
            \item Check $\mathtt{cond 3 :}$ 
            
            We first check case a of $\mathtt{cond 3}$. Check whether there exists a transaction $t''$ such that $t'' [\trf^y] t$ , $t' \in \tran^{wt,y}$, and $t'' [\tpo \cup \trf \cup \cfr]^+ t'$. 
            If there is no such $t''$ then proceed to case two of $\mathtt{cond 3}$. 
            If we find such $t''$, then $t' \notin \rbl(\tau , t, x)$. 
            This step will take $O(|\tran|)$ time.
            
             Next, we  check case b of $\mathtt{cond 3}$. Check whether there are transactions $t''$ and $t'''$ such that $t'' [\trf^y] t$, $t''' \in \tran^{wt,y}$, $t''' [\tpo \cup \trf]^+ t'$ and $t'' [\tpo \cup \trf \cup \cfr]^+ t'''$. 
            If there are no such $t''$ and $t'''$ then proceed to $\mathtt{cond 4}$. 
            If we find such $t''$ and $t'''$, then $t' \notin \rbl(\tau , t, x)$. 
            This step will take $O(|\tran|^2)$ time.
            \item Check $\mathtt{cond 4 :}$ 
            Check whether there are transactions $t''$ and $t''' \in \tran'$ such that $t'' [\trf^y] t$, $t''' \in \tran^{\wt,x}$, $t'\in \tran^{\wt,y}$ and $t'' [\tpo \cup \trf \cup \cfr]^+ t'''$. 
            If we find such $t''$ and $t'''$, 
            then $t' \notin$ $\rbl(\tau , t, x)$. 
            If there are no such $t''$ and $t'''$ then add $t'$ to $\rbl(\tau , t, x)$.
            This step will take $O(|\tran|^2)$ time.
        \end{itemize}
    \end{itemize}
    Similarly, we can compute $\vbl(\tau, t, x)$ in polynomial time. This completes the proof.
\end{proof}

\subsection{Properties of Trace semantics under $\ccvt$}
\label{app:satcons}
\begin{lemma}
    If $\tau_1 \models \ccvt$  and $\tau_1$ $\xrightarrow[]{\beginact(t)}_{\ccvt{-}\sat}$ $\tau_2$ then $\tau_2 \models \ccvt$.
    \label{lem:f3}
\end{lemma}
\begin{proof}
The proof follows trivially since we do not change the reads from and conflict relations and the transaction $t$ added to the partial order has no successors. Thus, if $\tau_1 \models \ccvt$, 
so does $\tau_2$.  
	\end{proof}

\begin{lemma}
    If $\tau_1 \models \ccvt$  and $\tau_1$ $\xrightarrow[]{\readact(t,t')}_{\ccvt{-}\sat}$ $\tau_2$ then $\tau_2 \models \ccvt$.
    \label{lem:f4}
\end{lemma}

\begin{proof}
    Let $\tau_1 = \langle \tran_{1} , \tpo_1 , \trf_1, \cfr_1 \rangle$, $\tau_2 = \langle \tran_{2} , \tpo_2 , \trf_2, \cfr_2 \rangle$, and $ev$ = $r(p,t,x,v)$ be a read event in transaction $t \in \tran_1$. Suppose $\tau_2 \nvDash \ccvt$, and  that $\tpo_2 \cup \trf_2 \cup \cfr_2$ is cyclic. Since $\tau_1 \models \ccvt$, and $t$ has no outgoing edges, 
   on adding $(t',t) \in \trf_2$, 
    it follows that the cycle in $\tau_2$ is as a result 
    of the newly added $\cfr_2$ edges. 
    
    We prove that 
    on adding $(t',t) \in \trf_2$, 
        such cycles are possible iff at least one of ${\tt{cond}_1}, \dots, {\tt{cond}_4}$ are true. 
    
 \subsection*{${\tt{cond1}}\vee \dots \vee {\tt{cond4}}$ induces $\tpo_2 \cup \trf_2 \cup \cfr_2$ cyclicity }
   \label{sec:cyc1}
    This direction is easy to see : assume 
    one of ${\tt{cond1}}, \dots, {\tt{cond4}}$ are true; then, as already 
    argued in the main paper, we will get a $\tpo_2 \cup \trf_2 \cup \cfr_2$ cycle on adding the $\trf_2$ edge from $t'$ to $t$ and we are done.

 \subsection*{$\tpo_2 \cup \trf_2 \cup \cfr_2$ cyclicity implies ${\tt{cond1}}\vee \dots \vee {\tt{cond4}}$ }
 \label{sec:cyc2}
  For the converse direction, assume that we add the $\trf_2$ edge from $t'$ to $t$ and obtain a $\tpo_2 \cup \trf_2 \cup \cfr_2$ cycle in $\tau_2$. We now argue that this cycle has been formed because one of 
 ${\tt{cond}_1}, \dots, {\tt{cond}_4}$ are true.  
 
 First of all, note that $t$ has no outgoing edges in $\tau_2$, since 
 it is the current transaction being executed. Also, 
 we know that $\tau_1$ has no $\tpo_1 \cup \trf_1 \cup \cfr_1$ cycles. 
 Thus, the cyclicity of $\tpo_2 \cup \trf_2 \cup \cfr_2$ is induced by the 
 newly added $\trf_2$ edge as well as the newly added 
  $\cfr_2$ edges. Note that adding the $\trf_2$ edge  to $\tau_1$ does not induce any cycle since $t$ has no outgoing edges. Lets analyze the 
  $\cfr_2$ edges added which induce cycles, and argue that 
  one of ${\tt{cond1}}, \dots, {\tt{cond4}}$ will be true.   
  
  \begin{enumerate}
  	\item We add  $\cfr_2$ edges from $t'' \in \vbl(\tau_1, t,x)$ to $t'$. 
  	For  these edges to induce a cycle, we should have a path from $t'$ to $t''$ in $\tau_1$. That is, we have $t' [\tpo_1 \cup \trf_1 \cup \cfr_1]^+ t''$. This is captured by ${\tt{cond1}}$.  
  	
  	\item Consider $y \neq x$. If $t' \in \tran^{wt,x} \cap \tran^{wt,y}$, 
  	and we have $t'''\trf_1^y t$.  Then we add 
  	$(t',t''') \in \cfr^y_2$. To get a cycle, we need a path from 
  	$t'''$ to $t'$. 
  	\begin{itemize}
  	\item If $t''' \in \tran^{wt,x} \cap \tran^{wt,y}$, then we add $(t''', t') \in \cfr^x_2$, resulting in a cycle. This is handled by  ${\tt{cond2}}$.
  	\item If we have a path 
  	$t''' [\tpo_1 \cup \trf_1 \cup \cfr_1]^+ t'$, then we get a cycle again. This is handled by ${\tt{cond3}}$. 
  \item As a last case, to obtain a path from $t'''$ to $t'$, 
  assume there is a path from $t'''$ to $t$ in $\tau_1$, and let $t'' \in \tran^{wt,x}$ be the last transaction writing to $x$ in this path.  Note that this will induce a path from $t'''$ to $t''$ to $t'$ : the path  from $t''$ to $t'$ comes by the $\cfr^x_2$ edge added from $t''$ to $t'$ since 
  we have $(t',t) \in \trf_2^x$. Once we get this path from $t'''$ to $t'$, we again have the path we were looking for to get the cycle. This is handled by ${\tt{cond4}}$. 
    		
  	\end{itemize}
   Thus, we have shown that obtaining a path from $t'''$ to $t'$ is covered 
by the conditions  ${\tt{cond}_1}, \dots, {\tt{cond}_4}$, and hence, 
a $\tpo_2 \cup \trf_2 \cup \cfr_2$ cycle. 

   \item Consider $y \neq x$. If $t_4 \in \tran^{wt,y}$ such that $t_4 [\tpo\cup\trf]^+ t'$, 
  	and we have $t'''\trf_1^y t$.  Then we add 
  	$(t_4,t''') \in \cfr^y_2$. To get a cycle, we need a path from 
  	$t'''$ to $t_4$. 
   If we have a path 
  	$t''' [\tpo_1 \cup \trf_1 \cup \cfr_1]^+ t_4$, then we get a cycle. This is handled by ${\tt{cond3}}$. 
\end{enumerate}

 Thus, we can think of  $\tpo_2 \cup \trf_2 \cup \cfr_2$ cycles as a result of the forbidden patterns described in ${\tt{cond}_i}$ $1 \leq i \leq 4$.  By construction of $\tau_1$ $\xrightarrow[]{\readact(t,t')}_{\ccvt{-}\sat}$ $\tau_2$, we ensure $\neg {\tt{cond}_1}, \dots, \neg {\tt{cond}_4}$. Hence, 
    $\tau_2 \models \ccvt$.  
\end{proof}

Define $\big[[ \sigma]\big]^{\ccvt{-}\sat}$ as the set of  traces  generated  using $\xrightarrow[]{}_{\ccvt{-}\sat}$ transitions, starting from an empty trace $\tau_{\o}$.

Consider a terminal trace $\tau$ generated by $\ccvt$ DPOR algorithm starting from the empty trace $\tau_\emptyset$.
That is, there is a sequence  
$\tau_0 \xrightarrow[]{}_{\ccvt{-}\sat} 
\tau_1 \xrightarrow[]{}_{\ccvt{-}\sat}  
\tau_2 \xrightarrow[]{}_{\ccvt{-}\sat} \dots 
 \xrightarrow[]{}_{\ccvt{-}\sat}  \tau_n$ with 
 $\tau_0=\tau_{\emptyset}$, 
 and $\tau_n=\tau$. 
 Since $\tau_{\emptyset}$ is a empty we have $\tau_{\emptyset} \models \ccvt$, 
 it follows by  
 Lemmas \ref{lem:f3} and \ref{lem:f4} that $\tau \models \ccvt$. 
 
 Hence, for each trace $\tau \in \big[[ \sigma]\big]^{\ccvt{-}\sat} $ we have $\tau \models \ccvt$.

\subsection{$\ccvt$ DPOR Completeness}
\label{app:complete-ccv}
In this section, we show the completeness of the DPOR algorithm. More precisely, for the input program under $\ccvt$ for any terminating run $\rho \in Runs(\conf)$ and trace $\tau$ s.t. $ \rho \models \tau$, 
we show that $\explore(\ccvt, \tau_{\emptyset}, \epsilon)$ will produce a recursive visit $\explore(\ccvt, \tau', \pi)$ for some terminal $\tau'$, and $\pi$ where $\tau'=\tau$.   First, we give some definitions and auxiliary lemmas.

Let  $\pi = \alpha_{t_1} \alpha_{t_2} \dots \alpha_{t_n}$ be an observation sequence where  $\alpha_t = \beginact(p,t) \dots$ \plog{end}$(p,t)$.
 $\alpha_t$ is called an \emph{observable}, and is a sequence 
 of events from transaction $t$.  
Given $\pi$ and a trace $\tau$, we define 
$\tau \vdash_{\tt{G}} \pi$ to represent a sequence 
$ \tau_0 \xrightarrow[]{\alpha_{t_1}}_{\ccvt{-}\sat} 
 \tau_1  \xrightarrow[]{\alpha_{t_2}}_{\ccvt{-}\sat} \cdot \cdot \cdot \xrightarrow[]{\alpha_{t_n}}_{\ccvt{-}\sat} \tau_n $, where 
 $\tau_0 = \tau$ and 
$\pi = \alpha_{t_1} \alpha_{t_2} \dots \alpha_{t_n}$. Moreover, we define $\pi (\tau) :=  \tau_n$.

 $\tau$ is $terminal$ if there is no event left to execute $i.e$ $\succof\tau = \emptyset$. 
 We define $\tau \vdash_{\tt{G_T}} \pi$ to say that $(i)$ $\tau \vdash_{\tt{G}} \pi$ and $(ii)$ $\pi(\tau)$ is terminal.

\begin{definition}($p$-free and $t$-free observation sequences)
For a process  $p \in \mathbf{P}$, we say that an observation sequence $\pi$ is $p$-free if   all observables $\alpha_{t}$ in $\pi$ pertain to 
transactions $t$ not in $p$. 
For a transaction $t$ issued in $p$ we say that $\pi$ is $t$-free if $\pi$ is $p$-free. 
\end{definition}

\begin{definition}(Independent Observables)
For observables $\alpha_{t_1}$ and $\alpha_{t_2}$, we write $\alpha_{t_1} \sim  \alpha_{t_2}$ to represent they are independent. 
This means 
$(i)$ $t_1,t_2$ are transactions issued in different processes, that is, 
 $t_1$ is issed in $p_1$, $t_2$ is issued  in $p_2$, with $p_1 \neq p_2$,  
$(ii)$ no read event $r(p_2,t_2,x,v)$ of transaction $t_2$ reads from $t_1$, and 
$(iii)$ no read event $r(p_1,t_1,x,v)$ of transaction $t_1$ reads from $t_2$.	
Thus, $\neg(\alpha_{t_1} \sim  \alpha_{t_2})$ if either $t_1, t_2$ are issued in the same process, or there is a $\trf$ relation between $t_1, t_2$. That is, $\neg(\alpha_{t_1} \sim  \alpha_{t_2})$ iff $t_1 [\tpo \cup \trf]^+ t_2$ or  $t_2 [\tpo \cup \trf]^+ t_1$.

\end{definition}
 
 \begin{definition}(Independent Observation Sequences)
Observation sequences $\pi^1, \pi^2$ are called independent written $\pi^1 \sim \pi^2$ if there are observables $\alpha_{t_1}$, $\alpha_{t_2}$, and observation sequences $\pi'$ 
and $\pi''$ such that $\pi^1 = \pi' \cdot \alpha_{t_1} \cdot \alpha_{t_2} \cdot \pi''$, 
$\pi^2 = \pi' \cdot \alpha_{t_2} \cdot \alpha_{t_1} \cdot \pi''$, 
and $\alpha_{t_1} \sim \alpha_{t_2}$. 
In other words, we get $\pi^2$ from $\pi^1$ by swapping neighbouring independent observables corresponding to transactions $t_1$ and $t_2$. 
\end{definition}
We use $\approx$ to denote reflexive transitive closure of $\sim$.

\begin{definition}(Equivalent Traces)
For traces $\tau_1 = \langle \tran_1, \tpo_1, \trf_1, \cfr_1 \rangle$ and $\tau_2 = \langle \tran_2, \tpo_2, \trf_2, \cfr_2 \rangle$, we say $\tau_1, \tau_2$ are equivalent denoted $\tau_1 \equiv \tau_2$ if $\tran_1 = \tran_2$, $\tpo_1 = \tpo_2$, $\trf_1 = \trf_2$ and for all $t_1$,$t_2 \in \tran_1^{w,x}$ for all variables $x$, we have $t_1$ [$\tpo_1 \cup \trf_1 \cup \cfr_1^x$] $t_2$ iff $t_1$ [$\tpo_2 \cup \trf_2 \cup \cfr_2^x$] $t_2$.
	
\end{definition}

\begin{lemma}
$((\tau_1 \equiv \tau_2) \wedge \tau_1 \vdash_{\tt{G}} \alpha_t) \Rightarrow (\tau_2 \vdash_{\tt{G}} \alpha_t \wedge (\alpha_t(\tau_1) \equiv \alpha_t(\tau_2)))$.
\label{lemma g-1}
\end{lemma}
\begin{proof}
Assume 	$\tau_1 \vdash_{\tt{G}} \alpha_t$ for an observable $\alpha_t$, and 
$\tau_1=(\tran_1, \tpo_1, \trf_1, \cfr_1)$, 
$\tau_2=(\tran_2, \tpo_2, \trf_2, \cfr_2)$ with $\tau_1 \equiv \tau_2$. Then we know that $\tran_1=\tran_2, \tpo_1=\tpo_2, \trf_2=\trf_1$. 
 
 Since $\tau_1 \vdash_{\tt{G}} \alpha_t$, let  $\tau_1 \xrightarrow[]{\alpha_t}_{\ccvt{-}\sat} \tau'$. 
 $\tau'=(\tran', \tpo', \trf', \cfr')$ where $\tran'=\tran_1 \cup \{t\}$, 
 $\tpo'=\tpo_1 \cup \{(t_1,t) \mid t_1 \in \tran_1$ is in the same process as $t\}$,  $\cfr_1 \subseteq \cfr'$ 
 and $\trf_1 \subseteq \trf'$. $\cfr'$ can contain $(t',t)$ for some 
 $t' \in \tran_1$, when $t',t \in {\tran'}^{w,x}$ for some variable $x$, based 
 on the trace semantics. In particular if we have 
 $t_1 {\trf'}^x t_2$ , $t_3 \in \tran^{wt,x}$ and $t_3 (\trf'\cup  \tpo')^+ t_2$, then    
$t_3 {\cfr'}^x t_1$.\footnote{We write $t~\cfr^x~t'$ to denote that this $\cfr$ relation is added because of a read on $x$ by a transaction $t''$ in the trace semantics.}

 Since $\tau_1 \equiv \tau_2$, for any transactions $t', t'' $ in $\tran_1=\tran_2$, $t' [\tpo_1 \cup \cfr_1^x \cup \trf_1] t''$ iff $t' [\tpo_2 \cup \cfr_2^x \cup \trf_2] t''$ for all variables $x$. In particular, $t' \cfr_1^x t''$ iff $t' \cfr_2^x t''$ for all $x$. 
 Now, let us construct a trace $\tau''=(\tran', \tpo', \trf', \cfr'')$,
 where $\cfr''$ is the smallest set such that $\cfr_2 \subseteq \cfr''$ and 
  whenever $t_1 {\trf'}^x t_2$ , $\tau_3 \in \tran^{\wt,x}$ and $t_3 (\trf'\cup  \tpo')^+ t_2$, then $t_3 {\cfr''}^x t_1$.

   Since $t' [\tpo_1 \cup \trf_1 \cup \cfr_1^x] t''$ iff $t' [\tpo_2 \cup \trf_2 \cup \cfr_2^x] t''$ for all $x$, and all $t', t'' \in \tran_1=\tran_2$, and $\cfr''$ is the smallest extension of $\cfr_2$  based 
   on the trace semantics, along with the fact that 
   $\tpo'=\tpo'', \trf'=\trf''$, we obtain for any two 
   transactions $t_1, t_2 \in \tran'=\tran''$, 
   $t_1 [\tpo' \cup \trf' \cup \cfr'] t_2$ iff  $t_1 [\tpo'' \cup \trf''
   \cup \cfr''] t_2$. This gives $\tau' \equiv \tau''$. 
    
      This also gives $\tau_2 \xrightarrow[]{\alpha_t}_{\ccvt{-}\sat} \tau''$,
  that is, $\tau_2 \vdash_{\tt{G}} \alpha_t$  and indeed $\alpha_t(\tau_1)=\tau' \equiv \tau''=\alpha_t(\tau_2)$.

\end{proof}

\begin{lemma}
If $\tau \vdash_{\tt{G}} \alpha_{t_1} \cdot \alpha_{t_2}$ and $\alpha_{t_1} {\sim} \alpha_{t_2}$ then $\tau \vdash_{\tt{G}} \alpha_{t_2} {\cdot} \alpha_{t_1}$ and $(\alpha_{t_1} {\cdot} \alpha_{t_2})(\tau) \equiv (\alpha_{t_2}{ \cdot} \alpha_{t_1})(\tau)$.
\label{lemma g-2}
\end{lemma}
\begin{proof}
Let $\tau = \langle \tran , \tpo, \trf,\cfr \rangle$ be a trace and let $t_1$ be a transaction issued in process $p_1$, and $t_2$ be a transaction issued in process $p_2$, with $p_1 \neq p_2$. Assume $\tau \vdash_{\tt{G}} \alpha_{t_1} \cdot \alpha_{t_2}$, with $\alpha_{t_1} {\sim} \alpha_{t_2}$.

We consider the following cases.
\begin{itemize}
    \item [$\bullet$] $t_1,t_2 \notin \tran^{rd}$. In this case $\tau \vdash_{\tt{G}} \alpha_{t_2} \cdot \alpha_{t_1}$ holds trivially.
    
     $(\alpha_{t_1} \cdot \alpha_{t_2})(\tau) = (\alpha_{t_2} \cdot \alpha_{t_1})(\tau) = \langle \tran' , \tpo' , \trf, \cfr \rangle$, where $\tran' = \tran \cup \{t_1,t_2\}$ and $\tpo' = \tpo \cup \{ (t,t_1) | t \in p_1\} \cup \{(t',t_2) | t' \in p_2 \}$.
    \item [$\bullet$] $t_1 \in \tran^{wt}$, $t_1 \notin \tran^{rd}$ and $t_2 \in \tran^{rd} \cap \tran^{wt}$. 
    Let $\tau_1 = \alpha_{t_1}(\tau)$. 
    We know that $\tau_1 = \langle \tran_1 , \tpo_1 , \trf_1,\cfr_1 \rangle$, where 
    $\tran_1 = \tran \cup \{ t_1 \}$, and $\tpo_1 = \tpo \cup \{ (t,t_1) | t \in p_1\}$, with $\trf_1 = \trf$ since there are no read events in $t_1$, and $\cfr_1 = \cfr$ since by the trace semantics, when $ \trf$ remain same, there are no $\cfr$ edges to be added.
    
    Consider $\tau_2 = \alpha_{t_2}(\tau_1) = \langle \tran_2, \tpo_2, \trf_2, \cfr_2 \rangle$. 
    Since $t_2 \in \tran^{rd}$, we know that $t_2$ has read events.
    \smallskip

    Let there be $n$ read events $ev_1 \dots ev_n$ in transaction $t_2$, 
    reading from transactions $t'_1, \dots, t'_n$ on variables $x_1, \dots, x_n$.  
    Hence we get a sequence $\tau_1 \xrightarrow[]{\beginact(t_2)}_{\ccvt{-}\sat} \dots \tau'_{i_1} \xrightarrow[]{\readact(t_2,t'_1)}_{\ccvt{-}\sat}  \tau_{i_1}  \xrightarrow[]{}_{\ccvt{-}\sat}\dots  
    \tau'_{i_n} \xrightarrow[]{\readact(t_2,t'_n)}_{\ccvt{-}\sat} \tau_{i_n} \dots \xrightarrow[]{\commitact(t_2)}_{\ccvt{-}\sat} \tau_2$.
    Hence $\tran_2 = \tran_1 \cup \{t_2\}$, $\tpo_2 = \tpo_1 \cup \{(t',t_2) | t' \in p_2\}$, $\trf_2 = \trf_1 \cup (\trf_{i_1} \cup \trf_{i_2} \dots \cup \trf_{i_n})$, and $\cfr_2 = \cfr_1 \cup (\cfr_{i_1} \cup \cfr_{i_2} \dots \cup \cfr_{i_n})$.
    Since $\tau_1 \vdash_{\tt{G}} \alpha_{t_2}$,  for all $ev_i : 1 \le i \le n$, we have $t'_i \in \rbl(\tau_1,t_2,x_i)$.
    
    Since $\alpha_{t_1} \sim \alpha_{t_2}$, we know that for all $ev_i : 1 \le i \le n$, we have $t'_i \in \rbl(\tau,t_2,x_i)$. It follows that $\tau \vdash_{\tt{G}} \alpha_{t_2}$.
    Define $\tau_3 = \langle \tran_3, \tpo_3, \trf_2, \cfr_2 \rangle$, where $
    \tran_3 = \tran \cup \{t_2\}$ and $\tpo_3 = \tpo \cup \{(t',t_2) | t' \in p_2\}$. It follows that $\alpha_{t_2}(\tau) = \tau_3$, $\tau_3 \vdash_{\tt{G}} \alpha_{t_1}$ and $\tau_2 = \alpha_{t_1}(\tau_3) = \alpha_{t_2}(\alpha_{t_1}(\tau))$.
    \item [$\bullet$] $t_2 \in \tran^{wt}$ and $t_1 \in \tran^{rd} \cap \tran^{wt}$. Similar to previous case.
    \item [$\bullet$] $t_1, t_2 \in \tran^{rd} \cap \tran^{wt}$. 
    Let $\tau_1 = \alpha_{t_1}(\tau) = \langle \tran_1, \tpo_1, \trf_1, \cfr_1 \rangle$ and $\tau_2 = \alpha_{t_2}(\tau_1) = \langle \tran_2, \tpo_2, \trf_2, \cfr_2 \rangle$. 
    Since $t_1 \in \tran^{rd}$, we know that $t_1$ has read events. 
    Let there be $n$ read events $ev_1 \dots ev_n$ in transaction $t_1$, reading from transactions $t'_1, \dots, t'_n$ on variables 
    $x_1, \dots, x_n$. 
    Hence we get the sequence $\tau \xrightarrow[]{\beginact(t_1)}_{\ccvt{-}\sat} \dots \tau'_{i_1} \xrightarrow[]{\readact(t_1,t'_1)}_{\ccvt{-}\sat}  \tau_{i_1} \xrightarrow[]{}_{\ccvt{-}\sat}\dots  
    \tau'_{i_n} \xrightarrow[]{\readact(t_1,t'_n)}_{\ccvt{-}\sat} \tau_{i_n} \dots \xrightarrow[]{\commitact(t_1)}_{\ccvt{-}\sat} \tau_1$.
    Since $t_2 \in \tran^{rd}$, we know that $t_2$ has read events. 
    Let there be  $m$ read events $ev'_1 \dots ev'_m$ in transaction $t_2$, reading from transactions $t''_1, \dots, t''_m$ on variables 
    $y_1, \dots, y_m$.  
    Then we get the sequence $\tau_1 \xrightarrow[]{\beginact(t_2)}_{\ccvt{-}\sat} \dots \tau'_{j_1} \xrightarrow[]{\readact(t_2,t''_1)}_{\ccvt{-}\sat}  \tau_{j_1}  \xrightarrow[]{}_{\ccvt{-}\sat} \dots 
    \tau'_{j_n} \xrightarrow[]{\readact(t_2,t''_m)}_{\ccvt{-}\sat} \tau_{j_n} \dots \xrightarrow[]{\commitact(t_2)}_{\ccvt{-}\sat} \tau_2$. 
    
    $\mathbf{Proving}$ $\mathtt{ \tau \vdash_{\tt{G}} \alpha_{t_2} {\cdot} \alpha_{t_1}}$
    
    Since $\tau_1 = \alpha_{t_1}(\tau)$, for all read events $ev_i$ such that $ev_i.trans = t_1$, we have $t'_i \in \rbl(\tau,t_1,x_i)$.
    Since $\tau_2 = \alpha_{t_2}(\tau_1)$, for all read events $ev'_j$ such that $ev'_j.trans = t_2$, we have $t''_j \in \rbl(\tau_1,t_2,y_j)$. 
    Since $\alpha_{t_1} \sim \alpha_{t_2}$, 
    it follows that, we have $t''_j \in \rbl(\tau,t_2,y_j)$ for all read events $ev'_j$ such that $ev'_j.trans = t_2$. Hence $\tau \vdash_{\tt{G}} \alpha_{t_2}$. 
    \smallskip

    Consider $\tau_3 = \alpha_{t_2}(\tau) = \langle \tran_3, \tpo_3, \trf_3, \cfr_3 \rangle $. To show that $\tau \vdash_{\tt{G}} \alpha_{t_2} \alpha_{t_1}$, 
    we show that for all read events $ev_i$ such that $ev_i.trans = t_1$,  $t'_i \in \rbl(\tau_3,t_1,x_i)$. We prove this using contradiction.
     Assume that $\exists t'_i$ s.t. $t'_i \notin \rbl(\tau_3,t_1,x_i)$.

    \medskip 
    \smallskip

    If $t'_i \notin \rbl(\tau_3,t_1,x_i)$, then 
    (1) there is a $t'_k$ such that 
    (i) $t'_k \in \vbl(\tau_3,t_1,x_i)$ and (ii) $t'_i [\tpo_3 \cup \trf_3 \cup \cfr_3]^+ t'_k$,
    or (2) there are transactions $t_3$ and $t_4$ such that $t_4 \in \tran^{\wt,y}$, $t_3 ~\trf^y_3 t_1$, 
    $t_4 [\tpo_3\cup\trf_3]^* t'_i $ and $t_3 [\tpo_3\cup\trf_3]^+ t_4 $,
    or (3) there are transactions $t_3$ and $t_4$ such that $t_3 \in \tran^{\wt,x}$, $t_3 [\tpo_3\cup\trf_3]^+ t_1$,
    $t_3 [\tpo_3\cup\trf_3]^+ t_4 $ and $t' \in \tran^{\wt,y}$.
    or (4) there is a  transaction $t_3$ such that $t_3 \in \tran^{\wt,x,y}$, $t_3 \trf^y_3 t_1$,
    and $t' \in \tran^{\wt,x,y}$.
     
     \medskip 
     \smallskip 
     
        Consider case (1).
       If   (i) is true, that is, $t'_k \in \vbl(\tau_3,t_1,x_i)$, since $\alpha_{t_1} \sim \alpha_{t_2}$, we also have $t'_k \in \vbl(\tau,t_1,x_i)$. 
        This, combined with $t'_i \in \rbl(\tau,t_1,x_i)$, gives according to the trace semantics, $t'_k [\cfr_1] t'_i$.  Going back to (ii), 
    we can have $t'_i [\tpo_3 \cup \trf_3 \cup \cfr_3]^+ t'_k$ only if we observed one of the following:
    \begin{itemize}
        \item [(a)] There is a path from $t'_i$ to $t'_k$ in $\tau$; that is, 
        $t'_i [\tpo \cup \trf \cup \cfr]^+ t'_k$. This implies that $t'_i \notin \rbl(\tau,t_1,x_i)$ which is a contradiction.
        \item [(b)] There is no direct path from $t'_i$ to $t'_k$, however, 
        there are  transactions  $t''_j,t''_l$ such that, in $\tau$ we had  
        $ t''_j [\tpo \cup \trf \cup \cfr]^* t'_k$ and $t'_i [\tpo \cup \trf \cup \cfr]^* t''_l$, along with $t''_j \in \rbl(\tau,t_2,y_j)$,  $t''_l \in \vbl(\tau,t_2,y_j)$. Let  $ev'_j = r(p_2,t_2,y_j,v)$ be the read event which reads from $t''_j$, justifying 
        $t''_j \in \rbl(\tau,t_2,y_j)$. 
        
        \smallskip
                Since $t''_j [\tpo \cup \trf \cup \cfr]^* t'_k$ and $t'_i [\tpo \cup \trf \cup \cfr]^* t''_l$ it follows that $ t''_j [\tpo_1 \cup \trf_1 \cup \cfr_1]^* t'_k$, $t'_i [\tpo_1 \cup \trf_1 \cup \cfr_1]^* t''_l$. By our assumption, we have $t'_i \in \rbl(\tau,t_1,x_i)$ and $t'_k \in \vbl(\tau,t_1,x_i)$, which  gives us  $t'_k [\cfr_1^{x_i}] t'_i$ (already observed in the para before (i)). 
                
                Thus we get  $t''_j [\tpo_1 \cup \trf_1 \cup \cfr_1]^*
                t'_k [\cfr_1^{x_i}] t'_i [\tpo_1 \cup \trf_1 \cup \cfr_1]^*
                 t''_l$, which gives  $t''_j [\tpo_1 \cup \trf_1 \cup \cfr_1]^+ t''_l$. Hence $t''_j \notin \rbl(\tau_1,t_2,y_j)$, a  contradiction.
    \end{itemize}

    Consider case (2). If (2) is true, then we have transactions $t_3$ and $t_4$ such that $t_4 \in \tran^{\wt,y}$, $t_3 ~\trf^y_3 t_1$, 
    $t_4 [\tpo_3\cup\trf_3]^* t'_i $ and $t_3 [\tpo_3\cup\trf_3]^+ t_4$.
    We can have  $t_4 [\tpo_3\cup\trf_3]^* t'_i $ and $t_3 [\tpo_3\cup\trf_3]^+ t_4$ only if we observe following:
    \begin{itemize}
        \item[(a)] We have $t_4 [\tpo\cup\trf]^* t'_i $ and $t_3 [\tpo\cup\trf]^+ t_4$ in $\tau$. 
        This implies that $t'_i \notin \rbl(\tau,t_1,x_i)$ which is a contradiction.
        \item[(b)] We have either $t_4 [\tpo_3\cup\trf_3]^* t'_i $ or $t_3 [\tpo_3\cup\trf_3]^+ t_4$ due to the new $\trf$ and $\cfr$ relations from reads in $t_2$ in $\tau$.
        
        Suppose we have $t_4 [\tpo\cup\trf]^* t'_i $ and get $t_3 [\tpo_3\cup\trf_3]^+ t_4$ due to transactions  $t''_j,t''_l$ such that, in $\tau$ we had  
        $ t''_j [\tpo \cup \trf \cup \cfr]^* t_4$ and $t_3 [\tpo \cup \trf \cup \cfr]^* t''_l$, along with $t''_j \in \rbl(\tau$ $,t_2,y_j)$,  $t''_l \in \vbl(\tau,t_2,y_j)$. 
        Let  $ev'_j = r(p_2,t_2,y_j,v)$ be the read event which reads from $t''_j$, justifying 
        $t''_j \in \rbl(\tau$ $,t_2,y_j)$.
        
        \smallskip
                Since $t''_j [\tpo \cup \trf \cup \cfr]^* t_4$ and $t_3 [\tpo \cup \trf \cup \cfr]^* t''_l$ 
                it follows that $ t''_j [\tpo_1 \cup \trf_1 \cup \cfr_1]^* t'_k$, $t'_i [\tpo_1 \cup \trf_1 \cup \cfr_1]^* t''_l$. 
                By our assumption, we have
                $t_i ~\trf^x_i~ t_1$, $t_3 ~\trf^y~ t_1$, $t_4 [\tpo\cup\trf]^+ t'_i$ and $t_4 \in \tran^{\wt,y}$
                which  gives us  $t_4 [\cfr_1^{y}] t_3$. 
                
                Thus we get  $t''_j [\tpo_1 \cup \trf_1 \cup \cfr_1]^*
                t_4 [\cfr_1^{y}] t_3 [\tpo_1 \cup \trf_1 \cup \cfr_1]^*
                 t''_l$, which gives  $t''_j [\tpo_1 \cup \trf_1 \cup \cfr_1]^+ t''_l$. Hence $t''_j \notin \rbl(\tau_1,t_2,y_j)$, a  contradiction.
    \end{itemize}

    A similar argument for other cases.
    
    Thus, we have $\neg(t'_i [\tpo_3 \cup \trf_3 \cup \cfr_3]^+ t'_k)$. This, for all $1 \leq i \leq n$, $t'_i \in \rbl(\tau_3,t_1,x_i)$. Thus, we now have $\tau \vdash_{\tt{G}} \alpha_{t_2} \alpha_{t_1}$. 
    It remains to show that $(\alpha_{t_1} {\cdot} \alpha_{t_2})(\tau) \equiv (\alpha_{t_2}{ \cdot} \alpha_{t_1})(\tau)$. 
    
    $\mathbf{Proving}$ $\mathtt{(\alpha_{t_1} {\cdot} \alpha_{t_2})(\tau) \equiv (\alpha_{t_2}{ \cdot} \alpha_{t_1})(\tau)}$
    
    Define $\tau_4 {=} \alpha_{t_1}(\tau_3) {=} \alpha_{t_1}.\alpha_{t_2}(\tau){=}
         \langle \tran_4, \tpo_4, \trf_4 , \cfr_4 \rangle$. Recall that 
    $\tau_2{=}$$\alpha_{t_2}.\alpha_{t_1}(\tau)$ =$
      \langle \tran_2,\tpo_2, \trf_2, \cfr_2 \rangle$.
    \smallskip

    To show that $(\alpha_{t_1} \cdot \alpha_{t_2})(\tau) \equiv (\alpha_{t_2} \cdot \alpha_{t_1})(\tau)$, we show that $\tran_2=\tran_4, \tpo_2=\tpo_4, \trf_2=\trf_4$, and, for any two transactions 
    $t'_k, t'_i \in \tran_2=\tran_4$,    $t'_k [\tpo_2 \cup \trf_2 \cup \cfr_2]^+ t'_i$ iff 
    $t'_k [\tpo_4 \cup \trf_4 \cup \cfr_4]^+ t'_i$. Of these, trivially, 
    $\tran_2=\tran_4, \tpo_2=\tpo_4$ follow. Since $\alpha_{t_2} \sim \alpha_{t_1}$, we have $\trf_2=\trf_4$. It remains to prove the last condition. 
    
    \medskip 
    
   Assume that $t'_i [\trf^{x_i}_1] t_1 \wedge t'_k \in \vbl(\tau_1,t_1,x_i)$, 
   where $\tau_1=\alpha_{t_1}(\tau)$. Then we have 
   $t'_i [\trf^{x_i}_1] t_1 \wedge t'_k \in \vbl(\tau,t_1,x_i)$, and by the trace semantics, we have $t'_k [\cfr_1] t'_i$, and hence 
   $t'_k [\cfr_2] t'_i$.  To obtain $\tau_4$, we execute $t_2$ first obtaining $\tau_3$ from $\tau$, and then $t_1$. We show that $t'_k [\tpo_4 \cup \trf_4 \cup \cfr_4]^+ t'_i$, proving $t'_k [\tpo_2 \cup \trf_2 \cup \cfr_2]^+ t'_i \Rightarrow
    t'_k [\tpo_4 \cup \trf_4 \cup \cfr_4]^+ t'_i$.

    \smallskip 
    
     If we have $t'_k \in \vbl(\tau_3,t_1, x_i)$, then we are done since we  have $t'_i \in \rbl(\tau,t_1,x_i)$, hence 
          $t'_i \in \rbl(\tau_3,t_1,x_i)$, thereby obtaining 
          $t'_k \cfr_3 t'_i$, and hence $t'_k \cfr_4 t'_i$.    
          
          \smallskip 
          
           Assume $t'_k \notin \vbl(\tau_3,t_1, x_i)$.   
   Assume that on executing $t_2$ from $\tau$, we have $t''_l \in \vbl(\tau, t_2, y_j)$, and let $t''_j [\trf^{y_j}_3] t_2$.      
     Then by the trace semantics,   $t''_l [\cfr_3] t''_j$ is a new edge which gets added.  
         When we execute $t_1$ next, assume  $t''_m \in \vbl(\tau_3,t_1,x_i)$.   
         Since we have $t'_i \in \rbl(\tau_3,t_1,x_i)$, we obtain $t''_m \cfr_4 t'_i$.

     \medskip 
     By assumption, $t'_k \neq t''_m$. However, $t'_k \in \vbl(\tau,t_1,x_i)$. \\
     $t'_k \notin  \vbl(\tau_3,t_1,x_i)$ points to some transaction that 
     blocked the visibility of $t'_k$ in $\tau_3$, by happening after $t'_k$, and which is in $\vbl(\tau_3,t_1,x_i)$.  
     
     \begin{enumerate}
     	\item This blocking transaction could be $t''_m$. If this is the case, we have $t'_k [\tpo_3 \cup \trf_3 \cup \cfr_3]^+ t''_m [\cfr_4] t'_i$, and hence  $t'_k [\tpo_4 \cup \trf_4 \cup \cfr_4]^+ t''_m [\cfr_4] t'_i$. 
     	\item The other possibility is that we have  $t''_j$
     	happens before $t''_m$, and $t''_l$ happens after $t'_k$ blocking 
     	$t'_k$ from being in $\vbl(\tau_3,t_1,x_i)$. That is,
     	     $t'_k [\tpo \cup \trf \cup \cfr]^* t''_l$ and 
     	     $t''_j [\tpo_4 \cup \trf_4 \cup \cfr_4]^* t''_m$.   
         Then we obtain $t'_k [\tpo_4 \cup \trf_4 \cup \cfr_4]^* t''_l [\cfr^{x_j}_4] t''_j [\tpo_4 \cup \trf_4 \cup \cfr_4]^* t''_m [\cfr_4] t'_i$. 
     
     \end{enumerate}
        
     Thus, we obtain $t'_k [\tpo_4 \cup \trf_4 \cup \cfr_4]^+ t'_i$ as desired.

\end{itemize}
 The converse direction, that is, whenever $t'_k [\tpo_4 \cup \trf_4 \cup \cfr_4]^+ t'_i$, we also have $t'_k [\tpo_2 \cup \trf_2 \cup \cfr_2]^+ t'_i$ is proved on similar lines. 
    
\end{proof}
\begin{lemma}
If $\tau_1 \equiv \tau_2$, $\alpha_{t_1} \sim \alpha_{t_2}$, and $\tau_1 \vdash_{\tt{G}} (\alpha_{t_1} \cdot \alpha_{t_2})$,  then (i) $\tau_2 \vdash_{\tt{G}} (\alpha_{t_2} \cdot \alpha_{t_1})$ and (ii) $(\alpha_{t_1} \cdot \alpha_{t_2})(\tau_1) \equiv (\alpha_{t_2} \cdot \alpha_{t_1})(\tau_2)$.
\label{lemma g-3}
\end{lemma}
\begin{proof}
Follows from Lemma \ref{lemma g-1} and Lemma \ref{lemma g-2}.
\end{proof}

From Lemma \ref{lemma g-3}, we get following lemma.
\begin{lemma}
If $\tau_1 \equiv \tau_2$, $\pi_{\tt{T^1}} \sim \pi_{\tt{T^2}}$, and $\tau_1 \vdash_{\tt{G}} \pi_{\tt{T^1}}$ then $\tau_2 \vdash_{\tt{G}} \pi_{\tt{T^2}}$ and $\pi_{\tt{T^1}}(\tau_1) \equiv \pi_{\tt{T^2}}(\tau_2)$.
\label{lemma g-4}
\end{lemma}
 
\begin{lemma}
If $\tau \vdash_{\tt{G_T}} \pi_{\tt{T}}$, $\tau \vdash_{\tt{G}} \alpha_t$ then $\pi_{\tt{T}} = \pi_{\tt{T^1}} \cdot \alpha_t \cdot \pi_{\tt{T^2}}$ for some $\pi_{\tt{T^1}}$ and $\pi_{\tt{T^2}}$ where $\pi_{\tt{T^1}}$ is $t$-free.
\label{lemma g-5}
\end{lemma}
Consider observables $\pi_{\tt{T}} = \alpha_{t_1},\alpha_{t_2} \dots \alpha_{t_n}$. 
We write $\pi_{\tt{T}} \lessapprox \pi'_{\tt{T}}$ to represent that $\pi'_{\tt{T}} = \pi_{\tt{T^0}} \cdot \alpha_{t_1} \cdot \pi_{\tt{T^1}} 
\cdot \alpha_{t_2} \cdot \pi_{\tt{T^2}} \dots \alpha_{t_n} \cdot \pi_{\tt{T^n}}$. 
In other words, $\pi_{\tt{T}}$ occurs as a non-contiguous subsequence in $\pi'_{\tt{T}}$. For such a $\pi_{\tt{T}}, \pi'_{\tt{T}}$, 
we define $\pi'_{\tt{T}} \oslash \pi_{\tt{T}} := \pi_{\tt{T^0}} \cdot \pi_{\tt{T^1}} \dots \pi_{\tt{T^n}}$. 
Since elements of $\pi_{\tt{T}}$ and $\pi'_{\tt{T}}$ are distinct, operation $\oslash$ is well defined. 
Let $\pi_{\tt{T}}[i]$ denote the $i$th observable 
in the observation sequence $\pi_{\tt{T}}$, and 
let  $|\pi_{\tt{T}}|$ denote  the number of observables 
in $\pi_{\tt{T}}$.

Let $\alpha_t = \pi_{\tt{T}}[i]$ for some $i: 1 \le i \le |\pi_{\tt{T}}|$. 
We define $\tt{Pre(\pi_{\tt{T}},\alpha_t)}$ as a subsequence $\pi'_{\tt{T}}$ of $\pi_{\tt{T}}$ such that  
(i) $\alpha_t \in \pi'_{\tt{T}}$,
(ii) $\alpha_{t_j} = \pi_{\tt{T}}[j] \in \pi'_{\tt{T}}$ for some $j: 1 \le j < i$ iff there exists $k: j < k \le i$ such that $\alpha_{t_k}=\pi_{\tt{T}}[k]  \in \pi'_{\tt{T}}$ and $\neg(\alpha_{t_j} \sim \alpha_{t_k})$. 
Thus, $\tt{Pre(\pi_{\tt{T}},\alpha_t)}$ consists of $\alpha_t$ and all 
$\alpha_{t'}$ appearing before $\alpha_t$ in $\pi_{\tt{T}}$ such that 
$t' [\tpo \cup \trf]^+ t$.

%
%
\begin{lemma}
If $\pi_{\tt{T^1}} = \tt{Pre(\pi_{\tt{T}},\alpha_t)}$ and $\pi_{\tt{T^2}} = \pi_{\tt{T}} \oslash \pi_{\tt{T^1}}$ then $\pi_{\tt{T}} \approx \pi_{\tt{T^1}} \cdot \pi_{\tt{T^2}}$.
\label{lemma g-6}
\end{lemma}
\begin{proof}
The proof is trivial since we can always execute in order, 
the $[\tpo \cup \trf]^+$ predecessors of $\alpha_t$ from $\tt{Pre(\pi_{\tt{T}},\alpha_t)}$, then $\alpha_t$, then the  observables 
in $\tt{Pre(\pi_{\tt{T}},\alpha_t)}$ which are independent from $\alpha_t$, followed by the suffix of $\pi_{\tt{T}}$ after $\alpha_t$.

\end{proof}

\begin{lemma}
If $\pi_{\tt{T}}= \pi'_{\tt{T}} \cdot \pi''_T$ and $\alpha_t \in \pi'_{\tt{T}}$ then $\tt{Pre(\pi_{\tt{T}},\alpha_t)} = \tt{Pre(\pi'_{\tt{T}},\alpha_t)}$.
\label{lemma g-7}
\end{lemma}
\begin{proof}
The proof is trivial once again, since 	all 
the transactions which are $[\tpo \cup \trf]^+ t$  are in the prefix  $\pi'_{\tt{T}}$.
\end{proof}

\begin{lemma}
If $\tau \vdash_{\tt{G}} \pi_{\tt{T}}$ then $\tau \vdash_{\tt{G_T}} \pi_{\tt{T}} \cdot \pi_{\tt{T^2}}$.
\label{lemma g-8}
\end{lemma}
\begin{proof}
This simply follows from the fact that we can extend the observation sequence 
$\pi_{\tt{T}}$ to obtain a terminal configuration, since the $\ccvt$ DPOR algorithm generates traces corresponding to terminating runs.   	
\end{proof}

\begin{lemma}
Consider a trace trace $\tau$ such that $\tau \models \ccvt$, $\tau \vdash_{\tt{G}} \pi_{\tt{T}} \cdot \alpha_t$. Let each read event $ev_i = r_i(p,t,x_i,v)$ in $\alpha_t$ read from some transaction $t_i$, and let $\pi_{\tt{T}}$ be $t$-free.  Then $\tau \vdash_{\tt{G}} \alpha'_t \cdot \pi_{\tt{T}}$, where $\alpha'_t$ is the same as  $\alpha_t$, with the exception that the sources of its read events can be different. That is, 
 each read event $ev_i = r_i(p,t,x_i,v) \in \alpha'_t$ can read from some transaction $t'_i \neq t_i$.
\label{lemma g-9}
\end{lemma}
\begin{proof}
Let $\pi_{\tt{T}} = \alpha_{t_1} \alpha_{t_2} \dots \alpha_{t_n}$ and $\tau_0 \xrightarrow[]{\alpha_{t_1}}_{\ccvt{-}\sat} \tau_1 \xrightarrow[]{\alpha_{t_2}}_{\ccvt{-}\sat} \dots \xrightarrow[]{\alpha_{t_n}}_{\ccvt{-}\sat} \tau_n \xrightarrow[]{\alpha_t}_{\ccvt{-}\sat} \tau_{n+1}$, where $\tau_0 = \tau$. Let $\tau_i = \langle \tran_i, \tpo_i, \trf_i, \cfr_i \rangle$. 
Let there be  $m$ read events in transaction $t$. Keeping in mind what we want to prove, where we want to execute $t$ first followed by $\pi_{\tt{T}}$, and obtain an execution $\alpha'_t \pi_{\tt{T}}$, we do the following.

For each read event $ev_i = r_i(p,t,x_i,v)$ we define $t'_i \in \tran^{wt,x_i}$ such that 
$(i)$ $ t'_i \in \rbl(\tau_0,t,x_i)$, and 
$(ii)$ there is no $t''_i \in \rbl(\tau_0,t,x_i)$ where $t'_i [\tpo_n \cup \trf_n \cup \cfr_n] t''_i$.  Note that this is possible since 
 $\pi_{\tt{T}}$ is $t$-free, so all the observables $\alpha_{t_i}$ occurring in $\pi_{\tt{T}}$ are such that $t_i$ is issued in a process other than that of $t$. Thus, when $\alpha_t$ is enabled, we can choose any of the 
 writes done earlier,  this fact is consistent with the $\ccvt$ semantics, since from Lemmas \ref{lem:f3}-\ref{lem:f4} we know $\tau_n \models \ccvt$.
Since $\tau_n \models \ccvt$ such a $t'_i \in  \rbl(\tau_0,t,x_i)$ exists for each $ev_i$.

\medskip 
Next we define a sequence of traces which can give the execution 
$\alpha'_t \pi_{\tt{T}}$. 
Define a sequence of traces $\tau'_0, \tau'_1, \dots \tau'_n$ where $\tau'_j = \langle \tran'_j, \tpo'_j, \trf'_j, \cfr'_j \rangle$ is such that 
\begin{enumerate}
	\item $\tau \vdash_{\tt{G}} {\alpha'_t}$ and $\alpha'_t(\tau)=\tau'_0$,
	\item $\tau'_j \vdash_{\tt{G}} \alpha_{t_{j+1}}$ and  $\tau'_{j+1} = \alpha_{t_{j+1}}(\tau'_j)$ for all $j: 0 \le j \le n$.
\end{enumerate}

\begin{itemize}
\item  Define $\tran'_i = \tran_i \cup \{t\}$ for all $0 \leq i \leq n$, 
\item 
Define $\tpo' = \{t' [\tpo] t \mid t'$ is a transaction in $\tau$, in the same process as $t\}$, and $\tpo'_i = \tpo_i \cup \tpo'$ for all $0 \leq i \leq n$, 

\item  $\trf' = \bigcup_{i:1\le i \le m} t'_i [\trf^{x_i}] t$ (reads from relation corresponding to each read event $ev_i \in \alpha'_t$), and 
 $\trf'_i = \trf_i \cup \trf'$ for all $i : 1 \le i \le n$, 
 
\item  $\cfr' = \bigcup_{j:1\le j \le m} \cfr'_j$ (each  $\cfr'_j$ corresponds to the updated $\cfr$ relation because of the read transition $\xrightarrow[]{\readact(t,t'_j)}_{\ccvt{-}\sat}$). 	
\smallskip 

For $1 \leq i \leq n$, we define $\cfr'_i$ inductively as $\cfr'_i = \cfr_i \cup \cfr'$ and show that $\tau'_i \vdash_{\tt{G}} \alpha_{t_{i+1}}$ and  $\tau'_{i+1} = \alpha_{t_{i+1}}(\tau'_i)$ holds good. 
 First define $\cfr'_0 = \cfr_0 \cup \cfr'$. Then define  $\cfr'_i$ to be the coherence order corresponding to $\alpha_{t_{i}}(\tau'_{i-1})$  where $\tau'_{i-1}=\langle \tran'_{{i-1}}, \tpo_{i-1}, \trf_{i-1}, \cfr_{i-1} \rangle$ for all $i: 1 \le i \le n$. 

 \end{itemize}

\medskip 
\noindent{\bf{Base case}}. The base case $\tau \vdash_{\tt{G}} {\alpha'_t}$ and $\alpha'_t(\tau)=\tau'_0$,
   holds trivially, by construction. 

\smallskip 

\noindent{\bf{Inductive hypothesis}}. Assume that $\tau'_j \vdash_{\tt{G}} \alpha_{t_{j+1}}$ and  $\tau'_{j+1} = \alpha_{t_{j+1}}(\tau'_j)$ for $0 \leq j \leq i-1$. 

\smallskip 

 We have to prove that 
$\tau'_i \vdash_{\tt{G}} \alpha_{t_{i+1}}$ and  $\tau'_{i+1} = \alpha_{t_{i+1}}(\tau'_i)$. 
\begin{enumerate}
	\item If $t_{i+1} \in \tran^{wt}$ and $t_{i+1} \notin \tran^{rd}$,  then the proof holds trivially from the inductive hypothesis. 
	\item Consider now $t_{i+1} \in \tran^{rd}$. Consider a read event 
	$ev^{i+1} = r(p,t_{i+1},x,v) \in \alpha_{t_{i+1}}$ which was reading from 
	transaction $t_x$ in $\pi_{\tt{T}}$. That is, we had $t_x \in \rbl(\tau_i, t_{i+1},x)$. 	If $t_x \in \rbl(\tau'_i, t_{i+1},x)$, then we are done, since we can simply extend the run from the inductive hypothesis. 
	
	\smallskip 
	
	Assume otherwise. That is, $t_x \notin \rbl(\tau'_i, t_{i+1},x)$.

	Since  $t_x$  $\in$   $\rbl(\tau_i, t_{i+1},x)$, we know that there are some blocking transactions in the new path which prevents $t_x$ from being readable. Basically, we have the blocking transactions since $t$ is moved before $\pi_{\tt{T}}$. 	
	
	\begin{itemize}
		\item Consider $tr_1 \in \vbl(\tau,t,y)$ for some variable $y$ such that,  in $\pi_{\tt{T}}.\alpha_t$, we had $t_x [\tpo_i \cup \trf_i \cup \cfr_i]^* tr_1$. This path
		 is possible since $\alpha_t$ comes last in in $\pi_{\tt{T}}.\alpha_t$,  after  $\alpha_{t_{i+1}}$. 
		\item  Consider $tr_2 \in \vbl(\tau_i,t_{i+1},x)$.  Since $t_x$ $\in$ $\rbl(\tau_i, t_{i+1},x)$, we have  $tr_2 [\cfr^x_{i+1}] t_x$. 
Now, in $\alpha'_t.\pi_{\tt{T}}$, assume $tr_3 [\trf'^y_0] t$, 
such that we have a path from $tr_3$ to $t_x$, as  $tr_3 [\tpo_i \cup \trf_i \cup \cfr_i]^* tr_2$. 
	\end{itemize}
	  		  Hence it follows that $tr_3 [\tpo_i \cup \trf_i \cup \cfr_i]^* tr_2$ 
$[\cfr^x_{i+1}] t_x$ $[\tpo_i \cup \trf_i \cup \cfr_i]^* tr_1$. 
Hence we have $tr_3 [\tpo_{i+1} \cup \trf_{i+1} \cup \cfr_{i+1}]^+ tr_1$, i.e 
$tr_3 [\tpo_n \cup \trf_n \cup \cfr_n]^+ tr_1$. This leads to the contradiction since $tr_1 \in \vbl(\tau,t,y)$ and  $tr_3 [\trf'^y_0] t$.  

Hence, $t_x$  $\in$   $\rbl(\tau_i, t_{i+1},x)$ and we can extend the run 
from the inductive hypothesis obtaining $\tau'_j \vdash_{\tt{G}} \alpha_{t_{j+1}}$,  $\tau'_{j+1} = \alpha_{t_{j+1}}(\tau'_j)$,  
for $0 \leq j \leq i-1$,
and 
$\tau'_i \vdash_{\tt{G}} \alpha_{t_{i+1}}$.

\end{enumerate}

Now we prove that $\cfr'_{i+1} = \cfr_{i+1} \cup \cfr'$. 
Assume we have $(tr1 , tr2) \in \cfr_{i+1}$. We show that $(tr1 , tr2) \in \cfr'_{i+1}$. We have the following cases:

\begin{itemize}
    \item[(i)] $(tr1 , tr2) \in \cfr_i$. From the inductive hypothesis it follows that $(tr1 , tr2) \in \cfr'_{i}$ and hence in $\cfr'_{i+1}$.
    \item[(ii)] $(tr1 , tr2) \notin \cfr_i$ and some read event $ev = r(p,t_{i+1},x,v)$ from $\alpha_{t_{i+1}}$ reads from $tr2$.  That is, $tr1,tr2 \in \rbl(\tau_i, t_{i+1},x)$ and $tr1 \in \vbl(\tau_i,t_{i+1},x)$.
    Assume $(tr1 , tr2) \notin \cfr'_{i+1}$.
    Since $(tr1, tr2) \notin \cfr'_{i+1}$, it follows that there exists  $tr3 \in \vbl(\tau'_i,t_{i+1},x)$ such that $tr2 [\tpo'_i \cup \trf'_i \cup \cfr'_i] tr3$. This means $tr2 \in \rbl(\tau_i,t_{i+1},x)$, but 
    $tr2 \notin \rbl(\tau'_i,t_{i+1},x)$. However,  as proved earlier, this leads to a contradiction. 

\end{itemize}
Thus, we have shown that $\tau \vdash_{\tt{G}} \alpha'_t \alpha_{t_1} \dots \alpha_{t_i}$ 
is such that $\alpha'_t \alpha_{t_1} \dots \alpha_{t_i}(\tau)=\tau'_i$ for all 
$0 \leq i \leq n$. When $i=n$ we obtain 
$\alpha'_t \pi_{\tt{T}}(\tau)=\tau'_n$, or $\tau \vdash_{\tt{G}} \alpha'_t \pi_{\tt{T}}$. 
\end{proof}


\begin{definition}[Linearization of  a Trace]
A observation sequence $\pi_{\tt{T}}$ is a \emph{linearization} of a trace $\tau = \langle \tran, \tpo , \trf, \cfr \rangle$ if $(i)$ $\pi_{\tt{T}}$ has the same transactions as $\tau$ and $(ii)$ $\pi_{\tt{T}}$ follows the ($\tpo \cup \trf)$ relation.
	
\end{definition}

We say that our DPOR algorithm \emph{generates an observation sequence} $\pi_{\tt{T}}$ from state $\tau$ where $\tau$ is a trace 
if 
 it invokes $\explore$ with parameters $\ccvt,\tau, \pi'$, where $\pi'$ is a \emph{linearization} of $\tau$, and generates  a sequence of recursive calls to $\explore$ resulting in $\pi_{\tt{T}}$.

\begin{lemma}
If $\tau \vdash_{\tt{G_T}} \pi_{\tt{T}}$ then, 
the DPOR algorithm  generates $\pi'_{\tt{T}}$ from state $\tau$ 
for some $\pi'_{\tt{T}} \approx \pi_{\tt{T}}$.
\label{lemma g-10}
\end{lemma}
\begin{proof}
We use induction on $|\pi_{\tt{T}}|$. If $\tau$ is $terminal$, then the  proof is trivial.
Assume that we have $\tau \vdash_{\tt{G_T}} \pi_{\tt{T}}$.
Assume that $\tau \vdash_{\tt{G}} \alpha_t$. It follows that $\tau \vdash_{\tt{G}} \alpha_t$.
Using Lemma \ref{lemma g-5} we get $\pi_{\tt{T}} = \pi_{\tt{T^1}} \cdot \alpha_t \cdot \pi_{\tt{T^2}}$, where $\pi_{\tt{T^1}}$ is $t$-free. 
We consider the following two cases:
\begin{itemize}
    \item In the first case, we assume that all read events in $t$ read from 
    transactions in $\tau$. So in this case, $\pi_{\tt{T^1}}$ can be empty. 
        Assume there are $n$ read events in $t$ and each read event $ev_i = r(p,t,x_i,v)$ 
    reads from transactions $t_i \in \tau$ for all $i: 1 \le i \le n$. 
In this case the DPOR algorithm will let each read event $ev_i$ read from all possible transactions $t' \in \rbl(\tau,t,x_i)$, including $t_i$.
    \item In the second case, assume that there exists at least one read event $ev' = r(p,t,x,v) \in \alpha_t$ which reads from a transaction $t' \in \pi_{\tt{T^1}}$ (hence, $\pi_{\tt{T^1}}$ is non empty).
    \smallskip 
     
    From Lemma \ref{lemma g-9}, we know that $\tau
    \vdash_{\tt{G}} \alpha'_t \cdot \pi_{\tt{T^1}}$. 
    From Lemma \ref{lemma g-8}, it follows that $\tau
    \vdash_{\tt{G_T}} \alpha'_t \cdot \pi_{\tt{T^1}} \cdot \pi_{\tt{T^4}}$, for some
    $\pi_{\tt{T^4}}$.
    Hence there exists $\tau'$ such that 
    $\tau' = \alpha'_t (\tau)$ and $\tau' \vdash_{\tt{G_T}} \pi_{\tt{T^1}} \cdot
    \pi_{\tt{T^4}}$. 
    
    Since $|\pi_{\tt{T^1}} \cdot
    \pi_{\tt{T^4}}| < |\pi_{\tt{T}}$, we can use the inductive hypothesis. It follows that the DPOR algorithm generates from state $\tau'$, 
    the observation sequence  $\pi_{\tt{T^5}}$ such that $\pi_{\tt{T^5}} \approx \pi_{\tt{T^1}} 
    \cdot \pi_{\tt{T^4}}$. 
    
    \smallskip 
    
    Let  $\pi_{\tt{T^3}}.\alpha_t = \tt{Pre(\pi_{\tt{T^1}}, \alpha_t)}$. Then  $\pi_{\tt{T^5}}
    \approx \pi_{\tt{T^1}} \cdot \pi_{\tt{T^4}}$ implies that  $\pi_{\tt{T^3}}
    \lessapprox \pi_{\tt{T^5}}$ (by applying Lemma \ref{lemma g-7}).
    Let $\pi_{\tt{T^6}} \approx \pi_{\tt{T}} \oslash (\pi_{\tt{T^3}} \cdot \alpha_t)$.
    By applying Lemma \ref{lemma g-6}, we get $\pi_{\tt{T}} \approx \pi_{\tt{T^3}} \cdot \alpha_t \cdot \pi_{\tt{T^6}}$.
    Since $\tau \vdash_{\tt{G_T}} \pi_{\tt{T}}$ and 
    $\pi_{\tt{T}} \approx \pi_{\tt{T^3}} \cdot \alpha_t \cdot \pi_{\tt{T^6}}$,
    by applying Lemma \ref{lemma g-3}, we get $\tau
    \vdash_{\tt{G_T}} \pi_{\tt{T^3}} \cdot \alpha_t \cdot \pi_{\tt{T^6}}$.
    Let $\tau' = (\pi_{\tt{T^3}} \cdot \alpha_t) (\tau)$. 
    Since $\tau \vdash_{\tt{G_T}} \pi_{\tt{T^3}} 
    \cdot \alpha_t \cdot \pi_{\tt{T^6}}$, we have $\tau'
    \vdash_{\tt{G_T}} \pi_{\tt{T^6}}$. 
    From inductive hypothesis it follows that $\tau'$ generates $\pi_{\tt{T^7}}$ such that $\pi_{\tt{T^7}} \approx \pi_{\tt{T^6}}$.\\
    In other words, $\tau$ generates $\pi_{\tt{T^3}} 
    \cdot \alpha_t \cdot \pi_{\tt{T^7}}$ where $\pi_{\tt{T}} \approx \pi_{\tt{T^3}} 
    \cdot \alpha_t \cdot \pi_{\tt{T^7}}$.
\end{itemize}

\end{proof}

%% file: cc-proofs.tex
\newpage
\centerline{\bf{Weak Causal Consistency $\cc$}} 
 \label{app:cc}

\section{$\cc$ Trace Semantics}
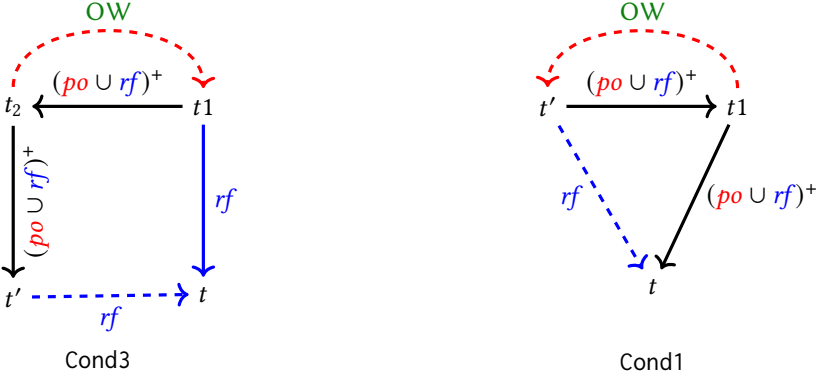
\begin{figure}[H]
        \begin{tikzpicture}[node distance=10mm , thick, main/.style= {draw,circle},];
        
        \node[] (3at2) [right= 5mm of t1] {$t_2$};
        \node[] (3at1) [right= 20mm of 3at2] {$t1$};
        \node[] (3at) [below=20mm of 3at1] {$t$};
        \node[] (3at') [below=20mm of 3at2] {$t'$};
        \node[] (31) [below=5mm of 3at] {};
        \node[] (3acap) [left=7mm of 31] {$\tt{Cond3}$};

        \draw[->,black,line width=1.2pt] (3at1) -- node[above] {$(\tpo \cup \trf)^+$} (3at2);
         \draw[->,blue,line width=1.2pt] (3at1) -- node[right] {$\trf$} (3at);
        \draw[dashed , ->,blue,line width=1.2pt] (3at') -- node[below] {$\trf$} (3at);
        \draw[dashed , ->,red,line width=1.2pt] (3at2) to [out=90, in=90] node[above] {$\ow$} (3at1);
        \draw[->,black,line width=1.2pt] (3at2) -- node[rotate=90,below] {$(\tpo \cup \trf)^+$} (3at');

        \node[] (1t') [right= 5mm of 2at1] {$t'$};
        \node[] (1t1) [right= 20mm of 1t'] {$t1$};
         \node[] (111) [below=20mm of 1t1] {};
        \node[] (1t) [left=8mm of 111] {$t$};
        \node[] (1cap) [below=5mm of 1t] {$\tt{Cond1}$};

        \draw[->,black,line width=1.2pt] (1t') -- node[above] {$(\tpo \cup \trf)^+$} (1t1);
        \draw[->,black,line width=1.2pt] (1t1) -- node[right] {$(\tpo \cup \trf)^+$} (1t);
        \draw[dashed , ->,blue,line width=1.2pt] (1t') -- node[left=1mm] {$\trf$} (1t);
        \draw[dashed , ->,red,line width=1.2pt] (1t1) to [out=90, in=90] node[above] {$\ow$} (1t');
        \end{tikzpicture}
        \caption{Readability for $\cc$}
        \label{read-cc}
    \end{figure}

\subsection{Readability for $\cc$}
\label{cc:r}
For $X=\cc$, $\rbl(\tau^t_X, t, x)$ is defined as the set of all transactions  $t' \in \tran^{\wt,x}$ s.t. 
\begin{itemize}
	\item[$\tt{cond1}$] 
	 there is no transaction $t'' \in \tran^{\wt,x}$  s.t. $t'~\co~t''~\co~t$ in $\tau^t_X$. 
Allowing $t'~\trf~t$ (wrt $x$) in the presence of such a $t''$ gives $t''~\ow~t'$ and $\cccyc$.
	\item[$\tt{cond3a}$]  In case $t' \in \tran^{\wt,y}$ ($t'$ also writes on some $y \neq x$) and $t \in \tran^{\rd,y}$,  
there is no transaction $t'' \in \tran^{\wt,y}$ such that $t'' ~\trf~t$ (wrt $y$)  and $t''~ \co ~t'$.  Note that having such a $t''$, and 
allowing $t'~\trf~t$ (wrt $x$) 
results in $t'~\ow~t''$ (wrt $y$)  and $\cccyc$.
\item[$\tt{cond3b}$]   if $t \in \tran^{\rd,y}$ (for some $y \neq x$) there are no transactions $t_1, t_2 \in \tran^{\wt,y}$ such that  
$t_1 ~\trf~ t$ (wrt $y$),  $t_1~\co~t_2~\co~t'$. Assuming we have  
this, then allowing $t'~\trf~t$ (wrt $x$) gives $t_1~\co~t_2~\co~t$ 
and hence $t_2~\ow~t_1$ (wrt $y$) and $\cccyc$. 
\end{itemize}	

Figure \ref{read-cc} explains cycles created after violating conditions $\tt{cond1}$ and $\tt{cond4}$.
$\tt{cond3} $ covers $\tt{cond3a} $  and $\tt{cond3b} $.
For $\tt{cond3a}$,  $t_2$ coincides with $t'$.

\begin{lemma}
    Given trace $\tau$, a transaction $t$ and a variable $x$, we can construct the set  $\rbl(\tau, t, x)$ in polynomial time.
    \label{lemma:cc-poly}
\end{lemma}
\begin{proof}
    We prove that set $\rbl(\tau, t, x)$ can be computed in polynomial time, by designing an algorithm to generate set $\rbl(\tau, t, x)$. The algorithm consists of the following  steps:
    \begin{itemize}
        \item[(i)] First we compute the  transitive closure of the relation 
        $[\tpo \cup \trf]$, i.e, we compute $[\tpo \cup \trf]^+$.  
        This will take $O(|\tran|^3)$ time.
        \item[(ii)] We compute the set $\tran'=\{t'\mid t' [\tpo \cup \trf]^+ t\}$. This takes $O(|\tran|)$ time.
        \item[(iii)] For each transaction $t' \in \tran^{wt,x}$, we perform following checks:
        \begin{itemize}
            \item Check $\mathtt{cond 1 :} $ Check whether there exists a transaction $t'' \in \tran' \cap \tran^{wt,x}$            
            with $t'$ $[\tpo  \cup \trf]^+$ $t''$. 
            If there is no such $t''$ then proceed to $\mathtt{cond 3}$. 
            If we find such $t''$, then $t' \notin \rbl(\tau , t, x)$. 
            This step will take $O(|\tran|^2)$ time.
            
            \item Check $\mathtt{cond 3 :}$ 
            
            We first check case a of $\mathtt{cond 3}$. Check whether there exists a transaction $t''$ such that $t'' [\trf^y] t$ , $t' \in \tran^{wt,y}$, and $t'' [\tpo \cup \trf]^+ t'$. 
            If there is no such $t''$ then proceed to case two of $\mathtt{cond 3}$. 
            If we find such $t''$, then $t' \notin \rbl(\tau , t, x)$. 
            This step will take $O(|\tran|)$ time.
            
             Next, we  check case b of $\mathtt{cond 3}$. Check whether there are transactions $t''$ and $t'''$ such that $t'' [\trf^y] t$, $t''' \in \tran^{wt,y}$, $t''' [\tpo \cup \trf]^+ t'$ and $t'' [\tpo \cup \trf]^+ t'''$. 
            If there are no such $t''$ and $t'''$ then $t' \in \rbl(\tau , t, x)$. 
            If we find such $t''$ and $t'''$, then $t' \notin \rbl(\tau , t, x)$. 
            This step will take $O(|\tran|^2)$ time.
        \end{itemize}
    \end{itemize}
\end{proof}

\subsection{Properties of Trace Semantics}
\label{app:ccful}

\begin{lemma}
    If $\tau_1 \models \cc$  and $\tau_1$ $\xrightarrow[]{\beginact(t)}_{\cc{-}\sat}$ $\tau_2$ then $\tau_2 \models \cc$.
    \label{lem:cc-3}
\end{lemma}
\begin{proof}
The proof follows trivially since we do not change the reads from relations, and the transaction $t$ added has no successors in $(\tpo_2 \cup \trf_2)$. 
Thus, if $\tau_1 \models \cc$, so does $\tau_2$.  
\end{proof}

\begin{lemma}
    If $\tau_1 \models \cc$  and $\tau_1$ $\xrightarrow[]{\readact(t,t')}_{\cc{-}\sat}$ $\tau_2$ then $\tau_2 \models \cc$.
    \label{lem:cc-4}
\end{lemma}

\begin{proof}
Let $\tau_1 = \langle \tran_{1} , \tpo_1 , \trf_1 , \ow_1 \rangle$, 
$\tau_2 = \langle \tran_{2} , \tpo_2 , \trf_2 , \ow_2 \rangle$\footnote{following trace semantics, $\ow$ relation will always be empty. },and 
$ev$ = $r(p,t,x,v)$ be a read event in transaction $t \in \tran_1$ such that $ev$ reads value of $x$ from the transaction $t'$. 
Suppose $\tau_2 \nvDash \cc$, and  that $\tpo_2 \cup \trf_2 \cup \ow_2$ is cyclic. 
Since $\tau_1 \models \cc$, and $t$ has no outgoing edges, just by adding $(t',t) \in \trf_2$ we can not get cycle in ($\tpo \cup \trf \cup \ow$).
Now, suppose that are transactions $t_1\in\tran^{W,y}, t_2$ and $t_3$ such that  $t_1$ [$\tpo_2 \cup \trf_2$]$^+$ $t_3$ and $t_2$ [$\trf_2^y$] $t_3$. 
We need to show that in trace $\tau_2$, we have $\neg$$(t_2 [\tpo \cup \trf]^+ t_1)$. 
We consider the following cases: 
\begin{itemize}
    \item [$\bullet$] $t_3 = t$, $t_2 = t'$.
    From the definition of $\xrightarrow[]{\readact(t,t')}_{\cc{-}\sat}$, 
    we  have $t_2 \in \rbl(\tau_1,t_3,x)$. 
    So we have $\neg$$(t_2 [\tpo_1 \cup \trf_1]^+ t_1)$.
    Since $t$ has no successor in $(\tpo_2 \cup \trf_2)$, it follows that in $\tau_2$ we will have $\neg$$(t_2 [\tpo_2 \cup \trf_2]^+ t_1)$.
    Thus, $\tau_2 \models \cc$. 
    
    \item [$\bullet$] $t_3 \neq t$, $t_2 \neq t'$.
    Since $\trf_2 = \trf_1 \cup t' [\trf] t$ and $t_1$ [$\tpo_2 \cup \trf_2$]$^+$ $t_3$ and $t_2$ [$\trf_2^y$] $t_3$, it follows that $t_1$ [$\tpo_1 \cup \trf_1$]$^+$ $t_3$ and $t_2$ [$\trf_1^y$] $t_3$.
    Since $\tau_1 \models \cc$ , we have $\neg$$(t_2 [\tpo_1 \cup \trf_1]^+ t_1)$.
    Since only new edge added in $\tau_2$ is $t' [\trf^x] t$, we will have $\neg$$(t_2 [\tpo_1 \cup \trf_1]^+ t_1)$ in $\tau_2$.
    Thus, $\tau_2 \models \cc$ is.
    \end{itemize}
\end{proof}

Define $\big[[ \sigma]\big]^{\cc{-}\sat}$ as the set of  traces  generated  using $\xrightarrow[]{}_{\cc{-}\sat}$ transitions, starting from an empty trace $\tau_{\o}$. 
Consider a terminal trace $\tau$ generated by $\cc$ DPOR algorithm starting from the empty trace $\tau_\emptyset$.
That is, there is a sequence  
$\tau_0 \xrightarrow[]{}_{\cc{-}\sat} 
\tau_1 \xrightarrow[]{}_{\cc{-}\sat}  
\tau_2 \xrightarrow[]{}_{\cc{-}\sat} \dots 
 \xrightarrow[]{}_{\cc{-}\sat}  \tau_n$ with 
 $\tau_0=\tau_{\emptyset}$, 
 and $\tau_n=\tau$. 
 Since $\tau_{\emptyset}$ is a empty we have $\tau_{\emptyset} \models \cc$, 
 it follows by  
 Lemmas \ref{lem:f3} and \ref{lem:f4} that $\tau \models \cc$. 
 
 Hence, for each trace $\tau \in \big[[ \sigma]\big]^{\cc{-}\sat}$  we have $\tau \models \cc$.

\subsection{Completeness of $\cc$ DPOR}
\label{app:complete-cc}
Proof for completeness of DPOR for $\cc$ follows the same schema as the completeness proof for $\cc$-DPOR. The only difference 
is in Lemmas \ref{lemma g-2} and \ref{lemma g-9}. 
We prove Lemmas \ref{lemma cc-comp-2} and \ref{lemma cc-comp-9}, which are the $\cc$ counterpart of Lemmas \ref{lemma g-2} and \ref{lemma g-9}, respectively.

\begin{lemma}
If $\tau \vdash_{\tt{G}} \alpha_{t_1} \cdot \alpha_{t_2}$ and $\alpha_{t_1} {\sim} \alpha_{t_2}$ then $\tau \vdash_{\tt{G}} \alpha_{t_2} {\cdot} \alpha_{t_1}$ and $(\alpha_{t_1} {\cdot} \alpha_{t_2})(\tau) \equiv (\alpha_{t_2}{ \cdot} \alpha_{t_1})(\tau)$.
\label{lemma cc-comp-2}
\end{lemma}

\begin{proof}
Let $\tau = \langle \tran , \tpo, \trf , \ow \rangle$ be a trace and let $t_1$ be a transaction issued in process $p_1$, and $t_2$ be a transaction issued in process $p_2$, with $p_1 \neq p_2$. 
Assume $\tau \vdash_{\tt{G}} \alpha_{t_1} \cdot \alpha_{t_2}$, with $\alpha_{t_1} {\sim} \alpha_{t_2}$. 

We consider the following cases.
\begin{itemize}
    \item [$\bullet$] $t_1,t_2 \notin \tran^r$. In this case $\tau \vdash_{\tt{G}} \alpha_{t_2} \cdot \alpha_{t_1}$ holds trivially.
    
    $(\alpha_{t_1} \cdot \alpha_{t_2})(\tau) = (\alpha_{t_2} \cdot \alpha_{t_1})(\tau) = \langle \tran' , \tpo' , \trf , \ow \rangle$, where $\tran' = \tran \cup \{t_1,t_2\}$ and $\tpo' = \tpo \cup \{ (t,t_1) | t \in p_1\} \cup \{(t',t_2) | t' \in p_2 \}$.
    
    \item [$\bullet$] $t_1 \in \tran^w$, $t_1 \notin \tran^r$ and $t_2 \in \tran^r \cap \tran^w$. 
    Let $\tau_1 = \alpha_{t_1}(\tau)$. 
    We know that $\tau_1 = \langle \tran_1 , \tpo_1 , \trf_1 , \ow_1 \rangle$, where 
    $\tran_1 = \tran \cup \{ t_1 \}$, and $\tpo_1 = \tpo \cup \{ (t,t_1) | t \in p_1\}$, with $\trf_1 = \trf$ since there are no read events in $t_1$.
    
    Consider $\tau_2 = \alpha_{t_2}(\tau_1) = \langle \tran_2, \tpo_2, \trf_2, \ow_2 \rangle$. 
    Since $t_2 \in \tran^r$, we know that $t_2$ has read events.
    \smallskip 
    
    Let there be $n$ read events $ev_1 \dots ev_n$ in transaction $t_2$, 
    reading from transactions $t'_1, \dots, t'_n$ on variables $x_1, \dots, x_n$.  
    Hence we get a sequence $\tau_1 \xrightarrow[]{\beginact(t_2)}_{\cc{-}\sat} \dots \tau'_{i_1} \xrightarrow[]{\readact(t_2,t'_1)}_{\cc{-}\sat}  \tau_{i_1}  \xrightarrow[]{}_{\cc{-}\sat}\dots  
    \tau'_{i_n} \xrightarrow[]{\readact(t_2,t'_n)}_{\cc{-}\sat} \tau_{i_n} \dots \xrightarrow[]{\commitact(t_2)}_{\cc{-}\sat} \tau_2$.
    Hence $\tran_2 = \tran_1 \cup \{t_2\}$, $\tpo_2 = \tpo_1 \cup \{(t',t_2) | t' \in p_2\}$, $\trf_2 = \trf_1 \cup (\trf_{i_1} \cup \trf_{i_2} \dots \cup \trf_{i_n})$.
    Since $\tau_1 \vdash_{\tt{G}} \alpha_{t_2}$,  for all $ev_i : 1 \le i \le n$, we have $t'_i \in \rbl(\tau_1,t_2,x_i)$.
    
    Since $\alpha_{t_1} \sim \alpha_{t_2}$, we know that for all $ev_i : 1 \le i \le n$, we have $t'_i \in \rbl(\tau,t_2,x_i)$. It follows that $\tau \vdash_{\tt{G}} \alpha_{t_2}$.
    Define $\tau_3 = \langle \tran_3, \tpo_3, \trf_2 , \ow_2 \rangle$, where $
    \tran_3 = \tran \cup \{t_2\}$ and $\tpo_3 = \tpo \cup \{(t',t_2) | t' \in p_2\}$. It follows that $\alpha_{t_2}(\tau) = \tau_3$, $\tau_3 \vdash_{\tt{G}} \alpha_{t_1}$ and $\tau_2 = \alpha_{t_1}(\tau_3) = \alpha_{t_2}(\alpha_{t_1}(\tau))$.
    \item [$\bullet$] $t_2 \in \tran^w$ and $t_1 \in \tran^r \cap \tran^w$. Similar to previous case.
    \item [$\bullet$] $t_1, t_2 \in \tran^r \cap \tran^w$. 
    Let $\tau_1 = \alpha_{t_1}(\tau) = \langle \tran_1, \tpo_1, \trf_1 \rangle$ and $\tau_2 = \alpha_{t_2}(\tau_1)$. Since $t_1 \in \tran^r$, we know that $t_1$ has read events. 
    Let there be $n$ read events $ev_1 \dots ev_n$ in transaction $t_1$, reading from transactions $t'_1, \dots, t'_n$ on variables 
    $x_1, \dots, x_n$. 
    Hence we get the sequence $\tau \xrightarrow[]{\beginact(t_1)}_{\cc{-}\sat} \dots \tau'_{i_1} \xrightarrow[]{\readact(t_1,t'_1)}_{\cc{-}\sat}  \tau_{i_1} \xrightarrow[]{}_{\cc{-}\sat}\dots  
    \tau'_{i_n} \xrightarrow[]{\readact(t_1,t'_n)}_{\cc{-}\sat} \tau_{i_n} \dots \xrightarrow[]{\commitact(t_1)}_{\cc{-}\sat} \tau_1$.
    Since $t_2 \in \tran^r$, we know that $t_2$ has read events. 
    Let there be  $m$ read events $ev'_1 \dots ev'_m$ in transaction $t_2$, reading from transactions $t''_1, \dots, t''_m$ on variables 
    $y_1, \dots, y_m$.  
    Then we get the sequence $\tau_1 \xrightarrow[]{\beginact(t_2)}_{\cc{-}\sat} \dots \tau'_{j_1} \xrightarrow[]{\readact(t_2,t''_1)}_{\cc{-}\sat}  \tau_{j_1}  \xrightarrow[]{}_{\cc{-}\sat} \dots 
    \tau'_{j_n} \xrightarrow[]{\readact(t_2,t''_m)}_{\cc{-}\sat} \tau_{j_n} \dots \xrightarrow[]{\commitact(t_2)}_{\cc{-}\sat} \tau_2$. 
    
    \smallskip
    
    $\mathbf{Proving}$ $\mathtt{ \tau \vdash_{\tt{G}} \alpha_{t_2} {\cdot} \alpha_{t_1}}$
    
    \smallskip
    
    Since $\tau_1 = \alpha_{t_1}(\tau)$, for all read events $ev_i$ such that $ev_i.trans = t_1$, we have $t'_i \in \rbl(\tau,t_1,x_i)$.
    Since $\tau_2 = \alpha_{t_2}(\tau_1)$, for all read events $ev'_j$ such that $ev'_j.trans = t_2$, we have $t''_j \in \rbl(\tau_1,t_2,y_j)$. 
    Since $\alpha_{t_1} \sim \alpha_{t_2}$, 
    it follows that, we have $t''_j \in \rbl(\tau,t_2,y_j)$ for all read events $ev'_j$ such that $ev'_j.trans = t_2$. Hence $\tau \vdash_{\tt{G}} \alpha_{t_2}$. 
    \smallskip

    Consider $\tau_3 = \alpha_{t_2}(\tau) = \langle \tran_3, \tpo_3, \trf_3 , \ow_3 \rangle $. To show that $\tau \vdash_{\tt{G}} \alpha_{t_2} \alpha_{t_1}$, 
    we show that for all read events $ev_i$ such that $ev_i.trans = t_1$,  $t'_i \in \rbl(\tau_3,t_1,x_i)$. We prove this using contradiction.
     Assume that $\exists t'_i$ s.t. $t'_i \notin \rbl(\tau_3,t_1,x_i)$.

    \medskip 
     
    If $t'_i \notin \rbl(\tau_3,t_1,x_i)$, then 
    there is a $t'_k$ such that $t'_k \in \rbl(\tau_3,t_1,x_i)$ and $t'_i [\tpo_3 \cup \trf_3]^+ t'_k$.
    
    \smallskip
    
    If $t'_k \in \rbl(\tau_3,t_1,x_i)$, since $\alpha_{t_1} \sim \alpha_{t_2}$, we also have $t'_k \in \rbl(\tau,t_1,x_i)$. 
     Since new relations added in $\tau_3$ are $\tpo_3 = \tpo \cup \{(t',t_2) \mid t' \in p_2 \}$ and $\trf_3 = \trf \bigcup_{j: 1 \le j \le m} t''_j [\trf^{y_j}] t_2$, 
       we can have $t'_i [\tpo_3 \cup \trf_3]^+ t'_k$ only if we there is a path from $t'_i$ to $t'_k$ in $\tau$; that is, 
        $t'_i [\tpo \cup \trf]^+ t'_k$. 
        This implies that $t'_i \notin \rbl(\tau,t_1,x_i)$ which is a contradiction.

    Thus, we have $\neg(t'_i [\tpo_3 \cup \trf_3]^+ t'_k)$. This, for all $1 \leq i \leq n$, $t'_i \in \rbl(\tau_3,t_1,x_i)$. Thus, we now have $\tau \vdash_{\tt{G}} \alpha_{t_2} \alpha_{t_1}$. 
    It remains to show that $(\alpha_{t_1} {\cdot} \alpha_{t_2})(\tau) \equiv (\alpha_{t_2}{ \cdot} \alpha_{t_1})(\tau)$. 
    
    $\mathbf{Proving}$ $\mathtt{(\alpha_{t_1} {\cdot} \alpha_{t_2})(\tau) \equiv (\alpha_{t_2}{ \cdot} \alpha_{t_1})(\tau)}$
    
    Define $\tau_4 {=} \alpha_{t_1}(\tau_3) {=} \alpha_{t_1}.\alpha_{t_2}(\tau){=}
         \langle \tran_4, \tpo_4, \trf_4 , \ow_4 \rangle$.
         
    Recall that 
    $\tau_2{=}$ $\alpha_{t_2}.\alpha_{t_1}(\tau)$ =$
      \langle \tran_2,\tpo_2, \trf_2 , \ow_2 \rangle$.
    \smallskip 
    Since, $\alpha_{t_2} \sim \alpha_{t_1}$, we have $\trf_2=\trf_4$. It follows trivially that $\tran_4 = \tran_2$ and $\tpo_4 = \tpo_4$.
    Since $\tran_2=\tran_4, \tpo_2=\tpo_4, \trf_2=\trf_4$ , and $\ow_4 = \ow_2$ it follows that $(\alpha_{t_1} \cdot \alpha_{t_2})(\tau) \equiv (\alpha_{t_2} \cdot \alpha_{t_1})(\tau)$.
    
\end{itemize}
\end{proof}

\begin{lemma}
Consider a trace trace $\tau$ such that $\tau \models \cc$, $\tau \vdash_{\tt{G}} \pi_{\tt{T}} \cdot \alpha_t$. Let each read event $ev_i = r_i(p,t,x_i,v)$ in $\alpha_t$ read from some transaction $t_i$, and let $\pi_{\tt{T}}$ be $t$-free.  Then $\tau \vdash_{\tt{G}} \alpha'_t \cdot \pi_{\tt{T}}$, where $\alpha'_t$ is the same as  $\alpha_t$, with the exception that the sources of its read events can be different. That is, 
 each read event $ev_i = r_i(p,t,x_i,v) \in \alpha'_t$ can read from some transaction $t'_i \neq t_i$.
\label{lemma cc-comp-9}
\end{lemma}

\begin{proof}
Let $\pi_{\tt{T}} = \alpha_{t_1} \alpha_{t_2} \dots \alpha_{t_n}$ and $\tau_0 \xrightarrow[]{\alpha_{t_1}}_{\cc{-}\sat} \tau_1 \xrightarrow[]{\alpha_{t_2}}_{\cc{-}\sat} \dots \xrightarrow[]{\alpha_{t_n}}_{\cc{-}\sat} \tau_n \xrightarrow[]{\alpha_t}_{\cc{-}\sat} \tau_{n+1}$, where $\tau_0 = \tau$. 
Let $\tau_i = \langle \tran_i, \tpo_i, \trf_i ,\ow_i \rangle$. 
Let there be  $m$ read events in transaction $t$. Keeping in mind what we want to prove, where we want to execute $t$ first followed by $\pi_{\tt{T}}$, and obtain an execution $\alpha'_t \pi_{\tt{T}}$, we do the following.

For each read event $ev_i = r_i(p,t,x_i,v)$ we define $t'_i \in \tran^{w,x_i}$ such that 
$(i)$ $ t'_i \in \rbl(\tau_0,t,x_i)$, and 
$(ii)$ there is no $t''_i \in \rbl(\tau_0,t,x_i)$ where $t'_i [\tpo_n \cup \trf_n] t''_i$.  
Note that this is possible since 
 $\pi_{\tt{T}}$ is $t$-free, so all the observables $\alpha_{t_i}$ occurring in $\pi_{\tt{T}}$ are such that $t_i$ is issued in a process other than that of $t$. 
 Thus, when $\alpha_t$ is enabled, we can choose any of the 
 writes done earlier,  this fact is consistent with the $\cc$ semantics, since from Lemmas \ref{lem:cc-3}-\ref{lem:cc-4} we know $\tau_n \models \cc$.
Since $\tau_n \models \cc$ such a $t'_i \in  \rbl(\tau_0,t,x_i)$ exists for each $ev_i$.

\medskip 
Next we define a sequence of traces which can give the execution 
$\alpha'_t \pi_{\tt{T}}$. 
Define a sequence of traces $\tau'_0, \tau'_1, \dots \tau'_n$ where $\tau'_j = \langle \tran'_j, \tpo'_j, \trf'_j , \ow \rangle$\footnote{$\ow$ relation will be empty following trace semantics.} is such that

\begin{enumerate}
	\item $\tau \vdash_{\tt{G}} {\alpha'_t}$ and $\alpha'_t(\tau)=\tau'_0$,
	\item $\tau'_j \vdash_{\tt{G}} \alpha_{t_{j+1}}$ and  $\tau'_{j+1} = \alpha_{t_{j+1}}(\tau'_j)$ for all $j: 0 \le j \le n$.
\end{enumerate}

\begin{itemize}
\item  Define $\tran'_i = \tran_i \cup \{t\}$ for all $0 \leq i \leq n$, 
\item 
Define $\tpo' = \{t' [\tpo] t \mid t'$ is a transaction in $\tau$, in the same process as $t\}$, and $\tpo'_i = \tpo_i \cup \tpo'$ for all $0 \leq i \leq n$, 

\item  $\trf' = \bigcup_{i:1\le i \le m} t'_i [\trf^{x_i}] t$ (reads from relation corresponding to each read event $ev_i \in \alpha'_t$), and 
 $\trf'_i = \trf_i \cup \trf'$ for all $i : 1 \le i \le n$, 
	
\smallskip 

We show that $\tau'_i \vdash_{\tt{G}} \alpha_{t_{i+1}}$ and  $\tau'_{i+1} = \alpha_{t_{i+1}}(\tau'_i)$ holds good. 
\end{itemize}

\medskip 
\noindent{\bf{Base case}}. The base case $\tau \vdash_{\tt{G}} {\alpha'_t}$ and $\alpha'_t(\tau)=\tau'_0$,
   holds trivially, by construction. 

\smallskip 

\noindent{\bf{Inductive hypothesis}}. Assume that $\tau'_j \vdash_{\tt{G}} \alpha_{t_{j+1}}$ and  $\tau'_{j+1} = \alpha_{t_{j+1}}(\tau'_j)$ for $0 \leq j \leq i-1$. 

\smallskip 

 We have to prove that 
$\tau'_i \vdash_{\tt{G}} \alpha_{t_{i+1}}$ and  $\tau'_{i+1} = \alpha_{t_{i+1}}(\tau'_i)$. 
\begin{enumerate}
	\item If $t_{i+1} \in \tran^{w}$ and $t_{i+1} \notin \tran^r$,  then the proof holds trivially from the inductive hypothesis. 
	\item Consider now $t_{i+1} \in \tran^r$. Consider a read event 
	$ev^{i+1} = r(p,t_{i+1},x,v) \in \alpha_{t_{i+1}}$ which was reading from 
	transaction $t_x$ in $\pi_{\tt{T}}$. That is, we had $t_x \in \rbl(\tau_i, t_{i+1},x)$. 	If $t_x \in \rbl(\tau'_i, t_{i+1},x)$, then we are done, since we can simply extend the run from the inductive hypothesis. 
	
	\smallskip 
	
	Assume otherwise. That is, $t_x \notin \rbl(\tau'_i, t_{i+1},x)$. 
	Since  $t_x$  $\in$   $\rbl(\tau_i, t_{i+1},x)$, we know that there are some blocking transactions in the new path which prevents $t_x$ from being readable. Basically, we have the blocking transactions since $t$ is moved before $\pi_{\tt{T}}$. 	
	Consider $tr_1 \in \rbl(\tau_i,t_{i+1},x)$.  
	Since $t_x$ $\in$ $\rbl(\tau_i, t_{i+1},x)$ and $t_x$ $\notin$ $\rbl(\tau'_i, t_{i+1},x)$, we have $t_x [\tpo'_i \cup \trf'_i] tr_1$.
	Since $\tpo'_i = \tpo_i \cup \tpo'$ , $\trf'_i = \trf_i \cup \trf'$, and $t$ has no successor in $\tau'_i$, it follows that $t_x [\tpo_i \cup \trf_i] tr_1$. This contradicts the fact that $t_x$ $\in$ $\rbl(\tau_i, t_{i+1},x)$.
\end{enumerate}

Thus, we have shown that $\tau \vdash_{\tt{G}} \alpha'_t \alpha_{t_1} \dots \alpha_{t_i}$ 
is such that $\alpha'_t \alpha_{t_1} \dots \alpha_{t_i}(\tau)=\tau'_i$ for all 
$0 \leq i \leq n$. When $i=n$ we obtain 
$\alpha'_t \pi_{\tt{T}}(\tau)=\tau'_n$, or $\tau \vdash_{\tt{G}} \alpha'_t \pi_{\tt{T}}$. 
\qed 
\end{proof}

%% file: cm-proofs.tex
\newpage
\centerline{\bf{Causal Memory $\cm$}} 

\section{Trace Semantics}
\label{app:cmsatcons}
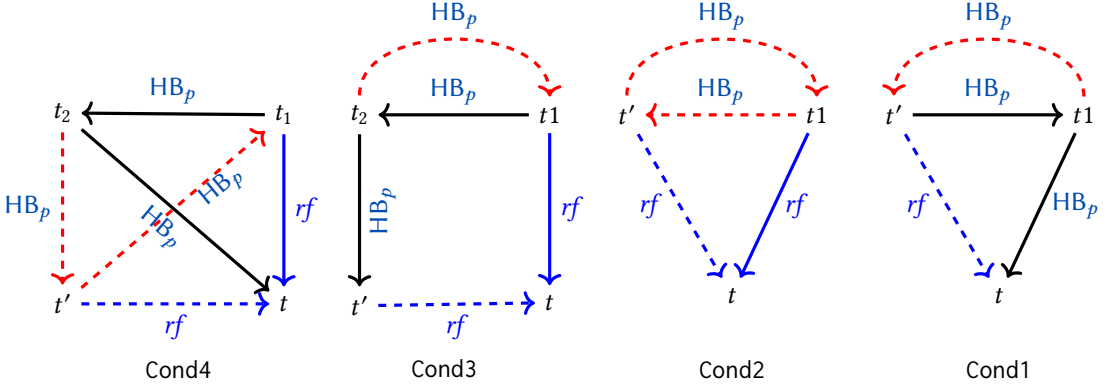
\begin{figure}[H]
        \begin{tikzpicture}[node distance=10mm , thick, main/.style= {draw,circle},];
        \node[] (t) {$t$};
        \node[] (t1) [above = 20mm of t] {$t_1$};
        \node[] (t') [left = 25mm of t] {$t'$};
        \node[] (t2) [above = 20mm of t']  {$t_2$};
        \node[] (1) [below=5mm of t] {};
        \node[] (cap) [left=7mm of 1] {$\tt{Cond4}$};
        \draw[->,blue,line width=1.2pt] (t1) -- node[right] {$\trf$} (t);
        \draw[->,black,line width=1.2pt] (t1) -- node[above] {$\hb{p}$} (t2);
        \draw[->,black,line width=1.2pt] (t2) --node[rotate=-40, below] {$\hb{p}$} (t);
        \draw[dashed , ->,blue,line width=1.2pt] (t') -- node[below] {$\trf$} (t);
        \draw[dashed , ->,red,line width=1.2pt] (t2) -- node[left] {$\hb{p}$} (t');
        \draw[dashed , ->,red,line width=1.2pt] (t') -- node[rotate=40,right=3mm] {$\hb{p}$} (t1);

        \node[] (3at2) [right= 5mm of t1] {$t_2$};
        \node[] (3at1) [right= 20mm of 3at2] {$t1$};
        \node[] (3at) [below=20mm of 3at1] {$t$};
        \node[] (3at') [below=20mm of 3at2] {$t'$};
        \node[] (31) [below=5mm of 3at] {};
        \node[] (3acap) [left=7mm of 31] {$\tt{Cond3}$};

        \draw[->,black,line width=1.2pt] (3at1) -- node[above] {$\hb{p}$} (3at2);
         \draw[->,blue,line width=1.2pt] (3at1) -- node[right] {$\trf$} (3at);
        \draw[dashed , ->,blue,line width=1.2pt] (3at') -- node[below] {$\trf$} (3at);
        \draw[dashed , ->,red,line width=1.2pt] (3at2) to [out=90, in=90] node[above] {$\hb{p}$} (3at1);
        \draw[->,black,line width=1.2pt] (3at2) -- node[rotate=90,below] {$\hb{p}$} (3at');
        
        \node[] (2at') [right= 5mm of 3at1] {$t'$};
        \node[] (2at1) [right= 20mm of 2at'] {$t1$};
        \node[] (21) [below=20mm of 2at1] {};
        \node[] (2at) [left=8mm of 21] {$t$};
        \node[] (2acap) [below=5mm of 2at] {$\tt{Cond2}$};
        
        \draw[dashed , ->,red,line width=1.2pt] (2at1) -- node[above] {$\hb{p}$} (2at');
         \draw[->,blue,line width=1.2pt] (2at1) -- node[right] {$\trf$} (2at);
        \draw[dashed , ->,blue,line width=1.2pt] (2at') -- node[left=1mm] {$\trf$} (2at);
        \draw[dashed , ->,red,line width=1.2pt] (2at') to [out=90, in=90] node[above] {$\hb{p}$} (2at1);

        \node[] (1t') [right= 5mm of 2at1] {$t'$};
        \node[] (1t1) [right= 20mm of 1t'] {$t1$};
         \node[] (111) [below=20mm of 1t1] {};
        \node[] (1t) [left=8mm of 111] {$t$};
        \node[] (1cap) [below=5mm of 1t] {$\tt{Cond1}$};

        \draw[->,black,line width=1.2pt] (1t') -- node[above] {$\hb{p}$} (1t1);
        \draw[->,black,line width=1.2pt] (1t1) -- node[right] {$\hb{p}$} (1t);
        \draw[dashed , ->,blue,line width=1.2pt] (1t') -- node[left=1mm] {$\trf$} (1t);
        \draw[dashed , ->,red,line width=1.2pt] (1t1) to [out=90, in=90] node[above] {$\hb{p}$} (1t');
        \end{tikzpicture}
        \caption{Readability for $\cm$}
         \label{read-cm}
    \end{figure}

\subsection{Readability and Visibility for  $\cm$}
\label{cm:rv}
Let $t$ be a transaction in process $p$. For $X=\cm$, $\rbl(\tau^t_X, t, x)$ is defined as the set of all transactions  $t' \in \tran^{\wt,x}$ s.t. 
\begin{itemize}
\item[$\tt{cond1}$] 	there is no transaction $t'' \in \tran^{\wt,x}$ such that 
	$t' ~\hb{p}~t''~\hb{p}~t$. Having such a $t''$ and allowing $t'~\trf~t$ gives $t''~\hb{p}~t'$ and $\cmcyc$.
	\item[$\tt{cond2}$]  If $t' \in \tran^{\wt,y}$ (for $y \neq x$), there is no transaction $t'' \in \tran^{\wt,x} \cap \tran^{\wt,y}$ such that $t''~\trf~t$ (wrt $y$).  
		Having such a $t''$ with $t'~\trf~t$ (wrt $x$) results in $t''~\hb{p}~t'$, $t~\hb{p}~t''$  and $\cmcyc$.
  
		\item[$\tt{cond3a}$] If $t' \in \tran^{\wt,y}$ (for $y \neq x$), there is no  transaction $t'' \in \tran^{\wt,y}$ such that $t''~\hb{p}~t'$ and $t''~\trf~t$ (wrt $y$).  Having such a $t''$ with  $t'~\trf~t$ gives $t'~\hb{p}~t''$ giving $\cmcyc$.
  
		\item[$\tt{cond3b}$]  There are no transactions $t_1, t_2 \in \tran^{\wt,y}$ such that  
$t_1~\trf~t$ (wrt $y$), and $t_1~\hb{p}~t_2~\hb{p}~t'$.
 Having such $t_1, t_2$ with $t'~\trf~t$ gives  $t_2~\hb{p}~t$. With 
 $t_1~\trf~t$ we get $t_2~\hb{p}~t_1$ and $\cmcyc$.	
 
\item[$\tt{cond4}$] If $t' \in \tran^{\wt,y}$ (for $y \neq x$), there are no transactions $t_1 \in \tran^{\wt,y}, t_2 \in \tran^{\wt,x}$ such that $t_1~\trf~t$ (wrt $y$), $t_1~\hb{p}~t_2~\hb{p}~t$. If so, allowing $t'~\trf~t$ gives 
$t_2~\hb{p}~t'$. Also, $t_1 ~\trf~ t, t_1~\hb{p}~t_2$ gives $t~\hb{p}~t_2$. 
Now, we have $t_2~\hb{p}~t'~\trf~t~\hb{p}~t_2$ obtaining $\cmcyc$.
\end{itemize}

Figure \ref{read-cm} explains cycles created after violating conditions $\tt{cond1}-\tt{cond4}$.
$\tt{cond3} $ covers $\tt{cond3a} $  and $\tt{cond3b} $.
For $\tt{cond3a}$,  $t_2$ coincides with $t'$.

After adding $t'~\trf~t$, we must check that 
there are no consistency violations.  
The check set $\vbl(\tau^t_X, t,x)$ is defined as the set of transactions which turn 	``sensitive''  
on adding the new  edge  $t'~\trf~t$. 
Unless  appropriate edges are added involving 
these sensitive transactions, we may get consistency violating cycles
in the resultant trace.
Let $\tau^{tt'}$ denote the trace obtained by adding the new transaction $t$ and the edge $t'~\trf~t$ to trace $\tau$. Now, we 
identify $\vbl(\tau^t_X, t,x)$ and the edges which must be added to $\tau^{tt'}$ to obtain a consistent extended trace.

For $X=\ccvt$, let $X$-edge=$\cfr$, $X$-rel=$\co$
and for $X=\cm$, let $X$-edge=$X$-rel=$\hb{p}$.

For $X =\cm$, $\vbl(\tau^t_X,t,x)$ is classified into three categories. 
	\begin{itemize}
	\item[(a)] The first kind of transactions in $\vbl(\tau^t_X,t,x)$  are  
$\{t''\in \rbl(\tau^t_X, t,x) \mid t''~\hb{p}~t\}$.  
	Then we add  from each $t'' \in \vbl(\tau^t_X,t,x)$ a $\hb{p}$ to $t'$ in $\extend_X(\tau^{tt'})$. 
	\item[(b)] The second kind of  transactions in $\vbl(\tau^t_X,t,x)$  are   $t_1 \in \tran^{\wt,y}$ such that $t_1~\trf~t$ (wrt $y$) when $t' \in \tran^{\wt,y}$. Then we  add  $t'~\hb{p}~t_1$ to $\extend_X(\tau^{tt'})$. 
	\item[(c)] The third kind of  transactions in $\vbl(\tau^t_X,t,x)$  are 
	 $t_2 \in \tran^{\wt,y}$ such that $t_2~\hb{p}~t'$ 
	and we have $t_1~\trf~t$ (wrt $y$) for some $t_1 \in \tran^{\wt,y}$. In this case, we add $t_2~\hb{p}~t_1$  to $\extend_X(\tau^{tt'})$. Adding the $\hb{p}$ edges 		$t_2~\hb{p}~t_1$ to $\extend_X(\tau^{tt'})$
				 can result in $t_3~\hb{p}~t_4$ for some $t_3,t_4\in \tran^{\wt,y}$.
If $t_3 ~\trf~ t_1$ and $t_1~\tpo~t$ then we add $t_1~\hb{p}~t_4$
to $\extend_X(\tau^{tt'})$. 
\end{itemize}

\begin{lemma}
    Given trace $\tau$, a transaction $t$ and a variable $x$, we can construct the sets  $\rbl(\tau, t, x)$ and $\vbl(\tau, t, x)$ in polynomial time.
    \label{lemma:cm-poly}
\end{lemma}
\begin{proof}
    Let $p$ be process of $t$.
    We prove that set $\rbl(\tau, t, x)$ can be computed in polynomial time ($O(|\tran|^3)$ time), by designing an algorithm to generate set $\rbl(\tau, t, x)$. The algorithm consists of the following  steps:
    \begin{itemize}
        \item[(i)] First we compute the  transitive closure of the relations
        $[\hb{p}]$ and $[\tpo \cup \trf]$, i.e, we compute $\hb{p}^+$ and $[\tpo \cup \trf]^+$. 
        This will take $O(|\tran|^3)$ time.
        \item[(ii)] We compute the set $\tran'=\{t'\mid t' ~\hb{p}~ t\}$. This takes $O(|\tran|)$ time.
        
        \item[(iii)] For each transaction $t' \in \tran^{wt,x}$, we perform following checks:
        \begin{itemize}
            \item Check $\mathtt{cond 1 :} $ Check whether there exists a transaction $t'' \in \tran' \cap \tran^{wt,x}$            
            with $t'$ $\hb{p}^+$ $t''$. 
            If there is no such $t''$ then proceed to $\mathtt{cond 2}$. 
            If we find such $t''$, then $t' \notin \rbl(\tau , t, x)$. 
            This step will take $O(|\tran|^2)$ time.
            \item Check $\mathtt{cond 2 :}$ Check whether there exists a transaction $t''$ such that $t'' [\trf^y] t$ and $t',t'' \in \tran^{wt,x} \cap \tran^{wt,y}$. 
            If there is no such $t''$ then proceed to $\mathtt{cond 3}$. 
            If we find such $t''$, then $t' \notin \rbl(\tau , t, x)$. 
            This step will take $O(|\tran|)$ time.
            \item Check $\mathtt{cond 3 :}$ 
            
            We first check case a of $\mathtt{cond 3}$. Check whether there exists a transaction $t''$ such that $t'' [\trf^y] t$ , $t' \in \tran^{wt,y}$, and $t'' [\hb{p}]^+ t'$. 
            If there is no such $t''$ then proceed to case two of $\mathtt{cond 3}$. 
            If we find such $t''$, then $t' \notin \rbl(\tau , t, x)$. 
            This step will take $O(|\tran|)$ time.
            
             Next, we  check case b of $\mathtt{cond 3}$. Check whether there are transactions $t''$ and $t'''$ such that $t'' [\trf^y] t$, $t''' \in \tran^{wt,y}$, $t''' ~\hb{p}^+~ t'$ and $t'' [\hb{p}]^+ t'''$. 
            If there are no such $t''$ and $t'''$ then proceed to $\mathtt{cond 4}$. 
            If we find such $t''$ and $t'''$, then $t' \notin \rbl(\tau , t, x)$. 
            This step will take $O(|\tran|^2)$ time.
            \item Check $\mathtt{cond 4 :}$ 
            Check whether there are transactions $t''$ and $t''' \in \tran'$ such that $t'' [\trf^y] t$, $t''' \in \tran^{\wt,x}$, $t'\in \tran^{\wt,y}$ and $t'' [\hb{p}]^+ t'''$. 
            If we find such $t''$ and $t'''$, 
            then $t' \notin$ $\rbl(\tau , t, x)$. 
            If there are no such $t''$ and $t'''$ then add $t'$ to $\rbl(\tau , t, x)$.
            This step will take $O(|\tran|^2)$ time.
        \end{itemize}
    \end{itemize}
    Similarly, we can compute $\vbl(\tau, t, x)$ in polynomial time. This completes the proof.
\end{proof}

\subsection{Properties of the Trace Semantics}
\label{cm:prop}
\begin{lemma}
    If $\tau_1 \models \cm$  and $\tau_1$ $\xrightarrow[]{\beginact(t)}_{\cm{-}\sat}$ $\tau_2$ then $\tau_2 \models \cm$.
    \label{lem:cm-f3}
\end{lemma}
\begin{proof}
The proof follows trivially since we do not change the reads from and coherence order relations and the transaction $t$ added to the partial order has no successors. Thus, if $\tau_1 \models \cm$, 
so does $\tau_2$.  
	\end{proof}

\begin{lemma}
    If $\tau_1 \models \cm$  and $\tau_1$ $\xrightarrow[]{\readact(t,t')}_{\cm{-}\sat}$ $\tau_2$ then $\tau_2 \models \cm$.
    \label{lem:cm-f4}
\end{lemma}

\begin{proof}
    Let $\tau_1 = \langle \tran_{1} , \tpo_1 , \trf_1, \hb{}_1 \rangle$, $\tau_2 = \langle \tran_{2} , \tpo_2 , \trf_2, \hb{}_2 \rangle$, and $ev$ = $r(p,t,x,v)$ be a read event in transaction $t \in \tran_1$. Suppose $\tau_2 \nvDash \cm$, and  that $\hb{2}$ is cyclic. Since $\tau_1 \models \cm$, and $t$ has no outgoing edges in $(\tpo \cup \trf)^+$, 
   on adding $(t',t) \in \trf_2$, 
    it follows that the cycle in $\tau_2$ is as a result 
    of the newly added $\hb{}_2$ edges. 
    
    We prove that 
    on adding $(t',t) \in \trf_2$, 
        such cycles are possible iff at least one of ${\tt{cond}_1}, \dots, {\tt{cond}_4}$ are true. 
    
 \subsection*{${\tt{cond1}}\vee \dots \vee {\tt{cond4}}$ induces $\hb{2}$ cyclicity }
   \label{sec:cm-cyc1}
    This direction is easy to see : assume 
    one of ${\tt{cond1}}, \dots, {\tt{cond4}}$ are true; then, as already 
    argued in the main paper, we will get a $\hb{2}$ cycle on adding the $\trf_2$ edge from $t'$ to $t$ and we are done.

 \subsection*{$\hb{2}$ cyclicity implies ${\tt{cond1}}\vee \dots \vee {\tt{cond4}}$ }
 \label{sec:cm-cyc2}
  For the converse direction, assume that we add the $\trf_2$ edge from $t'$ to $t$ and obtain a $\hb{2}$ cycle in $\tau_2$. We now argue that this cycle has been formed because one of 
 ${\tt{cond}_1}, \dots, {\tt{cond}_4}$ are true.  
 
 First of all, note that $t$ has no outgoing edges in $(\tpo_2 \cup \trf_2)^+$ in $\tau_2$, since 
 it is the current transaction being executed. Also, 
 we know that $\tau_1$ has no $\hb{}_1$ cycles. 
 Thus, the cyclicity of $\hb{2}$ is induced by the 
 newly added $\trf_2$ edge as well as the newly added 
  $\hb{}_2$ edges. Note that adding the $\trf_2$ edge  to $\tau_1$ does not induce any cycle since $t$ has no outgoing edges. Lets analyze the 
  $\hb{}_2$ edges added which induce cycles, and argue that 
  one of ${\tt{cond1}}, \dots, {\tt{cond4}}$ will be true.   
  
  \begin{enumerate}
  	\item We add  $\hb{}_2$ edges from $t'' \in \vbl(\tau_1, t,x)$ to $t'$. 
  	For  these edges to induce a cycle, we should have a path from $t'$ to $t''$ in $\tau_1$. That is, we have $t' [\hb{}_1]^+ t''$. This is captured by ${\tt{cond1}}$.  
  	
  	\item Consider $y \neq x$. If $t' \in \tran^{w,x} \cap \tran^{w,y}$, 
  	and we have $t'''\trf_1^y t$.  Then we add 
  	$(t',t''') \in \hb{}^y_2$. To get a cycle, we need a path from 
  	$t'''$ to $t'$. 
  	\begin{itemize}
  	\item If $t''' \in \tran^{w,x} \cap \tran^{w,y}$, then we add $(t''', t') \in \hb{}^x_2$, resulting in a cycle. This is handled by  ${\tt{cond2}}$.
  	\item If we have a path 
  	$t''' [\hb{}_1]^+ t'$, then we get a cycle again. This is handled by ${\tt{cond3}}$. 
  \item As a last case, to obtain a path from $t'''$ to $t'$, 
  assume there is a path from $t'''$ to $t$ in $\tau_1$, and let $t'' \in \tran^{w,x}$ be the last transaction writing to $x$ in this path.  Note that this will induce a path from $t'''$ to $t''$ to $t'$ : the path  from $t''$ to $t'$ comes by the $\hb{}^x_2$ edge added from $t''$ to $t'$ since 
  we have $(t',t) \in \trf_2^x$. Once we get this path from $t'''$ to $t'$, we again have the path we were looking for to get the cycle. This is handled by ${\tt{cond4}}$. 
    		
  	\end{itemize}
Thus, we have shown that obtaining a path from $t'''$ to $t'$ is covered 
by the conditions  ${\tt{cond}_1}, \dots, {\tt{cond}_4}$, and hence, 
a $\hb{2}$ cycle.  
\item Consider $y \neq x$. If $t_4 \in \tran^{w,y}$ such that $t_4 ~\hb{p}~ t'$, 
  	and we have $t'''\trf_1^y t$.  Then we add 
  	$(t_4,t''') \in ~\hb{2p}~$. To get a cycle, we need a path from 
  	$t'''$ to $t_4$. 
   If we have a path 
  	$t''' ~\hb{1p}~ t_4$, then we get a cycle. This is handled by ${\tt{cond3}}$. 
  	  \end{enumerate}
 
 Thus, we can think of  $\hb{2}$ cycles as a result of the forbidden patterns described in ${\tt{cond}_i}$ $1 \leq i \leq 4$.  By construction of $\tau_1$ $\xrightarrow[]{\readact(t,t')}_{\cm{-}\sat}$ $\tau_2$, we ensure $\neg {\tt{cond}_1}, \dots, \neg {\tt{cond}_4}$. Hence, 
    $\tau_2 \models \cm$. 
\end{proof}

Define $\big[[ \sigma]\big]^{\cm{-}\sat}$ as the set of  traces  generated  using $\xrightarrow[]{}_{\cm{-}\sat}$ transitions, starting from an empty trace $\tau_{\o}$.

Consider a terminal trace $\tau$ generated by $\cm$ DPOR algorithm starting from the empty trace $\tau_\emptyset$.
That is, there is a sequence  
$\tau_0 \xrightarrow[]{}_{\cm{-}\sat} 
\tau_1 \xrightarrow[]{}_{\cm{-}\sat}  
\tau_2 \xrightarrow[]{}_{\cm{-}\sat} \dots 
 \xrightarrow[]{}_{\cm{-}\sat}  \tau_n$ with 
 $\tau_0=\tau_{\emptyset}$, 
 and $\tau_n=\tau$. 
 Since $\tau_{\emptyset}$ is a empty we have $\tau_{\emptyset} \models \cm$, 
 it follows by  
 Lemmas \ref{lem:cm-f3} and \ref{lem:cm-f4} that $\tau \models \cm$. 
 
 Hence, for each trace $\tau \in \big[[ \sigma]\big]^{\cm{-}\sat} $ we have $\tau \models \cm$.

\subsection{$\cm$ DPOR Completeness}
\label{app:complete-cm}
In this section, we show the completeness of the DPOR algorithm. More precisely, for the input program under $\cm$ for any terminating run $\rho \in Runs(\conf)$ and trace $\tau$ s.t. $ \rho \models \tau$, 
we show that $\explore(\cm, \tau_{\emptyset}, \epsilon)$ will produce a recursive visit $\explore(\cm, \tau', \pi)$ for some terminal $\tau'$, and $\pi$ where $\tau'=\tau$.   First, we give some definitions and auxiliary lemmas.

Let  $\pi = \alpha_{t_1} \alpha_{t_2} \dots \alpha_{t_n}$ be an observation sequence where  $\alpha_t = \beginact(p,t) \dots$ \plog{end}$(p,t)$.
 $\alpha_t$ is called an \emph{observable}, and is a sequence 
 of events from transaction $t$.  
Given $\pi$ and a trace $\tau$, we define 
$\tau \vdash_{\tt{G}} \pi$ to represent a sequence 
$ \tau_0 \xrightarrow[]{\alpha_{t_1}}_{\cm{-}\sat} 
 \tau_1  \xrightarrow[]{\alpha_{t_2}}_{\cm{-}\sat} \cdot \cdot \cdot \xrightarrow[]{\alpha_{t_n}}_{\cm{-}\sat} \tau_n $, where 
 $\tau_0 = \tau$ and 
$\pi = \alpha_{t_1} \alpha_{t_2} \dots \alpha_{t_n}$. Moreover, we define $\pi (\tau) :=  \tau_n$.

 $\tau$ is $terminal$ if there is no event left to execute $i.e$ $\succof\tau = \emptyset$. 
 We define $\tau \vdash_{\tt{G_T}} \pi$ to say that $(i)$ $\tau \vdash_{\tt{G}} \pi$ and $(ii)$ $\pi(\tau)$ is terminal.

\begin{definition}($p$-free and $t$-free observation sequences)
For a process  $p \in \mathbf{P}$, we say that an observation sequence $\pi$ is $p$-free if   all observables $\alpha_{t}$ in $\pi$ pertain to 
transactions $t$ not in $p$. 
For a transaction $t$ issued in $p$ we say that $\pi$ is $t$-free if $\pi$ is $p$-free. 
\end{definition}

\begin{definition}(Independent Observables)
For observables $\alpha_{t_1}$ and $\alpha_{t_2}$, we write $\alpha_{t_1} \sim  \alpha_{t_2}$ to represent they are independent. 
This means 
$(i)$ $t_1,t_2$ are transactions issued in different processes, that is, 
 $t_1$ is issed in $p_1$, $t_2$ is issued  in $p_2$, with $p_1 \neq p_2$,  
$(ii)$ no read event $r(p_2,t_2,x,v)$ of transaction $t_2$ reads from $t_1$, and 
$(iii)$ no read event $r(p_1,t_1,x,v)$ of transaction $t_1$ reads from $t_2$.	
Thus, $\neg(\alpha_{t_1} \sim  \alpha_{t_2})$ if either $t_1, t_2$ are issued in the same process, or there is a $\trf$ relation between $t_1, t_2$. That is, $\neg(\alpha_{t_1} \sim  \alpha_{t_2})$ iff $t_1 ~\hb{p}~ t_2$ or  $t_2 ~\hb{p}~ t_1$.

\end{definition}
 
 \begin{definition}(Independent Observation Sequences)
Observation sequences $\pi^1, \pi^2$ are called independent written $\pi^1 \sim \pi^2$ if there are observables $\alpha_{t_1}$, $\alpha_{t_2}$, and observation sequences $\pi'$ 
and $\pi''$ such that $\pi^1 = \pi' \cdot \alpha_{t_1} \cdot \alpha_{t_2} \cdot \pi''$, 
$\pi^2 = \pi' \cdot \alpha_{t_2} \cdot \alpha_{t_1} \cdot \pi''$, 
and $\alpha_{t_1} \sim \alpha_{t_2}$. 
In other words, we get $\pi^2$ from $\pi^1$ by swapping neighbouring independent observables corresponding to transactions $t_1$ and $t_2$. 
\end{definition}
We use $\approx$ to denote reflexive transitive closure of $\sim$.

\begin{definition}(Equivalent Traces)
For traces $\tau_1 = \langle \tran_1, \tpo_1, \trf_1, \hb{}_1 \rangle$ and $\tau_2 = \langle \tran_2, \tpo_2, \trf_2, \hb{}_2 \rangle$, 
we say $\tau_1, \tau_2$ are equivalent denoted $\tau_1 \equiv \tau_2$ if $\tran_1 = \tran_2$, $\tpo_1 = \tpo_2$, $\trf_1 = \trf_2$ 
and for all $t_1$,$t_2 \in \tran_1^{w,x}$ for all variables $x$ for each process $p$, 
we have $t_1$ [$\hb{1p}^x$] $t_2$ iff $t_1$ [$\hb{2p}^x$] $t_2$.
	
\end{definition}

\begin{lemma}
$((\tau_1 \equiv \tau_2) \wedge \tau_1 \vdash_{\tt{G}} \alpha_t) \Rightarrow (\tau_2 \vdash_{\tt{G}} \alpha_t \wedge (\alpha_t(\tau_1) \equiv \alpha_t(\tau_2)))$.
\label{lemma cm-g-1}
\end{lemma}
\begin{proof}
Assume 	$\tau_1 \vdash_{\tt{G}} \alpha_t$ for an observable $\alpha_t$, and 
$\tau_1=(\tran_1, \tpo_1, \trf_1, \hb{}_1)$, 
$\tau_2=(\tran_2, \tpo_2, \trf_2, \hb{}_2)$ with $\tau_1 \equiv \tau_2$. Then we know that $\tran_1=\tran_2, \tpo_1=\tpo_2, \trf_2=\trf_1$. 
 
 Since $\tau_1 \vdash_{\tt{G}} \alpha_t$, let  $\tau_1 \xrightarrow[]{\alpha_t}_{\cm{-}\sat} \tau'$. 
 $\tau'=(\tran', \tpo', \trf', \hb{}')$ where $\tran'=\tran_1 \cup \{t\}$, 
 $\tpo'=\tpo_1 \cup \{(t_1,t) \mid t_1 \in \tran_1$ is in the same process as $t\}$,  $\hb{}_1 \subseteq \hb{}'$ 
 and $\trf_1 \subseteq \trf'$. $\hb{}'$ can contain $(t',t)$ for some 
 $t' \in \tran_1$, when $t',t \in {\tran'}^{w,x}$ for some variable $x$, based 
 on the trace semantics. In particular if we have 
 $t_1 {\trf'}^x t_2$ , $t_3 \in \tran^{wt,x}$ and $t_3 (\trf'\cup  \tpo')^+ t_2$, then    
$t_3 ~{\hb{}'}^x~ t_1$.\footnote{We write $t~\hb{}^x~t'$ to denote that this $\\hb{}$ relation is added because of a read on $x$ by a transaction $t''$ in the trace semantics.}

 Since $\tau_1 \equiv \tau_2$, for any transactions $t', t'' $ in $\tran_1=\tran_2$, $t' [\hb{}_1^x] t''$ iff $t' [\hb{}_2^x] t''$ for all variables $x$. In particular, $t' \hb{}_1^x t''$ iff $t' \hb{}_2^x t''$ for all $x$. 
 Now, let us construct a trace $\tau''=(\tran', \tpo', \trf', \hb{}'')$,
 where $\hb{}''$ is the smallest set such that $\hb{}_2 \subseteq \hb{}''$ and 
  whenever 
  (i) $t_2 \in p$,  $t_1 {\trf'}^x t_2$ , $\tau_3 \in \tran^{\wt,x}$ and $t_3 ~(\hb{p}')~t_2$, then $t_3 {\hb{p}''}^x t_1$ 
  and (ii) $t_2 \in p$,  $t_1 {\trf'}^x t_2$ , $\tau_3 \in \tran^{\wt,x}$ 
  and $t_1 ~(\hb{p}')~t_3$, 
  then $t_2 {\hb{p}''}^x t_3$ .

   Since $t' [\hb{1p}^x] t''$ iff $t' [\hb{2p}^x] t''$ for all $x$, and all $t', t'' \in \tran_1=\tran_2$, and $\hb{}''$ is the smallest extension of $\hb{}_2$  based 
   on the trace semantics, along with the fact that 
   $\tpo'=\tpo'', \trf'=\trf''$, 
   we obtain for each process $p$, for any two 
   transactions $t_1, t_2 \in \tran'=\tran''$, 
   $t_1 [\hb{p}'] t_2$ iff  $t_1 [\hb{p}''] t_2$.
   This gives $\tau' \equiv \tau''$. 
    
      This also gives $\tau_2 \xrightarrow[]{\alpha_t}_{\cm{-}\sat} \tau''$,
  that is, $\tau_2 \vdash_{\tt{G}} \alpha_t$  and indeed $\alpha_t(\tau_1)=\tau' \equiv \tau''=\alpha_t(\tau_2)$.

\end{proof}

\begin{lemma}
If $\tau \vdash_{\tt{G}} \alpha_{t_1} \cdot \alpha_{t_2}$ and $\alpha_{t_1} {\sim} \alpha_{t_2}$ then $\tau \vdash_{\tt{G}} \alpha_{t_2} {\cdot} \alpha_{t_1}$ and $(\alpha_{t_1} {\cdot} \alpha_{t_2})(\tau) \equiv (\alpha_{t_2}{ \cdot} \alpha_{t_1})(\tau)$.
\label{lemma cm-g-2}
\end{lemma}
\begin{proof}
Let $\tau = \langle \tran , \tpo, \trf,\hb{} \rangle$ be a trace and let $t_1$ be a transaction issued in process $p_1$, and $t_2$ be a transaction issued in process $p_2$, with $p_1 \neq p_2$. Assume $\tau \vdash_{\tt{G}} \alpha_{t_1} \cdot \alpha_{t_2}$, with $\alpha_{t_1} {\sim} \alpha_{t_2}$.

We consider the following cases.
\begin{itemize}
    \item [$\bullet$] $t_1,t_2 \notin \tran^{rd}$. In this case $\tau \vdash_{\tt{G}} \alpha_{t_2} \cdot \alpha_{t_1}$ holds trivially.
    
     $(\alpha_{t_1} \cdot \alpha_{t_2})(\tau) = (\alpha_{t_2} \cdot \alpha_{t_1})(\tau) = \langle \tran' , \tpo' , \trf, \hb{} \rangle$, 
     where $\tran' = \tran \cup \{t_1,t_2\}$ and $\tpo' = \tpo \cup \{ (t,t_1) | t \in p_1\} \cup \{(t',t_2) | t' \in p_2 \}$.
    \item [$\bullet$] $t_1 \in \tran^{wt}$, $t_1 \notin \tran^{rd}$ and $t_2 \in \tran^{rd} \cap \tran^{wt}$. 
    Let $\tau_1 = \alpha_{t_1}(\tau)$. 
    We know that $\tau_1 = \langle \tran_1 , \tpo_1 , \trf_1,\hb{}_1 \rangle$, where 
    $\tran_1 = \tran \cup \{ t_1 \}$, and $\tpo_1 = \tpo \cup \{ (t,t_1) | t \in p_1\}$, with $\trf_1 = \trf$ since there are no read events in $t_1$, and $\hb{}_1 = \hb{p}$ since by the trace semantics, when $ \trf$ remain same, there are no $\hb{p}$ edges to be added.
    
    Consider $\tau_2 = \alpha_{t_2}(\tau_1) = \langle \tran_2, \tpo_2, \trf_2, \hb{}_2 \rangle$. 
    Since $t_2 \in \tran^{rd}$, we know that $t_2$ has read events.
    \smallskip

    Let there be $n$ read events $ev_1 \dots ev_n$ in transaction $t_2$, 
    reading from transactions $t'_1, \dots, t'_n$ on variables $x_1, \dots, x_n$.  
    Hence we get a sequence $\tau_1 \xrightarrow[]{\beginact(t_2)}_{\cm{-}\sat} \dots \tau'_{i_1} \xrightarrow[]{\readact(t_2,t'_1)}_{\cm{-}\sat}  \tau_{i_1}  \xrightarrow[]{}_{\cm{-}\sat}\dots  
    \tau'_{i_n} \xrightarrow[]{\readact(t_2,t'_n)}_{\cm{-}\sat} \tau_{i_n} \dots \xrightarrow[]{\commitact(t_2)}_{\cm{-}\sat} \tau_2$.
    Hence $\tran_2 = \tran_1 \cup \{t_2\}$, $\tpo_2 = \tpo_1 \cup \{(t',t_2) | t' \in p_2\}$, $\trf_2 = \trf_1 \cup (\trf_{i_1} \cup \trf_{i_2} \dots \cup \trf_{i_n})$, and $\hb{}_2 = \hb{}_1 \cup (\hb{}_{i_1} \cup \hb{}_{i_2} \dots \cup \hb{}_{i_n})$.
    Since $\tau_1 \vdash_{\tt{G}} \alpha_{t_2}$,  for all $ev_i : 1 \le i \le n$, we have $t'_i \in \rbl(\tau_1,t_2,x_i)$.
    
    Since $\alpha_{t_1} \sim \alpha_{t_2}$, we know that for all $ev_i : 1 \le i \le n$, we have $t'_i \in \rbl(\tau,t_2,x_i)$. It follows that $\tau \vdash_{\tt{G}} \alpha_{t_2}$.
    Define $\tau_3 = \langle \tran_3, \tpo_3, \trf_2, \hb{}_2 \rangle$, where $
    \tran_3 = \tran \cup \{t_2\}$ and $\tpo_3 = \tpo \cup \{(t',t_2) | t' \in p_2\}$. It follows that $\alpha_{t_2}(\tau) = \tau_3$, $\tau_3 \vdash_{\tt{G}} \alpha_{t_1}$ and $\tau_2 = \alpha_{t_1}(\tau_3) = \alpha_{t_2}(\alpha_{t_1}(\tau))$.
    \item [$\bullet$] $t_2 \in \tran^{wt}$ and $t_1 \in \tran^{rd} \cap \tran^{wt}$. Similar to the previous case.
    \item [$\bullet$] $t_1, t_2 \in \tran^{rd} \cap \tran^{wt}$. 
    Let $\tau_1 = \alpha_{t_1}(\tau) = \langle \tran_1, \tpo_1, \trf_1, \hb{}_1 \rangle$ and $\tau_2 = \alpha_{t_2}(\tau_1) = \langle \tran_2, \tpo_2, \trf_2, \hb{}_2 \rangle$. 
    Since $t_1 \in \tran^{rd}$, we know that $t_1$ has read events. 
    Let there be $n$ read events $ev_1 \dots ev_n$ in transaction $t_1$, reading from transactions $t'_1, \dots, t'_n$ on variables 
    $x_1, \dots, x_n$. 
    Hence we get the sequence $\tau \xrightarrow[]{\beginact(t_1)}_{\cm{-}\sat} \dots \tau'_{i_1} \xrightarrow[]{\readact(t_1,t'_1)}_{\cm{-}\sat}  \tau_{i_1} \xrightarrow[]{}_{\cm{-}\sat}\dots  
    \tau'_{i_n} \xrightarrow[]{\readact(t_1,t'_n)}_{\cm{-}\sat} \tau_{i_n} \dots \xrightarrow[]{\commitact(t_1)}_{\cm{-}\sat} \tau_1$.
    Since $t_2 \in \tran^{rd}$, we know that $t_2$ has read events. 
    Let there be  $m$ read events $ev'_1 \dots ev'_m$ in transaction $t_2$, reading from transactions $t''_1, \dots, t''_m$ on variables 
    $y_1, \dots, y_m$.  
    Then we get the sequence $\tau_1 \xrightarrow[]{\beginact(t_2)}_{\cm{-}\sat} \dots \tau'_{j_1} \xrightarrow[]{\readact(t_2,t''_1)}_{\cm{-}\sat}  \tau_{j_1}  \xrightarrow[]{}_{\cm{-}\sat} \dots 
    \tau'_{j_n} \xrightarrow[]{\readact(t_2,t''_m)}_{\cm{-}\sat} \tau_{j_n} \dots \xrightarrow[]{\commitact(t_2)}_{\cm{-}\sat} \tau_2$. 
    
    $\mathbf{Proving}$ $\mathtt{ \tau \vdash_{\tt{G}} \alpha_{t_2} {\cdot} \alpha_{t_1}}$
    
    Since $\tau_1 = \alpha_{t_1}(\tau)$, for all read events $ev_i$ such that $ev_i.trans = t_1$, we have $t'_i \in \rbl(\tau,t_1,x_i)$.
    Since $\tau_2 = \alpha_{t_2}(\tau_1)$, for all read events $ev'_j$ such that $ev'_j.trans = t_2$, we have $t''_j \in \rbl(\tau_1,t_2,y_j)$. 
    Since $\alpha_{t_1} \sim \alpha_{t_2}$, 
    it follows that, we have $t''_j \in \rbl(\tau,t_2,y_j)$ for all read events $ev'_j$ such that $ev'_j.trans = t_2$. Hence $\tau \vdash_{\tt{G}} \alpha_{t_2}$. 
    \smallskip

    Consider $\tau_3 = \alpha_{t_2}(\tau) = \langle \tran_3, \tpo_3, \trf_3, \hb{}_3 \rangle $. To show that $\tau \vdash_{\tt{G}} \alpha_{t_2} \alpha_{t_1}$, 
    we show that for all read events $ev_i$ such that $ev_i.trans = t_1$,  $t'_i \in \rbl(\tau_3,t_1,x_i)$. We prove this using contradiction.
     Assume that $\exists t'_i$ s.t. $t'_i \notin \rbl(\tau_3,t_1,x_i)$.

    \medskip 
    \smallskip 
     
    Let $p$ be process of $t_1$ and $pi$ of $t_2$.
    If $t'_i \notin \rbl(\tau_3,t_1,x_i)$, then 
    (1) there is a $t'_k$ such that 
    (i) $t'_k \in \vbl(\tau_3,t_1,x_i)$ and (ii) $t'_i [\hb{p}_3] t'_k$,
    or (2) there are transactions $t_3$ and $t_4$ such that $t_4 \in \tran^{\wt,y}$, $t_3 ~\trf^y_3 t_1$, 
    $t_4 [\hb{p}_3] t'_i $ and $t_3 [\hb{p}_3] t_4 $,
    or (3) there are transactions $t_3$ and $t_4$ such that $t_3 \in \tran^{\wt,x}$, $t_3 [\hb{p}_3] t_1$,
    $t_3 [\hb{p}_3] t_4 $ and $t' \in \tran^{\wt,y}$.
    or (4) there is a  transaction $t_3$ such that $t_3 \in \tran^{\wt,x,y}$, $t_3 \trf^y_3 t_1$,
    and $t' \in \tran^{\wt,x,y}$.
     
     \medskip 
     \smallskip 
     
        Consider case (1).
       If   (i) is true, that is, $t'_k \in \vbl(\tau_3,t_1,x_i)$, since $\alpha_{t_1} \sim \alpha_{t_2}$, we also have $t'_k \in \vbl(\tau,t_1,x_i)$. 
        This, combined with $t'_i \in \rbl(\tau,t_1,x_i)$, gives according to the trace semantics, $t'_k [\hb{}_1] t'_i$.  Going back to (ii), 
    we can have $t'_i [\hb{p}_3] t'_k$ only if we observed one of the following:
    \begin{itemize}
        \item [(a)] There is a path from $t'_i$ to $t'_k$ in $\tau$; that is, 
        $t'_i [\hb{p}]^+ t'_k$. This implies that $t'_i \notin \rbl(\tau,t_1,x_i)$ which is a contradiction.
        \item [(b)] There is direct path from $t'_i$ to $t'_k$ in $\tau_3$ due to the new $\tpo',\trf$ and $\hb{p'}$ edges added due to $t_2$.
        But $t_2$ has no successor in $(\tpo_3\cup\trf_3)$ and $p' \neq p$.
        So, we can have $t'_i [\hb{p}_3] t'_k$ only if we have $t'_i [\hb{p}_3] t'_k$ which is covered in above case.

    \end{itemize}
    
    A similar argument for other cases.
    
    Thus, we have $\neg(t'_i [\hb{p}_3] t'_k)$. This, for all $1 \leq i \leq n$, $t'_i \in \rbl(\tau_3,t_1,x_i)$. Thus, we now have $\tau \vdash_{\tt{G}} \alpha_{t_2} \alpha_{t_1}$. 
    It remains to show that $(\alpha_{t_1} {\cdot} \alpha_{t_2})(\tau) \equiv (\alpha_{t_2}{ \cdot} \alpha_{t_1})(\tau)$. 
    
    $\mathbf{Proving}$ $\mathtt{(\alpha_{t_1} {\cdot} \alpha_{t_2})(\tau) \equiv (\alpha_{t_2}{ \cdot} \alpha_{t_1})(\tau)}$
    
    Define $\tau_4 {=} \alpha_{t_1}(\tau_3) {=} \alpha_{t_1}.\alpha_{t_2}(\tau){=}
         \langle \tran_4, \tpo_4, \trf_4 , \hb{}_4 \rangle$. Recall that 
    $\tau_2{=}$$\alpha_{t_2}.\alpha_{t_1}(\tau)$ =$
      \langle \tran_2,\tpo_2, \trf_2, \hb{}_2 \rangle$.
    
    To show that $(\alpha_{t_1} \cdot \alpha_{t_2})(\tau) \equiv (\alpha_{t_2} \cdot \alpha_{t_1})(\tau)$, we show that $\tran_2=\tran_4, \tpo_2=\tpo_4, \trf_2=\trf_4$, and, for any two transactions 
    $t'_k, t'_i \in \tran_2=\tran_4$,    $t'_k [\hb{2p}] t'_i$ iff 
    $t'_k [\hb{4p}] t'_i$. 
    Of these, trivially, 
    $\tran_2=\tran_4, \tpo_2=\tpo_4$ follow. Since $\alpha_{t_2} \sim \alpha_{t_1}$, we have $\trf_2=\trf_4$. 
    
    Since $p \neq p'$ and $t_2$ has no successor in $\tpo_3 \cup \trf_3$, it follows that $t'_k, t'_i \in \tran_2=\tran_4$,    $t'_k [\hb{2p}] t'_i$ iff 
    $t'_k [\hb{4p}] t'_i$. 
    
\end{itemize}

 The converse direction, that is, whenever $t'_k [\hb{4p}] t'_i$, we also have $t'_k [\hb{2p}] t'_i$ is proved on similar lines. 
    
\end{proof}
\begin{lemma}
If $\tau_1 \equiv \tau_2$, $\alpha_{t_1} \sim \alpha_{t_2}$, and $\tau_1 \vdash_{\tt{G}} (\alpha_{t_1} \cdot \alpha_{t_2})$,  then (i) $\tau_2 \vdash_{\tt{G}} (\alpha_{t_2} \cdot \alpha_{t_1})$ and (ii) $(\alpha_{t_1} \cdot \alpha_{t_2})(\tau_1) \equiv (\alpha_{t_2} \cdot \alpha_{t_1})(\tau_2)$.
\label{lemma cm-g-3}
\end{lemma}
\begin{proof}
Follows from Lemma \ref{lemma cm-g-1} and Lemma \ref{lemma cm-g-2}.
\end{proof}

From Lemma \ref{lemma cm-g-3}, we get following lemma.
\begin{lemma}
If $\tau_1 \equiv \tau_2$, $\pi_{\tt{T^1}} \sim \pi_{\tt{T^2}}$, and $\tau_1 \vdash_{\tt{G}} \pi_{\tt{T^1}}$ then $\tau_2 \vdash_{\tt{G}} \pi_{\tt{T^2}}$ and $\pi_{\tt{T^1}}(\tau_1) \equiv \pi_{\tt{T^2}}(\tau_2)$.
\label{lemma cm-g-4}
\end{lemma}
 
\begin{lemma}
If $\tau \vdash_{\tt{G_T}} \pi_{\tt{T}}$, $\tau \vdash_{\tt{G}} \alpha_t$ then $\pi_{\tt{T}} = \pi_{\tt{T^1}} \cdot \alpha_t \cdot \pi_{\tt{T^2}}$ for some $\pi_{\tt{T^1}}$ and $\pi_{\tt{T^2}}$ where $\pi_{\tt{T^1}}$ is $t$-free.
\label{lemma cm-g-5}
\end{lemma}
Consider observables $\pi_{\tt{T}} = \alpha_{t_1},\alpha_{t_2} \dots \alpha_{t_n}$. 
We write $\pi_{\tt{T}} \lessapprox \pi'_{\tt{T}}$ to represent that $\pi'_{\tt{T}} = \pi_{\tt{T^0}} \cdot \alpha_{t_1} \cdot \pi_{\tt{T^1}} 
\cdot \alpha_{t_2} \cdot \pi_{\tt{T^2}} \dots \alpha_{t_n} \cdot \pi_{\tt{T^n}}$. 
In other words, $\pi_{\tt{T}}$ occurs as a non-contiguous subsequence in $\pi'_{\tt{T}}$. For such a $\pi_{\tt{T}}, \pi'_{\tt{T}}$, 
we define $\pi'_{\tt{T}} \oslash \pi_{\tt{T}} := \pi_{\tt{T^0}} \cdot \pi_{\tt{T^1}} \dots \pi_{\tt{T^n}}$. 
Since elements of $\pi_{\tt{T}}$ and $\pi'_{\tt{T}}$ are distinct, operation $\oslash$ is well defined. 
Let $\pi_{\tt{T}}[i]$ denote the $i$th observable 
in the observation sequence $\pi_{\tt{T}}$, and 
let  $|\pi_{\tt{T}}|$ denote  the number of observables 
in $\pi_{\tt{T}}$.

Let $\alpha_t = \pi_{\tt{T}}[i]$ for some $i: 1 \le i \le |\pi_{\tt{T}}|$. 
We define $\tt{Pre(\pi_{\tt{T}},\alpha_t)}$ as a subsequence $\pi'_{\tt{T}}$ of $\pi_{\tt{T}}$ such that  
(i) $\alpha_t \in \pi'_{\tt{T}}$,
(ii) $\alpha_{t_j} = \pi_{\tt{T}}[j] \in \pi'_{\tt{T}}$ for some $j: 1 \le j < i$ iff there exists $k: j < k \le i$ such that $\alpha_{t_k}=\pi_{\tt{T}}[k]  \in \pi'_{\tt{T}}$ and $\neg(\alpha_{t_j} \sim \alpha_{t_k})$. 
Thus, $\tt{Pre(\pi_{\tt{T}},\alpha_t)}$ consists of $\alpha_t$ and all 
$\alpha_{t'}$ appearing before $\alpha_t$ in $\pi_{\tt{T}}$ such that 
$t' ~\hb{p}~ t$.

%
%
\begin{lemma}
If $\pi_{\tt{T^1}} = \tt{Pre(\pi_{\tt{T}},\alpha_t)}$ and $\pi_{\tt{T^2}} = \pi_{\tt{T}} \oslash \pi_{\tt{T^1}}$ then $\pi_{\tt{T}} \approx \pi_{\tt{T^1}} \cdot \pi_{\tt{T^2}}$.
\label{lemma cm-g-6}
\end{lemma}
\begin{proof}
The proof is trivial since we can always execute in order, 
the $~\hb{p}~$ predecessors of $\alpha_t$ from $\tt{Pre(\pi_{\tt{T}},\alpha_t)}$, then $\alpha_t$, then the  observables 
in $\tt{Pre(\pi_{\tt{T}},\alpha_t)}$ which are independent from $\alpha_t$, followed by the suffix of $\pi_{\tt{T}}$ after $\alpha_t$.

\end{proof}

\begin{lemma}
If $\pi_{\tt{T}}= \pi'_{\tt{T}} \cdot \pi''_T$ and $\alpha_t \in \pi'_{\tt{T}}$ then $\tt{Pre(\pi_{\tt{T}},\alpha_t)} = \tt{Pre(\pi'_{\tt{T}},\alpha_t)}$.
\label{lemma cm-g-7}
\end{lemma}
\begin{proof}
The proof is trivial once again, since 	all 
the transactions which are $~\hb{p}~ t$  are in the prefix  $\pi'_{\tt{T}}$.
\end{proof}

\begin{lemma}
If $\tau \vdash_{\tt{G}} \pi_{\tt{T}}$ then $\tau \vdash_{\tt{G_T}} \pi_{\tt{T}} \cdot \pi_{\tt{T^2}}$.
\label{lemma cm-g-8}
\end{lemma}
\begin{proof}
This simply follows from the fact that we can extend the observation sequence 
$\pi_{\tt{T}}$ to obtain a terminal configuration, since the $\cm$ DPOR algorithm generates traces corresponding to terminating runs.   	
\end{proof}

\begin{lemma}
Consider a trace trace $\tau$ such that $\tau \models \cm$, $\tau \vdash_{\tt{G}} \pi_{\tt{T}} \cdot \alpha_t$. 
Let each read event $ev_i = r_i(p,t,x_i,v)$ in $\alpha_t$ read from some transaction $t_i$, 
and let $\pi_{\tt{T}}$ be $t$-free.  
Then $\tau \vdash_{\tt{G}} \alpha'_t \cdot \pi_{\tt{T}}$, where $\alpha'_t$ is the same as  $\alpha_t$, 
with the exception that the sources of its read events can be different. That is, 
 each read event $ev_i = r_i(p,t,x_i,v) \in \alpha'_t$ can read from some transaction $t'_i \neq t_i$.
\label{lemma cm-g-9}
\end{lemma}
\begin{proof}
Let $\pi_{\tt{T}} = \alpha_{t_1} \alpha_{t_2} \dots \alpha_{t_n}$ and $\tau_0 \xrightarrow[]{\alpha_{t_1}}_{\cm{-}\sat} \tau_1 \xrightarrow[]{\alpha_{t_2}}_{\cm{-}\sat} \dots \xrightarrow[]{\alpha_{t_n}}_{\cm{-}\sat} \tau_n \xrightarrow[]{\alpha_t}_{\cm{-}\sat} \tau_{n+1}$, where $\tau_0 = \tau$. Let $\tau_i = \langle \tran_i, \tpo_i, \trf_i, \hb{}_i \rangle$. 
Let there be  $m$ read events in transaction $t$. 
Let $p$ be process $t$.
Keeping in mind what we want to prove, where we want to execute $t$ first followed by $\pi_{\tt{T}}$, and obtain an execution $\alpha'_t \pi_{\tt{T}}$, we do the following.

For each read event $ev_i = r_i(p,t,x_i,v)$ we define $t'_i \in \tran^{wt,x_i}$ such that 
$(i)$ $ t'_i \in \rbl(\tau_0,t,x_i)$, and 
$(ii)$ there is no $t''_i \in \rbl(\tau_0,t,x_i)$ 
where $t'_i [\hb{np}] t''_i$.  Note that this is possible since 
 $\pi_{\tt{T}}$ is $t$-free, so all the observables $\alpha_{t_i}$ occurring in $\pi_{\tt{T}}$ are such that $t_i$ is issued in a process other than that of $t$. Thus, when $\alpha_t$ is enabled, we can choose any of the 
 writes done earlier,  this fact is consistent with the $\cm$ semantics, since from Lemmas \ref{lem:cm-f3}-\ref{lem:cm-f4} we know $\tau_n \models \cm$.
Since $\tau_n \models \cm$ such a $t'_i \in  \rbl(\tau_0,t,x_i)$ exists for each $ev_i$.

\medskip 
Next we define a sequence of traces which can give the execution 
$\alpha'_t \pi_{\tt{T}}$. 
Define a sequence of traces $\tau'_0, \tau'_1, \dots \tau'_n$ where $\tau'_j = \langle \tran'_j, \tpo'_j, \trf'_j, \hb{}'_j \rangle$ is such that 
\begin{enumerate}
	\item $\tau \vdash_{\tt{G}} {\alpha'_t}$ and $\alpha'_t(\tau)=\tau'_0$,
	\item $\tau'_j \vdash_{\tt{G}} \alpha_{t_{j+1}}$ and  $\tau'_{j+1} = \alpha_{t_{j+1}}(\tau'_j)$ for all $j: 0 \le j \le n$.
\end{enumerate}

\begin{itemize}
\item  Define $\tran'_i = \tran_i \cup \{t\}$ for all $0 \leq i \leq n$, 
\item 
Define $\tpo' = \{t' [\tpo] t \mid t'$ is a transaction in $\tau$, in the same process as $t\}$, and $\tpo'_i = \tpo_i \cup \tpo'$ for all $0 \leq i \leq n$, 

\item  $\trf' = \bigcup_{i:1\le i \le m} t'_i [\trf^{x_i}] t$ (reads from relation corresponding to each read event $ev_i \in \alpha'_t$), and 
 $\trf'_i = \trf_i \cup \trf'$ for all $i : 1 \le i \le n$, 
 
\item  $\hb{p}' = \bigcup_{j:1\le j \le m} \hb{p}'_j$ (each  $\hb{p}'_j$ corresponds to the updated $\hb{p}$ relation because of the read transition $\xrightarrow[]{\readact(t,t'_j)}_{\cm{-}\sat}$). 	
\smallskip 

For $1 \leq i \leq n$, we define $\hb{}'_i$ inductively as $\hb{}'_i = \hb{}_i \cup \hb{}'$
and show that $\tau'_i \vdash_{\tt{G}} \alpha_{t_{i+1}}$ and  $\tau'_{i+1} = \alpha_{t_{i+1}}(\tau'_i)$ holds good. 
 First define $\hb{}'_0 = \hb{}_0 \cup \hb{p}'$. 
 Then define  $\hb{}'_i$ to be the coherence order corresponding to $\alpha_{t_{i}}(\tau'_{i-1})$  
 where $\tau'_{i-1}=\langle \tran'_{{i-1}}, \tpo'_{i-1}, \trf'_{i-1}, \hb{}'_{i-1} \rangle$
 for all $i: 1 \le i \le n$. 

 \end{itemize}

\medskip 
\noindent{\bf{Base case}}. The base case $\tau \vdash_{\tt{G}} {\alpha'_t}$ and $\alpha'_t(\tau)=\tau'_0$,
   holds trivially, by construction. 

\smallskip 

\noindent{\bf{Inductive hypothesis}}. Assume that $\tau'_j \vdash_{\tt{G}} \alpha_{t_{j+1}}$ and  $\tau'_{j+1} = \alpha_{t_{j+1}}(\tau'_j)$ for $0 \leq j \leq i-1$. 

\smallskip 

 We have to prove that 
$\tau'_i \vdash_{\tt{G}} \alpha_{t_{i+1}}$ and  $\tau'_{i+1} = \alpha_{t_{i+1}}(\tau'_i)$. 
\begin{enumerate}
	\item If $t_{i+1} \in \tran^{wt}$ and $t_{i+1} \notin \tran^{rd}$,  then the proof holds trivially from the inductive hypothesis. 
	\item Consider now $t_{i+1} \in \tran^{rd}$. Consider a read event 
	$ev^{i+1} = r(p,t_{i+1},x,v) \in \alpha_{t_{i+1}}$ which was reading from 
	transaction $t_x$ in $\pi_{\tt{T}}$. That is, we had $t_x \in \rbl(\tau_i, t_{i+1},x)$. 	If $t_x \in \rbl(\tau'_i, t_{i+1},x)$, then we are done, since we can simply extend the run from the inductive hypothesis. 
	
	\smallskip 
	
	Assume otherwise. That is, $t_x \notin \rbl(\tau'_i, t_{i+1},x)$.

	Since  $t_x$  $\in$   $\rbl(\tau_i, t_{i+1},x)$, there are some blocking transactions in the new path, which prevents $t_x$ from being readable. 
 We have the blocking transactions since $t$ is moved before $\pi_{\tt{T}}$. 	
 But $t$ updates only $\hb{p}$ and for each process of $t_{i+1}$ is not equal $p$. 
 And $t$ has no successor in  each $[\tpo'_j \cup \trf'_j]^+$.
	
Hence, $t_x$  $\in$   $\rbl(\tau_i, t_{i+1},x)$ and we can extend the run 
from the inductive hypothesis obtaining $\tau'_j \vdash_{\tt{G}} \alpha_{t_{j+1}}$,  $\tau'_{j+1} = \alpha_{t_{j+1}}(\tau'_j)$,  
for $0 \leq j \leq i-1$,
and 
$\tau'_i \vdash_{\tt{G}} \alpha_{t_{i+1}}$.

\end{enumerate}

Now we prove that $\hb{}'_{i+1} = \hb{}_{i+1} \cup \hb{}'$. 
Assume we have $(tr1 , tr2) \in \hb{}_{i+1}$. 
We show that $(tr1 , tr2) \in \hb{}'_{i+1}$.
Let $p1 \neq p$ be the process such that $(tr1 , tr2) \in \hb{}'_{i+1}$ is because of $(tr1 , tr2) \in \hb{p1}'_{i+1}$.
We have the following cases:

\begin{itemize}
    \item[(i)] $(tr1 , tr2) \in \hb{}_i$. From the inductive hypothesis it follows that $(tr1 , tr2) \in \hb{}'_{i}$ and hence in $\hb{}'_{i+1}$.
    \item[(ii)] $(tr1 , tr2) \notin \hb{}_i$. 
    Since $(tr1 , tr2) \notin \hb{}_i$ and $(tr1 , tr2) \in \hb{}'_{i}$, it follows that moving $t$ before $\pi_{\tt{T}}$ caused this.
    But $t$ only updates $\hb{p}$ edges and $p1 = p$. 
    Since $\pi_{\tt{T}}$ be $t$-free, we get a contradiction. 
\end{itemize}
Thus, we have shown that $\tau \vdash_{\tt{G}} \alpha'_t \alpha_{t_1} \dots \alpha_{t_i}$ 
is such that $\alpha'_t \alpha_{t_1} \dots \alpha_{t_i}(\tau)=\tau'_i$ for all 
$0 \leq i \leq n$. When $i=n$ we obtain 
$\alpha'_t \pi_{\tt{T}}(\tau)=\tau'_n$, or $\tau \vdash_{\tt{G}} \alpha'_t \pi_{\tt{T}}$. 
\end{proof}


\begin{definition}[Linearization of  a Trace]
A observation sequence $\pi_{\tt{T}}$ is a \emph{linearization} of a trace $\tau = \langle \tran, \tpo , \trf, \hb{p} \rangle$ if $(i)$ $\pi_{\tt{T}}$ has the same transactions as $\tau$ and $(ii)$ $\pi_{\tt{T}}$ follows the ($\tpo \cup \trf)$ relation.
	
\end{definition}

We say that our DPOR algorithm \emph{generates an observation sequence} $\pi_{\tt{T}}$ from state $\tau$ where $\tau$ is a trace 
if 
 it invokes $\explore$ with parameters $\cm,\tau, \pi'$, where $\pi'$ is a \emph{linearization} of $\tau$, and generates  a sequence of recursive calls to $\explore$ resulting in $\pi_{\tt{T}}$.

\begin{lemma}
If $\tau \vdash_{\tt{G_T}} \pi_{\tt{T}}$ then, 
the DPOR algorithm  generates $\pi'_{\tt{T}}$ from state $\tau$ 
for some $\pi'_{\tt{T}} \approx \pi_{\tt{T}}$.
\label{lemma cm-g-10}
\end{lemma}
\begin{proof}
We use induction on $|\pi_{\tt{T}}|$. If $\tau$ is $terminal$, then the  proof is trivial.
Assume that we have $\tau \vdash_{\tt{G_T}} \pi_{\tt{T}}$.
Assume that $\tau \vdash_{\tt{G}} \alpha_t$. It follows that $\tau \vdash_{\tt{G}} \alpha_t$.
Using Lemma \ref{lemma cm-g-5} we get $\pi_{\tt{T}} = \pi_{\tt{T^1}} \cdot \alpha_t \cdot \pi_{\tt{T^2}}$, where $\pi_{\tt{T^1}}$ is $t$-free. 
We consider the following two cases:
\begin{itemize}
    \item In the first case, we assume that all read events in $t$ read from 
    transactions in $\tau$. So in this case, $\pi_{\tt{T^1}}$ can be empty. 
        Assume there are $n$ read events in $t$ and each read event $ev_i = r(p,t,x_i,v)$ 
    reads from transactions $t_i \in \tau$ for all $i: 1 \le i \le n$. 
In this case the DPOR algorithm will let each read event $ev_i$ read from all possible transactions $t' \in \rbl(\tau,t,x_i)$, including $t_i$.
    \item In the second case, assume that there exists at least one read event $ev' = r(p,t,x,v) \in \alpha_t$ which reads from a transaction $t' \in \pi_{\tt{T^1}}$ (hence, $\pi_{\tt{T^1}}$ is non empty).
    \smallskip 
     
    From Lemma \ref{lemma cm-g-9}, we know that $\tau
    \vdash_{\tt{G}} \alpha'_t \cdot \pi_{\tt{T^1}}$. 
    From Lemma \ref{lemma cm-g-8}, it follows that $\tau
    \vdash_{\tt{G_T}} \alpha'_t \cdot \pi_{\tt{T^1}} \cdot \pi_{\tt{T^4}}$, for some
    $\pi_{\tt{T^4}}$.
    Hence there exists $\tau'$ such that 
    $\tau' = \alpha'_t (\tau)$ and $\tau' \vdash_{\tt{G_T}} \pi_{\tt{T^1}} \cdot
    \pi_{\tt{T^4}}$. 
    
    Since $|\pi_{\tt{T^1}} \cdot
    \pi_{\tt{T^4}}| < |\pi_{\tt{T}}$, we can use the inductive hypothesis. It follows that the DPOR algorithm generates from state $\tau'$, 
    the observation sequence  $\pi_{\tt{T^5}}$ such that $\pi_{\tt{T^5}} \approx \pi_{\tt{T^1}} 
    \cdot \pi_{\tt{T^4}}$. 
    
    \smallskip 
    
    Let  $\pi_{\tt{T^3}}.\alpha_t = \tt{Pre(\pi_{\tt{T^1}}, \alpha_t)}$. Then  $\pi_{\tt{T^5}}
    \approx \pi_{\tt{T^1}} \cdot \pi_{\tt{T^4}}$ implies that  $\pi_{\tt{T^3}}
    \lessapprox \pi_{\tt{T^5}}$ (by applying Lemma \ref{lemma cm-g-7}).
    Let $\pi_{\tt{T^6}} \approx \pi_{\tt{T}} \oslash (\pi_{\tt{T^3}} \cdot \alpha_t)$.
    By applying Lemma \ref{lemma cm-g-6}, we get $\pi_{\tt{T}} \approx \pi_{\tt{T^3}} \cdot \alpha_t \cdot \pi_{\tt{T^6}}$.
    Since $\tau \vdash_{\tt{G_T}} \pi_{\tt{T}}$ and 
    $\pi_{\tt{T}} \approx \pi_{\tt{T^3}} \cdot \alpha_t \cdot \pi_{\tt{T^6}}$,
    by applying Lemma \ref{lemma cm-g-3}, we get $\tau
    \vdash_{\tt{G_T}} \pi_{\tt{T^3}} \cdot \alpha_t \cdot \pi_{\tt{T^6}}$.
    Let $\tau' = (\pi_{\tt{T^3}} \cdot \alpha_t) (\tau)$. 
    Since $\tau \vdash_{\tt{G_T}} \pi_{\tt{T^3}} 
    \cdot \alpha_t \cdot \pi_{\tt{T^6}}$, we have $\tau'
    \vdash_{\tt{G_T}} \pi_{\tt{T^6}}$. 
    From inductive hypothesis it follows that $\tau'$ generates $\pi_{\tt{T^7}}$ such that $\pi_{\tt{T^7}} \approx \pi_{\tt{T^6}}$.\\
    In other words, $\tau$ generates $\pi_{\tt{T^3}} 
    \cdot \alpha_t \cdot \pi_{\tt{T^7}}$ where $\pi_{\tt{T}} \approx \pi_{\tt{T^3}} 
    \cdot \alpha_t \cdot \pi_{\tt{T^7}}$.
\end{itemize}

\end{proof}

%% file: ra-proofs.tex
\newpage
\centerline{\bf{Read Atomic }} 
\section{Trace Semantics}
\label{app:rasatcons}

  \begin{figure}[H]
        \begin{tikzpicture}[node distance=10mm , thick, main/.style= {draw,circle},];
        \node[] (t) {$t$};
        \node[] (t1) [above = 20mm of t] {$t_1$};
        \node[] (t') [left = 25mm of t] {$t'$};
        \node[] (t2) [above = 20mm of t']  {$t_2$};
        \node[] (1) [below=5mm of t] {};
        \node[] (cap) [left=7mm of 1] {$\tt{Cond4}$};
        \draw[->,blue,line width=1.2pt] (t1) -- node[right] {$\trf$} (t);
        \draw[->,black,line width=1.2pt] (t1) -- node[above] {$(\tpo \cup \trf \cup \cfr)^+$} (t2);
        \draw[->,black,line width=1.2pt] (t2) --node[rotate=-40, below] {$(\tpo \cup \trf)$} (t);
        \draw[dashed , ->,blue,line width=1.2pt] (t') -- node[below] {$\trf$} (t);
        \draw[dashed , ->,red,line width=1.2pt] (t2) -- node[left] {$\cfr$} (t');
        \draw[dashed , ->,red,line width=1.2pt] (t') -- node[rotate=40,right=3mm] {$\cfr$} (t1);

        \node[] (3at') [right= 5mm of t1] {$t'$};
        \node[] (3at1) [right= 20mm of 3at'] {$t1$};
        \node[] (31) [below=20mm of 3at1] {};
        \node[] (3at) [left=8mm of 31] {$t$};
        \node[] (3acap) [below=5mm of 3at] {$\tt{Cond3 (a)}$};

        \draw[->,black,line width=1.2pt] (3at1) -- node[above] {$(\tpo \cup \trf \cup \cfr)^+$} (3at');
         \draw[->,blue,line width=1.2pt] (3at1) -- node[right] {$\trf$} (3at);
        \draw[dashed , ->,blue,line width=1.2pt] (3at') -- node[left=1mm] {$\trf$} (3at);
        \draw[dashed , ->,red,line width=1.2pt] (3at') to [out=90, in=90] node[above] {$\cfr$} (3at1);
        
        \node[] (2at') [right= 5mm of 3at1] {$t'$};
        \node[] (2at1) [right= 20mm of 2at'] {$t1$};
        \node[] (21) [below=20mm of 2at1] {};
        \node[] (2at) [left=8mm of 21] {$t$};
        \node[] (2acap) [below=5mm of 2at] {$\tt{Cond2}$};
        
        \draw[dashed , ->,red,line width=1.2pt] (2at1) -- node[above] {$\cfr$} (2at');
         \draw[->,blue,line width=1.2pt] (2at1) -- node[right] {$\trf$} (2at);
        \draw[dashed , ->,blue,line width=1.2pt] (2at') -- node[left=1mm] {$\trf$} (2at);
        \draw[dashed , ->,red,line width=1.2pt] (2at') to [out=90, in=90] node[above] {$\cfr$} (2at1);

        \node[] (1t') [right= 5mm of 2at1] {$t'$};
        \node[] (1t1) [right= 20mm of 1t'] {$t1$};
         \node[] (111) [below=20mm of 1t1] {};
        \node[] (1t) [left=8mm of 111] {$t$};
        \node[] (1cap) [below=5mm of 1t] {$\tt{Cond1}$};

        \draw[->,black,line width=1.2pt] (1t') -- node[above] {$(\tpo \cup \trf \cup \cfr)^+$} (1t1);
        \draw[->,black,line width=1.2pt] (1t1) -- node[right] {$(\tpo \cup \trf)$} (1t);
        \draw[dashed , ->,blue,line width=1.2pt] (1t') -- node[left=1mm] {$\trf$} (1t);
        \draw[dashed , ->,red,line width=1.2pt] (1t1) to [out=90, in=90] node[above] {$\cfr$} (1t');
        \end{tikzpicture}
        \caption{Readability for $\readat$}
        \label{read-ra}
    \end{figure}

\subsection{Readability and Visibility  for $\readat$}
 \label{rato:rv}
 For $X=\readat$, $\rbl(\tau^t_X, t, x)$ is defined as the set of all transactions  $t' \in \tran^{\wt,x}$ s.t.
\begin{itemize}
	\item[$\tt{cond1}$] there is no transaction $t'' \in \tran^{\wt,x}$ such that 
	$t' ~(\co~\cup \cfr)^+~t''~(\tpo \cup \trf)~t$. Having such a $t''$ with $t'~\trf~t$ gives $t''~\cfr~t'$ and $\readatcyc$. 
	\item[$\tt{cond2}$]  If $t' \in \tran^{\wt,y}$ (for $y \neq x$) there is no transaction $t'' \in \tran^{\wt,x} \cap \tran^{\wt,y}$ such that $t''~\trf~t$ (wrt $y$). Allowing $t'~\trf~t$ in this case creates $t'~\cfr~t''$ and 
	$t''~\cfr~t'$ and $\readatcyc$.
	\item[$\tt{cond3a}$] If $t' \in \tran^{\wt,y}$ (for $y \neq x$), there is no transaction $t'' \in \tran^{\wt,y}$ such that $t''~\trf~t$ (wrt $y$) and $t''~(\co \cup \cfr)^+t'$. Having such a $t''$ with $t'~\trf~t$ gives 
	$t'~\cfr~t''$ and $\readatcyc$.
	\item[$\tt{cond4}$] If $t' \in \tran^{\wt,y}$ (for $y \neq x$), there are no transactions $t_1 \in \tran^{\wt,y}, t_2 \in \tran^{\wt,x}$ such that $t_1~\trf~t$ (wrt $y$), $t_1~(\co \cup \cfr)^+~t_2~(\tpo\cup \trf)~t$. 
	 If so, allowing $t'~\trf~t$ gives $t_2~\cfr~t'$ ($t_2, t' \in \tran^{\wt,x}$) as well as 
	 $t'~\cfr~t_1$  ($t_1, t' \in \tran^{\wt,y}$). This gives 
	 $t_1~(\co \cup \cfr)^+~t_2~\cfr~t'\cfr~t_1$ and $\readatcyc$. 
\end{itemize}

Figure \ref{read-ra} explains cycles created after violating conditions $\tt{cond1}-\tt{cond4}$.

After adding $t'~\trf~t$, we must check that 
there are no consistency violations.  
The check set $\vbl(\tau^t_X, t,x)$ is defined as the set of transactions which turn 	``sensitive''  
on adding the new  edge  $t'~\trf~t$. 
Unless  appropriate edges are added involving 
these sensitive transactions, we may get consistency violating cycles
in the resultant trace.
Let $\tau^{tt'}$ denote the trace obtained by adding the new transaction $t$ and the edge $t'~\trf~t$ to trace $\tau$. Now, we 
identify $\vbl(\tau^t_X, t,x)$ and the edges which must be added to $\tau^{tt'}$ to obtain a consistent extended trace.

 For $X=\readat$, $\vbl(\tau^t_X,t,x)$ is classified into two categories. 
	\begin{itemize}
	\item The first kind of transactions in $\vbl(\tau^t_X,t,x)$  are  
$\{t''\in \rbl(\tau^t_X, t,x) \mid t''~\tpo \cup \trf~t\}$.  Then we add  from each $t'' \in \vbl(\tau^t_X,t,x)$ a $\cfr$ edge to $t'$ in $\extend_X(\tau^{tt'})$. 
	\item The second kind of transactions in $\vbl(\tau^t_X,t,x)$  are
	transactions $t'' \in \tran^{\wt,y}$ such that $t''~\trf~t$ when 
	$t' \in  \tran^{\wt,y}$. Then we add $t'~\cfr~t''$ to $\extend_X(\tau^{tt'})$.
		\end{itemize}

\begin{lemma}
    Given trace $\tau$, a transaction $t$ and a variable $x$, we can construct the sets  $\rbl(\tau, t, x)$ and $\vbl(\tau, t, x)$ in polynomial time.
    \label{lemma:ra-poly}
\end{lemma}
\begin{proof}
    We prove that set $\rbl(\tau, t, x)$ can be computed in polynomial time ($O(|\tran|^3)$ time), by designing an algorithm to generate set $\rbl(\tau, t, x)$. The algorithm consists of the following  steps:
    \begin{itemize}
        \item[(i)] First we compute the  transitive closure of the relations
        $[\tpo \cup \trf \cup \cfr]$, i.e, we compute $[\tpo \cup \trf \cup \cfr ]^+$.  We can use the Floyd-Warshall algorithm \cite{10.5555/1614191} to compute the transitive closure. This will take $O(|\tran|^3)$ time.
        \item[(ii)] We compute the set $\tran'=\{t'\mid t' [\tpo \cup \trf] t\}$. This takes $O(|\tran|)$ time.
    
        \item[(iii)] For each transaction $t' \in \tran^{wt,x}$, we perform following checks:
        \begin{itemize}
            \item Check $\mathtt{cond 1 :} $ Check whether there exists a transaction $t'' \in \tran' \cap \tran^{wt,x}$            
            with $t'$ $[\tpo  \cup \trf \cup \cfr]^+$ $t''$. 
            If there is no such $t''$ then proceed to $\mathtt{cond 2}$. 
            If we find such $t''$, then $t' \notin \rbl(\tau , t, x)$. 
            This step will take $O(|\tran|^2)$ time.
            \item Check $\mathtt{cond 2 :}$ Check whether there exists a transaction $t''$ such that $t'' [\trf^y] t$ and $t',t'' \in \tran^{wt,x} \cap \tran^{wt,y}$. 
            If there is no such $t''$ then proceed to $\mathtt{cond 3}$. 
            If we find such $t''$, then $t' \notin \rbl(\tau , t, x)$. 
            This step will take $O(|\tran|)$ time.
            \item Check $\mathtt{cond 3a :}$ 
            
            Check whether there exists a transaction $t''$ such that $t'' [\trf^y] t$ , $t' \in \tran^{wt,y}$, and $t'' [\tpo \cup \trf \cup \cfr]^+ t'$. 
            If there is no such $t''$ then proceed to $\mathtt{cond 4}$. 
            If we find such $t''$, then $t' \notin \rbl(\tau , t, x)$. 
            This step will take $O(|\tran|)$ time.
            
            \item Check $\mathtt{cond 4 :}$ 
            Check whether there are transactions $t''$ and $t''' \in \tran'$ such that $t'' [\trf^y] t$, $t''' \in \tran^{\wt,x}$, $t' \in \tran^{\wt,y}$ and $t'' [\tpo \cup \trf \cup \cfr]^+ t'''$. 
            If we find such $t''$ and $t'''$, 
            then $t' \notin$ $\rbl(\tau , t, x)$. 
            If there are no such $t''$ and $t'''$ then add $t'$ to $\rbl(\tau , t, x)$.
            This step will take $O(|\tran|^2)$ time.
        \end{itemize}
    \end{itemize}
    Similarly, we can compute $\vbl(\tau, t, x)$ in polynomial time. This completes the proof.
\end{proof}

\subsection{Properties of the Trace Semantics}
\label{rato:prop}
\begin{lemma}
    If $\tau_1 \models \readat$  and $\tau_1$ $\xrightarrow[]{\beginact(t)}_{\readat{-}\sat}$ $\tau_2$ then $\tau_2 \models \readat$.
    \label{lem:ra-f3}
\end{lemma}
\begin{proof}
The proof follows trivially since we do not change the  	reads from and coherence order relations, and the transaction $t$ has no successors. Thus, if $\tau_1 \models \readat$, 
so does $\tau_2$.  
	\end{proof}

\begin{lemma}
    If $\tau_1 \models \readat$  and $\tau_1$ $\xrightarrow[]{\readact(t,t')}_{\readat{-}\sat}$ $\tau_2$ then $\tau_2 \models \readat$.
    \label{lem:ra-f4}
\end{lemma}

\begin{proof}
    Let $\tau_1 = \langle \tran_{1} , \tpo_1 , \trf_1, \cfr_1 \rangle$, $\tau_2 = \langle \tran_{2} , \tpo_2 , \trf_2, \cfr_2 \rangle$, and $ev$ = $r(p,t,x,v)$ be a read event in transaction $t \in \tran_1$. Suppose $\tau_2 \nvDash \readat$, and  that $\tpo_2 \cup \trf_2 \cup \cfr_2$ is cyclic. Since $\tau_1 \models \readat$, and $t$ has no outgoing edges, 
   on adding $(t',t) \in \trf_2$, 
    it follows that the cycle in $\tau_2$ is as a result 
    of the newly added $\cfr_2$ edges. 
    
    We prove that 
    on adding $(t',t) \in \trf_2$, 
        such cycles are possible iff at least one of ${\tt{cond}_1}, \dots, {\tt{cond}_4}$ are true. 
    
 \subsection*{${\tt{cond1}}\vee \dots \vee {\tt{cond4}}$ induces $\tpo_2 \cup \trf_2 \cup \cfr_2$ cyclicity }
   \label{sec:ra-cyc1}
    This direction is easy to see : assume 
    one of ${\tt{cond1}}, \dots, {\tt{cond4}}$ are true; then, as already 
    argued in the main paper, we will get a $\tpo_2 \cup \trf_2 \cup \cfr_2$ cycle on adding the $\trf_2$ edge from $t'$ to $t$ and we are done.

 \subsection*{$\tpo_2 \cup \trf_2 \cup \cfr_2$ cyclicity implies ${\tt{cond1}}\vee \dots \vee {\tt{cond4}}$ }
 \label{sec:ra-cyc2}
  For the converse direction, assume that we add the $\trf_2$ edge from $t'$ to $t$ and obtain a $\tpo_2 \cup \trf_2 \cup \cfr_2$ cycle in $\tau_2$. We now argue that this cycle has been formed because one of 
 ${\tt{cond}_1}, \dots, {\tt{cond}_4}$ are true.  
 
 First of all, note that $t$ has no outgoing edges in $\tau_2$, since 
 it is the current transaction being executed. Also, 
 we know that $\tau_1$ has no $\tpo_1 \cup \trf_1 \cup \cfr_1$ cycles. 
 Thus, the cyclicity of $\tpo_2 \cup \trf_2 \cup \cfr_2$ is induced by the 
 newly added $\trf_2$ edge as well as the newly added 
  $\cfr_2$ edges. Note that adding the $\trf_2$ edge  to $\tau_1$ does not induce any cycle since $t$ has no outgoing edges. Let's analyze the 
  $\cfr_2$ edges added, which induce cycles, and argue that 
  one of ${\tt{cond1}}, \dots, {\tt{cond4}}$ will be true.   
  
  \begin{enumerate}
  	\item We add  $\cfr_2$ edges from $t'' \in \vbl(\tau_1, t,x)$ to $t'$. 
  	For  these edges to induce a cycle, we should have a path from $t'$ to $t''$ in $\tau_1$. That is, we have $t' [\tpo_1 \cup \trf_1 \cup \cfr_1]^+ t''$. This is captured by ${\tt{cond1}}$.  
  	
  	\item Consider $y \neq x$. If $t' \in \tran^{w,x} \cap \tran^{w,y}$, 
  	and we have $t'''\trf_1^y t$.  Then we add 
  	$(t',t''') \in \cfr^y_2$. To get a cycle, we need a path from 
  	$t'''$ to $t'$. 
  	\begin{itemize}
  	\item If $t''' \in \tran^{w,x} \cap \tran^{w,y}$, then we add $(t''', t') \in \cfr^x_2$, resulting in a cycle. This is handled by  ${\tt{cond2}}$.
  	\item If we have a path 
  	$t''' [\tpo_1 \cup \trf_1 \cup \cfr_1]^+ t'$, then we get a cycle again. This is handled by ${\tt{cond3}}$. 
  \item As a last case, to obtain a path from $t'''$ to $t'$, 
  assume there is a path from $t'''$ to $t$ in $\tau_1$, and let $t'' \in \tran^{w,x}$ be the last transaction writing to $x$ in this path.  Note that this will induce a path from $t'''$ to $t''$ to $t'$ : the path  from $t''$ to $t'$ comes by the $\cfr^x_2$ edge added from $t''$ to $t'$ since 
  we have $(t',t) \in \trf_2^x$. Once we get this path from $t'''$ to $t'$, we again have the path we were looking for to get the cycle. This is handled by ${\tt{cond4}}$. 
    		
  	\end{itemize}
Thus, we have shown that obtaining a path from $t'''$ to $t'$ is covered 
by the conditions  ${\tt{cond}_1}, \dots, {\tt{cond}_4}$, and hence, 
a $\tpo_2 \cup \trf_2 \cup \cfr_2$ cycle.  
  	  \end{enumerate}
 
 Thus, we can think of  $\tpo_2 \cup \trf_2 \cup \cfr_2$ cycles as a result of the forbidden patterns described in ${\tt{cond}_i}$ $1 \leq i \leq 4$.  By construction of $\tau_1$ $\xrightarrow[]{\readact(t,t')}_{\readat{-}\sat}$ $\tau_2$, we ensure $\neg {\tt{cond}_1}, \dots, \neg {\tt{cond}_4}$. Hence, 
    $\tau_2 \models \readat$. 
\end{proof}
Define $\big[[ \sigma]\big]^{\readat{-}\sat}$ as the set of  traces  generated  using $\xrightarrow[]{}_{\readat{-}\sat}$ transitions, starting from an empty trace $\tau_{\o}$.

Consider a terminal trace $\tau$ generated by $\readat$ DPOR algorithm starting from the empty trace $\tau_\emptyset$.
That is, there is a sequence  
$\tau_0 \xrightarrow[]{}_{\readat{-}\sat} 
\tau_1 \xrightarrow[]{}_{\readat{-}\sat}  
\tau_2 \xrightarrow[]{}_{\readat{-}\sat} \dots 
 \xrightarrow[]{}_{\readat{-}\sat}  \tau_n$ with 
 $\tau_0=\tau_{\emptyset}$, 
 and $\tau_n=\tau$. 
 Since $\tau_{\emptyset}$ is a empty we have $\tau_{\emptyset} \models \readat$, 
 it follows by  
 Lemmas \ref{lem:ra-f3} and \ref{lem:ra-f4} that $\tau \models \readat$. 
 
 Hence, for each trace $\tau \in \big[[ \sigma]\big]^{\readat{-}\sat} $ we have $\tau \models \readat$.

\subsection{$\readat$ DPOR Completeness}
\label{rato:comp}
In this section, we show the completeness of the $\readat$ DPOR algorithm. More precisely, for the input program under $\readat$ for any terminating run $\rho \in Runs(\conf)$ and trace $\tau$ s.t. $ \rho \models \tau$, 
we show that $\explore(\readat, \tau_{\emptyset}, \epsilon)$ will produce a recursive visit $\explore(\readat, \tau', \pi)$ for some terminal $\tau'$, and $\pi$ where $\tau'=\tau$.   First, we give some definitions and auxiliary lemmas.

Let  $\pi = \alpha_{t_1} \alpha_{t_2} \dots \alpha_{t_n}$ be an observation sequence where  $\alpha_t = \beginact(p,t) \dots$ \plog{end}$(p,t)$.
 $\alpha_t$ is called an \emph{observable}, and is a sequence 
 of events from transaction $t$.  
Given $\pi$ and a trace $\tau$, we define 
$\tau \vdash_{\tt{G}} \pi$ to represent a sequence 
$ \tau_0 \xrightarrow[]{\alpha_{t_1}}_{\readat{-}\sat} 
 \tau_1  \xrightarrow[]{\alpha_{t_2}}_{\readat{-}\sat} \cdot \cdot \cdot \xrightarrow[]{\alpha_{t_n}}_{\readat{-}\sat} \tau_n $, where 
 $\tau_0 = \tau$ and 
$\pi = \alpha_{t_1} \alpha_{t_2} \dots \alpha_{t_n}$. Moreover, we define $\pi (\tau) :=  \tau_n$.

 $\tau$ is $terminal$ if there is no event left to execute $i.e$ $\succof\tau = \emptyset$. 
 We define $\tau \vdash_{\tt{G_T}} \pi$ to say that $(i)$ $\tau \vdash_{\tt{G}} \pi$ and $(ii)$ $\pi(\tau)$ is terminal.

\begin{definition}($p$-free and $t$-free observation sequences)
For a process  $p \in \mathbf{P}$, we say that an observation sequence $\pi$ is $p$-free if   all observables $\alpha_{t}$ in $\pi$ pertain to 
transactions $t$ not in $p$. 
For a transaction $t$ issued in $p$ we say that $\pi$ is $t$-free if $\pi$ is $p$-free. 
\end{definition}

\begin{definition}(Independent Observables)
For observables $\alpha_{t_1}$ and $\alpha_{t_2}$, we write $\alpha_{t_1} \sim  \alpha_{t_2}$ to represent they are independent. 
This means 
$(i)$ $t_1,t_2$ are transactions issued in different processes, that is, 
 $t_1$ is issed in $p_1$, $t_2$ is issued  in $p_2$, with $p_1 \neq p_2$,  
$(ii)$ no read event $r(p_2,t_2,x,v)$ of transaction $t_2$ reads from $t_1$, and 
$(iii)$ no read event $r(p_1,t_1,x,v)$ of transaction $t_1$ reads from $t_2$.	
Thus, $\neg(\alpha_{t_1} \sim  \alpha_{t_2})$ if either $t_1, t_2$ are issued in the same process, or there is a $\trf$ relation between $t_1, t_2$. That is, $\neg(\alpha_{t_1} \sim  \alpha_{t_2})$ iff $t_1 [\tpo \cup \trf]^+ t_2$ or  $t_2 [\tpo \cup \trf]^+ t_1$.

\end{definition}
 
 \begin{definition}(Independent Observation Sequences)
Observation sequences $\pi^1, \pi^2$ are called independent written $\pi^1 \sim \pi^2$ if there are observables $\alpha_{t_1}$, $\alpha_{t_2}$, and observation sequences $\pi'$ 
and $\pi''$ such that $\pi^1 = \pi' \cdot \alpha_{t_1} \cdot \alpha_{t_2} \cdot \pi''$, 
$\pi^2 = \pi' \cdot \alpha_{t_2} \cdot \alpha_{t_1} \cdot \pi''$, 
and $\alpha_{t_1} \sim \alpha_{t_2}$. 
In other words, we get $\pi^2$ from $\pi^1$ by swapping neighbouring independent observables corresponding to transactions $t_1$ and $t_2$. 
\end{definition}
We use $\approx$ to denote reflexive transitive closure of $\sim$.

\begin{definition}(Equivalent Traces)
For traces $\tau_1 = \langle \tran_1, \tpo_1, \trf_1, \cfr_1 \rangle$ and $\tau_2 = \langle \tran_2, \tpo_2, \trf_2, \cfr_2 \rangle$, we say $\tau_1, \tau_2$ are equivalent denoted $\tau_1 \equiv \tau_2$ if $\tran_1 = \tran_2$, $\tpo_1 = \tpo_2$, $\trf_1 = \trf_2$ and for all $t_1$,$t_2 \in \tran_1^{w,x}$ for all variables $x$, we have $t_1$ [$\tpo_1 \cup \trf_1 \cup \cfr_1^x$] $t_2$ iff $t_1$ [$\tpo_2 \cup \trf_2 \cup \cfr_2^x$] $t_2$.
	
\end{definition}

\begin{lemma}
$((\tau_1 \equiv \tau_2) \wedge \tau_1 \vdash_{\tt{G}} \alpha_t) \Rightarrow (\tau_2 \vdash_{\tt{G}} \alpha_t \wedge (\alpha_t(\tau_1) \equiv \alpha_t(\tau_2)))$.
\label{lemma ra-g-1}
\end{lemma}
\begin{proof}
Assume 	$\tau_1 \vdash_{\tt{G}} \alpha_t$ for an observable $\alpha_t$, and 
$\tau_1=(\tran_1, \tpo_1, \trf_1, \cfr_1)$, 
$\tau_2=(\tran_2, \tpo_2, \trf_2, \cfr_2)$ with $\tau_1 \equiv \tau_2$. Then we know that $\tran_1=\tran_2, \tpo_1=\tpo_2, \trf_2=\trf_1$. 
 
 Since $\tau_1 \vdash_{\tt{G}} \alpha_t$, let  $\tau_1 \xrightarrow[]{\alpha_t}_{\readat{-}\sat} \tau'$. 
 $\tau'=(\tran', \tpo', \trf', \cfr')$ where $\tran'=\tran_1 \cup \{t\}$, 
 $\tpo'=\tpo_1 \cup \{(t_1,t) \mid t_1 \in \tran_1$ is in the same process as $t\}$,  $\cfr_1 \subseteq \cfr'$ 
 and $\trf_1 \subseteq \trf'$. $\cfr'$ can contain $(t',t)$ for some 
 $t' \in \tran_1$, when $t',t \in {\tran'}^{w,x}$ for some variable $x$, based 
 on the trace semantics. In particular if we have 
 $t_1 {\trf'}^x t_2$ , $t_3 \in \tran^{wt,x}$ and $t_3 (\trf'\cup  \tpo') t_2$, then    
 $t_3 {\cfr'}^x t_1$.

 Since $\tau_1 \equiv \tau_2$, for any transactions $t', t'' $ in $\tran_1=\tran_2$, $t' [\tpo_1 \cup \cfr_1^x \cup \trf_1] t''$ iff $t' [\tpo_2 \cup \cfr_2^x \cup \trf_2] t''$ for all variables $x$. In particular, $t' \cfr_1^x t''$ iff $t' \cfr_2^x t''$ for all $x$. 
 Now, let us construct a trace $\tau''=(\tran', \tpo', \trf', \cfr'')$,
 where $\cfr''$ is the smallest set such that $\cfr_2 \subseteq \cfr''$ and 
  whenever $t_1 {\trf'}^x t_2$ , $\tau_3 \in \tran^{\wt,x}$ and $t_3 (\trf'\cup  \tpo') t_2$, then $t_3 {\cfr''}^x t_1$.

   Since $t' [\tpo_1 \cup \trf_1 \cup \cfr_1^x] t''$ iff $t' [\tpo_2 \cup \trf_2 \cup \cfr_2^x] t''$ for all $x$, and all $t', t'' \in \tran_1=\tran_2$, and $\cfr''$ is the smallest extension of $\cfr_2$  based 
   on the trace semantics, along with the fact that 
   $\tpo'=\tpo'', \trf'=\trf''$, we obtain for any two 
   transactions $t_1, t_2 \in \tran'=\tran''$, 
   $t_1 [\tpo' \cup \trf' \cup \cfr'] t_2$ iff  $t_1 [\tpo'' \cup \trf''
   \cup \cfr''] t_2$. This gives $\tau' \equiv \tau''$. 
    
      This also gives $\tau_2 \xrightarrow[]{\alpha_t}_{\readat{-}\sat} \tau''$,
  that is, $\tau_2 \vdash_{\tt{G}} \alpha_t$  and indeed $\alpha_t(\tau_1)=\tau' \equiv \tau''=\alpha_t(\tau_2)$.

\end{proof}

\begin{lemma}
If $\tau \vdash_{\tt{G}} \alpha_{t_1} \cdot \alpha_{t_2}$ and $\alpha_{t_1} {\sim} \alpha_{t_2}$ then $\tau \vdash_{\tt{G}} \alpha_{t_2} {\cdot} \alpha_{t_1}$ and $(\alpha_{t_1} {\cdot} \alpha_{t_2})(\tau) \equiv (\alpha_{t_2}{ \cdot} \alpha_{t_1})(\tau)$.
\label{lemma ra-g-2}
\end{lemma}
\begin{proof}
Let $\tau = \langle \tran , \tpo, \trf,\cfr \rangle$ be a trace and let $t_1$ be a transaction issued in process $p_1$, and $t_2$ be a transaction issued in process $p_2$, with $p_1 \neq p_2$. Assume $\tau \vdash_{\tt{G}} \alpha_{t_1} \cdot \alpha_{t_2}$, with $\alpha_{t_1} {\sim} \alpha_{t_2}$.

We consider the following cases.
\begin{itemize}
    \item [$\bullet$] $t_1,t_2 \notin \tran^{rd}$. In this case $\tau \vdash_{\tt{G}} \alpha_{t_2} \cdot \alpha_{t_1}$ holds trivially.
    
     $(\alpha_{t_1} \cdot \alpha_{t_2})(\tau) = (\alpha_{t_2} \cdot \alpha_{t_1})(\tau) = \langle \tran' , \tpo' , \trf, \cfr \rangle$, where $\tran' = \tran \cup \{t_1,t_2\}$ and $\tpo' = \tpo \cup \{ (t,t_1) | t \in p_1\} \cup \{(t',t_2) | t' \in p_2 \}$.
    \item [$\bullet$] $t_1 \in \tran^{wt}$, $t_1 \notin \tran^{rd}$ and $t_2 \in \tran^{rd} \cap \tran^{wt}$. 
    Let $\tau_1 = \alpha_{t_1}(\tau)$. 
    We know that $\tau_1 = \langle \tran_1 , \tpo_1 , \trf_1,\cfr_1 \rangle$, where 
    $\tran_1 = \tran \cup \{ t_1 \}$, and $\tpo_1 = \tpo \cup \{ (t,t_1) | t \in p_1\}$, with $\trf_1 = \trf$ since there are no read events in $t_1$, and $\cfr_1 = \cfr$ since by the trace semantics, when $ \trf$ remain same, there are no $\cfr$ edges to be added.
    
    Consider $\tau_2 = \alpha_{t_2}(\tau_1) = \langle \tran_2, \tpo_2, \trf_2, \cfr_2 \rangle$. 
    Since $t_2 \in \tran^{rd}$, we know that $t_2$ has read events.
    \smallskip

    Let there be $n$ read events $ev_1 \dots ev_n$ in transaction $t_2$, 
    reading from transactions $t'_1, \dots, t'_n$ on variables $x_1, \dots, x_n$.  
    Hence we get a sequence $\tau_1 \xrightarrow[]{\beginact(t_2)}_{\readat{-}\sat} \dots \tau'_{i_1} \xrightarrow[]{\readact(t_2,t'_1)}_{\readat{-}\sat}  \tau_{i_1}  \xrightarrow[]{}_{\readat{-}\sat}\dots  
    \tau'_{i_n} \xrightarrow[]{\readact(t_2,t'_n)}_{\readat{-}\sat} \tau_{i_n} \dots \xrightarrow[]{\commitact(t_2)}_{\readat{-}\sat} \tau_2$.
    Hence $\tran_2 = \tran_1 \cup \{t_2\}$, $\tpo_2 = \tpo_1 \cup \{(t',t_2) | t' \in p_2\}$, $\trf_2 = \trf_1 \cup (\trf_{i_1} \cup \trf_{i_2} \dots \cup \trf_{i_n})$, and $\cfr_2 = \cfr_1 \cup (\cfr_{i_1} \cup \cfr_{i_2} \dots \cup \cfr_{i_n})$.
    Since $\tau_1 \vdash_{\tt{G}} \alpha_{t_2}$,  for all $ev_i : 1 \le i \le n$, we have $t'_i \in \rbl(\tau_1,t_2,x_i)$.
    
    Since $\alpha_{t_1} \sim \alpha_{t_2}$, we know that for all $ev_i : 1 \le i \le n$, we have $t'_i \in \rbl(\tau,t_2,x_i)$. It follows that $\tau \vdash_{\tt{G}} \alpha_{t_2}$.
    Define $\tau_3 = \langle \tran_3, \tpo_3, \trf_2, \cfr_2 \rangle$, where $
    \tran_3 = \tran \cup \{t_2\}$ and $\tpo_3 = \tpo \cup \{(t',t_2) | t' \in p_2\}$. It follows that $\alpha_{t_2}(\tau) = \tau_3$, $\tau_3 \vdash_{\tt{G}} \alpha_{t_1}$ and $\tau_2 = \alpha_{t_1}(\tau_3) = \alpha_{t_2}(\alpha_{t_1}(\tau))$.
    \item [$\bullet$] $t_2 \in \tran^{wt}$ and $t_1 \in \tran^{rd} \cap \tran^{wt}$. Similar to previous case.
    \item [$\bullet$] $t_1, t_2 \in \tran^{rd} \cap \tran^{wt}$. 
    Let $\tau_1 = \alpha_{t_1}(\tau) = \langle \tran_1, \tpo_1, \trf_1, \cfr_1 \rangle$ and $\tau_2 = \alpha_{t_2}(\tau_1)$. Since $t_1 \in \tran^{rd}$, we know that $t_1$ has read events. 
    Let there be $n$ read events $ev_1 \dots ev_n$ in transaction $t_1$, reading from transactions $t'_1, \dots, t'_n$ on variables 
    $x_1, \dots, x_n$. 
    Hence we get the sequence $\tau \xrightarrow[]{\beginact(t_1)}_{\readat{-}\sat} \dots \tau'_{i_1} \xrightarrow[]{\readact(t_1,t'_1)}_{\readat{-}\sat}  \tau_{i_1} \xrightarrow[]{}_{\readat{-}\sat}\dots  
    \tau'_{i_n} \xrightarrow[]{\readact(t_1,t'_n)}_{\readat{-}\sat} \tau_{i_n} \dots \xrightarrow[]{\commitact(t_1)}_{\readat{-}\sat} \tau_1$.
    Since $t_2 \in \tran^{rd}$, we know that $t_2$ has read events. 
    Let there be  $m$ read events $ev'_1 \dots ev'_m$ in transaction $t_2$, reading from transactions $t''_1, \dots, t''_m$ on variables 
    $y_1, \dots, y_m$.  
    Then we get the sequence $\tau_1 \xrightarrow[]{\beginact(t_2)}_{\readat{-}\sat} \dots \tau'_{j_1} \xrightarrow[]{\readact(t_2,t''_1)}_{\readat{-}\sat}  \tau_{j_1}  \xrightarrow[]{}_{\readat{-}\sat} \dots 
    \tau'_{j_n} \xrightarrow[]{\readact(t_2,t''_m)}_{\readat{-}\sat} \tau_{j_n} \dots \xrightarrow[]{\commitact(t_2)}_{\readat{-}\sat} \tau_2$. 
    
    $\mathbf{Proving}$ $\mathtt{ \tau \vdash_{\tt{G}} \alpha_{t_2} {\cdot} \alpha_{t_1}}$
    
    Since $\tau_1 = \alpha_{t_1}(\tau)$, for all read events $ev_i$ such that $ev_i.trans = t_1$, we have $t'_i \in \rbl(\tau,t_1,x_i)$.
    Since $\tau_2 = \alpha_{t_2}(\tau_1)$, for all read events $ev'_j$ such that $ev'_j.trans = t_2$, we have $t''_j \in \rbl(\tau_1,t_2,y_j)$. 
    Since $\alpha_{t_1} \sim \alpha_{t_2}$, 
    it follows that, we have $t''_j \in \rbl(\tau,t_2,y_j)$ for all read events $ev'_j$ such that $ev'_j.trans = t_2$. Hence $\tau \vdash_{\tt{G}} \alpha_{t_2}$. 
    \smallskip

    Consider $\tau_3 = \alpha_{t_2}(\tau) = \langle \tran_3, \tpo_3, \trf_3, \cfr_3 \rangle $. To show that $\tau \vdash_{\tt{G}} \alpha_{t_2} \cdot \alpha_{t_1}$, 
    we show that for all read events $ev_i$ such that $ev_i.trans = t_1$,  $t'_i \in \rbl(\tau_3,t_1,x_i)$. We prove this using contradiction.
     Assume that $\exists t'_i$ s.t. $t'_i \notin \rbl(\tau_3,t_1,x_i)$.

    \medskip 
    \smallskip

    If $t'_i \notin \rbl(\tau_3,t_1,x_i)$, then 
    there is a $t'_k$ such that (i) $t'_k \in \vbl(\tau_3,t_1,x_i)$ and (ii) $t'_i [\tpo_3 \cup \trf_3 \cup \cfr_3]^+ t'_k$.
     
     \medskip 
     \smallskip

       If   (i) is true, that is, $t'_k \in \vbl(\tau_3,t_1,x_i)$, since $\alpha_{t_1} \sim \alpha_{t_2}$, we also have $t'_k \in \vbl(\tau,t_1,x_i)$. 
        This, combined with $t'_i \in \rbl(\tau,t_1,x_i)$, gives according to the trace semantics, $t'_k [\cfr_1] t'_i$.  Going back to (ii), 
    we can have $t'_i [\tpo_3 \cup \trf_3 \cup \cfr_3]^+ t'_k$ only if we observed one of the following:
    \begin{itemize}
        \item [(a)] There is a path from $t'_i$ to $t'_k$ in $\tau$; that is, 
        $t'_i [\tpo \cup \trf \cup \cfr]^+ t'_k$. This implies that $t'_i \notin \rbl(\tau,t_1,x_i)$ which is a contradiction.
        \item [(b)] There is no direct path from $t'_i$ to $t'_k$, however, 
        there are  transactions  $t''_j,t''_l$ such that, in $\tau$ we had  
        $ t''_j [\tpo \cup \trf \cup \cfr]^* t'_k$ and $t'_i [\tpo \cup \trf \cup \cfr]^* t''_l$, along with $t''_j \in \rbl(\tau,t_2,y_j)$,  $t''_l \in \vbl(\tau,t_2,y_j)$. Let  $ev'_j = r(p_2,t_2,y_j,v)$ be the read event which reads from $t''_j$, justifying 
        $t''_j \in \rbl(\tau,t_2,y_j)$. 
        
        \smallskip
                Since $t''_j [\tpo \cup \trf \cup \cfr]^* t'_k$ and $t'_i [\tpo \cup \trf \cup \cfr]^* t''_l$ it follows that $ t''_j [\tpo_1 \cup \trf_1 \cup \cfr_1]^* t'_k$, $t'_i [\tpo_1 \cup \trf_1 \cup \cfr_1]^* t''_l$. By our assumption, we have $t'_i \in \rbl(\tau,t_1,x_i)$ and $t'_k \in \vbl(\tau,t_1,x_i)$, which  gives us  $t'_k [\cfr_1^{x_i}] t'_i$ (already observed in the para before (i)). 
                
                Thus we get  $t''_j [\tpo_1 \cup \trf_1 \cup \cfr_1]^*
                t'_k [\cfr_1^{x_i}] t'_i [\tpo_1 \cup \trf_1 \cup \cfr_1]^*
                 t''_l$, which gives  $t''_j [\tpo_1 \cup \trf_1 \cup \cfr_1]^+ t''_l$. Hence $t''_j \notin \rbl(\tau_1,t_2,y_j)$, a  contradiction.

                 A similar argument for other cases.
    \end{itemize}
    Thus, we have $\neg(t'_i [\tpo_3 \cup \trf_3 \cup \cfr_3]^+ t'_k)$. This, for all $1 \leq i \leq n$, $t'_i \in \rbl(\tau_3,t_1,x_i)$. Thus, we now have $\tau \vdash_{\tt{G}} \alpha_{t_2} \cdot \alpha_{t_1}$. 
    It remains to show that $(\alpha_{t_1} {\cdot} \alpha_{t_2})(\tau) \equiv (\alpha_{t_2}{ \cdot} \alpha_{t_1})(\tau)$. 
    
    $\mathbf{Proving}$ $\mathtt{(\alpha_{t_1} {\cdot} \alpha_{t_2})(\tau) \equiv (\alpha_{t_2}{ \cdot} \alpha_{t_1})(\tau)}$
    
    Define $\tau_4 {=} \alpha_{t_1}(\tau_3) {=} \alpha_{t_1}.\alpha_{t_2}(\tau){=}
         \langle \tran_4, \tpo_4, \trf_4 , \cfr_4 \rangle$. Recall that 
    $\tau_2{=}$$\alpha_{t_2}.\alpha_{t_1}(\tau)$ =$
      \langle \tran_2,\tpo_2, \trf_2, \cfr_2 \rangle$.
    \smallskip

    To show that $(\alpha_{t_1} \cdot \alpha_{t_2})(\tau) \equiv (\alpha_{t_2} \cdot \alpha_{t_1})(\tau)$, we show that $\tran_2=\tran_4, \tpo_2=\tpo_4, \trf_2=\trf_4$, and, for any two transactions 
    $t'_k, t'_i \in \tran_2=\tran_4$,    $t'_k [\tpo_2 \cup \trf_2 \cup \cfr_2]^+ t'_i$ iff 
    $t'_k [\tpo_4 \cup \trf_4 \cup \cfr_4]^+ t'_i$. Of these, trivially, 
    $\tran_2=\tran_4, \tpo_2=\tpo_4$ follow. Since $\alpha_{t_2} \sim \alpha_{t_1}$, we have $\trf_2=\trf_4$. It remains to prove the last condition. 
    
    \medskip 
    
   Assume that $t'_i [\trf^{x_i}_1] t_1 \wedge t'_k \in \vbl(\tau_1,t_1,x_i)$, 
   where $\tau_1=\alpha_{t_1}(\tau)$. Then we have 
   $t'_i [\trf^{x_i}_1] t_1 \wedge t'_k \in \vbl(\tau,t_1,x_i)$, and by the trace semantics, we have $t'_k [\cfr_1] t'_i$, and hence 
   $t'_k [\cfr_2] t'_i$.  To obtain $\tau_4$, we execute $t_2$ first obtaining $\tau_3$ from $\tau$, and then $t_1$. We show that $t'_k [\tpo_4 \cup \trf_4 \cup \cfr_4]^+ t'_i$, proving $t'_k [\tpo_2 \cup \trf_2 \cup \cfr_2]^+ t'_i \Rightarrow
    t'_k [\tpo_4 \cup \trf_4 \cup \cfr_4]^+ t'_i$.

    \smallskip 
    
     If we have $t'_k \in \vbl(\tau_3,t_1, x_i)$, then we are done since we  have $t'_i \in \rbl(\tau,t_1,x_i)$, hence 
          $t'_i \in \rbl(\tau_3,t_1,x_i)$, thereby obtaining 
          $t'_k \cfr_3 t'_i$, and hence $t'_k \cfr_4 t'_i$.    
          
          \smallskip 
          
           Assume $t'_k \notin \vbl(\tau_3,t_1, x_i)$.   
   Assume that on executing $t_2$ from $\tau$, we have $t''_l \in \vbl(\tau, t_2, y_j)$, and let $t''_j [\trf^{y_j}_3] t_2$.      
     Then by the trace semantics,   $t''_l [\cfr_3] t''_j$ is a new edge which gets added.  
         When we execute $t_1$ next, assume  $t''_m \in \vbl(\tau_3,t_1,x_i)$.   
         Since we have $t'_i \in \rbl(\tau_3,t_1,x_i)$, we obtain $t''_m \cfr_4 t'_i$.

     \medskip 
     By assumption, $t'_k \neq t''_m$. However, $t'_k \in \vbl(\tau,t_1,x_i)$. \\
     $t'_k \notin  \vbl(\tau_3,t_1,x_i)$ points to some transaction that 
     blocked the visibility of $t'_k$ in $\tau_3$, by happening after $t'_k$, and which is in $\vbl(\tau_3,t_1,x_i)$.  
     
     \begin{enumerate}
     	\item This blocking transaction could be $t''_m$. If this is the case, we have $t'_k [\tpo_3 \cup \trf_3 \cup \cfr_3]^+ t''_m [\cfr_4] t'_i$, and hence  $t'_k [\tpo_4 \cup \trf_4 \cup \cfr_4]^+ t''_m [\cfr_4] t'_i$. 
     	\item The other possibility is that we have  $t''_j$
     	happens before $t''_m$, and $t''_l$ happens after $t'_k$ blocking 
     	$t'_k$ from being in $\vbl(\tau_3,t_1,x_i)$. That is,
     	     $t'_k [\tpo \cup \trf \cup \cfr]^* t''_l$ and 
     	     $t''_j [\tpo_4 \cup \trf_4 \cup \cfr_4]^* t''_m$.   
         Then we obtain $t'_k [\tpo_4 \cup \trf_4 \cup \cfr_4]^* t''_l [\cfr^{x_j}_4] t''_j [\tpo_4 \cup \trf_4 \cup \cfr_4]^* t''_m [\cfr_4] t'_i$. 
     
     \end{enumerate}
        
     Thus, we obtain $t'_k [\tpo_4 \cup \trf_4 \cup \cfr_4]^+ t'_i$ as desired.

\end{itemize}
 The converse direction, that is, whenever $t'_k [\tpo_4 \cup \trf_4 \cup \cfr_4]^+ t'_i$, we also have $t'_k [\tpo_2 \cup \trf_2 \cup \cfr_2]^+ t'_i$ is proved on similar lines. 
    
\end{proof}
\begin{lemma}
If $\tau_1 \equiv \tau_2$, $\alpha_{t_1} \sim \alpha_{t_2}$, and $\tau_1 \vdash_{\tt{G}} (\alpha_{t_1} \cdot \alpha_{t_2})$,  then (i) $\tau_2 \vdash_{\tt{G}} (\alpha_{t_2} \cdot \alpha_{t_1})$ and (ii) $(\alpha_{t_1} \cdot \alpha_{t_2})(\tau_1) \equiv (\alpha_{t_2} \cdot \alpha_{t_1})(\tau_2)$.
\label{lemma ra-g-3}
\end{lemma}
\begin{proof}
Follows from Lemma \ref{lemma ra-g-1} and Lemma \ref{lemma ra-g-2}.
\end{proof}

From Lemma \ref{lemma ra-g-3}, we get following lemma.
\begin{lemma}
If $\tau_1 \equiv \tau_2$, $\pi_{\tt{T^1}} \sim \pi_{\tt{T^2}}$, and $\tau_1 \vdash_{\tt{G}} \pi_{\tt{T^1}}$ then $\tau_2 \vdash_{\tt{G}} \pi_{\tt{T^2}}$ and $\pi_{\tt{T^1}}(\tau_1) \equiv \pi_{\tt{T^2}}(\tau_2)$.
\label{lemma ra-g-4}
\end{lemma}
 
\begin{lemma}
If $\tau \vdash_{\tt{G_T}} \pi_{\tt{T}}$, $\tau \vdash_{\tt{G}} \alpha_t$ then $\pi_{\tt{T}} = \pi_{\tt{T^1}} \cdot \alpha_t \cdot \pi_{\tt{T^2}}$ for some $\pi_{\tt{T^1}}$ and $\pi_{\tt{T^2}}$ where $\pi_{\tt{T^1}}$ is $t$-free.
\label{lemma ra-g-5}
\end{lemma}
Consider observables $\pi_{\tt{T}} = \alpha_{t_1},\alpha_{t_2} \dots \alpha_{t_n}$. 
We write $\pi_{\tt{T}} \lessapprox \pi'_{\tt{T}}$ to represent that $\pi'_{\tt{T}} = \pi_{\tt{T^0}} \cdot \alpha_{t_1} \cdot \pi_{\tt{T^1}} 
\cdot \alpha_{t_2} \cdot \pi_{\tt{T^2}} \dots \alpha_{t_n} \cdot \pi_{\tt{T^n}}$. 
In other words, $\pi_{\tt{T}}$ occurs as a non-contiguous subsequence in $\pi'_{\tt{T}}$. For such a $\pi_{\tt{T}}, \pi'_{\tt{T}}$, 
we define $\pi'_{\tt{T}} \oslash \pi_{\tt{T}} := \pi_{\tt{T^0}} \cdot \pi_{\tt{T^1}} \dots \pi_{\tt{T^n}}$. 
Since elements of $\pi_{\tt{T}}$ and $\pi'_{\tt{T}}$ are distinct, operation $\oslash$ is well defined. 
Let $\pi_{\tt{T}}[i]$ denote the $i$th observable 
in the observation sequence $\pi_{\tt{T}}$, and 
let  $|\pi_{\tt{T}}|$ denote  the number of observables 
in $\pi_{\tt{T}}$.

Let $\alpha_t = \pi_{\tt{T}}[i]$ for some $i: 1 \le i \le |\pi_{\tt{T}}|$. 
We define $\tt{Pre(\pi_{\tt{T}},\alpha_t)}$ as a subsequence $\pi'_{\tt{T}}$ of $\pi_{\tt{T}}$ such that  
(i) $\alpha_t \in \pi'_{\tt{T}}$,
(ii) $\alpha_{t_j} = \pi_{\tt{T}}[j] \in \pi'_{\tt{T}}$ for some $j: 1 \le j < i$ iff there exists $k: j < k \le i$ such that $\alpha_{t_k}=\pi_{\tt{T}}[k]  \in \pi'_{\tt{T}}$ and $\neg(\alpha_{t_j} \sim \alpha_{t_k})$. 
Thus, $\tt{Pre(\pi_{\tt{T}},\alpha_t)}$ consists of $\alpha_t$ and all 
$\alpha_{t'}$ appearing before $\alpha_t$ in $\pi_{\tt{T}}$ such that 
$t' [\tpo \cup \trf]^+ t$.

%
%
\begin{lemma}
If $\pi_{\tt{T^1}} = \tt{Pre(\pi_{\tt{T}},\alpha_t)}$ and $\pi_{\tt{T^2}} = \pi_{\tt{T}} \oslash \pi_{\tt{T^1}}$ then $\pi_{\tt{T}} \approx \pi_{\tt{T^1}} \cdot \pi_{\tt{T^2}}$.
\label{lemma ra-g-6}
\end{lemma}
\begin{proof}
The proof is trivial since we can always execute in order, 
the $[\tpo \cup \trf]^+$ predecessors of $\alpha_t$ from $\tt{Pre(\pi_{\tt{T}},\alpha_t)}$, then $\alpha_t$, then the  observables 
in $\tt{Pre(\pi_{\tt{T}},\alpha_t)}$ which are independent from $\alpha_t$, followed by the suffix of $\pi_{\tt{T}}$ after $\alpha_t$.

\end{proof}

\begin{lemma}
If $\pi_{\tt{T}}= \pi'_{\tt{T}} \cdot \pi''_T$ and $\alpha_t \in \pi'_{\tt{T}}$ then $\tt{Pre(\pi_{\tt{T}},\alpha_t)} = \tt{Pre(\pi'_{\tt{T}},\alpha_t)}$.
\label{lemma ra-g-7}
\end{lemma}
\begin{proof}
The proof is trivial once again, since 	all 
the transactions which are $[\tpo \cup \trf]^+ t$  are in the prefix  $\pi'_{\tt{T}}$.
\end{proof}

\begin{lemma}
If $\tau \vdash_{\tt{G}} \pi_{\tt{T}}$ then $\tau \vdash_{\tt{G_T}} \pi_{\tt{T}} \cdot \pi_{\tt{T^2}}$.
\label{lemma ra-g-8}
\end{lemma}
\begin{proof}
This simply follows from the fact that we can extend the observation sequence 
$\pi_{\tt{T}}$ to obtain a terminal configuration since the $\readat$ DPOR algorithm generates traces corresponding to terminating runs.   	
\end{proof}

\begin{lemma}
Consider a trace trace $\tau$ such that $\tau \models \readat$, $\tau \vdash_{\tt{G}} \pi_{\tt{T}} \cdot \alpha_t$. Let each read event $ev_i = r_i(p,t,x_i,v)$ in $\alpha_t$ read from some transaction $t_i$, and let $\pi_{\tt{T}}$ be $t$-free.  Then $\tau \vdash_{\tt{G}} \alpha'_t \cdot \pi_{\tt{T}}$, where $\alpha'_t$ is the same as  $\alpha_t$, with the exception that the sources of its read events can be different. That is, 
 each read event $ev_i = r_i(p,t,x_i,v) \in \alpha'_t$ can read from some transaction $t'_i \neq t_i$.
\label{lemma ra-g-9}
\end{lemma}
\begin{proof}
Let $\pi_{\tt{T}} = \alpha_{t_1} \alpha_{t_2} \dots \alpha_{t_n}$ and $\tau_0 \xrightarrow[]{\alpha_{t_1}}_{\readat{-}\sat} \tau_1 \xrightarrow[]{\alpha_{t_2}}_{\readat{-}\sat} \dots \xrightarrow[]{\alpha_{t_n}}_{\readat{-}\sat} \tau_n \xrightarrow[]{\alpha_t}_{\readat{-}\sat} \tau_{n+1}$, where $\tau_0 = \tau$. Let $\tau_i = \langle \tran_i, \tpo_i, \trf_i, \cfr_i \rangle$. 
Let there be  $m$ read events in transaction $t$. Keeping in mind what we want to prove, where we want to execute $t$ first followed by $\pi_{\tt{T}}$, and obtain an execution $\alpha'_t \pi_{\tt{T}}$, we do the following.

For each read event $ev_i = r_i(p,t,x_i,v)$ we define $t'_i \in \tran^{wt,x_i}$ such that 
$(i)$ $ t'_i \in \rbl(\tau_0,t,x_i)$, and 
$(ii)$ there is no $t''_i \in \rbl(\tau_0,t,x_i)$ where $t'_i [\tpo_n \cup \trf_n \cup \cfr_n] t''_i$.  Note that this is possible since 
 $\pi_{\tt{T}}$ is $t$-free, so all the observables $\alpha_{t_i}$ occurring in $\pi_{\tt{T}}$ are such that $t_i$ is issued in a process other than that of $t$. Thus, when $\alpha_t$ is enabled, we can choose any of the 
 writes done earlier,  this fact is consistent with the $\readat$ semantics, since from Lemmas \ref{lem:f3}-\ref{lem:f4} we know $\tau_n \models \readat$.
Since $\tau_n \models \readat$ such a $t'_i \in  \rbl(\tau_0,t,x_i)$ exists for each $ev_i$.

\medskip 
Next we define a sequence of traces which can give the execution 
$\alpha'_t \pi_{\tt{T}}$. 
Define a sequence of traces $\tau'_0, \tau'_1, \dots \tau'_n$ where $\tau'_j = \langle \tran'_j, \tpo'_j, \trf'_j, \cfr'_j \rangle$ is such that 
\begin{enumerate}
	\item $\tau \vdash_{\tt{G}} {\alpha'_t}$ and $\alpha'_t(\tau)=\tau'_0$,
	\item $\tau'_j \vdash_{\tt{G}} \alpha_{t_{j+1}}$ and  $\tau'_{j+1} = \alpha_{t_{j+1}}(\tau'_j)$ for all $j: 0 \le j \le n$.
\end{enumerate}

\begin{itemize}
\item  Define $\tran'_i = \tran_i \cup \{t\}$ for all $0 \leq i \leq n$, 
\item 
Define $\tpo' = \{t' [\tpo] t \mid t'$ is a transaction in $\tau$, in the same process as $t\}$, and $\tpo'_i = \tpo_i \cup \tpo'$ for all $0 \leq i \leq n$, 

\item  $\trf' = \bigcup_{i:1\le i \le m} t'_i [\trf^{x_i}] t$ (reads from relation corresponding to each read event $ev_i \in \alpha'_t$), and 
 $\trf'_i = \trf_i \cup \trf'$ for all $i : 1 \le i \le n$, 
 
\item  $\cfr' = \bigcup_{j:1\le j \le m} \cfr'_j$ (each  $\cfr'_j$ corresponds to the updated $\cfr$ relation because of the read transition $\xrightarrow[]{\readact(t,t'_j)}_{\readat{-}\sat}$). 	
\smallskip 

For $1 \leq i \leq n$, we define $\cfr'_i$ inductively as $\cfr'_i = \cfr_i \cup \cfr'$ and show that $\tau'_i \vdash_{\tt{G}} \alpha_{t_{i+1}}$ and  $\tau'_{i+1} = \alpha_{t_{i+1}}(\tau'_i)$ holds good. 
 First define $\cfr'_0 = \cfr_0 \cup \cfr'$. Then define  $\cfr'_i$ to be the coherence order corresponding to $\alpha_{t_{i}}(\tau'_{i-1})$  where $\tau'_{i-1}=\langle \tran'_{{i-1}}, \tpo_{i-1}, \trf_{i-1}, \cfr_{i-1} \rangle$ for all $i: 1 \le i \le n$. 

 \end{itemize}

\medskip 
\noindent{\bf{Base case}}. The base case $\tau \vdash_{\tt{G}} {\alpha'_t}$ and $\alpha'_t(\tau)=\tau'_0$,
   holds trivially, by construction. 

\smallskip 

\noindent{\bf{Inductive hypothesis}}. Assume that $\tau'_j \vdash_{\tt{G}} \alpha_{t_{j+1}}$ and  $\tau'_{j+1} = \alpha_{t_{j+1}}(\tau'_j)$ for $0 \leq j \leq i-1$. 

\smallskip 

 We have to prove that 
$\tau'_i \vdash_{\tt{G}} \alpha_{t_{i+1}}$ and  $\tau'_{i+1} = \alpha_{t_{i+1}}(\tau'_i)$. 
\begin{enumerate}
	\item If $t_{i+1} \in \tran^{wt}$ and $t_{i+1} \notin \tran^{rd}$,  then the proof holds trivially from the inductive hypothesis. 
	\item Consider now $t_{i+1} \in \tran^{rd}$. Consider a read event 
	$ev^{i+1} = r(p,t_{i+1},x,v) \in \alpha_{t_{i+1}}$ which was reading from 
	transaction $t_x$ in $\pi_{\tt{T}}$. That is, we had $t_x \in \rbl(\tau_i, t_{i+1},x)$. 	If $t_x \in \rbl(\tau'_i, t_{i+1},x)$, then we are done, since we can simply extend the run from the inductive hypothesis. 
	
	\smallskip 
	
	Assume otherwise. That is, $t_x \notin \rbl(\tau'_i, t_{i+1},x)$.

	Since  $t_x$  $\in$   $\rbl(\tau_i, t_{i+1},x)$, we know that there are some blocking transactions in the new path which prevents $t_x$ from being readable. Basically, we have the blocking transactions since $t$ is moved before $\pi_{\tt{T}}$. 	
	
	\begin{itemize}
		\item Consider $tr_1 \in \vbl(\tau,t,y)$ for some variable $y$ such that,  in $\pi_{\tt{T}}.\alpha_t$, we had $t_x [\tpo_i \cup \trf_i \cup \cfr_i]^* tr_1$. This path
		 is possible since $\alpha_t$ comes last in in $\pi_{\tt{T}}.\alpha_t$,  after  $\alpha_{t_{i+1}}$. 
		\item  Consider $tr_2 \in \vbl(\tau_i,t_{i+1},x)$.  Since $t_x$ $\in$ $\rbl(\tau_i, t_{i+1},x)$, we have  $tr_2 [\cfr^x_{i+1}] t_x$. 
Now, in $\alpha'_t.\pi_{\tt{T}}$, assume $tr_3 [\trf'^y_0] t$, 
such that we have a path from $tr_3$ to $t_x$, as  $tr_3 [\tpo_i \cup \trf_i \cup \cfr_i]^* tr_2$. 
	\end{itemize}
	  		  Hence it follows that $tr_3 [\tpo_i \cup \trf_i \cup \cfr_i]^* tr_2$ 
$[\cfr^x_{i+1}] t_x$ $[\tpo_i \cup \trf_i \cup \cfr_i]^* tr_1$. 
Hence we have $tr_3 [\tpo_{i+1} \cup \trf_{i+1} \cup \cfr_{i+1}]^+ tr_1$, i.e 
$tr_3 [\tpo_n \cup \trf_n \cup \cfr_n]^+ tr_1$. This leads to the contradiction since $tr_1 \in \vbl(\tau,t,y)$ and  $tr_3 [\trf'^y_0] t$.  

Hence, $t_x$  $\in$   $\rbl(\tau_i, t_{i+1},x)$ and we can extend the run 
from the inductive hypothesis obtaining $\tau'_j \vdash_{\tt{G}} \alpha_{t_{j+1}}$,  $\tau'_{j+1} = \alpha_{t_{j+1}}(\tau'_j)$,  
for $0 \leq j \leq i-1$,
and 
$\tau'_i \vdash_{\tt{G}} \alpha_{t_{i+1}}$.

\end{enumerate}

Now we prove that $\cfr'_{i+1} = \cfr_{i+1} \cup \cfr'$. 
Assume we have $(tr1 , tr2) \in \cfr_{i+1}$. We show that $(tr1 , tr2) \in \cfr'_{i+1}$. We have the following cases:

\begin{itemize}
    \item[(i)] $(tr1 , tr2) \in \cfr_i$. From the inductive hypothesis it follows that $(tr1 , tr2) \in \cfr'_{i}$ and hence in $\cfr'_{i+1}$.
    \item[(ii)] $(tr1 , tr2) \notin \cfr_i$ and some read event $ev = r(p,t_{i+1},x,v)$ from $\alpha_{t_{i+1}}$ reads from $tr2$.  That is, $tr1,tr2 \in \rbl(\tau_i, t_{i+1},x)$ and $tr1 \in \vbl(\tau_i,t_{i+1},x)$.
    Assume $(tr1 , tr2) \notin \cfr'_{i+1}$.
    Since $(tr1, tr2) \notin \cfr'_{i+1}$, it follows that there exists  $tr3 \in \vbl(\tau'_i,t_{i+1},x)$ such that $tr2 [\tpo'_i \cup \trf'_i \cup \cfr'_i] tr3$. This means $tr2 \in \rbl(\tau_i,t_{i+1},x)$, but 
    $tr2 \notin \rbl(\tau'_i,t_{i+1},x)$. However,  as proved earlier, this leads to a contradiction. 

\end{itemize}
Thus, we have shown that $\tau \vdash_{\tt{G}} \alpha'_t \alpha_{t_1} \dots \alpha_{t_i}$ 
is such that $\alpha'_t \alpha_{t_1} \dots \alpha_{t_i}(\tau)=\tau'_i$ for all 
$0 \leq i \leq n$. When $i=n$ we obtain 
$\alpha'_t \pi_{\tt{T}}(\tau)=\tau'_n$, or $\tau \vdash_{\tt{G}} \alpha'_t \pi_{\tt{T}}$. 
\end{proof}


\begin{definition}[Linearization of  a Trace]
A observation sequence $\pi_{\tt{T}}$ is a \emph{linearization} of a trace $\tau = \langle \tran, \tpo , \trf, \cfr \rangle$ if $(i)$ $\pi_{\tt{T}}$ has the same transactions as $\tau$ and $(ii)$ $\pi_{\tt{T}}$ follows the ($\tpo \cup \trf)$ relation.
	
\end{definition}

We say that our $\readat$ DPOR algorithm \emph{generates an observation sequence} $\pi_{\tt{T}}$ from state $\tau$ where $\tau$ is a trace 
if 
 it invokes $\explore$ with parameters $\readat,\tau, \pi'$, where $\pi'$ is a \emph{linearization} of $\tau$, and generates  a sequence of recursive calls to $\explore$ resulting in $\pi_{\tt{T}}$.

\begin{lemma}
If $\tau \vdash_{\tt{G_T}} \pi_{\tt{T}}$ then, 
the DPOR algorithm  generates $\pi'_{\tt{T}}$ from state $\tau$ 
for some $\pi'_{\tt{T}} \approx \pi_{\tt{T}}$.
\label{lemma ra-g-10}
\end{lemma}
\begin{proof}
We use induction on $|\pi_{\tt{T}}|$. If $\tau$ is $terminal$, then the  proof is trivial.
Assume that we have $\tau \vdash_{\tt{G_T}} \pi_{\tt{T}}$.
Assume that $\tau \vdash_{\tt{G}} \alpha_t$. It follows that $\tau \vdash_{\tt{G}} \alpha_t$.
Using Lemma \ref{lemma ra-g-5} we get $\pi_{\tt{T}} = \pi_{\tt{T^1}} \cdot \alpha_t \cdot \pi_{\tt{T^2}}$, where $\pi_{\tt{T^1}}$ is $t$-free. 
We consider the following two cases:
\begin{itemize}
    \item In the first case, we assume that all read events in $t$ read from 
    transactions in $\tau$. So in this case, $\pi_{\tt{T^1}}$ can be empty. 
        Assume there are $n$ read events in $t$ and each read event $ev_i = r(p,t,x_i,v)$ 
    reads from transactions $t_i \in \tau$ for all $i: 1 \le i \le n$. 
In this case the DPOR algorithm will let each read event $ev_i$ read from all possible transactions $t' \in \rbl(\tau,t,x_i)$, including $t_i$.
    \item In the second case, assume that there exists at least one read event $ev' = r(p,t,x,v) \in \alpha_t$ which reads from a transaction $t' \in \pi_{\tt{T^1}}$ (hence, $\pi_{\tt{T^1}}$ is non empty).
    \smallskip 
     
    From Lemma \ref{lemma ra-g-9}, we know that $\tau
    \vdash_{\tt{G}} \alpha'_t \cdot \pi_{\tt{T^1}}$. 
    From Lemma \ref{lemma ra-g-8}, it follows that $\tau
    \vdash_{\tt{G_T}} \alpha'_t \cdot \pi_{\tt{T^1}} \cdot \pi_{\tt{T^4}}$, for some
    $\pi_{\tt{T^4}}$.
    Hence there exists $\tau'$ such that 
    $\tau' = \alpha'_t (\tau)$ and $\tau' \vdash_{\tt{G_T}} \pi_{\tt{T^1}} \cdot
    \pi_{\tt{T^4}}$. 
    
    Since $|\pi_{\tt{T^1}} \cdot
    \pi_{\tt{T^4}}| < |\pi_{\tt{T}}$, we can use the inductive hypothesis. It follows that the DPOR algorithm generates from state $\tau'$, 
    the observation sequence  $\pi_{\tt{T^5}}$ such that $\pi_{\tt{T^5}} \approx \pi_{\tt{T^1}} 
    \cdot \pi_{\tt{T^4}}$. 
    
    \smallskip 
    
    Let  $\pi_{\tt{T^3}}.\alpha_t = \tt{Pre(\pi_{\tt{T^1}}, \alpha_t)}$. Then  $\pi_{\tt{T^5}}
    \approx \pi_{\tt{T^1}} \cdot \pi_{\tt{T^4}}$ implies that  $\pi_{\tt{T^3}}
    \lessapprox \pi_{\tt{T^5}}$ (by applying Lemma \ref{lemma ra-g-7}).
    Let $\pi_{\tt{T^6}} \approx \pi_{\tt{T}} \oslash (\pi_{\tt{T^3}} \cdot \alpha_t)$.
    By applying Lemma \ref{lemma ra-g-6}, we get $\pi_{\tt{T}} \approx \pi_{\tt{T^3}} \cdot \alpha_t \cdot \pi_{\tt{T^6}}$.
    Since $\tau \vdash_{\tt{G_T}} \pi_{\tt{T}}$ and 
    $\pi_{\tt{T}} \approx \pi_{\tt{T^3}} \cdot \alpha_t \cdot \pi_{\tt{T^6}}$,
    by applying Lemma \ref{lemma ra-g-3}, we get $\tau
    \vdash_{\tt{G_T}} \pi_{\tt{T^3}} \cdot \alpha_t \cdot \pi_{\tt{T^6}}$.
    Let $\tau' = (\pi_{\tt{T^3}} \cdot \alpha_t) (\tau)$. 
    Since $\tau \vdash_{\tt{G_T}} \pi_{\tt{T^3}} 
    \cdot \alpha_t \cdot \pi_{\tt{T^6}}$, we have $\tau'
    \vdash_{\tt{G_T}} \pi_{\tt{T^6}}$. 
    From inductive hypothesis it follows that $\tau'$ generates $\pi_{\tt{T^7}}$ such that $\pi_{\tt{T^7}} \approx \pi_{\tt{T^6}}$.\\
    In other words, $\tau$ generates $\pi_{\tt{T^3}} 
    \cdot \alpha_t \cdot \pi_{\tt{T^7}}$ where $\pi_{\tt{T}} \approx \pi_{\tt{T^3}} 
    \cdot \alpha_t \cdot \pi_{\tt{T^7}}$.
\end{itemize}

\end{proof}

%% file: rc-proofs.tex
\newpage
\centerline{\bf{Read Committed }} 
\section{Properties of the Trace Semantics}
\label{app:rcsatcons}
\begin{figure}[H]
        \begin{tikzpicture}[node distance=10mm , thick, main/.style= {draw,circle},];
        
        \node[] (1t') [] {$t'$};
        \node[] (1t1) [right= 20mm of 1t'] {$t1$};
         \node[] (111) [below=20mm of 1t1] {};
        \node[] (1t) [left=8mm of 111] {$t$};
        \node[] (1cap) [below=5mm of 1t] {$\tt{Cond1}$};

        \draw[->,black,line width=1.2pt] (1t') -- node[above] {$(\tpo \cup \trf \cup \cfr)^+$} (1t1);
        \draw[->,black,line width=1.2pt] (1t1) -- node[right] {$\trf_{\tto\;\tt{before\; event}}$ } (1t);
        \draw[dashed , ->,blue,line width=1.2pt] (1t') -- node[left=1mm] {$\trf$} (1t);
        \draw[dashed , ->,red,line width=1.2pt] (1t1) to [out=90, in=90] node[above] {$\cfr$} (1t');
        \end{tikzpicture}
        \caption{Readability for $\readc$}
    \end{figure}

\subsection{Readability and Visibility for $\readc$}
\label{rcom:rv}
 For $X=\readc$, we define $\rbl(\tau^t_X, t, x)$ 
as the set of all transactions  $t' \in \tran^{\wt,x}$ provided we do not have a transaction $t1 \in \tran^{\wt,x}$ such that the following is true. 
Assume  $\beta~\tto~\alpha$ are two read events in $t$ such that $\beta$ reads from $t1$, $\alpha$ (current read event) reads from $t'$ and 
$t' ~(\co~\cup \cfr)^+~t1$. Note that having such a $t1$ induces  
$t1~\cfr~t'$ and a $\readccyc$.
Call ${\tt{Cond1}}$ the existence of such a $t1$.

After adding $t'~\trf~t$, we must check that 
there are no consistency violations.  
The check set $\vbl(\tau^t_X, t,x)$ is defined as the set of transactions which turn 	``sensitive''  
on adding the new  edge  $t'~\trf~t$. 
Unless  appropriate edges are added involving 
these sensitive transactions, we may get consistency violating cycles
in the resultant trace.
Let $\tau^{tt'}$ denote the trace obtained by adding the new transaction $t$ and the edge $t'~\trf~t$ to trace $\tau$. Now, we 
identify $\vbl(\tau^t_X, t,x)$ and the edges which must be added to $\tau^{tt'}$ to obtain a consistent extended trace.

For $X=\readc$, $\vbl(\tau^t_X,t,x)=\{t1  \mid$ there is a read event $\beta$ in $t$ reading from $t1$, 
$\beta~\tto~\alpha$, and $\alpha$ reads from $t'\}$. The newly added $t'~\rf~t$ is due to this $\alpha$. 
 This necessitates 
adding $t1~\cfr~t'$ to preserve consistency. 	Then $\extend_X(\tau^{tt'})$ contains $\{(t1,t') \mid t1 \in \vbl(\tau^t_X, t,x)\}$. 

\begin{lemma}
    Given trace $\tau$, a transaction $t$ and a variable $x$, we can construct the set  $\rbl(\tau, t, x)$ and $\vbl(\tau, t, x)$ in polynomial time.
    \label{lemma:rc-poly}
\end{lemma}
\begin{proof}
    We prove that set $\rbl(\tau, t, x)$ can be computed in polynomial time ($O(|\tran|^3)$ time), by designing an algorithm to generate set $\rbl(\tau, t, x)$. The algorithm consists of the following  steps:
    \begin{itemize}
        \item[(i)] First we compute the  transitive closure of the relations
        $[\tpo \cup \trf \cup \cfr]$, i.e, we compute $[\tpo \cup \trf \cup \cfr ]^+$.  We can use the Floyd-Warshall algorithm \cite{10.5555/1614191} to compute the transitive closure. This will take $O(|\tran|^3)$ time.
        \item[(ii)] We compute the set $\tran'=\{t'\mid t' [\trf] t\}$. This takes $O(|\tran|)$ time.
        
        \item[(iii)] For each transaction $t' \in \tran^{wt,x}$, we perform following checks:
        \begin{itemize}
            \item Check $\mathtt{cond 1 :} $ Check whether there exists a transaction $t'' \in \tran' \cap \tran^{wt,x}$            
            with $t'$ $[\tpo  \cup \trf \cup \cfr]^+$ $t''$. 
            If there is no such $t''$ then $t' \in \rbl(\tau , t, x)$. 
            If we find such $t''$, then $t' \notin \rbl(\tau , t, x)$. 
            This step will take $O(|\tran|^2)$ time.
        \end{itemize}
    \end{itemize}
    Similarly, we can compute $\vbl(\tau, t, x)$ in polynomial time. This completes the proof.
\end{proof}

\subsection{Properties of the Trace Semantics}
\label{rcom:prop}
\begin{lemma}
    If $\tau_1 \models \readc$  and $\tau_1$ $\xrightarrow[]{\beginact(t)}_{\readc{-}\sat}$ $\tau_2$ then $\tau_2 \models \readc$.
    \label{lem:rcom-f3}
\end{lemma}
\begin{proof}
The proof follows trivially since we do not change the  	reads from and coherence order relations, and the transaction $t$ added to the partial order has no successors. Thus, if $\tau_1 \models \readc$, 
so does $\tau_2$.  
	\end{proof}

\begin{lemma}
    If $\tau_1 \models \readc$  and $\tau_1$ $\xrightarrow[]{\readact(t,t')}_{\readc{-}\sat}$ $\tau_2$ then $\tau_2 \models \readc$.
    \label{lem:rcom-f4}
\end{lemma}

\begin{proof}
    Let $\tau_1 = \langle \tran_{1} , \tpo_1 , \trf_1, \cfr_1 \rangle$, $\tau_2 = \langle \tran_{2} , \tpo_2 , \trf_2, \cfr_2 \rangle$, and $ev$ = $r(p,t,x,v)$ be a read event in transaction $t \in \tran_1$. Suppose $\tau_2 \nvDash \readc$, and  that $\trf_{2\tto{-}before}$ is cyclic. Since $\tau_1 \models \readc$, and $t$ has no outgoing edges, 
   on adding $(t',t) \in \trf_2$, 
    it follows that the cycle in $\tau_2$ is as a result 
    of the newly added $\cfr_2$ edges. 
    
    We prove that 
    on adding $(t',t) \in \trf_2$, 
        such cycles are possible iff at least one of ${\tt{cond}_1}, \dots, {\tt{cond}_4}$ are true. 
    
 \subsection*{${\tt{cond1}}$ induces $(\tpo\cup\trf\cup\cfr)$ cyclicity }
   \label{sec:rc-cyc1}
    This direction is easy to see : assume 
    one of ${\tt{cond1}}$ is true; then, as already 
    argued in the main paper, we will get a $(\tpo\cup\trf\cup\cfr)$ cycle on adding the $\trf_2$ edge from $t'$ to $t$ and we are done.

 \subsection*{$(\tpo\cup\trf\cup\cfr)$ cyclicity implies ${\tt{cond1}}$ }
 \label{sec:rc-cyc2}
  For the converse direction, assume that we add the $\trf_2$ edge from $t'$ to $t$ and obtain a $(\tpo\cup\trf\cup\cfr)$ cycle in $\tau_2$. We now argue that this cycle has been formed because
 ${\tt{cond}_1}$ is true.  
 
 First, note that $t$ has no outgoing edges in $\tau_2$ since 
 it is the current transaction being executed. Also, 
 we know that $\tau_1$ has no $(\tpo_1\cup\trf_1\cup\cfr_2)$ cycles. 
 Thus, the cyclicity of $(\tpo_2\cup\trf_2\cup\cfr_2)$ is induced by the 
 newly added $\trf_2$ edge as well as the newly added 
  $\cfr_2$ edges. 
  Adding the $\trf_2$ edge  to $\tau_1$ does not induce any cycle since $t$ has no outgoing edges. 
  Let's analyze the 
  $\cfr_2$ edges added, which induce cycles, and argue that ${\tt{cond1}}$ will be true.   
  
  \begin{enumerate}
  	\item We add  $\cfr_2$ edges from $t'' \in \vbl(\tau_1, t,x)$ to $t'$. 
  	For  these edges to induce a cycle, we should have a path from $t'$ to $t''$ in $\tau_1$. That is, we have $t' [\tpo\cup\trf\cup\cfr]^+ t''$. This is captured by ${\tt{cond1}}$.  
  	
\end{enumerate}
 
 Thus, we can think of  $\tpo\cup\trf\cup\cfr$ cycles due to the forbidden pattern described in ${\tt{cond}_i}$.  By construction of $\tau_1$ $\xrightarrow[]{\readact(t,t')}_{\readc{-}\sat}$ $\tau_2$, we ensure $\neg {\tt{cond}_1}$. 
 Hence, 
    $\tau_2 \models \readc$. 
    
\end{proof}

Define $\big[[ \sigma]\big]^{\readc{-}\sat}$ as the set of  traces  generated  using $\xrightarrow[]{}_{\readc{-}\sat}$ transitions, starting from an empty trace $\tau_{\o}$.

Consider a terminal trace $\tau$ generated by $\readc$ DPOR algorithm starting from the empty trace $\tau_\emptyset$.
That is, there is a sequence  
$\tau_0 \xrightarrow[]{}_{\readc{-}\sat} 
\tau_1 \xrightarrow[]{}_{\readc{-}\sat}  
\tau_2 \xrightarrow[]{}_{\readc{-}\sat} \dots 
 \xrightarrow[]{}_{\readc{-}\sat}  \tau_n$ with 
 $\tau_0=\tau_{\emptyset}$, 
 and $\tau_n=\tau$. 
 Since $\tau_{\emptyset}$ is a empty we have $\tau_{\emptyset} \models \readc$, 
 it follows by  
 Lemmas \ref{lem:rcom-f3} and \ref{lem:rcom-f4} that $\tau \models \readc$. 
 
 Hence, for each trace $\tau \in \big[[ \sigma]\big]^{\readc{-}\sat} $ we have $\tau \models \readc$.

\subsection{$\readc$ DPOR Completeness}
\label{rcom:comp}
In this section, we show the completeness of the $\readc$ DPOR algorithm. More precisely, for the input program under $\readc$ for any terminating run $\rho \in Runs(\conf)$ and trace $\tau$ s.t. $ \rho \models \tau$, 
we show that $\explore(\readc, \tau_{\emptyset}, \epsilon)$ will produce a recursive visit $\explore(\readc, \tau', \pi)$ for some terminal $\tau'$, and $\pi$ where $\tau'=\tau$.   First, we give some definitions and auxiliary lemmas.

Let  $\pi = \alpha_{t_1} \alpha_{t_2} \dots \alpha_{t_n}$ be an observation sequence where  $\alpha_t = \beginact(p,t) \dots$ \plog{end}$(p,t)$.
 $\alpha_t$ is called an \emph{observable}, and is a sequence 
 of events from transaction $t$.  
Given $\pi$ and a trace $\tau$, we define 
$\tau \vdash_{\tt{G}} \pi$ to represent a sequence 
$ \tau_0 \xrightarrow[]{\alpha_{t_1}}_{\readc{-}\sat} 
 \tau_1  \xrightarrow[]{\alpha_{t_2}}_{\readc{-}\sat} \cdot \cdot \cdot \xrightarrow[]{\alpha_{t_n}}_{\readc{-}\sat} \tau_n $, where 
 $\tau_0 = \tau$ and 
$\pi = \alpha_{t_1} \alpha_{t_2} \dots \alpha_{t_n}$. Moreover, we define $\pi (\tau) :=  \tau_n$.

 $\tau$ is $terminal$ if there is no event left to execute $i.e$ $\succof\tau = \emptyset$. 
 We define $\tau \vdash_{\tt{G_T}} \pi$ to say that $(i)$ $\tau \vdash_{\tt{G}} \pi$ and $(ii)$ $\pi(\tau)$ is terminal.

\begin{definition}($p$-free and $t$-free observation sequences)
For a process  $p \in \mathbf{P}$, we say that an observation sequence $\pi$ is $p$-free if   all observables $\alpha_{t}$ in $\pi$ pertain to 
transactions $t$ not in $p$. 
For a transaction $t$ issued in $p$ we say that $\pi$ is $t$-free if $\pi$ is $p$-free. 
\end{definition}

\begin{definition}(Independent Observables)
For observables $\alpha_{t_1}$ and $\alpha_{t_2}$, we write $\alpha_{t_1} \sim  \alpha_{t_2}$ to represent they are independent. 
This means 
$(i)$ $t_1,t_2$ are transactions issued in different processes, that is, 
 $t_1$ is issed in $p_1$, $t_2$ is issued  in $p_2$, with $p_1 \neq p_2$,  
$(ii)$ no read event $r(p_2,t_2,x,v)$ of transaction $t_2$ reads from $t_1$, and 
$(iii)$ no read event $r(p_1,t_1,x,v)$ of transaction $t_1$ reads from $t_2$.	
Thus, $\neg(\alpha_{t_1} \sim  \alpha_{t_2})$ if either $t_1, t_2$ are issued in the same process, or there is a $\trf$ relation between $t_1, t_2$. That is, $\neg(\alpha_{t_1} \sim  \alpha_{t_2})$ iff $t_1 [\tpo \cup \trf]^+ t_2$ or  $t_2 [\tpo \cup \trf]^+ t_1$.

\end{definition}
 
 \begin{definition}(Independent Observation Sequences)
Observation sequences $\pi^1, \pi^2$ are called independent written $\pi^1 \sim \pi^2$ if there are observables $\alpha_{t_1}$, $\alpha_{t_2}$, and observation sequences $\pi'$ 
and $\pi''$ such that $\pi^1 = \pi' \cdot \alpha_{t_1} \cdot \alpha_{t_2} \cdot \pi''$, 
$\pi^2 = \pi' \cdot \alpha_{t_2} \cdot \alpha_{t_1} \cdot \pi''$, 
and $\alpha_{t_1} \sim \alpha_{t_2}$. 
In other words, we get $\pi^2$ from $\pi^1$ by swapping neighbouring independent observables corresponding to transactions $t_1$ and $t_2$. 
\end{definition}
We use $\approx$ to denote reflexive transitive closure of $\sim$.

\begin{definition}(Equivalent Traces)
For traces $\tau_1 = \langle \tran_1, \tpo_1, \trf_1, \cfr_1 \rangle$ and $\tau_2 = \langle \tran_2, \tpo_2, \trf_2, \cfr_2 \rangle$, we say $\tau_1, \tau_2$ are equivalent denoted $\tau_1 \equiv \tau_2$ if $\tran_1 = \tran_2$, $\tpo_1 = \tpo_2$, $\trf_1 = \trf_2$ and for all $t_1$,$t_2 \in \tran_1^{w,x}$ for all variables $x$, we have $t_1$ [$\tpo_1 \cup \trf_1 \cup \cfr_1^x$] $t_2$ iff $t_1$ [$\tpo_2 \cup \trf_2 \cup \cfr_2^x$] $t_2$.
	
\end{definition}

\begin{lemma}
$((\tau_1 \equiv \tau_2) \wedge \tau_1 \vdash_{\tt{G}} \alpha_t) \Rightarrow (\tau_2 \vdash_{\tt{G}} \alpha_t \wedge (\alpha_t(\tau_1) \equiv \alpha_t(\tau_2)))$.
\label{lemma rc-g-1}
\end{lemma}
\begin{proof}
Assume 	$\tau_1 \vdash_{\tt{G}} \alpha_t$ for an observable $\alpha_t$, and 
$\tau_1=(\tran_1, \tpo_1, \trf_1, \cfr_1)$, 
$\tau_2=(\tran_2, \tpo_2, \trf_2, \cfr_2)$ with $\tau_1 \equiv \tau_2$. Then we know that $\tran_1=\tran_2, \tpo_1=\tpo_2, \trf_2=\trf_1$. 
 
 Since $\tau_1 \vdash_{\tt{G}} \alpha_t$, let  $\tau_1 \xrightarrow[]{\alpha_t}_{\readc{-}\sat} \tau'$. 
 $\tau'=(\tran', \tpo', \trf', \cfr')$ where $\tran'=\tran_1 \cup \{t\}$, 
 $\tpo'=\tpo_1 \cup \{(t_1,t) \mid t_1 \in \tran_1$ is in the same process as $t\}$,  $\cfr_1 \subseteq \cfr'$ 
 and $\trf_1 \subseteq \trf'$. $\cfr'$ can contain $(t',t)$ for some 
 $t' \in \tran_1$, when $t',t \in {\tran'}^{w,x}$ for some variable $x$, based 
 on the trace semantics. 
 In particular, if we have 
 $t_1 {\trf'}^x t_2$ , $t_3 \in \vbl(\tau',t,x)$, then    
 $t_3 ~{\cfr'}^x~ t_1$.

 Since $\tau_1 \equiv \tau_2$, for any transactions $t', t'' $ in $\tran_1=\tran_2$, $t' [\tpo_1 \cup \cfr_1^x \cup \trf_1] t''$ iff $t' [\tpo_2 \cup \cfr_2^x \cup \trf_2] t''$ for all variables $x$. In particular, $t' \cfr_1^x t''$ iff $t' \cfr_2^x t''$ for all $x$. 
 Now, let us construct a trace $\tau''=(\tran', \tpo', \trf', \cfr'')$,
 where $\cfr''$ is the smallest set such that $\cfr_2 \subseteq \cfr''$ and 
  whenever $t_1 {\trf'}^x t_2$ , $t_3 \in \vbl(\tau',t,x)$, then $t_3 {\cfr''}^x t_1$.

   Since $t' [\tpo_1 \cup \trf_1 \cup \cfr_1^x] t''$ iff $t' [\tpo_2 \cup \trf_2 \cup \cfr_2^x] t''$ for all $x$, and all $t', t'' \in \tran_1=\tran_2$, and $\cfr''$ is the smallest extension of $\cfr_2$  based 
   on the trace semantics, along with the fact that 
   $\tpo'=\tpo'', \trf'=\trf''$, we obtain for any two 
   transactions $t_1, t_2 \in \tran'=\tran''$, 
   $t_1 [\tpo' \cup \trf' \cup \cfr'] t_2$ iff  $t_1 [\tpo'' \cup \trf''
   \cup \cfr''] t_2$. This gives $\tau' \equiv \tau''$. 
    
      This also gives $\tau_2 \xrightarrow[]{\alpha_t}_{\readc{-}\sat} \tau''$,
  that is, $\tau_2 \vdash_{\tt{G}} \alpha_t$  and indeed $\alpha_t(\tau_1)=\tau' \equiv \tau''=\alpha_t(\tau_2)$.

\end{proof}

\begin{lemma}
If $\tau \vdash_{\tt{G}} \alpha_{t_1} \cdot \alpha_{t_2}$ and $\alpha_{t_1} {\sim} \alpha_{t_2}$ then $\tau \vdash_{\tt{G}} \alpha_{t_2} {\cdot} \alpha_{t_1}$ and $(\alpha_{t_1} {\cdot} \alpha_{t_2})(\tau) \equiv (\alpha_{t_2}{ \cdot} \alpha_{t_1})(\tau)$.
\label{lemma rc-g-2}
\end{lemma}
\begin{proof}
Let $\tau = \langle \tran , \tpo, \trf,\cfr \rangle$ be a trace and let $t_1$ be a transaction issued in process $p_1$, and $t_2$ be a transaction issued in process $p_2$, with $p_1 \neq p_2$. Assume $\tau \vdash_{\tt{G}} \alpha_{t_1} \cdot \alpha_{t_2}$, with $\alpha_{t_1} {\sim} \alpha_{t_2}$.

We consider the following cases.
\begin{itemize}
    \item [$\bullet$] $t_1,t_2 \notin \tran^{rd}$. In this case $\tau \vdash_{\tt{G}} \alpha_{t_2} \cdot \alpha_{t_1}$ holds trivially.
    
     $(\alpha_{t_1} \cdot \alpha_{t_2})(\tau) = (\alpha_{t_2} \cdot \alpha_{t_1})(\tau) = \langle \tran' , \tpo' , \trf, \cfr \rangle$, where $\tran' = \tran \cup \{t_1,t_2\}$ and $\tpo' = \tpo \cup \{ (t,t_1) | t \in p_1\} \cup \{(t',t_2) | t' \in p_2 \}$.
    \item [$\bullet$] $t_1 \in \tran^{wt}$, $t_1 \notin \tran^{rd}$ and $t_2 \in \tran^{rd} \cap \tran^{wt}$. 
    Let $\tau_1 = \alpha_{t_1}(\tau)$. 
    We know that $\tau_1 = \langle \tran_1 , \tpo_1 , \trf_1,\cfr_1 \rangle$, where 
    $\tran_1 = \tran \cup \{ t_1 \}$, and $\tpo_1 = \tpo \cup \{ (t,t_1) | t \in p_1\}$, with $\trf_1 = \trf$ since there are no read events in $t_1$, and $\cfr_1 = \cfr$ since by the trace semantics, when $ \trf$ remain same, there are no $\cfr$ edges to be added.
    
    Consider $\tau_2 = \alpha_{t_2}(\tau_1) = \langle \tran_2, \tpo_2, \trf_2, \cfr_2 \rangle$. 
    Since $t_2 \in \tran^{rd}$, we know that $t_2$ has read events.
    \smallskip

    Let there be $n$ read events $ev_1 \dots ev_n$ in transaction $t_2$, 
    reading from transactions $t'_1, \dots, t'_n$ on variables $x_1, \dots, x_n$.  
    Hence we get a sequence $\tau_1 \xrightarrow[]{\beginact(t_2)}_{\readc{-}\sat} \dots \tau'_{i_1} \xrightarrow[]{\readact(t_2,t'_1)}_{\readc{-}\sat}  \tau_{i_1}  \xrightarrow[]{}_{\readc{-}\sat}\dots  
    \tau'_{i_n} \xrightarrow[]{\readact(t_2,t'_n)}_{\readc{-}\sat} \tau_{i_n} \dots \xrightarrow[]{\commitact(t_2)}_{\readc{-}\sat} \tau_2$.
    Hence $\tran_2 = \tran_1 \cup \{t_2\}$, $\tpo_2 = \tpo_1 \cup \{(t',t_2) | t' \in p_2\}$, $\trf_2 = \trf_1 \cup (\trf_{i_1} \cup \trf_{i_2} \dots \cup \trf_{i_n})$, and $\cfr_2 = \cfr_1 \cup (\cfr_{i_1} \cup \cfr_{i_2} \dots \cup \cfr_{i_n})$.
    Since $\tau_1 \vdash_{\tt{G}} \alpha_{t_2}$,  for all $ev_i : 1 \le i \le n$, we have $t'_i \in \rbl(\tau_1,t_2,x_i)$.
    
    Since $\alpha_{t_1} \sim \alpha_{t_2}$, we know that for all $ev_i : 1 \le i \le n$, we have $t'_i \in \rbl(\tau,t_2,x_i)$. It follows that $\tau \vdash_{\tt{G}} \alpha_{t_2}$.
    Define $\tau_3 = \langle \tran_3, \tpo_3, \trf_2, \cfr_2 \rangle$, where $
    \tran_3 = \tran \cup \{t_2\}$ and $\tpo_3 = \tpo \cup \{(t',t_2) | t' \in p_2\}$. It follows that $\alpha_{t_2}(\tau) = \tau_3$, $\tau_3 \vdash_{\tt{G}} \alpha_{t_1}$ and $\tau_2 = \alpha_{t_1}(\tau_3) = \alpha_{t_2}(\alpha_{t_1}(\tau))$.
    \item [$\bullet$] $t_2 \in \tran^{wt}$ and $t_1 \in \tran^{rd} \cap \tran^{wt}$. Similar to previous case.
    \item [$\bullet$] $t_1, t_2 \in \tran^{rd} \cap \tran^{wt}$. 
    Let $\tau_1 = \alpha_{t_1}(\tau) = \langle \tran_1, \tpo_1, \trf_1, \cfr_1 \rangle$ and $\tau_2 = \alpha_{t_2}(\tau_1)$. Since $t_1 \in \tran^{rd}$, we know that $t_1$ has read events. 
    Let there be $n$ read events $ev_1 \dots ev_n$ in transaction $t_1$, reading from transactions $t'_1, \dots, t'_n$ on variables 
    $x_1, \dots, x_n$. 
    Hence we get the sequence $\tau \xrightarrow[]{\beginact(t_1)}_{\readc{-}\sat} \dots \tau'_{i_1} \xrightarrow[]{\readact(t_1,t'_1)}_{\readc{-}\sat}  \tau_{i_1} \xrightarrow[]{}_{\readc{-}\sat}\dots  
    \tau'_{i_n} \xrightarrow[]{\readact(t_1,t'_n)}_{\readc{-}\sat} \tau_{i_n} \dots \xrightarrow[]{\commitact(t_1)}_{\readc{-}\sat} \tau_1$.
    Since $t_2 \in \tran^{rd}$, we know that $t_2$ has read events. 
    Let there be  $m$ read events $ev'_1 \dots ev'_m$ in transaction $t_2$, reading from transactions $t''_1, \dots, t''_m$ on variables 
    $y_1, \dots, y_m$.  
    Then we get the sequence $\tau_1 \xrightarrow[]{\beginact(t_2)}_{\readc{-}\sat} \dots \tau'_{j_1} \xrightarrow[]{\readact(t_2,t''_1)}_{\readc{-}\sat}  \tau_{j_1}  \xrightarrow[]{}_{\readc{-}\sat} \dots 
    \tau'_{j_n} \xrightarrow[]{\readact(t_2,t''_m)}_{\readc{-}\sat} \tau_{j_n} \dots \xrightarrow[]{\commitact(t_2)}_{\readc{-}\sat} \tau_2$. 
    
    $\mathbf{Proving}$ $\mathtt{ \tau \vdash_{\tt{G}} \alpha_{t_2} {\cdot} \alpha_{t_1}}$
    
    Since $\tau_1 = \alpha_{t_1}(\tau)$, for all read events $ev_i$ such that $ev_i.trans = t_1$, we have $t'_i \in \rbl(\tau,t_1,x_i)$.
    Since $\tau_2 = \alpha_{t_2}(\tau_1)$, for all read events $ev'_j$ such that $ev'_j.trans = t_2$, we have $t''_j \in \rbl(\tau_1,t_2,y_j)$. 
    Since $\alpha_{t_1} \sim \alpha_{t_2}$, 
    it follows that, we have $t''_j \in \rbl(\tau,t_2,y_j)$ for all read events $ev'_j$ such that $ev'_j.trans = t_2$. Hence $\tau \vdash_{\tt{G}} \alpha_{t_2}$. 
    \smallskip

    Consider $\tau_3 = \alpha_{t_2}(\tau) = \langle \tran_3, \tpo_3, \trf_3, \cfr_3 \rangle $. To show that $\tau \vdash_{\tt{G}} \alpha_{t_2} \cdot \alpha_{t_1}$, 
    we show that for all read events $ev_i$ such that $ev_i.trans = t_1$,  $t'_i \in \rbl(\tau_3,t_1,x_i)$. We prove this using contradiction.
     Assume that $\exists t'_i$ s.t. $t'_i \notin \rbl(\tau_3,t_1,x_i)$.

    \medskip 
    \smallskip

    If $t'_i \notin \rbl(\tau_3,t_1,x_i)$, then 
    there is a $t'_k$ such that (i) $t'_k \in \vbl(\tau_3,t_1,x_i)$ and (ii) $t'_i [\tpo_3 \cup \trf_3 \cup \cfr_3]^+ t'_k$.
     
     \medskip 
     \smallskip

       If   (i) is true, that is, $t'_k \in \vbl(\tau_3,t_1,x_i)$, since $\alpha_{t_1} \sim \alpha_{t_2}$, we also have $t'_k \in \vbl(\tau,t_1,x_i)$. 
        This, combined with $t'_i \in \rbl(\tau,t_1,x_i)$, gives according to the trace semantics, $t'_k [\cfr_1] t'_i$.  Going back to (ii), 
    we can have $t'_i [\tpo_3 \cup \trf_3 \cup \cfr_3]^+ t'_k$ only if we observed one of the following:
    \begin{itemize}
        \item [(a)] There is a path from $t'_i$ to $t'_k$ in $\tau$; that is, 
        $t'_i [\tpo \cup \trf \cup \cfr]^+ t'_k$. This implies that $t'_i \notin \rbl(\tau,t_1,x_i)$ which is a contradiction.
        \item [(b)] There is no direct path from $t'_i$ to $t'_k$, however, 
        there are  transactions  $t''_j,t''_l$ such that, in $\tau$ we had  
        $ t''_j [\tpo \cup \trf \cup \cfr]^* t'_k$ and $t'_i [\tpo \cup \trf \cup \cfr]^* t''_l$, along with $t''_j \in \rbl(\tau,t_2,y_j)$,  $t''_l \in \vbl(\tau,t_2,y_j)$. Let  $ev'_j = r(p_2,t_2,y_j,v)$ be the read event which reads from $t''_j$, justifying 
        $t''_j \in \rbl(\tau,t_2,y_j)$. 
        
        \smallskip
                Since $t''_j [\tpo \cup \trf \cup \cfr]^* t'_k$ and $t'_i [\tpo \cup \trf \cup \cfr]^* t''_l$ it follows that $ t''_j [\tpo_1 \cup \trf_1 \cup \cfr_1]^* t'_k$, $t'_i [\tpo_1 \cup \trf_1 \cup \cfr_1]^* t''_l$. By our assumption, we have $t'_i \in \rbl(\tau,t_1,x_i)$ and $t'_k \in \vbl(\tau,t_1,x_i)$, which  gives us  $t'_k [\cfr_1^{x_i}] t'_i$ (already observed in the para before (i)). 
                
                Thus we get  $t''_j [\tpo_1 \cup \trf_1 \cup \cfr_1]^*
                t'_k [\cfr_1^{x_i}] t'_i [\tpo_1 \cup \trf_1 \cup \cfr_1]^*
                 t''_l$, which gives  $t''_j [\tpo_1 \cup \trf_1 \cup \cfr_1]^+ t''_l$. Hence $t''_j \notin \rbl(\tau_1,t_2,y_j)$, a  contradiction.
    \end{itemize}
    Thus, we have $\neg(t'_i [\tpo_3 \cup \trf_3 \cup \cfr_3]^+ t'_k)$. This, for all $1 \leq i \leq n$, $t'_i \in \rbl(\tau_3,t_1,x_i)$. Thus, we now have $\tau \vdash_{\tt{G}} \alpha_{t_2} \cdot \alpha_{t_1}$. 
    It remains to show that $(\alpha_{t_1} {\cdot} \alpha_{t_2})(\tau) \equiv (\alpha_{t_2}{ \cdot} \alpha_{t_1})(\tau)$. 
    
    $\mathbf{Proving}$ $\mathtt{(\alpha_{t_1} {\cdot} \alpha_{t_2})(\tau) \equiv (\alpha_{t_2}{ \cdot} \alpha_{t_1})(\tau)}$
    
    Define $\tau_4 {=} \alpha_{t_1}(\tau_3) {=} \alpha_{t_1}.\alpha_{t_2}(\tau){=}
         \langle \tran_4, \tpo_4, \trf_4 , \cfr_4 \rangle$. Recall that 
    $\tau_2{=}$$\alpha_{t_2}.\alpha_{t_1}(\tau)$ =$
      \langle \tran_2,\tpo_2, \trf_2, \cfr_2 \rangle$.
    \smallskip

    To show that $(\alpha_{t_1} \cdot \alpha_{t_2})(\tau) \equiv (\alpha_{t_2} \cdot \alpha_{t_1})(\tau)$, we show that $\tran_2=\tran_4, \tpo_2=\tpo_4, \trf_2=\trf_4$, and, for any two transactions 
    $t'_k, t'_i \in \tran_2=\tran_4$,    $t'_k [\tpo_2 \cup \trf_2 \cup \cfr_2]^+ t'_i$ iff 
    $t'_k [\tpo_4 \cup \trf_4 \cup \cfr_4]^+ t'_i$. Of these, trivially, 
    $\tran_2=\tran_4, \tpo_2=\tpo_4$ follow. Since $\alpha_{t_2} \sim \alpha_{t_1}$, we have $\trf_2=\trf_4$. It remains to prove the last condition. 
    
    \medskip 
    
   Assume that $t'_i [\trf^{x_i}_1] t_1 \wedge t'_k \in \vbl(\tau_1,t_1,x_i)$, 
   where $\tau_1=\alpha_{t_1}(\tau)$. Then we have 
   $t'_i [\trf^{x_i}_1] t_1 \wedge t'_k \in \vbl(\tau,t_1,x_i)$, and by the trace semantics, we have $t'_k [\cfr_1] t'_i$, and hence 
   $t'_k [\cfr_2] t'_i$.  To obtain $\tau_4$, we execute $t_2$ first obtaining $\tau_3$ from $\tau$, and then $t_1$. We show that $t'_k [\tpo_4 \cup \trf_4 \cup \cfr_4]^+ t'_i$, proving $t'_k [\tpo_2 \cup \trf_2 \cup \cfr_2]^+ t'_i \Rightarrow
    t'_k [\tpo_4 \cup \trf_4 \cup \cfr_4]^+ t'_i$.

    \smallskip 
    
     If we have $t'_k \in \vbl(\tau_3,t_1, x_i)$, then we are done since we  have $t'_i \in \rbl(\tau,t_1,x_i)$, hence 
          $t'_i \in \rbl(\tau_3,t_1,x_i)$, thereby obtaining 
          $t'_k \cfr_3 t'_i$, and hence $t'_k \cfr_4 t'_i$.    
          
          \smallskip 
          
           Assume $t'_k \notin \vbl(\tau_3,t_1, x_i)$.   
   Assume that on executing $t_2$ from $\tau$, we have $t''_l \in \vbl(\tau, t_2, y_j)$, and let $t''_j [\trf^{y_j}_3] t_2$.      
     Then by the trace semantics,   $t''_l [\cfr_3] t''_j$ is a new edge which gets added.  
         When we execute $t_1$ next, assume  $t''_m \in \vbl(\tau_3,t_1,x_i)$.   
         Since we have $t'_i \in \rbl(\tau_3,t_1,x_i)$, we obtain $t''_m \cfr_4 t'_i$.

     \medskip 
     By assumption, $t'_k \neq t''_m$. However, $t'_k \in \vbl(\tau,t_1,x_i)$. \\
     $t'_k \notin  \vbl(\tau_3,t_1,x_i)$ points to some transaction that 
     blocked the visibility of $t'_k$ in $\tau_3$, by happening after $t'_k$, and which is in $\vbl(\tau_3,t_1,x_i)$.  
     
     \begin{enumerate}
     	\item This blocking transaction could be $t''_m$. If this is the case, we have $t'_k [\tpo_3 \cup \trf_3 \cup \cfr_3]^+ t''_m [\cfr_4] t'_i$, and hence  $t'_k [\tpo_4 \cup \trf_4 \cup \cfr_4]^+ t''_m [\cfr_4] t'_i$. 
     	\item The other possibility is that we have  $t''_j$
     	happens before $t''_m$, and $t''_l$ happens after $t'_k$ blocking 
     	$t'_k$ from being in $\vbl(\tau_3,t_1,x_i)$. That is,
     	     $t'_k [\tpo \cup \trf \cup \cfr]^* t''_l$ and 
     	     $t''_j [\tpo_4 \cup \trf_4 \cup \cfr_4]^* t''_m$.   
         Then we obtain $t'_k [\tpo_4 \cup \trf_4 \cup \cfr_4]^* t''_l [\cfr^{x_j}_4] t''_j [\tpo_4 \cup \trf_4 \cup \cfr_4]^* t''_m [\cfr_4] t'_i$. 
     
     \end{enumerate}
        
     Thus, we obtain $t'_k [\tpo_4 \cup \trf_4 \cup \cfr_4]^+ t'_i$ as desired.

\end{itemize}
 The converse direction, that is, whenever $t'_k [\tpo_4 \cup \trf_4 \cup \cfr_4]^+ t'_i$, we also have $t'_k [\tpo_2 \cup \trf_2 \cup \cfr_2]^+ t'_i$ is proved on similar lines. 
    
\end{proof}
\begin{lemma}
If $\tau_1 \equiv \tau_2$, $\alpha_{t_1} \sim \alpha_{t_2}$, and $\tau_1 \vdash_{\tt{G}} (\alpha_{t_1} \cdot \alpha_{t_2})$,  then (i) $\tau_2 \vdash_{\tt{G}} (\alpha_{t_2} \cdot \alpha_{t_1})$ and (ii) $(\alpha_{t_1} \cdot \alpha_{t_2})(\tau_1) \equiv (\alpha_{t_2} \cdot \alpha_{t_1})(\tau_2)$.
\label{lemma rc-g-3}
\end{lemma}
\begin{proof}
Follows from Lemma \ref{lemma rc-g-1} and Lemma \ref{lemma rc-g-2}.
\end{proof}

From Lemma \ref{lemma rc-g-3}, we get following lemma.
\begin{lemma}
If $\tau_1 \equiv \tau_2$, $\pi_{\tt{T^1}} \sim \pi_{\tt{T^2}}$, and $\tau_1 \vdash_{\tt{G}} \pi_{\tt{T^1}}$ then $\tau_2 \vdash_{\tt{G}} \pi_{\tt{T^2}}$ and $\pi_{\tt{T^1}}(\tau_1) \equiv \pi_{\tt{T^2}}(\tau_2)$.
\label{lemma rc-g-4}
\end{lemma}
 
\begin{lemma}
If $\tau \vdash_{\tt{G_T}} \pi_{\tt{T}}$, $\tau \vdash_{\tt{G}} \alpha_t$ then $\pi_{\tt{T}} = \pi_{\tt{T^1}} \cdot \alpha_t \cdot \pi_{\tt{T^2}}$ for some $\pi_{\tt{T^1}}$ and $\pi_{\tt{T^2}}$ where $\pi_{\tt{T^1}}$ is $t$-free.
\label{lemma rc-g-5}
\end{lemma}
Consider observables $\pi_{\tt{T}} = \alpha_{t_1},\alpha_{t_2} \dots \alpha_{t_n}$. 
We write $\pi_{\tt{T}} \lessapprox \pi'_{\tt{T}}$ to represent that $\pi'_{\tt{T}} = \pi_{\tt{T^0}} \cdot \alpha_{t_1} \cdot \pi_{\tt{T^1}} 
\cdot \alpha_{t_2} \cdot \pi_{\tt{T^2}} \dots \alpha_{t_n} \cdot \pi_{\tt{T^n}}$. 
In other words, $\pi_{\tt{T}}$ occurs as a non-contiguous subsequence in $\pi'_{\tt{T}}$. For such a $\pi_{\tt{T}}, \pi'_{\tt{T}}$, 
we define $\pi'_{\tt{T}} \oslash \pi_{\tt{T}} := \pi_{\tt{T^0}} \cdot \pi_{\tt{T^1}} \dots \pi_{\tt{T^n}}$. 
Since elements of $\pi_{\tt{T}}$ and $\pi'_{\tt{T}}$ are distinct, operation $\oslash$ is well defined. 
Let $\pi_{\tt{T}}[i]$ denote the $i$th observable 
in the observation sequence $\pi_{\tt{T}}$, and 
let  $|\pi_{\tt{T}}|$ denote  the number of observables 
in $\pi_{\tt{T}}$.

Let $\alpha_t = \pi_{\tt{T}}[i]$ for some $i: 1 \le i \le |\pi_{\tt{T}}|$. 
We define $\tt{Pre(\pi_{\tt{T}},\alpha_t)}$ as a subsequence $\pi'_{\tt{T}}$ of $\pi_{\tt{T}}$ such that  
(i) $\alpha_t \in \pi'_{\tt{T}}$,
(ii) $\alpha_{t_j} = \pi_{\tt{T}}[j] \in \pi'_{\tt{T}}$ for some $j: 1 \le j < i$ iff there exists $k: j < k \le i$ such that $\alpha_{t_k}=\pi_{\tt{T}}[k]  \in \pi'_{\tt{T}}$ and $\neg(\alpha_{t_j} \sim \alpha_{t_k})$. 
Thus, $\tt{Pre(\pi_{\tt{T}},\alpha_t)}$ consists of $\alpha_t$ and all 
$\alpha_{t'}$ appearing before $\alpha_t$ in $\pi_{\tt{T}}$ such that 
$t' [\tpo \cup \trf]^+ t$.

%
%
\begin{lemma}
If $\pi_{\tt{T^1}} = \tt{Pre(\pi_{\tt{T}},\alpha_t)}$ and $\pi_{\tt{T^2}} = \pi_{\tt{T}} \oslash \pi_{\tt{T^1}}$ then $\pi_{\tt{T}} \approx \pi_{\tt{T^1}} \cdot \pi_{\tt{T^2}}$.
\label{lemma rc-g-6}
\end{lemma}
\begin{proof}
The proof is trivial since we can always execute in order, 
the $[\tpo \cup \trf]^+$ predecessors of $\alpha_t$ from $\tt{Pre(\pi_{\tt{T}},\alpha_t)}$, then $\alpha_t$, then the  observables 
in $\tt{Pre(\pi_{\tt{T}},\alpha_t)}$ which are independent from $\alpha_t$, followed by the suffix of $\pi_{\tt{T}}$ after $\alpha_t$.

\end{proof}

\begin{lemma}
If $\pi_{\tt{T}}= \pi'_{\tt{T}} \cdot \pi''_T$ and $\alpha_t \in \pi'_{\tt{T}}$ then $\tt{Pre(\pi_{\tt{T}},\alpha_t)} = \tt{Pre(\pi'_{\tt{T}},\alpha_t)}$.
\label{lemma rc-g-7}
\end{lemma}
\begin{proof}
The proof is trivial once again, since 	all 
the transactions which are $[\tpo \cup \trf]^+ t$  are in the prefix  $\pi'_{\tt{T}}$.
\end{proof}

\begin{lemma}
If $\tau \vdash_{\tt{G}} \pi_{\tt{T}}$ then $\tau \vdash_{\tt{G_T}} \pi_{\tt{T}} \cdot \pi_{\tt{T^2}}$.
\label{lemma rc-g-8}
\end{lemma}
\begin{proof}
This simply follows from the fact that we can extend the observation sequence 
$\pi_{\tt{T}}$ to obtain a terminal configuration since the $\readc$ DPOR algorithm generates traces corresponding to terminating runs.   	
\end{proof}

\begin{lemma}
Consider a trace trace $\tau$ such that $\tau \models \readc$, $\tau \vdash_{\tt{G}} \pi_{\tt{T}} \cdot \alpha_t$. Let each read event $ev_i = r_i(p,t,x_i,v)$ in $\alpha_t$ read from some transaction $t_i$, and let $\pi_{\tt{T}}$ be $t$-free.  Then $\tau \vdash_{\tt{G}} \alpha'_t \cdot \pi_{\tt{T}}$, where $\alpha'_t$ is the same as  $\alpha_t$, with the exception that the sources of its read events can be different. That is, 
 each read event $ev_i = r_i(p,t,x_i,v) \in \alpha'_t$ can read from some transaction $t'_i \neq t_i$.
\label{lemma rc-g-9}
\end{lemma}
\begin{proof}
Let $\pi_{\tt{T}} = \alpha_{t_1} \alpha_{t_2} \dots \alpha_{t_n}$ and $\tau_0 \xrightarrow[]{\alpha_{t_1}}_{\readc{-}\sat} \tau_1 \xrightarrow[]{\alpha_{t_2}}_{\readc{-}\sat} \dots \xrightarrow[]{\alpha_{t_n}}_{\readc{-}\sat} \tau_n \xrightarrow[]{\alpha_t}_{\readc{-}\sat} \tau_{n+1}$, where $\tau_0 = \tau$. Let $\tau_i = \langle \tran_i, \tpo_i, \trf_i, \cfr_i \rangle$. 
Let there be  $m$ read events in transaction $t$. Keeping in mind what we want to prove, where we want to execute $t$ first followed by $\pi_{\tt{T}}$, and obtain an execution $\alpha'_t \pi_{\tt{T}}$, we do the following.

For each read event $ev_i = r_i(p,t,x_i,v)$ we define $t'_i \in \tran^{wt,x_i}$ such that 
$(i)$ $ t'_i \in \rbl(\tau_0,t,x_i)$, and 
$(ii)$ there is no $t''_i \in \rbl(\tau_0,t,x_i)$ where $t'_i [\tpo_n \cup \trf_n \cup \cfr_n] t''_i$.  Note that this is possible since 
 $\pi_{\tt{T}}$ is $t$-free, so all the observables $\alpha_{t_i}$ occurring in $\pi_{\tt{T}}$ are such that $t_i$ is issued in a process other than that of $t$. Thus, when $\alpha_t$ is enabled, we can choose any of the 
 writes done earlier,  this fact is consistent with the $\readc$ semantics, since from Lemmas \ref{lem:f3}-\ref{lem:f4} we know $\tau_n \models \readc$.
Since $\tau_n \models \readc$ such a $t'_i \in  \rbl(\tau_0,t,x_i)$ exists for each $ev_i$.

\medskip 
Next we define a sequence of traces which can give the execution 
$\alpha'_t \pi_{\tt{T}}$. 
Define a sequence of traces $\tau'_0, \tau'_1, \dots \tau'_n$ where $\tau'_j = \langle \tran'_j, \tpo'_j, \trf'_j, \cfr'_j \rangle$ is such that 
\begin{enumerate}
	\item $\tau \vdash_{\tt{G}} {\alpha'_t}$ and $\alpha'_t(\tau)=\tau'_0$,
	\item $\tau'_j \vdash_{\tt{G}} \alpha_{t_{j+1}}$ and  $\tau'_{j+1} = \alpha_{t_{j+1}}(\tau'_j)$ for all $j: 0 \le j \le n$.
\end{enumerate}

\begin{itemize}
\item  Define $\tran'_i = \tran_i \cup \{t\}$ for all $0 \leq i \leq n$, 
\item 
Define $\tpo' = \{t' [\tpo] t \mid t'$ is a transaction in $\tau$, in the same process as $t\}$, and $\tpo'_i = \tpo_i \cup \tpo'$ for all $0 \leq i \leq n$, 

\item  $\trf' = \bigcup_{i:1\le i \le m} t'_i [\trf^{x_i}] t$ (reads from relation corresponding to each read event $ev_i \in \alpha'_t$), and 
 $\trf'_i = \trf_i \cup \trf'$ for all $i : 1 \le i \le n$, 
 
\item  $\cfr' = \bigcup_{j:1\le j \le m} \cfr'_j$ (each  $\cfr'_j$ corresponds to the updated $\cfr$ relation because of the read transition $\xrightarrow[]{\readact(t,t'_j)}_{\readc{-}\sat}$). 	
\smallskip 

For $1 \leq i \leq n$, we define $\cfr'_i$ inductively as $\cfr'_i = \cfr_i \cup \cfr'$ and show that $\tau'_i \vdash_{\tt{G}} \alpha_{t_{i+1}}$ and  $\tau'_{i+1} = \alpha_{t_{i+1}}(\tau'_i)$ holds good. 
 First define $\cfr'_0 = \cfr_0 \cup \cfr'$. Then define  $\cfr'_i$ to be the coherence order corresponding to $\alpha_{t_{i}}(\tau'_{i-1})$  where $\tau'_{i-1}=\langle \tran'_{{i-1}}, \tpo_{i-1}, \trf_{i-1}, \cfr_{i-1} \rangle$ for all $i: 1 \le i \le n$. 

 \end{itemize}

\medskip 
\noindent{\bf{Base case}}. The base case $\tau \vdash_{\tt{G}} {\alpha'_t}$ and $\alpha'_t(\tau)=\tau'_0$,
   holds trivially, by construction. 

\smallskip 

\noindent{\bf{Inductive hypothesis}}. Assume that $\tau'_j \vdash_{\tt{G}} \alpha_{t_{j+1}}$ and  $\tau'_{j+1} = \alpha_{t_{j+1}}(\tau'_j)$ for $0 \leq j \leq i-1$. 

\smallskip 

 We have to prove that 
$\tau'_i \vdash_{\tt{G}} \alpha_{t_{i+1}}$ and  $\tau'_{i+1} = \alpha_{t_{i+1}}(\tau'_i)$. 
\begin{enumerate}
	\item If $t_{i+1} \in \tran^{wt}$ and $t_{i+1} \notin \tran^{rd}$,  then the proof holds trivially from the inductive hypothesis. 
	\item Consider now $t_{i+1} \in \tran^{rd}$. Consider a read event 
	$ev^{i+1} = r(p,t_{i+1},x,v) \in \alpha_{t_{i+1}}$ which was reading from 
	transaction $t_x$ in $\pi_{\tt{T}}$. That is, we had $t_x \in \rbl(\tau_i, t_{i+1},x)$. 	If $t_x \in \rbl(\tau'_i, t_{i+1},x)$, then we are done, since we can simply extend the run from the inductive hypothesis. 
	
	\smallskip 
	
	Assume otherwise. That is, $t_x \notin \rbl(\tau'_i, t_{i+1},x)$.

	Since  $t_x$  $\in$   $\rbl(\tau_i, t_{i+1},x)$, we know that there are some blocking transactions in the new path which prevents $t_x$ from being readable. Basically, we have the blocking transactions since $t$ is moved before $\pi_{\tt{T}}$. 	
	
	\begin{itemize}
		\item Consider $tr_1 \in \vbl(\tau,t,y)$ for some variable $y$ such that,  in $\pi_{\tt{T}}.\alpha_t$, we had $t_x [\tpo_i \cup \trf_i \cup \cfr_i]^* tr_1$. This path
		 is possible since $\alpha_t$ comes last in in $\pi_{\tt{T}}.\alpha_t$,  after  $\alpha_{t_{i+1}}$. 
		\item  Consider $tr_2 \in \vbl(\tau_i,t_{i+1},x)$.  Since $t_x$ $\in$ $\rbl(\tau_i, t_{i+1},x)$, we have  $tr_2 [\cfr^x_{i+1}] t_x$. 
Now, in $\alpha'_t.\pi_{\tt{T}}$, assume $tr_3 [\trf'^y_0] t$, 
such that we have a path from $tr_3$ to $t_x$, as  $tr_3 [\tpo_i \cup \trf_i \cup \cfr_i]^* tr_2$. 
	\end{itemize}
	  		  Hence it follows that $tr_3 [\tpo_i \cup \trf_i \cup \cfr_i]^* tr_2$ 
$[\cfr^x_{i+1}] t_x$ $[\tpo_i \cup \trf_i \cup \cfr_i]^* tr_1$. 
Hence we have $tr_3 [\tpo_{i+1} \cup \trf_{i+1} \cup \cfr_{i+1}]^+ tr_1$, i.e 
$tr_3 [\tpo_n \cup \trf_n \cup \cfr_n]^+ tr_1$. This leads to the contradiction since $tr_1 \in \vbl(\tau,t,y)$ and  $tr_3 [\trf'^y_0] t$.  

Hence, $t_x$  $\in$   $\rbl(\tau_i, t_{i+1},x)$ and we can extend the run 
from the inductive hypothesis obtaining $\tau'_j \vdash_{\tt{G}} \alpha_{t_{j+1}}$,  $\tau'_{j+1} = \alpha_{t_{j+1}}(\tau'_j)$,  
for $0 \leq j \leq i-1$,
and 
$\tau'_i \vdash_{\tt{G}} \alpha_{t_{i+1}}$.

\end{enumerate}

Now we prove that $\cfr'_{i+1} = \cfr_{i+1} \cup \cfr'$. 
Assume we have $(tr1 , tr2) \in \cfr_{i+1}$. We show that $(tr1 , tr2) \in \cfr'_{i+1}$. We have the following cases:

\begin{itemize}
    \item[(i)] $(tr1 , tr2) \in \cfr_i$. From the inductive hypothesis it follows that $(tr1 , tr2) \in \cfr'_{i}$ and hence in $\cfr'_{i+1}$.
    \item[(ii)] $(tr1 , tr2) \notin \cfr_i$ and some read event $ev = r(p,t_{i+1},x,v)$ from $\alpha_{t_{i+1}}$ reads from $tr2$.  That is, $tr1,tr2 \in \rbl(\tau_i, t_{i+1},x)$ and $tr1 \in \vbl(\tau_i,t_{i+1},x)$.
    Assume $(tr1 , tr2) \notin \cfr'_{i+1}$.
    Since $(tr1, tr2) \notin \cfr'_{i+1}$, it follows that there exists  $tr3 \in \vbl(\tau'_i,t_{i+1},x)$ such that $tr2 [\tpo'_i \cup \trf'_i \cup \cfr'_i] tr3$. This means $tr2 \in \rbl(\tau_i,t_{i+1},x)$, but 
    $tr2 \notin \rbl(\tau'_i,t_{i+1},x)$. However,  as proved earlier, this leads to a contradiction. 

\end{itemize}
Thus, we have shown that $\tau \vdash_{\tt{G}} \alpha'_t \alpha_{t_1} \dots \alpha_{t_i}$ 
is such that $\alpha'_t \alpha_{t_1} \dots \alpha_{t_i}(\tau)=\tau'_i$ for all 
$0 \leq i \leq n$. When $i=n$ we obtain 
$\alpha'_t \pi_{\tt{T}}(\tau)=\tau'_n$, or $\tau \vdash_{\tt{G}} \alpha'_t \pi_{\tt{T}}$. 
\end{proof}


\begin{definition}[Linearization of  a Trace]
A observation sequence $\pi_{\tt{T}}$ is a \emph{linearization} of a trace $\tau = \langle \tran, \tpo , \trf, \cfr \rangle$ if $(i)$ $\pi_{\tt{T}}$ has the same transactions as $\tau$ and $(ii)$ $\pi_{\tt{T}}$ follows the ($\tpo \cup \trf)$ relation.
	
\end{definition}

We say that our $\readc$ DPOR algorithm \emph{generates an observation sequence} $\pi_{\tt{T}}$ from state $\tau$ where $\tau$ is a trace 
if 
 it invokes $\explore$ with parameters $\readc,\tau, \pi'$, where $\pi'$ is a \emph{linearization} of $\tau$, and generates  a sequence of recursive calls to $\explore$ resulting in $\pi_{\tt{T}}$.

\begin{lemma}
If $\tau \vdash_{\tt{G_T}} \pi_{\tt{T}}$ then, 
the DPOR algorithm  generates $\pi'_{\tt{T}}$ from state $\tau$ 
for some $\pi'_{\tt{T}} \approx \pi_{\tt{T}}$.
\label{lemma rc-g-10}
\end{lemma}
\begin{proof}
We use induction on $|\pi_{\tt{T}}|$. If $\tau$ is $terminal$, then the  proof is trivial.
Assume that we have $\tau \vdash_{\tt{G_T}} \pi_{\tt{T}}$.
Assume that $\tau \vdash_{\tt{G}} \alpha_t$. It follows that $\tau \vdash_{\tt{G}} \alpha_t$.
Using Lemma \ref{lemma rc-g-5} we get $\pi_{\tt{T}} = \pi_{\tt{T^1}} \cdot \alpha_t \cdot \pi_{\tt{T^2}}$, where $\pi_{\tt{T^1}}$ is $t$-free. 
We consider the following two cases:
\begin{itemize}
    \item In the first case, we assume that all read events in $t$ read from 
    transactions in $\tau$. So in this case, $\pi_{\tt{T^1}}$ can be empty. 
        Assume there are $n$ read events in $t$ and each read event $ev_i = r(p,t,x_i,v)$ 
    reads from transactions $t_i \in \tau$ for all $i: 1 \le i \le n$. 
In this case the DPOR algorithm will let each read event $ev_i$ read from all possible transactions $t' \in \rbl(\tau,t,x_i)$, including $t_i$.
    \item In the second case, assume that there exists at least one read event $ev' = r(p,t,x,v) \in \alpha_t$ which reads from a transaction $t' \in \pi_{\tt{T^1}}$ (hence, $\pi_{\tt{T^1}}$ is non empty).
    \smallskip 
     
    From Lemma \ref{lemma rc-g-9}, we know that $\tau
    \vdash_{\tt{G}} \alpha'_t \cdot \pi_{\tt{T^1}}$. 
    From Lemma \ref{lemma rc-g-8}, it follows that $\tau
    \vdash_{\tt{G_T}} \alpha'_t \cdot \pi_{\tt{T^1}} \cdot \pi_{\tt{T^4}}$, for some
    $\pi_{\tt{T^4}}$.
    Hence there exists $\tau'$ such that 
    $\tau' = \alpha'_t (\tau)$ and $\tau' \vdash_{\tt{G_T}} \pi_{\tt{T^1}} \cdot
    \pi_{\tt{T^4}}$. 
    
    Since $|\pi_{\tt{T^1}} \cdot
    \pi_{\tt{T^4}}| < |\pi_{\tt{T}}$, we can use the inductive hypothesis. It follows that the DPOR algorithm generates from state $\tau'$, 
    the observation sequence  $\pi_{\tt{T^5}}$ such that $\pi_{\tt{T^5}} \approx \pi_{\tt{T^1}} 
    \cdot \pi_{\tt{T^4}}$. 
    
    \smallskip 
    
    Let  $\pi_{\tt{T^3}}.\alpha_t = \tt{Pre(\pi_{\tt{T^1}}, \alpha_t)}$. Then  $\pi_{\tt{T^5}}
    \approx \pi_{\tt{T^1}} \cdot \pi_{\tt{T^4}}$ implies that  $\pi_{\tt{T^3}}
    \lessapprox \pi_{\tt{T^5}}$ (by applying Lemma \ref{lemma rc-g-7}).
    Let $\pi_{\tt{T^6}} \approx \pi_{\tt{T}} \oslash (\pi_{\tt{T^3}} \cdot \alpha_t)$.
    By applying Lemma \ref{lemma rc-g-6}, we get $\pi_{\tt{T}} \approx \pi_{\tt{T^3}} \cdot \alpha_t \cdot \pi_{\tt{T^6}}$.
    Since $\tau \vdash_{\tt{G_T}} \pi_{\tt{T}}$ and 
    $\pi_{\tt{T}} \approx \pi_{\tt{T^3}} \cdot \alpha_t \cdot \pi_{\tt{T^6}}$,
    by applying Lemma \ref{lemma rc-g-3}, we get $\tau
    \vdash_{\tt{G_T}} \pi_{\tt{T^3}} \cdot \alpha_t \cdot \pi_{\tt{T^6}}$.
    Let $\tau' = (\pi_{\tt{T^3}} \cdot \alpha_t) (\tau)$. 
    Since $\tau \vdash_{\tt{G_T}} \pi_{\tt{T^3}} 
    \cdot \alpha_t \cdot \pi_{\tt{T^6}}$, we have $\tau'
    \vdash_{\tt{G_T}} \pi_{\tt{T^6}}$. 
    From inductive hypothesis it follows that $\tau'$ generates $\pi_{\tt{T^7}}$ such that $\pi_{\tt{T^7}} \approx \pi_{\tt{T^6}}$.\\
    In other words, $\tau$ generates $\pi_{\tt{T^3}} 
    \cdot \alpha_t \cdot \pi_{\tt{T^7}}$ where $\pi_{\tt{T}} \approx \pi_{\tt{T^3}} 
    \cdot \alpha_t \cdot \pi_{\tt{T^7}}$.
\end{itemize}

\end{proof}

%% file: appendix-expts.tex
\section{More Details for Experimental Evaluation}
Here we describe details about the performance on extra versions of classical benchmarks.  

\subsection{Execution times}
\label{app:ext-tables}
Table \ref{param-ver2} and \ref{param-ver3} gives execution time for Version-2 and Version-3 of the classical benchmarks. 
\footnotesize{
\begin{table}[t]
    \centering
    \caption{Classical Benchmarks version 2 execution time(in seconds) }
    \label{param-ver2}
\centering
\begin{tabular}{@{}l c c c c c c c c c c@{}}
\hline\hline
\multicolumn{1}{c}{} & \multicolumn{2}{c}{$\cc$} & \multicolumn{2}{c}{$\ccvt$} & \multicolumn{2}{c}{$\cm$} & \multicolumn{2}{c}{$\readat$} & \multicolumn{2}{c}{$\readc$} \\
\multicolumn{1}{c}{Program} &  \texttt{Traces} &  \texttt{Time} &  \texttt{Traces} &  \texttt{Time} &  \texttt{Traces} &  \texttt{Time} &  \texttt{Traces} & \texttt{Time} &  \texttt{Traces} &  \texttt{Time} \\
\hline
 Causality Violation 	 &560	&0.08	&469	&0.08	&540	&0.14	&469	&0.09	&4341	&0.65\\
Causal Violation 	&10020	&0.84	&9330	&1.01	&10020	&2.78	&12540	&1.57	&47730	&7.01\\
Delivery Order	&99	&0.05	&99	&0.06	&92	&0.07	&147	&0.05	&496	&0.08\\
Long Fork	&52608	&4.72	&49968	&6.16	&52608	&25.07	&49968	&10.11	&108336	&21.12\\
Lost Update	&77	&0.05	&72	&0.06	&71	&0.06	&94	&0.05	&326	&0.07\\
Message Passing	&72	&0.05	&68	&0.06	&72	&0.06	&88	&0.05	&580	&0.12\\
Modification Order	&8072	&0.65	&6823	&0.75	&6823	&3.14	&7823	&1.6	&13120	&2.51\\
Conflict violation	&28764	&2.44	&18162	&2.08	&22347	&6.47	&25451	&3.66	&551781	&84.08\\
Read Atomicity	&1692	&0.15	&1624	&0.2	&1692	&0.62	&1824	&0.32	&3994	&0.9\\
Read Committed	&1576	&0.2	&1377	&0.23	&1452	&1.09	&4399	&1.29	&20428	&7.34\\
Repeated Read	&2042	&0.23	&1439	&0.22	&1509	&0.6	&2674	&0.45	&57724	&11.55\\
Load Buffer	&1551	&0.16	&1230	&0.17	&1410	&0.42	&1752	&0.24	&37701	&7.16\\
Store Buffer	&90	&0.05	&79	&0.06	&61	&0.06	&273	&0.06	&4905	&0.41\\
Writeskew	&285	&0.07	&285	&0.08	&285	&0.18	&2064	&0.43	&75744	&14.81\\
\hline
\end{tabular}
\end{table}
}

\footnotesize{
\begin{table}[t]
    \centering
    \caption{Classical Benchmarks version 3 execution time(in seconds) }
    \label{param-ver3}
\centering
\begin{tabular}{@{}l c c c c c c c c c c@{}}
\hline\hline
\multicolumn{1}{c}{} & \multicolumn{2}{c}{$\cc$} & \multicolumn{2}{c}{$\ccvt$} & \multicolumn{2}{c}{$\cm$} & \multicolumn{2}{c}{$\readat$} & \multicolumn{2}{c}{$\readc$} \\
\multicolumn{1}{c}{Program} &  \texttt{Traces} &  \texttt{Time} &  \texttt{Traces} &  \texttt{Time} &  \texttt{Traces} &  \texttt{Time} &  \texttt{Traces} & \texttt{Time} &  \texttt{Traces} &  \texttt{Time} \\
\hline
 Causality Violation 	&	24874	&	3.02	&	21653	&	3.11	&	23040	&	9.38	&	25317	&	9.19	&	117078	&	31.07	\\
Causal Violation 	&	215580	&	33.25	&	215580	&	33.51	&	215580	&	124.75	&	341516	&	153.67	&	540758	&	164.5	\\
Delivery Order	&	14954	&	2.1	&	14954	&	2.18	&	13245	&	7.73	&	18194	&	8.17	&	128454	&	40.67	\\
Long Fork	&	614656	&	204.84	&	614656	&	206.52	&	614656	&	562.37	&	614656	&	416.87	&	614656	&	295.99	\\
Lost Update	&	66782	&	8.9	&	63650	&	8.73	&	65087	&	26.51	&	119666	&	30.54	&	238664	&	47.08	\\
Message Passing	&	78708	&	12.95	&	74940	&	12.81	&	78708	&	58.48	&	143136	&	76.03	&	688440	&	290.23	\\
Modification Order	&	82908	&	10.9	&	69090	&	9.45	&	69090	&	28.56	&	69090	&	18.67	&	164787	&	34.23	\\
Conflict violation	&	457152	&	77.5	&	452104	&	76.69	&	457152	&	307.16	&	616239	&	294.52	&	859822	&	326.73	\\
Read Atomicity	&	78708	&	10.94	&	74940	&	10.87	&	78708	&	47.67	&	97740	&	37.74	&	362835	&	110.21	\\
Read Committed	&	10981	&	1.59	&	10927	&	1.57	&	10942	&	4.63	&	19207	&	5.67	&	20617	&	4.79	\\
Repeated Read	&	27238	&	4	&	24107	&	3	&	25960	&	7	&	54482	&	9	&	170864	&	26	\\
Load Buffer	&	7947	&	0.92	&	7947	&	0.95	&	7947	&	3.41	&	9693	&	2.8	&	13572	&	2.94	\\
Store Buffer	&	32096	&	4.63	&	22288	&	3.33	&	20720	&	10.57	&	46180	&	15.21	&	742416	&	188.29	\\
Writeskew	&	325260	&	47.03	&	121451	&	18.43	&	154562	&	71.23	&	935710	&	312.53	&	7726230	&	2357.01	\\
\hline
\end{tabular}
\end{table}
}

\newpage 

\subsection{Programs for Classical Benchmarks}
\label{app:tab3-p}
Fig.\ref{fig:programs} gives the programs Load Buffer, Store Buffer, Modification Order.
We have assertion conditions on the values given in comments in all of these programs.

\begin{figure}[ht]
    \centering
        \subfigure[Load Buffer]{
          \epsfig{figure=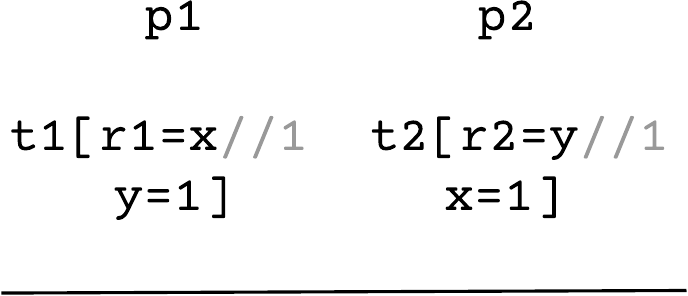,width=2in}
          \label{fig:load buffer}
          }
    \subfigure[Store Buffer]{
         \epsfig{figure=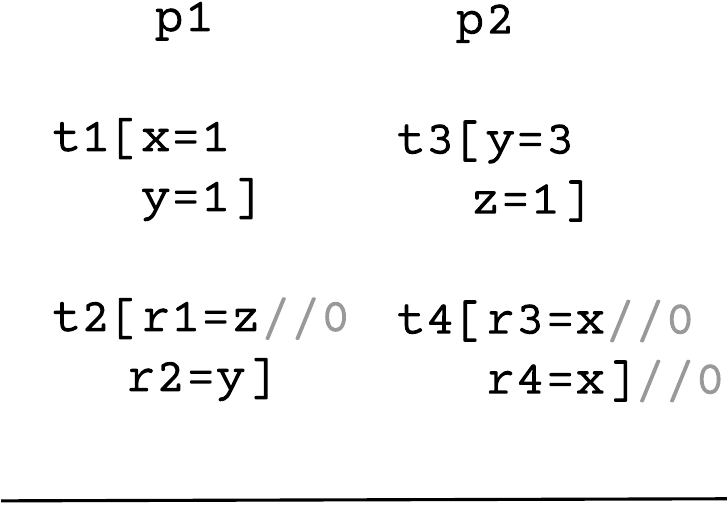,width=2in}
          \label{fig:store buffer}
          }
    \vspace{5pt}
    \subfigure[Modification Order]{
          \epsfig{figure=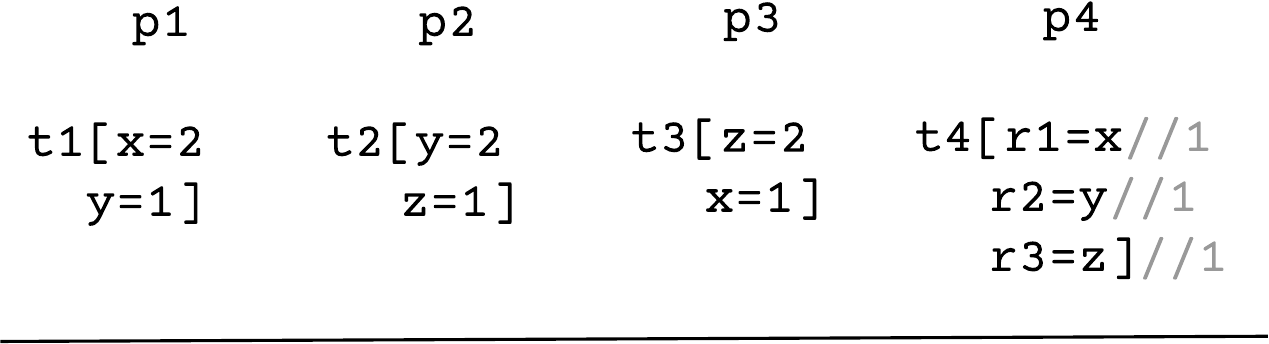,height=1in}
          \label{fig:modification order}
          }
    \caption{Our own programs used in Classical Benchmarks table in main paper}
    \label{fig:programs}
\end{figure}